\documentclass[10pt]{amsart}
\usepackage{amsthm}
\usepackage{amsmath}
\usepackage{amssymb}
\usepackage{booktabs}
\usepackage{comment}
\usepackage{caption}
\usepackage{dsfont}
\usepackage{graphicx}
\usepackage{marginnote}
\usepackage{natbib}
\usepackage{setspace}
\usepackage[dvipsnames]{xcolor}
\usepackage[colorlinks,citecolor=OliveGreen,linkcolor=blue,hypertexnames=false]{hyperref}

\numberwithin{table}{section}
\numberwithin{figure}{section}
\newtheorem{theorem}{Theorem}
\newtheorem{lemma}{Lemma}

\newtheorem*{definition*}{Definition}

\newtheorem{assumption}{Assumption}

\newcommand{\eps}{\varepsilon}
\newcommand{\beps}{\boldsymbol{\eps}}
\newcommand{\F}{\mathbf{F}}
\newcommand{\R}{\mathbf{R}}
\newcommand{\E}{\mathbf{E}}
\newcommand{\C}{\mathbf{C}}
\newcommand{\D}{\mathcal{D}}
\newcommand{\X}{\mathbf{X}}
\newcommand{\bH}{\mathbf{H}}
\newcommand{\M}{\mathbf{M}}
\newcommand{\I}{\mathbf{I}}
\newcommand{\U}{\mathbf{U}}
\newcommand{\W}{\mathbf{W}}
\DeclareMathOperator{\Ve}{Vec}
\newcommand{\blue}[1]{{\textcolor{blue}{#1}}}
\newcommand{\orange}[1]{{\textcolor{orange}{#1}}}
\newcommand{\proj}[1]{\mathring{#1}}
\DeclareMathOperator{\tr}{tr}
\DeclareMathOperator{\cov}{Cov}

\begin{document}
\title[Nonstationary matrix factor models]{Inference in matrix-valued time series with common stochastic trends and multifactor error structure}
\author{Rong Chen}
\address{Department of Statistics, School of Arts and Sciences, Rutgers University, NJ USA}
\email{rongchen@stat.rutgers.edu}
\author{Simone Giannerini}
\address{Dipartimento di scienze economiche e statistiche, Università di Udine, Italy}
\email{simone.giannerini@uniud.it}
\author{Greta Goracci}
\address{Free University of Bozen-Bolzano, Faculty of Economics and Management, Bolzano, Italy}
\email{greta.goracci@unibz.it}
\author{Lorenzo Trapani}
\address{University of Leicester Business School, University Road, Leicester
LE1 7RH, UK, and Department of Economics and Management, University of
Pavia, 27100 Pavia, Italy}
\email{lt285@leicester.ac.uk}

\thanks{\textbf{Acknowledgements} We are grateful to the participants to: the Workshop on the Analysis of Complex Data: Tensors, Networks, and Dynamic Systems, (Banff, May 12-17, 2024), in particular Elynn Chen and Qiwei Yao; to the Econometrics with Data Science conference (University of Reading, 16 September, 2024); and to the 1st CAM-Risk conference New risks and policy challenges (Universita' di Pavia, 18-20 December, 2024).}

\begin{abstract}
We develop an estimation methodology for a factor model for high-dimensional matrix-valued time series, where common stochastic trends and common stationary factors can be present. We study, in particular, the estimation of (row and column) loading spaces, of the common stochastic trends and of the common stationary factors, and the row and column ranks thereof. In a set of (negative) preliminary results, we show that a projection-based technique fails to improve the rates of convergence compared to a ``flattened'' estimation technique which does not take into account the matrix nature of the data. Hence, we develop a three-step algorithm where: (i) we first project the data onto the orthogonal complement to the (row and column) loadings of the common stochastic trends; (ii) we subsequently use such ``trend free'' data to estimate the stationary common component; (iii) we remove the estimated common stationary component from
the data, and re-estimate, using a projection-based estimator, the row and column common stochastic trends and their loadings. We show that this estimator succeeds in refining the rates of convergence of the initial, ``flattened'' estimator. As a by-product, we develop consistent eigenvalue-ratio based estimators for the number of stationary and nonstationary common factors.
\end{abstract}

\maketitle

\doublespacing
\section{Introduction\label{intro}}

In this paper, we study inference for a Matrix Factor Model (MFM)\ where common stochastic trends may be present as well as stationary common
factors, viz.%
\begin{equation}
\underset{p_{1}\times p_{2}}{\X_{t}}=\underset{p_{1}\times h_{R_{1}}}{\R_{1}}\underset{h_{R_{1}}\times h_{C_{1}}}{\F_{1,t}} \underset{h_{C_{1}}\times p_{2}}{\C_{1}^{\prime}}+\underset{p_{1}\times h_{R_{0}}}{\R_{0}}\underset{h_{R_{0}}\times h_{C_{0}}}{\F_{0,t}}\underset{h_{C_{0}}\times p_{2}}{\C_{0}^{\prime}}+\underset{p_{1}\times p_{2}}{\E_{t}},  \label{model}
\end{equation}%
where: $1\leq t\leq T$, $\min \left\{ p_{1},p_{2},T\right\} \rightarrow \infty $, $0\leq h_{R_{1}}$, $h_{C_{1}}$, $h_{R_{0}}$, $h_{C_{0}}<\infty $,
the common factors $\left\{ \F_{0,t},-\infty <t<\infty \right\} $form a stationary sequence, and the common stochastic trends $\F_{1,t}$ satisfy%
\begin{equation}
\F_{1,t}=\F_{1,t-1}+\beps_{t},
\label{model_ft}
\end{equation}%
with $\left\{\beps_{t},-\infty <t<\infty \right\} $ a stationary sequence. In particular, we propose a methodology to estimate the row and column loadings spaces for both the stationary and the nonstationary common factor structures (resp. $\R_{0}$, $\C_{0}$, $\R_{1}$ and $\C_{1}$), the common factors $\F_{1,t}$ and $\F_{0,t}$, and the dimensions of all factor spaces ($h_{R_{1}}$, $h_{C_{1}}$, $h_{R_{0}}$ and $h_{C_{0}}$).
\subsection*{Matrix Factor Models: a brief literature review}
In recent years, MFMs have been studied extensively as a way of modelling parsimoniously large datasets, and as an alternative to vectorising the data $\X_{t}$. \citet{wang2019factor} and \citet{Chen2020StatisticalIF} make powerful cases in favour of exploiting the matrix structure of $X_{t}$, when there is a ``two-way''\ factor structure, for the purpose of dimension reduction (see also \citealp{he2023one}). In addition to dimension reduction, several datasets lend themselves naturally to be modelled as matrix-valued time series, with examples in health sciences (such as electronic health records and ICU data) and 2-D image data processing (see, \textit{inter alia}, \citealp{Chen2020StatisticalIF} and \citealp{Gao2021A}), in macroeconomics (see e.g. \citealp{wang2019factor}, where several macroeconomic indicators are modelled for different countries; or \citealp{chen2021factor}, who consider import-export data), and in finance (see e.g. \citealp{wang2019factor}, where several portfolios are modelled through several indicators such as size or BE\ ratio). There is now a plethora of contributions on inference for \textit{stationary} MFMs. The determination of the number of common factors has been studied in various contributions, including, e.g. \citet{han2020rank} and \citet{he2023one}. The estimation of loadings and common factors has been developed in several articles, including \citet{Chen2020StatisticalIF}, who propose an estimation technique based on the spectrum of a weighted average of the mean and the column (row) covariance matrix of the data; \citet{hkyz2021}, who refine the rates of convergence of the estimated loadings via iterative projections (see also \citealp{he2023one}); and also \citet{chen2020semiparametric} and \citet{chen2021factor}. 
All the references above, however, consider models where only stationary, $I(0)$ common factors are present, thus ruling out the presence
of $I(1)$ common factors as described by equation (\ref{model} ). This can be viewed as an important gap in the literature: many datasets are well-known to be driven by stochastic trends: macroeconomic indicators are typically $I(1)$; and yield curves are often modelled as being driven by common stochastic trends, at least in the vector-valued case (see e.g. \citealp{bt2}, and the empirical application in  \citealp{hamilton2024principal}). Indeed, not only we are not aware of any contributions dealing with common stochastic trends in the context of
matrix-valued time series, but contributions in the context of vector-valued time series that consider $I(1)$ common factors are also rare:  \citet{bai04} develops the full-blown inferential theory for loadings and common factors; \citet{bt2} propose a family of randomised tests to determine the number of common trends and stationary factors; and \citet{massacci2022high} consider a threshold model where, in each regime, there are latent $I(1)$ common factors. Related contributions, \textit{lato sensu}, have also been developed in the context of high-dimensional cointegration (e.g. \citealp{onatski2018alternative}, \citealp{onatski2019extreme}, \citealp{bykhovskaya2022cointegration}, and \citealp{bct}). Naturally, in order to estimate (\ref{model}), it is always possible to take the first difference of the matrix-valued time series $\X_t$, and apply one of the techniques discussed above for stationary data; however, this would not afford the separate estimation of the $I(1)$ and the $I(0)$ components; in fact, it would not even be possible to understand whether there are any $I(1)$ common factors at all. 
\par
In this paper, we fill the aforementioned gap by developing the full-fledged inferential theory for model (\ref{model}); as we argue below, this is not a mere extension of existing techniques developed for the vector-valued case, as the problem calls for an entirely novel methodology. 

\subsection*{The structure of (\ref{model})}
We discuss two possible interpretations of (\ref{model})-(\ref{model_ft}). The first one goes along similar lines as in \citet{wang2019factor}, and it considers a ``two-step'' hierarchical factor model. Let the $j$-th column of $\X_{t}$ be denoted as $\X_{\cdot j,t}$, $1\leq j\leq p_{2}$, and consider the following factor model 
\begin{equation}
\X_{\cdot j,t}=\R_{1}\mathbf{g}_{j1,t}+\R_{0}\mathbf{g}_{j0,t}+\widetilde{\E}_{\cdot j,t},  \label{bai_v}
\end{equation}%
where $\mathbf{g}_{j1,t}$ is a $h_{R_{1}}$-dimensional $I(1)$ process, $\mathbf{g}_{j0,t}$ is a $h_{R_{0}}$-dimensional $I(0)$ process, and $\widetilde{\E}_{\cdot j,t}$ is an idiosyncratic term; (\ref{bai_v}) is exactly the same model as in \citet{bai04} for a
vector-valued time series with common stochastic trends. Define now the $%
1\leq i_{1}\leq h_{R_{1}}$ rows of $\mathbf{g}_{j1,t}$ as $\mathbf{g}%
_{j1,t}^{\left( i_{1}\right) }$, and the $1\leq i_{0}\leq h_{R_{0}}$ rows of 
$\mathbf{g}_{j0,t}$ as $\mathbf{g}_{j0,t}^{\left( i_{0}\right) }$, and
consider the ``nested''\ factor model for
the $p_{2}$-dimensional series $\mathbf{g}_{1,t}^{\left( i_{1}\right) }$, $%
1\leq i_{1}\leq h_{R_{1}}$:%
\begin{equation*}
\mathbf{g}_{1,t}^{\left( i_{1}\right) }=\left( \mathbf{g}_{11,t}^{\left(
i_{1}\right) },...,\mathbf{g}_{p_{2}1,t}^{\left( i_{1}\right) }\right)
^{\prime}=\underset{p_{2}\times h_{C_{1}}}{\C_{1}^{\left(
i_{1}\right) }}\underset{h_{C_{1}}\times 1}{\mathbf{h}_{1,i_{1},t}}+\boldsymbol{\nu}_{i_{1},t}^{(1) },
\end{equation*}%
where $\mathbf{h}_{1,i_{1},t}$ is a vector of $I(1)$ common
factors, $\C_{1}^{\left( i_{1}\right) }\ $a loadings matrix, and$\ 
\boldsymbol{\nu}_{i_{1},t}^{(1) }$ a $p_{2}$-dimensional stationary
idiosyncratic component. By the same token, we also define the nested\
factor model for the $p_{2}$-dimensional stationary series $\mathbf{g}%
_{0,t}^{\left( i_{0}\right) }$, $1\leq i_{0}\leq h_{R_{0}}$:%
\begin{equation*}
\mathbf{g}_{0,t}^{\left( i_{0}\right) }=\underset{p_{2}\times h_{C_{0}}}{%
\C_{0}^{\left( i_{0}\right) }}\underset{h_{C_{0}}\times 1}{\mathbf{h}%
_{0,i_{0},t}}+\boldsymbol{\nu}_{i_{0},t}^{\left( 0\right) },
\end{equation*}%
where $\mathbf{h}_{0,i_{0},t}$ is a vector of stationary, $I(0)$ common factors. Let us now put the above together. Assume $\C%
_{1}^{\left( i_{1}\right) }=\C_{1}$ and $\C_{0}^{\left(
i_{0}\right) }=\C_{0}$; define $\F_{1,t}$\ by stacking the
vectors $\mathbf{h}_{1,i_{1},t}^{\prime}$, and $\F_{0,t}$\ by
stacking the vectors $\mathbf{h}_{0,i_{1},t}^{\prime}$; define $\boldsymbol{\nu}_{t}^{(1) }$ by stacking the vectors $\boldsymbol{\nu}%
_{i_{1},t}^{(1) \prime}$, and $\boldsymbol{\nu}_{t}^{\left(
0\right) }$\ by stacking the vectors $\boldsymbol{\nu}_{i_{0},t}^{\left(
0\right) \prime}$. We finally receive%
\begin{equation*}
\X_{t}=\R_{1}\F_{1,t}\C_{1}^{\prime}+%
\R_{1}\boldsymbol{\nu}_{t}^{(1) }+\R_{0}\F%
_{0,t}\C_{0}^{\prime}+\R_{0}\boldsymbol{\nu}_{t}^{\left(
0\right) }+\widetilde{\E}_{t}=\R_{1}\F_{1,t}\C_{1}^{\prime}+\R_{0}\F_{0,t}\C_{0}^{\prime}+%
\E_{t},
\end{equation*}%
where $\E_{t}\equiv \R_{1}\boldsymbol{\nu}_{t}^{(1)
}+\R_{0}\boldsymbol{\nu}_{t}^{\left( 0\right) }+\widetilde{\E}%
_{t}$. 

As a second example, we note that, in (\ref{model})-(\ref%
{model_ft}), a $p_{1}\times p_{2}$ valued $I(1)$ time series $%
\X_{t}$ is driven by a (small) number of common stochastic trends.
Hence, (\ref{model}) represents a case of ``two-way''\ cointegration, in that it is possible to
construct vector-valued time series as linear combinations of both the rows
and the columns of $\X_{t}$ which are stationary (in essence, by
pre- or post- multiplying $\X_{t}$ by the orthogonal complements to $%
\R_{1}$ and $\C_{1}$ respectively). In this respect, (\ref%
{model}) can be viewed, heuristically, as an extension of the common
stochastic trends representation of a cointegrated system as discussed in %
\citet{stock1988variable}. Indeed, two recent contributions (\citealp{li2024cointegratedmatrixautoregressionmodels}, and \citealp{hecq2024detecting}) consider the extension of cointegrated Vector AutoRegressions to matrix-valued time series, but only for the case where the cross-sectional dimensions $p_{1}$ and $p_{2}$ are fixed. As a word of warning, however, we would like to point out that (\ref{model}) is not entirely aligned to a cointegrated system in the sense of \citet{johansen1991estimation}, and we refer to our concluding remarks in
Section \ref{conclusion} for a more thorough analysis.

\subsection*{The estimation methodology}
We now offer a preview of how our methodology works and of our results. We begin with an account of the problem at hand;
the details are in Section \ref{flattened}. Given the number of $I(1)$ common factors $h_{R_{1}}$ and $h_{C_{1}}$, we begin by noting
that, when estimating $\R_{1}$ using a ``flattened'' approach based on the second moment matrix $\sum_{t=1}^{T}\X_{t}\X_{t}^{\prime}$, the estimator has rate $O_{P}\left( p_{1}^{1/2}T^{-1}\right) $ - see Section \ref{flat}. Modulo the dimensionality effect represented by the $O_{P}\left(p_{1}^{1/2}\right) $ term, such a ``superconsistency'' is typical of the estimation of a cointegrated system; however, especially if $T$ is small, this rate may not be sufficiently fast. In order to refine it, a possible, and natural, way of estimating $\R_{1}$ in (\ref{model}) would be to use the iterative projection-based estimator considered in \citet{hkyz2021} - that is, given the initial, ``flattened''\ estimator of $\C_{1}$ (say $\hat{\C}_{1}$), one could define the projected data $\X_{t}\hat{\C}_{1}$, and re-estimate $\R_{1}$ as the eigenvectors corresponding to the $h_{R_{1}}$ largest eigenvalues of the (suitably rescaled) $\sum_{t=1}^{T}\X_{t}\hat{\C}_{1}\left( \X_{t}\hat{\C}_{1}\right)^{\prime}$. However, as we show in Section \ref{proj}, this estimator fails to improve the rate of convergence of the initial, non-projection-based,
estimate of $\R_{1}$. Heuristically, this can be explained by noting that, in (\ref{model}), the term $\R_{0}\F_{0,t}\C _{0}^{\prime}$ is present. For the purpose of the projection-based estimator of $\R_{1}$ and $\C_{1}$, this is a component of the error term; however, the projection-based estimator essentially works by attenuating the error by averaging it cross-sectionally through its projection onto $\hat{\C}_{1}$. Indeed, when $\X_{t}$ is multiplied by $\hat{\C}_{1}$, the ``signal''\ component $\R_{1}\F_{1,t}\C_{1}^{\prime}\hat{\C}_{1}$ contains the term $\C_{1}^{\prime}\hat{\C}_{1}$, which is proportional to $p_{2}$; conversely, the error component $\E_{t}\hat{\C}_{1}$ (provided that the errors are \textit{weakly} cross-sectionally dependent) heuristically becomes proportional to $p_{2}^{1/2}$ - hence, projecting results in a reduction of the noise-to-signal ratio. However, this argument fails in the presence of stationary common factors: the component $\R_{0}\F_{0,t}\C_{0}^{\prime}\hat{\C}_{1}$, in general, is proportional to $p_{2}$ due to the \textit{strong} cross-sectional dependence induced by the common factors $\F_{0,t}$; seeing as this component is effectively part of the error term, the noise-to-signal is not attenuated, and no refinement of the rates of convergence of the estimates of $\R_{1}$ (or $\C_{1}$) can be expected. In light of the above, it would be desirable to eliminate the $\R_{0}\F_{0,t}\C_{0}^{\prime}$ component prior to applying the projection method to the estimation of $\R_{1}$ (or $\C_{1}$). This, too, is not straightforward: a consistent estimate of $\R_{0}\F_{0,t} \C_{0}^{\prime}$ is required, but this cannot be obtained by simply
estimating $\R_{1}$ and $\C_{1}$ (and the common factors $ \F_{1,t}$) using the first-stage, flattened estimator mentioned above: the rate of convergence of the estimated $I(1)$\ common component $\R_{1}\F_{1,t}\C_{1}^{\prime}$ is not fast enough to be able to get rid of it without an impact on the subsequent estimation of $\R_{0}$, $\C_{0}$, and $\F_{0,t}$.

Hence, in this paper we propose a different iterative procedure, which we describe henceforth; the details are in Section \ref{orthogonal}. After obtaining the initial, flattened estimator of $\C_{1}$ (resp. $\R_{1}$), denoted as $\hat{\C}_{1}$, we construct its orthogonal complement $\hat{%
\C}_{1,\perp}$; this is a ``huge'' matrix, since both the numbers of its rows and columns grow with $p_{2}$. In order to estimate the stationary common component $\R_{0}\F_{0,t}\C_{0}^{\prime}$, we firstly get rid of the $I(1)$ common component $\R_{1}\F_{1,t}\C_{1}^{\prime}$
by projecting the data $\X_{t}$ onto $\hat{\C}_{1,\perp}$, and subsequently using the second moment matrix $\sum_{t=1}^{T}\X_{t}\hat{\C}_{1,\perp}\left( \X_{t}\hat{\C}_{1,\perp}\right) ^{\prime}$ to estimate $\R_{0}$, $\C_{0}$, and $\F_{0,t}$. Interestingly, this approach is the complete
opposite to the projection-based estimator (and, in general, to the philosophy of the Johnson-Lindenstrauss Lemma, and of the ``sketching'' approach, see e.g. \citealp{matouvsek2008variants} as a comprehensive review): instead of projecting the data onto a small dimensional space which is ``parallel''\ to $\C_{1}$ (so as to conserve the information contained in it), we project onto a large dimensional space which is orthogonal (so as to get rid of $\C_{1}$). As we show in Section \ref{inferencercf1}, this procedure yields an
estimator of the stationary common component $\R_{0}\F_{0,t}\C_{0}^{\prime}$ (say $\hat{\R}_{0}\hat{\F}_{0,t}\hat{\C}_{0}^{\prime}$) whose rate of convergence is sufficiently fast to be able to filter it out from the data $\X_{t}$. We then construct the ``purified''\ data $\overset{\diamond }{\X}_{t}=\X_{t}-\hat{\R}_{0}\hat{\F}_{0,t}\hat{\C}_{0}^{\prime}$, and apply the
projection based estimator thereto, using the second moment matrix $%
\sum_{t=1}^{T}\overset{\diamond }{\X}_{t}\hat{\C}%
_{1}\left( \overset{\diamond }{\X}_{t}\hat{\C}%
_{1}\right) ^{\prime}$. The resulting estimator of $\R_{1}$ refines
the rate of the initial estimator $\hat{\R}_{1}$, with - in
particular - the $O_{P}\left( p_{1}^{1/2}T^{-1}\right) $ component in the
error term becoming of order $O_{P}\left(
p_{1}^{1/2}p_{2}^{-1/2}T^{-1}\right) $. This is exactly what would be
expected when using a projection-based estimator in the absence of strong
cross-sectional dependence in the error term. In Section \ref{projectCR}, we
show that refinements are also available for the corresponding estimator of $%
\C_{1}$ (as can be expected), and for the estimator of the $I\left(
1\right) $ common factors $\F_{1,t}$. As a by-product, we also
derive consistent estimation of $\R_{0}$, $\C_{0}$, and $%
\F_{0,t}$. Finally, building on the spectra of the second moment
matrices $\sum_{t=1}^{T}\overset{\diamond }{\X}_{t}\hat{\C}_{1}\left( \overset{\diamond }{\X}_{t}\hat{\C}%
_{1}\right) ^{\prime}$ and $\sum_{t=1}^{T}\X_{t}\hat{\C}%
_{1,\perp}\left( \X_{t}\hat{\C}_{1,\perp}\right)
^{\prime}$, we are able to propose estimators of the ranks $h_{R_{1}}$, $%
h_{C_{1}}$, $h_{R_{0}}$ and $h_{C_{0}}$ based on the eigenvalue ratio
principle. 

In conclusion, this is the first attempt to carry out inference on a MFM
with common stationary and non-stationary, $I(1)$, factors. We
make at least three contributions. First, we derive the full-blown
estimation theory for the stationary and the non-stationary factor spaces;
the ``anti-projection''\ approach which we
develop is, to the best of our knowledge, entirely novel. Secondly, we study
the estimation of the dimensions of the stationary and the non-stationary
factor spaces $h_{R_{0}}$, $h_{C_{0}}$, $h_{R_{1}}$, and $h_{C_{1}}$; whilst
this is an application, as mentioned above, of the eigenvalue ratio
principle, however this paper is the first contribution to address this
issue in the context of MFMs. Thirdly and finally, in the Supplement we study the spectrum of the second moment matrices studied hereafter; building on these, a test for the null hypothesis that the matrix-valued time series $\X_{t}$ can be readily derived, e.g. building on the randomised tests discussed in \citet{bt2}.

\medskip
The remainder of the paper is organised as follows. In Section \ref{assumptions}, we discuss our model and the main assumptions required for
our methodology. In Section \ref{estimation}, we report the full-fledged inferential theory. In particular, in Section \ref{flattened} we report a
set of preliminary, ``negative''\ results concerning the estimation of the $I(1)$ factor structure, and the failure of the iterative projection-based estimator; in Section \ref{orthogonal} we report the ``anti-projection''-based methodology, and the rates of convergence of the estimated non-stationary and stationary factor structures; and in Section \ref{number}, we propose an estimation technique for the ranks $h_{R_{0}}$, $h_{C_{0}}$, $h_{R_{1}}$, and $h_{C_{1}}$. Monte Carlo studies are reported in Section \ref{simulation}.
Section \ref{conclusion} concludes, also discussing possible extensions to e.g. the estimation of a cointegrated system. Technical lemmas, proofs and further evidence from synthetic data is contained in the Supplement.
\par
NOTATION. We use $\log \left( x\right) $ to denote the natural logarithm of $x$; we denote matrices using capitalised bold-face, e.g. $\mathbf{A}$, their elements using lower-case (e.g. $a_{ij}$ denotes the element of $\mathbf{A}$
in position $\left( i,j\right) $), and, for a generic $n\times m$ matrix $\mathbf{A}$, we define the space orthogonal to its column space as $\mathbf{A}_{\perp}$; the Frobenius norm is denoted as $\left\Vert \mathbf{A}\right\Vert _{F}=\left( \sum_{i=1}^{n}\sum_{j=1}^{m}a_{ij}^{2}\right) ^{1/2}$. Given a random variable $Y$, we use $\left\vert Y\right\vert _{\nu }$ for its $\mathcal{L}_{\nu }$-norm, i.e. $\left\vert Y\right\vert _{\nu }=\left(E\left\vert Y\right\vert ^{\nu }\right) ^{1/\nu }$, $\nu \geq 1$. Other, relevant notation is introduced later on in the paper.
\par
\section{Model and assumptions\label{assumptions}}
Recall (\ref{model})-(\ref{model_ft}):
\begin{align*}
&\X_{t}=\R_{1}\F_{1,t}\C_{1}^{\prime}+\R_{0}\F_{0,t}\C_{0}^{\prime}+\E_{t},\\
&\F_{1,t}=\F_{1,t-1}+\beps_{t}.
\end{align*}%
In the spirit of \textit{approximate} factor models (\citealp{chamberlainrothschild83}), we assume (weak) serial and cross sectional dependence. As far as the former is concerned, we will rely on the following
\begin{definition*}
The $d$-dimensional sequence $\left\{ m_{t},-\infty <t<\infty \right\} $ forms an $\mathcal{L}_{\nu }$-decomposable Bernoulli shift if and only if $m_{t}=h\left( \eta _{t},\eta _{t-1},\dots\right) $, where: $\left\{ \eta _{t},-\infty <t<\infty \right\} $ is an \textit{i.i.d.} sequence with values in a measurable space $S$; $h\left( \cdot \right) :S^{\mathbb{N}}\rightarrow \mathbb{R}^{d}$ is a non random measurable function; $\left\vert m_{t}\right\vert
_{\nu }<\infty $; and $\left\vert m_{t}-m_{t,\ell }^{\ast }\right\vert _{\nu}\leq c_{0}\ell ^{-a}$, for some $c_{0}>0$ and $a>0$, where $m_{t,\ell}^{\ast }=h\left( \eta _{t},\dots,\eta _{t-\ell +1},\eta _{t-\ell ,t,\ell}^{\ast },\right. $ $\left. \eta _{t-\ell -1,t,\ell }^{\ast }\dots\right) $, with $\left\{ \eta _{s,t,\ell }^{\ast },-\infty <s,\ell ,t<\infty \right\} $ \textit{i.i.d.} copies of $\eta _{0}$, independent of $\left\{\eta_{t},-\infty <t<\infty \right\} $.
\end{definition*}

The concepts of Bernoulli shift and decomposability appeared first in \citet{ibragimov1962some}; see also \citet{wu2005} and \citet{berkeshormann}. Bernoulli shifts have proven a convenient way to model dependent time series, mainly due to their generality and to the fact that they are much easier to verify than e.g. mixing conditions: \citet{aue09} and \citet{linliu}, \textit{inter alia}, provide numerous examples of such DGPs, which include ARMA models, ARCH/GARCH sequences, and other nonlinear time series models (e.g. random coefficient autoregressive models and
threshold models). \\
\par
We are now ready to present our assumptions. Prior to doing so, we note that - for the sake of transparency of the proofs - we have tried to write \textit{primitive} assumptions. However, all our assumptions could be replaced by more high-level conditions, as we discuss after each assumption. Recall that the orthogonal complements to $\R_1$ and $\C_1$ are denoted as $\R_{1,\perp}$ and $\C_{1,\perp}$ respectively. 
\begin{assumption}
\label{as-1}It holds that: \textit{(i)} $\left\{ \Ve\left(\beps_{t}\right) ,-\infty <t<\infty \right\} $ is an $\mathcal{L} _{2+\delta }$-decomposable Bernoulli shift with $a>2$; 
\textit{(ii)} (a) $ \lim_{T\rightarrow \infty }E\left( T^{-1/2}\sum_{t=1}^{T}\beps_{t}\right) \left( T^{-1/2}\sum_{t=1}^{T}\beps_{t}\right)
^{\prime}=\Sigma _{F}^{\left( a\right) }$ with $\Sigma _{F}^{\left(a\right) }$ a positive definite $h_{R_{1}}\times h_{R_{1}}$ matrix; (b) $ \lim_{T\rightarrow \infty }E\left( T^{-1/2}\sum_{t=1}^{T}\beps_{t}\right) ^{\prime}\left(T^{-1/2}\sum_{t=1}^{T}\beps%
_{t}\right) =\Sigma _{F}^{\left( b\right) }$ with $\Sigma _{F}^{\left(b\right) }$ a positive definite $h_{C_{1}}\times h_{C_{1}}$ matrix.
\end{assumption}

\begin{assumption}
\label{as-2}It holds that: 
\textit{(i)} $\left\{ \Ve\left( \F_{0,t}\right) ,-\infty <t<\infty \right\} $ is an $\mathcal{L}_{4}$-decomposable Bernoulli shift with $a>2$;
\textit{(ii)} (a) $E\left( \F_{0,t}\F_{0,t}^{\prime}\right) =\Sigma _{F,1}^{\left( a\right) }$ with $\Sigma _{F,1}^{\left( a\right) }$ a positive definite $h_{R_{0}}\times h_{R_{0}}$ matrix; (b) $E\left( \F_{0,t}^{\prime}\F_{0,t}\right) =\Sigma _{F,1}^{\left( b\right) }$ with $\Sigma _{F,1}^{\left(
b\right) }$ a positive definite $h_{C_{0}}\times h_{C_{0}}$ matrix.
\end{assumption}

Assumptions \ref{as-1} and \ref{as-2} require $\left\{ \F_{0,t},-\infty <t<\infty \right\} $
and $\left\{ \beps_{t},-\infty <t<\infty \right\} $ to be
stationary sequences - hence, whilst \textit{conditional} heteroskedasticity
is allowed for, \textit{unconditional} heteroskedasticity is not.
In principle, it would be possible to consider this case too, by letting
- as suggested in Section 3.2.2 in \citet{horvath2023changepoint} - $\left\{ \F_{0,t},1\leq t\leq T\right\}
=\bigcup_{\ell =1}^{L}\left\{ \F_{0,t}^{\left( \ell \right)
},m_{\ell -1}\leq t\leq m_{\ell }\right\} $ with $m_{0}=1$ and $m_{L}=T$,
assuming that each sequence $\left\{ \F_{0,t}^{\left( \ell \right)
},-\infty <t<\infty \right\} $ satisfies Assumption \ref{as-1}. The main reason to have this assumption to model serial dependence is to be able to obtain bounds on the growth rates of partial sums, and other limiting theorems for summations involving $\left\{ \F_{0,t},-\infty <t<\infty \right\} $and $\left\{ \beps_{t},-\infty <t<\infty \right\} $. As mentioned above, all our technical results could be directly assumed (instead of shown using Assumptions \ref{as-1} and \ref{as-2}); this would make the set-up more general, but it would be less transparent. 

\begin{assumption}
\label{as-3}It holds that: 
\textit{(i)} $E\left( e_{ij,t}\right) =0$ and $E\left\vert e_{ij,t}\right\vert ^{4}\leq c_{0}$ for some $c_{0}<\infty $ and all $1\leq i\leq p_{1}$ and $1\leq j\leq p_{2}$; 
\textit{(ii)} (a) $ \sum_{t=1}^{T}\left\vert E\left( e_{ij,t}e_{i^{\prime}j^{\prime},s}\right) \right\vert \leq c_{0}$ for all $1\leq t\neq s\leq T$, $1\leq i,i^{\prime}\leq p_{1}$ and $1\leq j,j^{\prime}\leq p_{2}$; (b) $\sum_{i=1}^{p_{1}} \left\vert E\left(e_{ij,t}e_{i^{\prime}j^{\prime},s}\right) \right\vert\leq c_{0}$ for all $1\leq t,s\leq T$, $1\leq i\neq i^{\prime}\leq p_{1}$ and $1\leq j,j^{\prime}\leq p_{2}$; 
(c) $\sum_{j=1}^{p_{2}}\left\vert E\left( e_{ij,t}e_{i^{\prime}j^{\prime},s}\right) \right\vert \leq c_{0}$
for all $1\leq t,s\leq T$, $1\leq i,i^{\prime}\leq p_{1}$ and $1\leq j\neq j^{\prime}\leq p_{2}$; (d) $\sum_{j=1}^{p_{2}}\sum_{t=1}^{T}\left\vert E\left( e_{hj,t}e_{h^{\prime}k,s}\right) \right\vert \leq c_{0}$ for all $ 1\leq t\neq s\leq T$, $1\leq h,h^{\prime}\leq p_{1}$ and $1\leq j\neq k\leq p_{2}$; (e) $\sum_{i=1}^{p_{1}}\sum_{j=1}^{p_{2}}\left\vert E\left(e_{hj,t}e_{h^{\prime}j^{\prime},t}\right) \right\vert \leq c_{0}$ for all $ 1\leq t\leq T$, $1\leq i\neq i^{\prime}\leq p_{1}$ and $1\leq j\neq j^{\prime}\leq p_{2}$ 
\textit{(iii)} (a) $\sum_{i=1}^{p_{1}}\sum_{h=1}^{p_{2}}\sum_{t=1}^{T}\left\vert \cov\left(
e_{ik,t}e_{jk,t},e_{ih,s}e_{jh,s}\right) \right\vert \leq c_{0}$ for all $%
1\leq t\neq s\leq T$, $1\leq i\neq j^{\prime}\leq p_{1}$ and $1\leq h\neq
k\leq p_{2}$.
\end{assumption}

Assumption \ref{as-3} is a standard high-level requirement in this literature: in essence, it allows for the idiosyncratic components to be cross-sectionally
correlated, but only weakly, and it is virtually the same as Assumption D in \citet{hkyz2021} and Assumption B3 in \citet{he2023one}. The only difference with the extant literature is that we require the existence of only $4$ moments for the idiosyncratic components (as opposed to $8$); this is a direct consequence of Assumption \ref{as-5} below. 

\begin{assumption}
\label{as-4}It holds that: \textit{(i)} (a) $\left\Vert \R_{1}\right\Vert _{\max }<\infty $ and $\left\Vert \C_{1}\right\Vert _{\max
}<\infty $; (b) $\left\Vert \R_{0}\right\Vert _{\max }<\infty $ and $%
\left\Vert \C_{0}\right\Vert _{\max }<\infty $; \textit{(ii)} (a) $%
\R_{0}^{\prime}\R_{1,\perp}\neq 0$; (b) $\C_{0}^{\prime}\C_{1,\perp}\neq 0$.
\end{assumption}

Part \textit{(i)} of the assumption is standard. As far as part \textit{(ii)}
is concerned, we require it in order to avoid the case where, when
anti-projecting onto the orthogonal spaces $\R_{1,\perp}$ and $%
\C_{1,\perp}$, this annihilates also the common stationary
component, as well as the nonstationary one. 

\begin{assumption}
\label{as-5}It holds that: $\left\{ \beps_{t},1\leq t\leq
T\right\} $, $\left\{ \F_{0,t},1\leq t\leq T\right\} $ and $\left\{
e_{ij,t},1\leq t\leq T\right\} $ are three mutually independent groups, for
all $1\leq i\leq p_{1}$ and $1\leq j\leq p_{2}$.
\end{assumption}

Assumption 5 is the same as Assumption D in \citet{bai04}, and in principle it could
be relaxed, by replacing some of the assumptions above with more high-level
requirements (and strengthening the moment conditions).

\section{Estimation\label{estimation}}

We begin by presenting our ``negative'' results on the estimation (and of possible refinements thereof) of the row and column loading spaces associated with the common stochastic trends $\F_{1,t}$, and on the estimation of $\F_{1,t}$ itself, in Section \ref{flattened}. In Section \ref{flat}, we derive, as a benchmark, the results for the flattened estimators; in Section \ref{proj}, we show that the rates of convergence cannot be improved by applying the projection-based method directly. In Section \ref{orthogonal}, we present our methodology to refine the rates of convergence: in Section \ref{inferencercf1}, we estimate the stationary common component after projecting the nonstationary one onto its orthogonal complement, and remove them from the data; in Section \ref{projectCR}, we apply the projection-based methodology to refine the rates of convergence of the row and column loadings associated with the common stochastic trends; and, in Section \ref{proj-i0}, we consider a further iteration of this procedure to investigate whether it is possible to refine the estimates of the stationary common component. 

\subsection{Preliminary theory: negative results on the factor structures
estimation\label{flattened}}

In this section, we report a set of \textit{negative }results, which serve as motivation for our proposed algorithm. In particular, we begin by
studying ``flattened''\ estimators of the factor structure corresponding to the $I(1)$ component of equation (\ref{model}), i.e. estimators based on, essentially, vectorising the matrix-valued series $\X_{t}$, in Section \ref{flat}. We then consider ``projection-based''\ estimators of the aforementioned factor structure, based on preliminarily projecting the data $\X_{t}$ onto the space spanned by the columns of $\C_{1}$ (or, equivalently, the space spanned by the rows of $\R_{1}$), in Section \ref{proj}. In both cases, we show that, owing to the strong cross-sectional dependence induced by the factor structure in the $I(0) $ component of $\X_{t}$, estimation results in two major problems: (1) the common $I(1)$ factors cannot be estimated consistently (not even after a linear transformation), thus also making it impossible to estimate consistently the common $I(1)$ component $\R_{1}\F_{1,t}\C_{1}^{\prime}$, in turn making it impossible to estimate the $I(0)$ common factor structure; and (2) even
though the spaces spanned by the columns of $\C_{1}$ or $\R_{1}$ can be estimated consistently, projecting onto $\C_{1}$ or $\R_{1}$ does not improve the rates of convergence of such estimators.

\subsubsection{The flattened estimators\label{flat}}

Consider the ``flattened'' sample covariance matrices%
\begin{equation}
\M_{R_1}=\frac{1}{p_{1}p_{2}T^{2}}\sum_{t=1}^{T}\X_{t} \X_{t}^{\prime},\text{ \ \ and \ \ }\M_{C_1}=\frac{1}{p_{1}p_{2}T^{2}}\sum_{t=1}^{T}\X_{t}^{\prime}\X_{t}.
\label{m-flat}
\end{equation}%
The estimator of $\R_{1}$ ($\C_{1}$) is defined as the eigenvectors corresponding to the largest $h_{R_{1}}$ (resp. $h_{C_{1}}$) eigenvalues of $M_{R_1}$ (resp. $\M_{C_1}$), viz. 
\begin{equation}
\M_{R_1}\hat{\R}_{1} =\hat{\R}_{1}\Lambda _{R_{1}} \text{, \ \ and \ \ } \M_{C_1}\hat{\C}_{1} =\hat{\C}_{1}\Lambda _{C_{1}},
\label{def-c-hat}
\end{equation}%
where $\Lambda _{R_{1}}$ is a $h_{R_{1}}\times h_{R_{1}}$ diagonal matrix containing the largest $h_{R_{1}}$ eigenvalues of $\M_{R_1}$, and $\Lambda _{C_{1}}$\ is defined similarly, under the constraints $\hat{\R}_1^{\prime} \hat{\R}_1=p_{1}\I_{h_{R_{1}}}$ and $\hat{\C}_{1}%
^{\prime}\hat{\C}_{1}=p_{2}\I_{h_{C_{1}}}$.

\begin{theorem}\label{hat-estimates}
We assume that Assumptions \ref{as-1}-\ref{as-5} are satisfied. Then there exist: a $h_{R_{1}}\times h_{R_{1}}$ matrix $\bH_{R_{1}}$, with $\left\Vert \bH_{R_{1}}\right\Vert _{F}=O_{P}(1)$ and $\left\Vert \left(\bH_{R_{1}}\right)^{-1}\right\Vert _{F}=O_{P}(1)$; and a $h_{C_{1}}\times h_{C_{1}}$ matrix $\bH_{C_{1}}$, with $\left\Vert \bH_{C_{1}}\right\Vert _{F}=O_{P}(1)$ and $\left\Vert \left( \bH_{C_{1}}\right) ^{-1}\right\Vert _{F}=O_{P}(1) $, such that 
\begin{equation}
\left\Vert \hat{\R}_1-\R_1\bH_{R_{1}}\right\Vert _{F}=O_{P}\left( \frac{p_{1}^{1/2}}{T}\right) \text{, \ \ and \ \ }  
\left\Vert \hat{\C}_{1}-\C_{1}\bH_{C_{1}}\right\Vert _{F}=O_{P}\left( \frac{p_{2}^{1/2}}{T}\right) .  \notag
\end{equation}
\end{theorem}

The results in Theorem \ref{hat-estimates} are ``standard'': the $O_P(T^{-1})$ rate is a consequence of having cointegration, and it corresponds to the well-known notion of ``superconsistency'' in time series econometrics (\citealp{stock1987}); the main difference, in our context, is the lack of identification which is typical of factor models, so that $\hat{\R}_1$ and $\hat{\C}_{1}$ are only able to estimate a transformation of $\R_{1}$ and $\C_{1}$ respectively. The impact of the dimensionality (given by the terms $p_{1}^{1/2}$ and $p_{2}^{1/2}$ respectively) is also a standard
feature of high dimensional factor models: e.g., a similar result is found
in \citet{bai04} in the context of vector-valued time series.
\par
As we show in Lemma \ref{f-hat-negative} below, the rates in Theorem \ref{hat-estimates} are generally not enough to estimate consistently the space spanned by the common nonstationary factors $\F_{1,t}$. We consider the following, Least-Squares-based, estimator 
\begin{equation}
\hat{\F}_{1,t}=\frac{1}{p_{1}p_{2}}\hat{\R}_1^{\prime}%
\X_{t}\hat{\C}_{1}.  \label{f-hat}
\end{equation}

\begin{lemma}
\label{f-hat-negative}We assume that Assumptions \ref{as-1}-\ref{as-5} are
satisfied. Then it holds that $|| \hat{\F}_{1,t}-\left( \bH_{R_{1}}\right) ^{-1}$
$\F_{1,t}\left( \bH_{C_{1}}^{\prime}\right) ^{-1} ||_{F}=O_{P}(1) .$
\end{lemma}

Lemma \ref{f-hat-negative} does state that $\hat{\F}_{1,t}$ is
consistent: the estimation error is of order $O_{P}(1) $, which
is of a smaller order of magnitude than the signal $\F_{1,t}$ - a
standard application of the Functional Central Limit Theorem yields $%
\left\Vert \F_{1,t}\right\Vert _{F}=O_{P}\left( T^{1/2}\right) $.
However, the rate of convergence is slower than e.g. the one derived in
Theorem 2 in \citet{bai04}, where it is shown that - for an $N$-dimensional
vector-valued time series - the rate of convergence is found to be $%
O_{P}\left( N^{-1/2}\right) +O_{P}\left( T^{-3/2}\right) $ $=$ $o_{P}\left(
1\right) $.

In the case of Lemma \ref{f-hat-negative}, the $O_{P}(1) $ order
arises from the fact that, in (\ref{model}), the remainder $\U_{t}$
defined as%
\begin{equation}
\X_{t}=\R_1\F_{1,t}\C_{1}^{\prime}+\U_{t}=\R_1\F%
_{t}\C_{1}^{\prime}+\left( \R_{0}\F_{0,t}\C%
_{1}^{\prime}+\E_{t}\right) ,  \label{u-model}
\end{equation}%
also contains a factor structure. In turn, upon inspecting the proof of
Lemma \ref{f-hat-negative} (and comparing it with e.g. the proof of Theorem
2 in \citealp{bai04}), when applying cross-sectional averaging to $\U%
_{t}$, the strong cross-correlation arising from the presence of $\F%
_{0,t}$ prevents it from drifting to zero. Intuitively, this indicates that,
as can be expected, cross-sectional averaging does not help in the presence
of common factors.

\subsubsection{Projection-based estimation\label{proj}}

We now show that the same problems as in Lemma \ref{f-hat-negative} affects
the projection-based estimators of $\R_{1}$ and $\C_{1}$. These
could be constructed along the lines studied in \citet{he2023one}, \textit{inter
alia}, using 
\begin{equation*}
\hat{\M}_{R_{1}}^{\dagger}=\frac{1}{p_{1}p_{2}^{2}T^{2}}%
\sum_{t=1}^{T}\X_{t}\hat{\C}_{1}\hat{\C}_{1}^{\prime}\X_{t}^{\prime}\text{, \ \ and \ \ }\hat{\M}%
_{C_{1}}^{\dagger}=\frac{1}{p_{1}^{2}p_{2}T^{2}}\sum_{t=1}^{T}\X%
_{t}^{\prime}\hat{\R}_1\hat{\R}_1^{\prime}\X%
_{t},
\end{equation*}%
as the eigenvectors corresponding to the largest $h_{R_{1}}$ (resp. $h_{C_{1}}$)
eigenvalues of $\hat{\M}_{R_{1}}^{\dagger}$ (resp. $\hat{%
\M}_{C_{1}}^{\dagger}$), viz. 
\begin{equation}
\hat{\M}_{R_{1}}^{\dagger}\hat{\R}_1^{\dagger}=\hat{\R}_1^{\dagger}\Lambda _{R_{1}}^{\dagger},\text{ \ \ and \ \ }\hat{%
\M}_{C_{1}}^{\dagger}\hat{\C}_{1}^{\dagger}=\hat{\C}_{1}^{\dagger}\Lambda _{C_{1}}^{\dagger},  \label{proj-neg-est}
\end{equation}%
where $\Lambda _{R_{1}}^{\dagger}$ is a $h_{R_{1}}\times h_{R_{1}}$ diagonal matrix
containing the largest $h_{R_{1}}$ eigenvalues of $\M_{R_1}$, and $%
\Lambda _{C_{1}}^{\dagger}$\ is defined similarly, under the constraints $%
\left( \hat{\R}_1^{\dagger}\right) ^{\prime}\hat{\R}_1^{\dagger}=p_{1}\I_{h_{R_{1}}}$ and $\left( \hat{\C}_{1}^{\dagger}\right) ^{\prime}\hat{\C}_{1}^{\dagger}=p_{2}\I%
_{h_{C_{1}}}$.

\begin{lemma}
\label{proj-negative}We assume that Assumptions \ref{as-1}-\ref{as-5} are
satisfied. Then there exists a $h_{R_{1}}\times h_{R_{1}}$ matrix $\bH%
_{R_{1}}^{\dagger}$, with $\left\Vert \bH_{R_{1}}^{\dagger}\right\Vert
_{F}=O_{P}(1) $ and $\left\Vert \left( \bH_{R_{1}}^{\dagger}\right) ^{-1}\right\Vert _{F}=O_{P}(1) $, and a $h_{C_{1}}\times h_{C_{1}}$ matrix $\bH_{C_{1}}^{\dagger}$,
with $\left\Vert \bH_{C_{1}}^{\dagger}\right\Vert _{F}=O_{P}\left(
1\right) $ and $\left\Vert \left( \bH_{C_{1}}^{\dagger}\right)
^{-1}\right\Vert _{F}=O_{P}(1) $, such that%
\begin{equation}
\left\Vert \hat{\R}_1^{\dagger}-\R_1\bH_{R_{1}}^{\dagger}\right\Vert _{F}=O_{P}\left( \frac{p_{1}^{1/2}}{T}\right) , \text{ \ \ and \ \ }  \left\Vert \hat{\C}_{1}^{\dagger}-\C_{1}\bH_{C_{1}}^{\dagger}\right\Vert _{F}=O_{P}\left( \frac{p_{2}^{1/2}}{T}\right) .
\notag
\end{equation}
\end{lemma}

Lemma \ref{proj-negative} is, in essence, a negative result: despite
projecting $\X_{t}$ onto the space spanned by the columns of $\C_{1}$, the rate of convergence of the new estimator $\hat{\R}_1^{\dagger}$ does not improve over that of $\hat{\R}_1
$. Intuitively, this is due to the fact that, when projecting $\X%
_{t} $ onto $\C_{1}$, the effect on the ``signal''\ component $\R_1\F_{1,t}\C_{1}^{\prime}\C_{1}$ is to make it grow by a factor $\C_{1}^{\prime}\C_{1}\sim p_{2}$; on the other hand, the effect of such projecting on $\U%
_{t}$ in (\ref{u-model}) depends on the extent of cross-sectional dependence
in $\U_{t}$. If the columns of $\U_{t}$ are weakly
cross-correlated, the effect of projecting is that $\U_{t}\C_{1}$
will grow at a rate $O\left( p_{2}^{1/2}\right) $; in such a case, with the
signal growing as $p_{2}$, the signal-to-noise ratio would be enhanced,
thereby resulting in an estimate with a faster rate of convergence.
Conversely, in the presence of strong dependence among the columns of $%
\U_{t}$, the cross-sectional averaging in $\U_{t}\C_{1}$
will result in a rate proportional to $p_{2}$; in this case, the signal and
the noise would grow by the same factor, hence resulting in no enhancement
of the rates of convergence of the projection-based estimator.

\subsection{Inferential theory based on anti-projections\label{orthogonal}}

The (negative) results in Lemmas \ref{f-hat-negative} and \ref{proj-negative} suggest that, in order to enhance the rates of convergence of the estimated common factors and loadings, the stationary common factor structure needs to be filtered out first, and then a projection-based technique can be applied. Hence, in this section, we present the three stages of our algorithms and the corresponding theory. First, we propose an estimator of $\R_{0}$, $\C_{0}$ and $\F_{0,t}$, obtained after projecting away the $I(1)$ component onto the space orthogonal to the columns
of $\C_{1}$ or, equivalently, $\R_{1}$ (Section \ref{inferencercf1}); the output is a set of consistent (modulo a linear transformation)
estimators of $\R_{0}$, $\C_{0}$ and $\F_{0,t}$, and therefore of the common $I(0)$ component $\R_{0}\F_{0,t}\C_{0}^{\prime}$ - albeit with improvable rates of convergence. Second, we study the estimation of $\R_{1}$, $\C_{1}$ and $\F_{1,t}$, after subtracting the estimated common $I(0)$ component $\R_{0}\F_{0,t}\C_{0}^{\prime}$ from the data $\X_{t}$, and projecting these onto the space spanned by (the estimated) $\C_{1}$ or equivalently $\R_{1}$, thus taking advantage of the fact that, after removing the common $I(0)$ component from the data, cross-sectional dependence becomes substantially weaker (Section \ref{projectCR}); the output is a set of consistent (modulo a linear transformation) estimators of $\R_{1}$, $\C_{1}$ and $\F_{1,t}$, and therefore of the common $I(1)$ component $\R_1\F_{1,t}\C_{1}^{\prime}$ - with faster rates of convergence than the ones derived in Section \ref{flat} for $\R_{1}$ and $\C_{1}$. Third, we refine the rates of convergence obtained in the first step, by projecting the data $\X_{t}$ (minus the estimated common $I(1)$ component $\R_1\F_{1,t}\C_{1}^{\prime}$) onto the space
spanned by (the estimated) $\C_{0}$ or equivalently $\R_{0}$
(Section \ref{proj-i0}); the output is a set of consistent (modulo a linear
transformation) estimators of $\R_{0}$, $\C_{0}$ and $%
\F_{0,t}$, and therefore of the common $I(0)$ component 
$\R_{0}\F_{0,t}\C_{0}^{\prime}$, with faster rates
of convergence than the ones derived in the first step.

\subsubsection{Anti-projection based estimation of $\R_{0}$, $\C_{0}$ and $\F_{0,t}$\label{inferencercf1}}

Define the orthogonal (to the columns of $\C_{1}$) space and its corresponding sample version 
\begin{align}
&\C_{1,\perp}=\I_{p_{2}}-\C_{1}\left( \C%
_{1}^{\prime}\C_{1}\right) ^{-1}\C_{1}^{\prime},
\label{c-orthogonal} \\
&\hat{\C}_{1,\perp}=\I_{p_{2}}-\hat{\C}%
_{1}\left( \hat{\C}_{1}^{\prime}\hat{\C}_{1}\right)
^{-1}\hat{\C}_{1}^{\prime}.  \label{c-hat-orthogonal}
\end{align}%
The two matrices are: $p_{2}\times p_{2}$; symmetric; and idempotent. By the
same token, we can also define $\hat{\R}_{1,\perp}$ (as an
estimator of the space $\R_{1,\perp}$, orthogonal to the columns of 
$\R$), and study its use and its properties; 
Define%
\begin{align}
&\hat{\X}_{t}^{C_{1}} =\X_{t}\hat{\C}_{1,\perp}, \text{ \ \ and \ \ }  \hat{\X}_{t}^{R_{1}} =\X_{t}^{\prime}\hat{\R}_{1,\perp},  \label{x-hat-2} \\
&\M_{R_{1},\perp}=\frac{1}{p_{1}p_{2}^{2}T}\sum_{t=1}^{T}\hat{\X}_{t}^{C_{1}}\left( \hat{\X}_{t}^{C_{1}}\right) ^{\prime},\text{ \ \ and \ \ }\M_{C_{1},\perp}=\frac{1}{p_{1}^{2}p_{2}T}\sum_{t=1}^{T}\hat{\X}_{t}^{R_{1}}\left(\hat{\X}_{t}^{R_{1}}\right)^{\prime}.\label{m-orthogonal}
\end{align}%
The estimator of $\R_{0}$ ($\C_{0}$) is defined as the
eigenvectors corresponding to the largest $h_{R_{1}}$ (resp. $h_{C_{1}}$)
eigenvalues of $\M_{R_{1},\perp}$ (resp. $\M_{C_{1},\perp}$), viz. 
\begin{equation}
\M_{R_{1},\perp}\hat{\R}_{0} =\hat{\R}_{0}\Lambda _{R_{0}}, \text{ \ \ and \ \ } 
\M_{C_{1},\perp}\hat{\C}_{0} =\hat{\C}_{0}\Lambda _{C_{0}},  \label{def-c1-hat}
\end{equation}%
where $\Lambda _{R_{0}}$ is a $h_{R_{1}}\times h_{R_{1}}$ diagonal matrix
containing the largest $h_{R_{1}}$ eigenvalues of $\M{R_{1},\perp}$, and $\Lambda _{C_{0}}$\ is defined similarly, under
the constraints $\hat{\R}_{0}^{\prime}\hat{\R}_{0}=p_{1}\I_{h_{R_{1}}}$ and $\hat{\C}_{0}^{\prime}\hat{\C}_{0}=p_{2}\I_{h_{C_{1}}}$.

\begin{theorem}
\label{rc1-hat}We assume that Assumptions \ref{as-1}-\ref{as-5} are satisfied. Then there exist: a $h_{R_{1}}\times h_{R_{1}}$ matrix $\bH%
_{R_0}$, with $\left\Vert \bH_{R_0}\right\Vert _{F}=O_{P}(1) $ and $\left\Vert \left( \bH_{R_0}\right) ^{-1}\right\Vert_{F}=O_{P}(1) $; and a $h_{C_{1}}\times h_{C_{1}}$ matrix $\bH_{C_0}$, with $\left\Vert \bH_{C_0}\right\Vert_{F}=O_{P}(1) $ and $\left\Vert \left( \bH_{C_0}\right)
^{-1}\right\Vert _{F}=O_{P}(1) $, such that
\begin{eqnarray}
\left\Vert \hat{\R}_{0}-\R_{0}\bH_{R_0}\right\Vert _{F} &=&O_{P}\left( \frac{p_{1}^{1/2}}{p_{2}^{1/2}T^{1/2}}\right) +O_{P}\left( \frac{p_{1}^{1/2}}{p_{1}p_{2}}\right) ,  \label{r1-hat}
\\
\left\Vert \hat{\C}_{0}-\C_{0}\bH_{C_0}\right\Vert _{F} &=&O_{P}\left( \frac{p_{2}^{1/2}}{p_{1}^{1/2}T^{1/2}}\right) +O_{P}\left( \frac{p_{2}^{1/2}}{p_{1}p_{2}}\right) .  \label{c1-hat}
\end{eqnarray}
\end{theorem}

Equations (\ref{r1-hat}) and (\ref{c1-hat}) contain the rates of convergence
of $\hat{\R}_{0}$ and $\hat{\C}_{0}$; again, $\R_{0}$ and $\C_{0}$ are estimated modulo a transformation.
The estimators $\hat{\R}_{0}$ and $\hat{\C}_{0}$ are, in essence, projection-based estimators;
hence, their rates can be compared with the ones obtained e.g. in Theorem
3.1 in \citet{he2023one}. The two terms in (\ref{r1-hat}) and (\ref{c1-hat}) are
the same as found in \citet{he2023one}; we would like to point out that \citet{he2023one} obtain also further error terms, which in our case are absent. This is, essentially, due to the fact that the projection matrix, $\hat{\C}_{1,\perp}$, has a very fast rate of convergence to $\C_{1,\perp}$,\footnote{See Lemmas \ref{r-orth} and \ref{c-orth}.} with $\left\Vert \hat{\C}_{1,\perp}-\C_{1,\perp}\right\Vert _{F}^{2}=O_{P}\left(
T^{-2}\right)$.
\par
We now turn to the estimation of the stationary common factors $\F_{0,t}$. Using the Least Squares principle, we can define the following
estimator of $\F_{0,t}$
\begin{eqnarray}
\Ve\hat{\F}_{0,t} &=&\left[ \left( \hat{\C}_{0}^{\prime}\hat{\C}_{1,\perp}\left( \hat{\C}_{1,\perp}\right) ^{\prime}\hat{\C}_{0}\right) ^{-1}\otimes \left( \hat{\R}_{0}^{\prime}\hat{\R}_{1,\perp}\left( \hat{\R}_{1,\perp}\right) ^{\prime}\hat{\R}_{0}\right) ^{-1}\right]   \label{fi-hat-estimator} \\
&&\times \left[ \left( \hat{\C}_{0}^{\prime}\hat{\C}_{1,\perp}\otimes \hat{\R}_{0}^{\prime}\hat{\R}_{1,\perp}\right) \left( \left( \hat{\C}_{1,\perp}\right)^{\prime}\otimes \left( \hat{\R}_{1,\perp}\right) ^{\prime}\right) \right] \Ve\X_{t}.  \notag
\end{eqnarray}
Let for short $p_{1\wedge 2}=\min \left\{ p_{1},p_{2}\right\}$.
\begin{theorem}
\label{f1-hat}We assume that Assumptions \ref{as-1}-\ref{as-5} are
satisfied. Then%
\begin{equation}
\left\Vert \hat{\F}_{0,t}-\left( \bH_{R_{0}}\right) ^{-1}%
\F_{0,t}\left( \bH_{C_{0}}^{\prime}\right) ^{-1}\right\Vert
_{F}=O_{P}\left( \frac{1}{\sqrt{p_{1}p_{2}}}\right) +O_{P}\left( \frac{1}{%
p_{1\wedge 2}T}\right) ,  \label{f1-hat-point}
\end{equation}%
where $\bH_{R_{0}}$ and $\bH_{C_{0}}$,are defined in Theorem %
\ref{rc1-hat}, and%
\begin{equation}
\frac{1}{T}\sum_{t=1}^{T}\left\Vert \hat{\F}_{0,t}-\left( 
\bH_{R_{0}}\right) ^{-1}\F_{0,t}\left( \bH%
_{C_{0}}^{\prime}\right) ^{-1}\right\Vert _{F}^{2}=O_{P}\left( \frac{1}{%
p_{1}p_{2}}\right) +O_{P}\left( \frac{1}{\left( p_{1\wedge 2}T\right) ^{2}}%
\right) .  \label{f1-hat-l2}
\end{equation}
\end{theorem}

Theorem \ref{f1-hat} states the consistency of the estimate of the space
spanned by the common stationary factors. The rates in the theorem can be
compared with Theorem 3.5(1) in \citet{he2023one}: the $O_{P}\left( \left(
p_{1}p_{2}\right)^{-1/2}\right) $ component is the same as in our case, and it can
be viewed as a non-improvable component of the estimator. Conversely, in %
\citet{he2023one} the $O_{P}\left( \left(p_{1\wedge 2}T\right)^{-1}\right) $ component is
replaced by an $O_{P}\left( \left(p_{1\wedge 2}T^{1/2}\right)^{-1}\right) $ term. In
our case, this difference arises from using $\hat{\C}_{1,\perp}$
and $\hat{\R}_{1,\perp}$.

\subsubsection{Projected estimation of $\C_{1}$, $\R_{1}$ and $\F_{1,t}$\label{projectCR}}
Consider now the ``filtered''\ data%
\begin{equation}
\proj{\X}_{t}=\X_{t}-\hat{\R}_{0}\hat{\F}_{0,t}\hat{\C}_{0}^{\prime},  \label{project-x}
\end{equation}%
and the corresponding projected covariance matrix%
\begin{equation*}
\proj{\M}_{R_{1}}=\frac{1}{p_{1}p_{2}^{2}T^{2}}\sum_{t=1}^{T}\proj{\X}_{t}\hat{\C}_{1}\hat{\C}_{1}^{\prime}\proj{\X}_{t}^{\prime}.
\end{equation*}%
Letting $\widetilde{\Lambda }_{R_{1}}$ be a $h_{R_{1}}\times h_{R_{1}}$ diagonal matrix
containing the largest $h_{R_{1}}$ eigenvalues of $\proj{\M}_{R_{1}}$, we
can define the estimator of $\R_{1}$ as the solution to the
eigenvalue/eigenvector problem 
\begin{equation}
\proj{\M}_{R_{1}}\widetilde{\R}_1=\widetilde{\R}_1\widetilde{\Lambda }_{R_{1}}.  \label{r-tilde}
\end{equation}%
We can define, analogously%
\begin{equation*}
\proj{\M}_{C_{1}}=\frac{1}{p_{1}^{2}p_{2}T^{2}}\sum_{t=1}^{T}\proj{\X}_t^{\prime}\hat{\R}_1\hat{\R}_1^{\prime}\proj{\X}_t,
\end{equation*}%
and subsequently obtain the projected estimator of $\C_{1}$\ as the
solution of 
\begin{equation}
\proj{\M}_{C_{1}}\widetilde{\C}_{1}=\widetilde{\C}_{1}\widetilde{\Lambda }_{C_{1}},  \label{c-tilde}
\end{equation}%
where $\widetilde{\Lambda }_{C_{1}}$\ is defined, similarly to $\widetilde{%
\Lambda }_{R_{1}}$, as a $h_{C_{1}}\times h_{C_{1}}$ diagonal matrix containing the
largest $h_{R_{1}}$ eigenvalues of $\proj{\M}_{R_{1}}$.

\begin{theorem}
\label{rc-tilde}We assume that Assumptions \ref{as-1}-\ref{as-5} are
satisfied. Then there exist: a $h_{R_{1}}\times h_{R_{1}}$ matrix $\widetilde{%
\bH}_{R_{1}}$, with $\left\Vert \widetilde{\bH}_{R_{1}}\right\Vert
_{F}=O_{P}(1) $ and $\left\Vert \left( \widetilde{\bH}%
_{R_{1}}\right) ^{-1}\right\Vert _{F}=O_{P}(1) $; and a $h_{C_{1}}\times
h_{C_{1}}$ matrix $\widetilde{\bH}_{C_{1}}$, with $\left\Vert \widetilde{%
\bH}_{C_{1}}\right\Vert _{F}=O_{P}(1) $ and $\left\Vert
\left( \widetilde{\bH}_{C_{1}}\right) ^{-1}\right\Vert _{F}=O_{P}\left(
1\right) $, such that%
\begin{align}
\left\Vert \widetilde{\R}_1-\R_1\widetilde{\bH}%
_{R_{1}}\right\Vert _{F} &=O_{P}\left( \frac{p_{1}^{1/2}}{p_{2}^{1/2}T}\right)
+O_{P}\left( \frac{p_{1}^{1/2}}{T^{2}}\right) +O_{P}\left( \frac{1}{T^{3/2}}%
\right) ,  \label{r-hat-tilde} \\
\left\Vert \widetilde{\C}_{1}-\C_{1}\widetilde{\bH}%
_{C_{1}}\right\Vert _{F} &=O_{P}\left( \frac{p_{2}^{1/2}}{p_{1}^{1/2}T}\right)
+O_{P}\left( \frac{p_{2}^{1/2}}{T^{2}}\right) +O_{P}\left( \frac{1}{T^{3/2}}%
\right) .  \label{c-hat-tilde}
\end{align}
\end{theorem}

Using the Least Squares principle, we can propose the following estimator of $\F_{1,t}$%
\begin{equation}
\widetilde{\F}_{1,t}=\frac{1}{p_{1}p_{2}}\widetilde{\R}_1%
^{\prime}\proj{\X}_t\widetilde{\C}_{1}.
\label{f-tilde-estimator}
\end{equation}

Let $\widetilde{\zeta }_{1}=\min \left\{ p_{1}^{1/2}p_{2}^{1/2},p_{1\wedge
2}^{1/2}T^{1/2},T^{3/2}\right\} $.

\begin{theorem}
\label{f-tilde}We assume that Assumptions \ref{as-1}-\ref{as-5} are
satisfied. Then%
\begin{equation}
\left\Vert \widetilde{\F}_{1,t}-\left( \widetilde{\bH}%
_{R_{1}}\right) ^{-1}\F_{1,t}\left( \widetilde{\bH}%
_{C_{1}}^{\prime}\right) ^{-1}\right\Vert _{F}  \label{f-tilde-point}
=O_{P}\left( \widetilde{\zeta }_{1}^{-1}\right)  , 
\notag
\end{equation}%
where $\widetilde{\bH}_{R_{1}}$ and $\widetilde{\bH}_{C_{1}}$%
,are defined in Theorem \ref{rc-tilde}, and%
\begin{equation}
\frac{1}{T}\sum_{t=1}^{T}\left\Vert \widetilde{\F}_{1,t}-\left( 
\widetilde{\bH}_{R_{1}}\right) ^{-1}\F_{1,t}\left( \widetilde{%
\bH}_{C_{1}}^{\prime}\right) ^{-1}\right\Vert _{F}^{2}
\label{f-tilde-l2} =O_{P}\left( \widetilde{\zeta }_{1}^{-2}\right) .  \notag
\end{equation}
\end{theorem}

\subsubsection{Projected estimation of $\R_{0}$, $\C_{0}$
and $\F_{0,t}$\label{proj-i0}}

Finally, it is possible to iterate the ``anti-projection''
approach to re-estimate $\R_{0}$, $\C_{0}$ and $\F%
_{0,t}$. Define 
\begin{equation*}
\widetilde{\X}_{t}^{C_1}=\X_{t}\widetilde{\C}_{1,\perp}^{s},\text{ and }\widetilde{\X}_{t}^{R_1}=\X_{t}^{\prime}%
\widetilde{\R}_{1,\perp}^{s},
\end{equation*}%
\begin{equation*}
\widetilde{\M}_{R_1,\perp}=\frac{1}{p_{1}p_{2}^{2}T}\sum_{t=1}^{T}%
\widetilde{\X}_{t}^{C_1}\left( \widetilde{\X}_{t}^{C_1}\right)
^{\prime},\text{ \ \ and \ \ }\widetilde{\M}_{C_1,\perp}=\frac{1%
}{p_{1}^{2}p_{2}T}\sum_{t=1}^{T}\widetilde{\X}_{t}^{R_1}\left( 
\widetilde{\X}_{t}^{R_1}\right) ^{\prime}.
\end{equation*}%
The estimator of $\R_{0}$ ($\C_{0}$) is defined as the
eigenvectors corresponding to the largest $h_{R_{0}}$ (resp. $h_{C_{0}}$)
eigenvalues of $\widetilde{\M}_{R_1,\perp}$ (resp. $\widetilde{\M}_{C_1,\perp}$), viz. 
\begin{equation*}
\widetilde{\M}_{R_1,\perp}\widetilde{\R}_{0}=\widetilde{%
\R}_{0}\widetilde{\Lambda }_{R_0},\text{ \ \ and \ \ }\widetilde{\M}_{C_1,\perp}\widetilde{\C}_{0}=\widetilde{\C}_{0}\widetilde{\Lambda }_{C_0}.
\end{equation*}%
where $\widetilde{\Lambda }_{R_0}$ is a $h_{R_{0}}\times h_{R_{0}}$ diagonal matrix
containing the largest $h_{R_{0}}$ eigenvalues of $\widetilde{\M}_{R_1,\perp}$, and $\widetilde{\Lambda }_{C_0}$\ is defined similarly,
under the constraints $\widetilde{\R}_{0}^{\prime}\widetilde{%
\R}_{0}=p_{1}\I_{h_{R_{0}}}$ and $\widetilde{\C}_{0}^{\prime}\widetilde{\C}_{0}=p_{2}\I_{h_{C_{0}}}$.

The next lemma shows that $\widetilde{\R}_{0}$ and $\widetilde{%
\C}_{0}$ do not improve with respect to $\hat{\R}_{0}$
and $\hat{\C}_{0}$ (at least, as far as the dominating terms are
concerned).

\begin{lemma}
\label{neg-rtilde}We assume that Assumptions \ref{as-1}-\ref{as-5} hold.
Then there exist a $h_{R_{0}}\times h_{R_{0}}$ matrix $\widetilde{\bH}_{R_0}$%
, with $\left\Vert \widetilde{\bH}_{R_0}\right\Vert _{F}=O_{P}\left(
1\right) $ and $\left\Vert \left( \widetilde{\bH}_{R_0}\right)
^{-1}\right\Vert _{F}$, and a $h_{C_{0}}\times h_{C_{0}}$ matrix $\widetilde{\bH}_{C_0}$, with $\left\Vert \widetilde{\bH}_{C_0}\right\Vert
_{F}=O_{P}(1) $ and $\left\Vert \left( \widetilde{\bH}%
_{C_0}\right) ^{-1}\right\Vert _{F}$ such that%
\begin{align}
\left\Vert \widetilde{\R}_{0}-\R_{0}\widetilde{\bH}%
_{R_0}\right\Vert _{F} &=O_{P}\left( \frac{p_{1}^{1/2}}{p_{1}p_{2}}\right)
+O_{P}\left( \frac{p_{1}^{1/2}}{p_{2}^{1/2}T^{1/2}}\right) ,
\label{r-tilde-neg} \\
\left\Vert \widetilde{\C}_{0}-\C_{0}\widetilde{\bH}%
_{C_0}\right\Vert _{F} &=O_{P}\left( \frac{p_{2}^{1/2}}{p_{1}p_{2}}\right)
+O_{P}\left( \frac{p_{2}^{1/2}}{p_{1}^{1/2}T^{1/2}}\right) .
\label{c-tilde-neg}
\end{align}
\end{lemma}

It is however possible to refine the rates of the estimated stationary
common factors $\F_{0,t}$. These are defined as%
\begin{eqnarray}
\Ve\widetilde{\F}_{0,t} &=&\left[ \left( \hat{\C}%
_{0}^{\prime}\widetilde{\C}_{1,\perp}^{s}\left( \widetilde{\C%
}_{1,\perp}^{s}\right) ^{\prime}\hat{\C}_{0}\right) ^{-1}\otimes
\left( \hat{\R}_{0}^{\prime}\widetilde{\R}%
_{1,\perp}^{s}\left( \widetilde{\R}_{1,\perp}^{s}\right) ^{\prime}%
\hat{\R}_{0}\right) ^{-1}\right]  \label{f1-tilde} \\
&&\times \left[ \left( \hat{\C}_{0}^{\prime}\widetilde{\C%
}_{1,\perp}^{s}\otimes \hat{\R}_{0}^{\prime}\widetilde{\R%
}_{1,\perp}^{s}\right) \left( \left( \widetilde{\C}%
_{1,\perp}^{s}\right) ^{\prime}\otimes \left( \widetilde{\R}%
_{1,\perp}^{s}\right) ^{\prime}\right) \right] \Ve\X_{t}.  \notag
\end{eqnarray}%
Note that we are using $\hat{\R}_{0}$ and $\hat{\C}%
_{0}$; in principle, it is possible to also use $\widetilde{\R}_{0}$
and $\widetilde{\C}_{0}$, but the results in Lemma \ref{neg-rtilde}
cast doubts over the effectiveness of such a choice. Let $\widetilde{\zeta }_{0}=\min \left\{ p_{1}^{1/2}p_{2}^{1/2},p_{1\wedge
2}T^{2},p_{1\wedge 2}^{2}T,p_{1\wedge 2}^{3/2}T^{3/2}\right\} $.

\begin{theorem}
\label{f1-tilde-rates}We assume that Assumptions \ref{as-1}-\ref{as-5} are
satisfied. Then%
\begin{equation*}
\left\Vert \widetilde{\F}_{0,t}-\left( \widetilde{\bH}%
_{R_{0}}\right) ^{-1}\F_{0,t}\left( \widetilde{\bH}%
_{C_{0}}^{\prime}\right) ^{-1}\right\Vert _{F} =O_{P}\left( \widetilde{\zeta }_{0}^{-1}\right) ,
\end{equation*}%
where $\widetilde{\bH}_{R_{0}}$ and $\widetilde{\bH}_{C_{0}}$%
,are defined in Lemma \ref{neg-rtilde}, and%
\begin{equation*}
\frac{1}{T}\sum_{t=1}^{T}\left\Vert \widetilde{\F}_{0,t}-\left( 
\widetilde{\bH}_{R_{0}}\right) ^{-1}\F_{0,t}\left( \widetilde{%
\bH}_{C_{0}}^{\prime}\right) ^{-1}\right\Vert _{F}^{2}=O_{P}\left( \widetilde{\zeta }_{0}^{-2}\right) .
\end{equation*}
\end{theorem}

The rates can be compared with those in Theorem \ref{f1-hat}: the
non-improvable rate $O_{P}\left( (p_{1}p_{2})^{-1/2}\right) $ is
still present; however, the estimator-dependent rate has been refined. 

\subsection{Estimation of the number of common factors\label{number}}

In the above, we have (implicitly) assumed that the number of common stationary and nonstationary factors $h_{R_{1}}$, $h_{C_{1}}$, $h_{R_{0}}$, and $h_{C_{0}}$ are known. In practice, this is seldom the case, and an estimate of $h_{R_{1}}$, $h_{C_{1}}$, $h_{R_{0}}$, and $h_{C_{0}}$ is required as the preliminary
step in order to use our methodology. In this section, we discuss this issue, proposing a family of consistent estimators for the numbers of common factors.

Our first result shows that, as long as $h_{R_{1}}$, $h_{C_{1}}$, $h_{R_{0}}$, and $h_{C_{0}} $ are estimated consistently, all the theory derived above still holds. Let $\breve{h}_{R_{1}}$, $\breve{h}_{C_{1}}$, $\breve{h}_{R_{1}}$, and $\breve{h}_{C_{1}}$ denote such estimators.

\begin{lemma}
\label{bai03}We assume that Assumptions \ref{as-1}-\ref{as-5} are satisfied. Then, if $\breve{h}_{R_{1}}=h_{R_{1}}+o_{P}(1) $, $\breve{h}_{C_{1}}=h_{C_{1}}+o_{P}(1) $, $\breve{h}_{R_{1}}=h_{R_{0}}+o_{P}(1)$, and $\breve{h}_{C_{1}}=h_{C_{0}}+o_{P}(1) $, Theorems \ref%
{hat-estimates}-\ref{f-tilde},\ and Lemmas \ref{f-hat-negative}-\ref%
{neg-rtilde} still hold.
\end{lemma}
\par
Several possible estimators can be proposed for $h_{R_{1}}$, $h_{C_{1}}$, $h_{R_{0}}$, and $h_{C_{0}}$: our results in Lemmas \ref{eig-mRx}, \ref{eig-mCx}, \ref{eig-orth-r}, \ref{eig-tilde-c}, \ref{spec-m-r-tilde} and \ref{spec-m-c-tilde} lend themselves to extending, to the matrix-valued time series context, both the information criteria proposed in \citet{bai04} and the sequential randomised tests proposed in \citet{bt2}. Other methodologies, specifically developed for the case of stationary matrix- or tensor-valued time series, could be also extended to our context - whilst an exhaustive treatment goes beyond the scope of this section, we refer to the article by \citet{he2023one} for a comprehensive review of the state of the art on this important issue. Here, we propose a methodology based on the eigenvalue ratio (ER) principle (see \citealp{lam2012factor}, and \citealp{ahn2013eigenvalue}). We introduce the following estimators for $h_{R_{1}}$ and $h_{C_{1}}$:%
\begin{equation}
\hat{h}_{R_{1}}=\mathop{\mathrm{arg\ max}}_{0\leq j\leq h_{\max }}\frac{\lambda _{j}\left( \M_{R_1}\right) }{\lambda _{j+1}\left( \M_{R_1}\right) +\hat{c}_{R_{1}}\delta_{R_1,p_{1},p_{2},T}}\text{ \ \
and \ \ }\hat{h}_{C_{1}}=\mathop{\mathrm{arg\ max}}_{0\leq j\leq h_{\max }} \frac{\lambda _{j}\left( \M_{C_1}\right) }{\lambda _{j+1}\left(\M_{C_1}\right) +\hat{c}_{C_{1}}\delta_{C_1,p_{1},p_{2},T}},
\label{k-hat}
\end{equation}%
which are based on the ``flattened'' covariance matrices $\M_{R_1}$ and $\M_{C_1}$ respectively - an alternative to $\hat{h}_{R_{1}}$ and $\hat{h}_{C_{1}}$ can also be based on the eigenvalues of $\hat{\M}_{R_{1}}^{\dagger}$ and $\hat{\M}_{C_{1}}^{\dagger}$ respectively, and in this case we
use the notation $\hat{h}_{R_{1}}^{\dagger}$ and $\hat{h}_{C_{1}}^{\dagger}
$; and%
\begin{equation}
\widetilde{h}_{R_{1}}=\mathop{\mathrm{arg\ max}}_{0\leq j\leq h_{\max }}\frac{\lambda _{j}\left( \proj{\M}_{R_{1}}\right) }{\lambda _{j+1}\left(\proj{\M}_{R_{1}}\right) +\widetilde{c}_{R_{1}}\delta_{R_1\diamond,p_{1},p_{2},T}} \text{ \ \ and \ \ }\widetilde{h}_{C_{1}}= \mathop{\mathrm{arg\ max}}_{0\leq j\leq h_{\max }}\frac{\lambda _{j}\left( 
\proj{\M}_{C_{1}}\right) }{\lambda _{j+1}\left( \proj{\M}_{C_{1}}\right) +\widetilde{c}_{C_{1}}\delta_{C_1\diamond,p_{1},p_{2},T}},
\label{k-tilde}
\end{equation}%
which are based on the projected\ covariance matrices $\proj{\M}_{R_{1}}$
and $\proj{\M}_{C_{1}}$ respectively. It can be envisaged that, given
that the eigen-gap is wider in the case of $\proj{\M}_{R_{1}}$ and $%
\proj{\M}_{C_{1}}$ as opposed to $\M_{R_1}$ and $\M_{C_1}$ (and $\hat{\M}_{R_{1}}^{\dagger}$ and $\hat{\M}_{C_{1}}^{\dagger}$), $\widetilde{h}_{R_{1}}$ and $\widetilde{h}_{C_{1}}$ may offer a
better performance than $\hat{h}_{R_{1}}$ and $\hat{h}_{C_{1}}$ (and $%
\hat{h}_{R_{1}}^{\dagger}$ and $\hat{h}_{C_{1}}^{\dagger}$); we explore
this in simulations. In both (\ref{k-hat}) and (\ref{k-tilde}), the
``mock''\ eigenvalues $\lambda _{0}\left( 
\mathbf{\cdot }\right) $ are designed to as to pick up the case of no common
factor; following \citet{ahn2013eigenvalue}, we construct these as $\lambda _{0}\left( \M_{\iota }\right) =\omega_{p_{1},p_{2},T}^{\iota}$, with $\iota \in \left\{R_1, C_1,R_{1\diamond},C_{1\diamond}\right\}$, with the convention that $\M_{R_{1\diamond}}= \proj{\M}_{R_1}$, $\M_{C_{1\diamond}}= \proj{\M}_{C_1}$, and $\omega _{p_{1},p_{2},T}^{\iota }$ is a
sequence such that, as $\min \left\{ p_{1},p_{2},T\right\} \rightarrow\infty $, $\omega _{p_{1},p_{2},T}^{\iota }\rightarrow 0\text{ \ \ and \ \ }\left(\delta _{\iota ,p_{1},p_{2},T}\right) ^{-1}\omega_{p_{1},p_{2},T}^{\iota }\rightarrow \infty $. Further, in both (\ref{k-hat}) and (\ref{k-tilde}), $h_{\max }$ is a user-chosen upper bound such that $h_{\max }<\min \left\{ p_{1},p_{2},T\right\} $, and we use the following sequences
\begin{align*}
\delta _{R_1,p_{1},p_{2},T} &=\delta _{C_1,p_{1},p_{2},T}=\frac{1}{T},\\
\delta _{R_1\diamond,p_{1},p_{2},T} &=\frac{1}{p_{1\wedge 2}^{1/2}T^{3/2}}+ \frac{1}{p_{2}T}+\frac{1}{T^{2}}+\frac{1}{p_{1}^{1/2}p_{2}^{1/2}T}, \\
\delta _{C_1\diamond,p_{1},p_{2},T} &=\frac{1}{p_{1\wedge 2}^{1/2}T^{3/2}}+ \frac{1}{p_{1}T}+\frac{1}{T^{2}}+\frac{1}{p_{1}^{1/2}p_{2}^{1/2}T}.
\end{align*}

Similarly, we propose the following estimators for $h_{R_{0}}$ and $h_{C_{0}}$: 
\begin{equation}
\hat{h}_{R_{0}}=\mathop{\mathrm{arg\ max}}_{0\leq j\leq h_{\max }}\frac{\lambda _{j}\left( \M_{R_1,\perp}\right) }{\lambda _{j+1}\left(\M_{R_1,\perp}\right) +c_{R_{0}}\delta _{R_0,p_{1},p_{2},T}} \text{\; and \;}\hat{h}_{C_{0}}=\mathop{\mathrm{arg\ max}}_{0\leq j\leq h_{\max }}\frac{\lambda _{j}\left( \M_{C_1,\perp}\right) }{\lambda _{j+1}\left( \M_{C_1,\perp}\right) +c_{C_{0}}\delta_{C_0,p_{1},p_{2},T}},\label{h-hat}
\end{equation}%
which are based on the ``anti-projected'' covariance matrices $\M_{R_1,\perp}$ and $\M_{C_1,\perp}$. Even in this case, $h_{\max }$ is a user-chosen upper bound such that $h_{\max }<\min \left\{ p_{1},p_{2},T\right\} $, and
\begin{equation*}
\delta _{R_0,p_{1},p_{2},T}=\frac{1}{p_{2}^{1/2}T^{1/2}}+\frac{1}{p_{1}p_{2}}\text{ \ \ and \ \ }\delta _{C_0,p_{1},p_{2},T}=\frac{1}{p_{1}^{1/2}T^{1/2}}+\frac{1}{p_{1}p_{2}}.
\end{equation*}%
As above, the mock eigenvalues are defined as $\lambda _{j}\left( \M_{\iota ,\perp}\right) =\omega_{p_{1},p_{2},T}^{\iota} $, with, as $\min \left\{ p_{1},p_{2},T\right\} \rightarrow \infty $, $\omega _{p_{1},p_{2},T}^{\iota }\rightarrow 0$, and $\left(\delta _{\iota ,p_{1},p_{2},T}\right) ^{-1}\omega_{p_{1},p_{2},T}^{\iota }\rightarrow \infty$, for $\iota \in \left\{ \left( R_0\right) ,\left( C_0\right) \right\} $. The constants $\hat{c}_{R_{1}}$, $\hat{c}_{C_{1}}$, $\widetilde{c}_{R_{1}}$, $\widetilde{c}_{C_{1}}$, $c_{R_{0}}$ and $c_{C_{0}}$ can e.g. be chosen adaptively, using different subsamples and choosing the values of the constants which offer stable estimates across such subsamples,
in a similar spirit to \citet{hallinliska07} and \citet{ABC10}.

\begin{theorem}
\label{er}We assume that Assumptions \ref{as-1}-\ref{as-5} are satisfied. Then, as $\min \left\{ p_{1},p_{2},T\right\} \rightarrow \infty $, it holds that all the estimators defined in (\ref{k-hat}), (\ref{k-tilde}) and (\ref{h-hat}) are consistent - i.e., $\hat{h}_{R_{1}}=h_{R_{1}}+o_{P}(1) $, and so on.
\end{theorem}

\section{Monte Carlo evidence\label{simulation}}
In this Section we show the results of a series of Monte Carlo studies to showcase the performance of our methodology. Section~\ref{sec:MCestimation} contains some key results for the estimators of the factor loadings and a comparison of their convergence rates; Section~\ref{sec:MCnfact} is devoted to the estimation of the number of factors. The full set of detailed results can be found in Appendix~\ref{app:MC}.
\par
We simulate from the following Data Generating Process (DGP)
\begin{equation}\label{modelsim}
\underset{\blue{p_{1}\times p_{2}}}{\X_{t}} = \underset{\blue{p_{1}\times h_{R_1}}}{\R_1} \underset{\orange{h_{R_1}\times h_{C_1}}}{\F_{1,t}}\underset{\blue{h_{C_1}\times p_{2}}}{\C_1^{\prime}}+ \underset{\blue{p_{1}\times h_{R_0}}}{\R_{0}} \underset{\orange{h_{R_0}\times h_{C_0}}}{\F_{0,t}}\underset{\blue{h_{C_0}\times p_{2}}}{\C_{0}^{\prime}}+\underset{\blue{p_{1}\times p_{2}}}{\E_{t}}, \qquad t=1,\dots,T
\end{equation}%
where:
\begin{itemize}
  \item[$\bullet$] $\F_{1,t} = \F_{1,t-1} + \beps_{t}$ with $\beps_{t}\sim N(\mathbf{0},\mathrm{I}_{h_{R_1} h_{C_1}})$;
  \item[$\bullet$] $\Ve(\F_{0,t}) \sim N(\mathbf{0},\sigma_0\mathrm{I}_{h_{R_0} h_{C_0}})$;
  \item[$\bullet$] $\R_0,\C_0 \sim \mathcal{U}[-a_0,a_0]$, \quad  $\R_1,\C_1 \sim \mathcal{U}[-a_1,a_1]$
\end{itemize}

\subsection{Estimation and convergence rates}\label{sec:MCestimation}
We explore the cases presented in Table~\ref{tab:mcpar},
\begingroup
\begin{table}
\renewcommand{\arraystretch}{1}
\begin{tabular}{crrrrrrr}
 Case & $h_{R_0}$ & $h_{C_0}$ & $h_{R_1}$ & $h_{C_1}$ & $a_0$ & $a_1$ & $\sigma_0$ \\ 
  \cmidrule{2-8} 
  1.1 & 1 & 1 & 1 & 1 & 1 & 1 & 1 \\ 
  1.2 & 1 & 1 & 1 & 1 & 1 & 1 & 2 \\ 
  2.1 & 1 & 1 & 2 & 2 & 1 & 1 & 1 \\ 
  2.2 & 1 & 1 & 2 & 2 & 1 & 1 & 2 \\ 
  3.1 & 2 & 2 & 1 & 1 & 1 & 1 & 1 \\ 
  3.2 & 2 & 2 & 1 & 1 & 1 & 1 & 2 \\ 
  4.1 & 1 & 1 & 1 & 1 & 10 & 10 & 1 \\ 
  4.2 & 1 & 1 & 1 & 1 & 10 & 10 & 2 \\ 
  \cmidrule{2-8} 
\end{tabular}
\caption{Parameters' combination for the simulation study on the rates of convergence.}\label{tab:mcpar}
\end{table}
\endgroup
and we combine them with, with $p_1 = 10,  20,  50, 100$, $p_2 = 20$, $T = 20,  50, 100, 200$. For each parameters' combination we consider the average over 1000 Monte Carlo replications. Due to identification indeterminacy, we measure the performance of the estimators using the distance between subspaces. Given two orthogonal matrices $O_1$ and $O_2$ of sizes $p \times q_1$ and $p \times q_2$, define
\begin{equation}\label{DS}
\D(O_1,O_2) = \left(1 - \frac{1}{\max(q_1,q_2)}\tr\left(O_1 O_1^\prime O_2 O_2^\prime\right)  \right)^{\frac{1}{2}}.
\end{equation}
$\D(O_1,O_2)$ ranges between 0 and 1. It is equal to 0 if the column spaces of $O_1$ and $O_2$ are the same, and 1 if they are orthogonal.
\par
Figure~\ref{fig:mc2box1T} shows the boxplots of the ratio $\D_{\text{flat}}/\D_{\text{proj}}$ between the initial flattened and the refined projected estimators for $\R_1$ (left) and $\C_1$ (right) against series' length $T$. Each boxplot contains 32 values for a specific $T$, corresponding to the 8 cases (1.1 - 4.2) times the 4 values of $p_1$. In turn, each of the 32 values is the average over 1000 Monte Carlo replications.  Clearly, the refined projected estimator improves uniformly over the initial ``flattened'' estimator and the gain increases with the length of the series, reaching a median ratio of about $2.4$ for $T=200$. The ratio reaches 3.2 in some instances. The boxplots of the ratio $\D_{\text{flat}}/\D_{\text{proj}}$ against row-dimension $p_1$ is reported in Figure~\ref{fig:mc2box1p} in Appendix~\ref{app:MC} and the results are also consistent with the theoretical convergence rates in that, for fixed $p_2$ and $T$, as $p_1$ increases, the ratio should decrease for $\R_1$ and increase for $\C_1$.  
\begin{figure}
    \centering
    \includegraphics[width=0.45\linewidth]{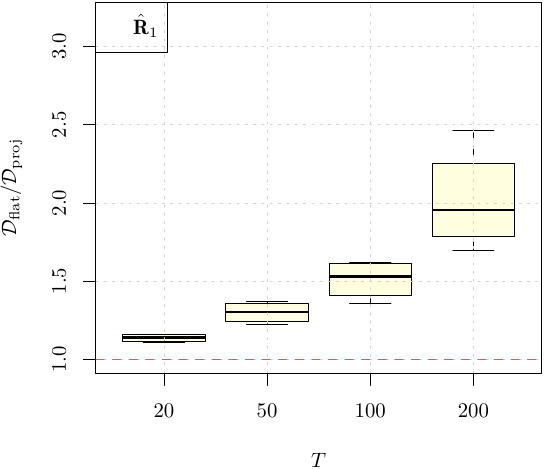}
    \includegraphics[width=0.45\linewidth]{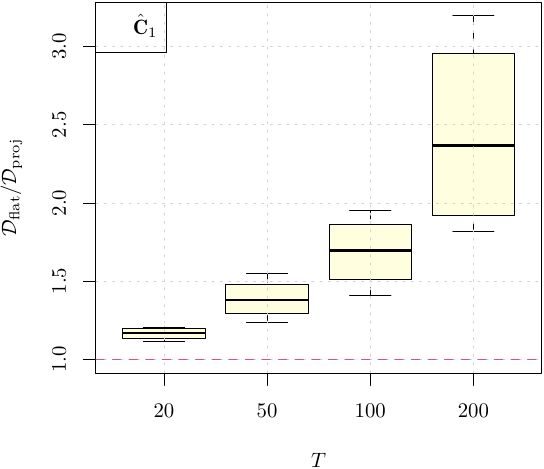}
    \caption{Boxplots of the ratio $\D_{\text{flat}}/\D_{\text{proj}}$ between the initial flattened  and the refined projected estimators for $\R_1$ (left) and $\C_1$ (right) against series' length $T$.}
    \label{fig:mc2box1T}
\end{figure}
A different, non trivial, behaviour is expected for $\R_0$ and $\C_0$ since the theoretical convergence rates of the initial and refined estimators are the same, see Figure~\ref{fig:mc2box0p} where the boxplots of the ratio are plotted against $p_1$. The plot against $T$ can be found in Figure~\ref{fig:mc2box0T} of the Appendix. In any case, also in this instance, the refined estimators improve uniformly over the initial ones, albeit by a tighter margin with respect to those for $\R_1$ and $\C_1$.  
\begin{figure}
    \centering
    \includegraphics[width=0.45\linewidth]{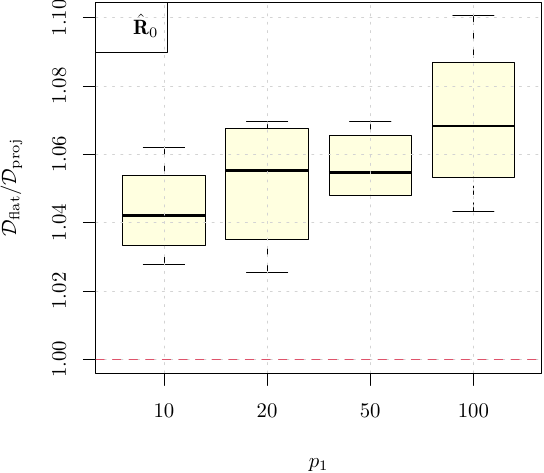}
    \includegraphics[width=0.45\linewidth]{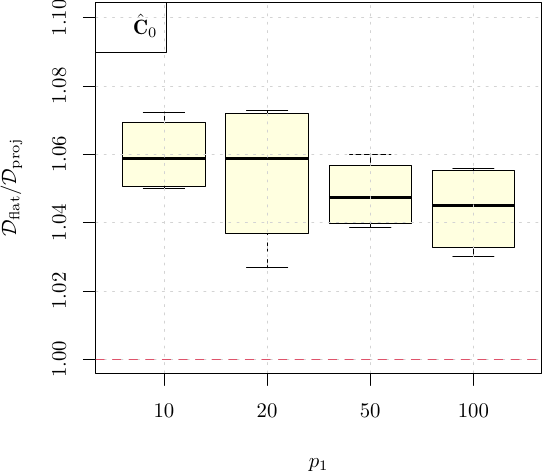}
    \caption{Boxplots of the ratio $\D_{\text{flat}}/\D_{\text{proj}}$ between the initial flattened  and the refined projected estimators for $\R_0$ (left) and $\C_0$ (right) against $p_1$.}
    \label{fig:mc2box0p}
\end{figure}

\begin{figure}
    \centering
\includegraphics[width=0.45\linewidth,keepaspectratio]{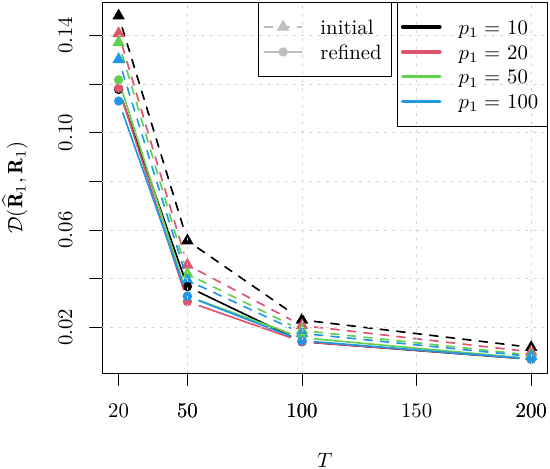} 
\includegraphics[width=0.45\linewidth,keepaspectratio]{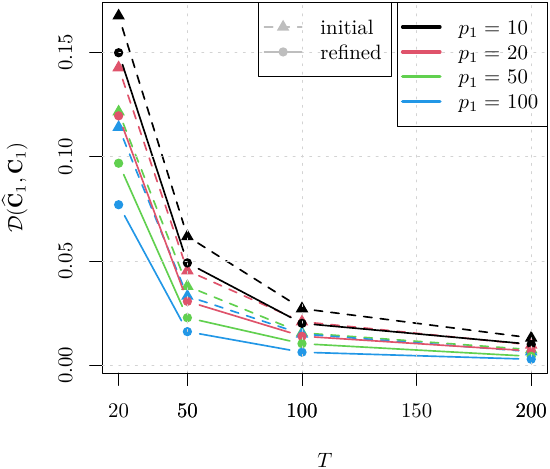} 
    \caption{Case 1.1: estimation of $\R_1$ (left) and $\C_1$ (right) for varying series length $T$ and row dimension $p_1$. Also, $p_2=20$, whereas $\R_0,\C_0 \sim \mathcal{U}[-1,1]$, \quad  $\R_1,\C_1 \sim \mathcal{U}[-1,1]$. Triangles with dashed lines indicate the initial ``flattened'' estimator, circles with full lines indicate the refined projected estimator.}\label{fig:1}
\end{figure}
The detailed results for Case 1.1 are shown in Figure~\ref{fig:1}, where the distances $\D(\hat\R_1,\R_1)$ (left) and $\D(\hat\C_1,\C_1)$ (right) are plotted against the length of the series $T$ for different values of the row-dimension $p_1$. Clearly, the refined projected estimator (circles) is always superior to the initial ``flattened'' estimator (triangles) and the rate also depends on $p_1$. This latter finding can also be appreciated from Figure~\ref{fig:2}, which shows the results of the estimation of $\R_0$ (left) and $\C_0$ (right). The full set of detailed results can be found in the Appendix~\ref{app:MC}.
\begin{figure}
    \centering
\includegraphics[width=0.45\linewidth,keepaspectratio]{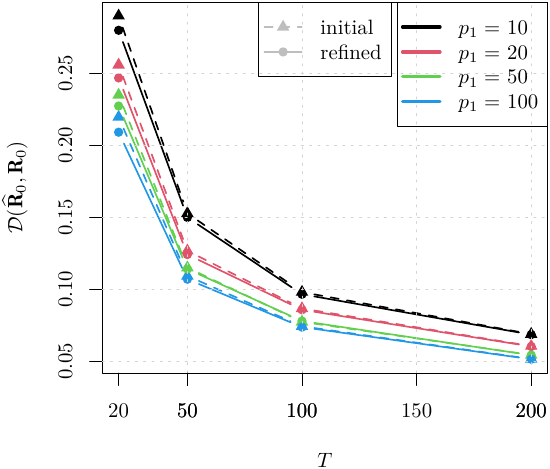} 
\includegraphics[width=0.45\linewidth,keepaspectratio]{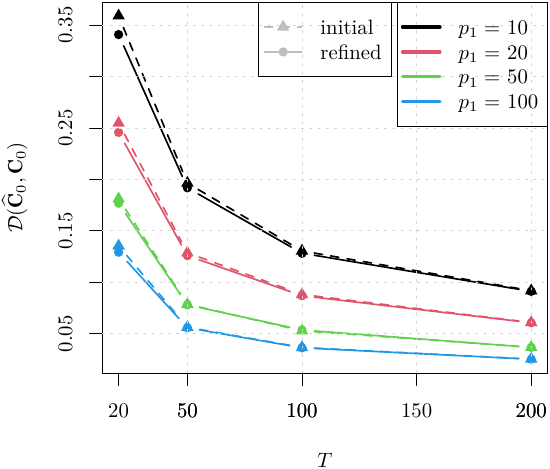} 
    \caption{Case 1.1: as Figure~\ref{fig:1} but for $\R_0$ and $\C_0$.}\label{fig:2}
\end{figure}
\subsection{Estimating the number of factors}\label{sec:MCnfact}

We simulate from the DGP of Eq.~(\ref{modelsim}), with

\begin{itemize}
  \item[$\bullet$] $\F_{1,t} = \F_{1,t-1} + \boldsymbol{\eps}_{t}$ with $\boldsymbol{\eps}_{t}\sim N(\mathbf{0},\mathrm{I}_{h_{R_1} h_{C_1}})$;
  \item[$\bullet$] $\Ve(\F_{0,t}) \sim N(\mathbf{0},\mathrm{I}_{h_{R_0} h_{C_0}})$;
  \item[$\bullet$] $\R_0,\C_0 \sim \mathcal{U}[-1,1]$, \quad  $\R_1,\C_1 \sim \mathcal{U}[-1,1]$
\end{itemize}
We explore the following cases

\begingroup
\begin{center}
\renewcommand{\arraystretch}{0.6}
\begin{tabular}{rrrrr}
Case & $h_{R_0}$ & $h_{C_0}$ & $h_{R_1}$ & $h_{C_1}$ \\
\cmidrule{2-5}
1 & 1 & 1 & 1 & 1\\
2 & 2 & 2 & 2 & 2\\
3 & 1 & 1 & 2 & 2\\
4 & 1 & 1 & 3 & 3\\
5 & 2 & 2 & 1 & 1\\
6 & 3 & 3 & 1 & 1\\
7 & 2 & 1 & 1 & 1\\
8 & 1 & 1 & 2 & 1\\
\cmidrule{2-5}
\end{tabular}
\end{center}
\endgroup
and, as in the previous section, we combine the above with $p_1 = 10,  20,  50, 100$, $p_2 = 20$, $T = 20,  50, 100, 200$. We focus on the frequency of correct identification of $h_{R_0}$, $h_{C_0}$, $h_{R_1}$, $h_{C_1}$, based upon 1000 Monte Carlo replications. Here we report the results of different implementations of the Eigenvalue Ratio criterion:
\begin{description}
   \item[static]: $h_{R_0}$, $h_{C_0}$, $h_{R_1}$, $h_{C_1}$ are estimated once in the procedure. 
   \item[it0] starts from the above static estimates. Then, it uses the estimated number of $I(1)$ factors as starting values for the procedure, which stops when either the final estimated numbers of $I(1)$ factors coincide with the initial ones, or the maximum number of iterations is reached.
   \item[it1] starts from the \textbf{it0} estimates. If the initial estimate is not a fixed point i.e. the initial and refined estimates of the $I(1)$ parameters are the same at the first iteration, then, tries to refine as follows.
   \begin{enumerate}
     \item computes the static estimates on a grid of initial values for $h_{R_1}$, $h_{C_1}$;
     \item derives the graph of the combinations from the grid and retains the fixed point of the graph as candidates;
     \item if there is at least one candidate, updates the initial \textbf{it0} estimates if either the max number of iterations is reached or there is a new candidate/parameter combination with max average ER value and max cluster size. 
   \end{enumerate}
     \item[it2] starts from the \textbf{it0} estimates and keeps the values of $h_{R_0}$, $h_{C_0}$. If the initial estimate is not a fixed point i.e. the initial and refined estimates of the $I(1)$ parameters are the same at the first iteration, then, tries to refine the $I(1)$ parameters as follows.
   \begin{enumerate}
     \item computes the static estimates on a grid of initial values for $h_{R_1}$, $h_{C_1}$.
     \item derives the graph of the combinations from the grid and retains the fixed point of the graph as candidates.
     \item if there is at least one candidate,  updates the initial \textbf{it0} estimates by choosing the values of $h_{R_1}$, $h_{C_1}$ individually as the maximizers of the ER value among the parameters' combinations. 
   \end{enumerate}
\end{description}
In practice, criteria \textbf{it1} and \textbf{it2} differ only in step (3) when it comes to updating the initial estimates based upon \textbf{it0}. To showcase the advantages of refined iterative procedures over the simple iterative estimator \textbf{it0} we simulate a series from with the following parameters: 
\begin{align*}
  h_{R_0} = 1;\quad h_{C_0} = 1;\quad h_{R_1} = 2;\quad h_{C_1} = 1; \qquad\text{(case 8)}\\
  p_1 = 100;\quad p_2 = 20;\quad T = 50;
\end{align*}
Both the static and the simple iterative estimators for the number of factors incorrectly estimate  $h_{R_0} = h_{C_0} = h_{R_1} = h_{C_1} = 1$. The reason for this can be appreciated in Figure~\ref{fig:graph} where we show the associated graph where the nodes are the grid of initial values for $h_{R_1}$, $h_{C_1}$ and the arrows show the node connecting the initial and the refined estimate. There are two clusters of connected nodes but only the right hand side cluster has a fixed point (node 2), which corresponds to the true parameters' value. Hence, since the static estimator starts from node 1, the simple iterative estimator \textbf{it0} does not converge and falls back to the static solution. This is why, when a fixed point is not reached, estimators \textbf{it1-it2} try to refine over the initial \textbf{it0} estimate by looking for the fixed points of the graph (if any) and selecting the solution according to (slightly) different criteria, as explained above. This improves the initial iterative estimate \textbf{it0} in case of lack of convergence and solves the problem of the dependence on initial conditions.    
\begin{figure}
  \centering
    \begin{minipage}{.3\linewidth}
    \begingroup
  \renewcommand{\arraystretch}{0.9}
\begin{tabular}[t]{ccc}
\toprule
node & $h_{R_1}$ & $h_{C_1}$\\
\midrule
1 & 1 & 1\\
2 & 2 & 1\\
3 & 3 & 1\\
4 & 4 & 1\\
5 & 1 & 2\\
6 & 2 & 2\\
7 & 3 & 2\\
8 & 4 & 2\\
9 & 1 & 3\\
10 & 2 & 3\\
11 & 3 & 3\\
12 & 4 & 3\\
13 & 1 & 4\\
14 & 2 & 4\\
15 & 3 & 4\\
16 & 4 & 4\\
\bottomrule
\end{tabular}
\endgroup
\end{minipage}
\begin{minipage}{.6\linewidth}
\centering 
\includegraphics[width=0.9\linewidth,keepaspectratio]{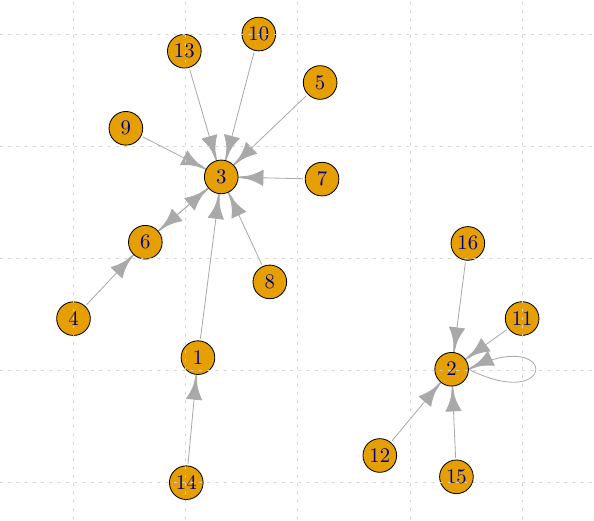} 
\end{minipage}
  \caption{Graph associated to the iterative estimators for the number of factors. Each node corresponds to a combination of initial values for $h_{R_1}$, $h_{C_1}$ reported in the table. Node 2 is a fixed point and corresponds to the true parameters' value. Procedures starting from nodes belonging to the cluster without a fixed point (left) do not converge. In such a case, estimators \textbf{it1-it2} correctly identify the fixed point (node 2) as the candidate solution.}\label{fig:graph}
\end{figure}
\noindent
The results of the Monte Carlo exercise are reported in Figure~\ref{fig:nf1t}, which contains the boxplots of the percentages of correct selection for each criterion and for each loadings matrix. For the sake of presentation, each boxplot aggregates 128 values (8 cases $\times$ 4 values of $p1$ $\times$ 4 values of $T$). The full set of results, stratified by $p_1$ and $T$ are available in Figure~\ref{fig:nf1b} of the Appendix~\ref{app:MC}. The iterative estimators improve noticeably over the static estimator, especially for $\R_1$ and $\C_1$. This is best appreciated in Figure~\ref{fig:nf2t2}, which shows the boxplots of the differences of percentages of correct estimation of the number of factors for the iterative criteria w.r.t. the static criterion.
\begin{figure}
    \centering
    \includegraphics[width=0.9\linewidth]{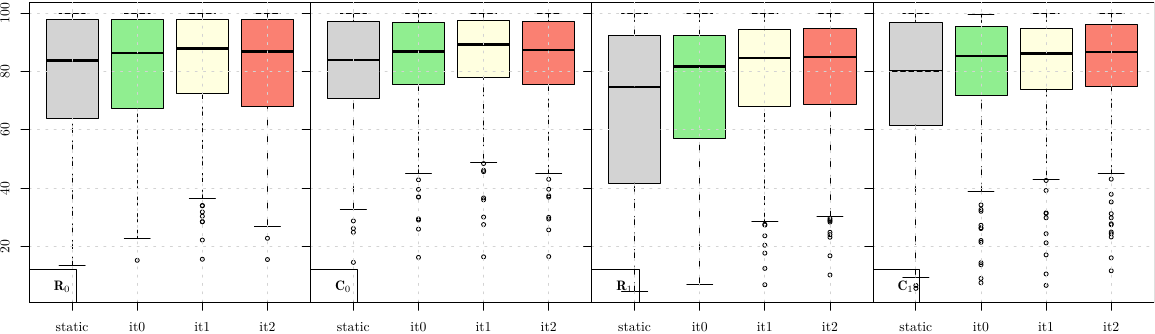}
\caption{Boxplots of the percentages of correct estimation of the number of factors for the 4 criteria for different parameters. Each boxplot contains 128 percentages:  8 cases $\times$ 4 values for $p_1$ $\times$ 4 values for $T$. The extended results, stratified by $p_1$ and $T$ are available in Figure~\ref{fig:nf1b} of the Appendix~\ref{app:MC}. }\label{fig:nf1t}
\end{figure}
\begin{figure}
    \centering
    \includegraphics[width=0.9\linewidth]{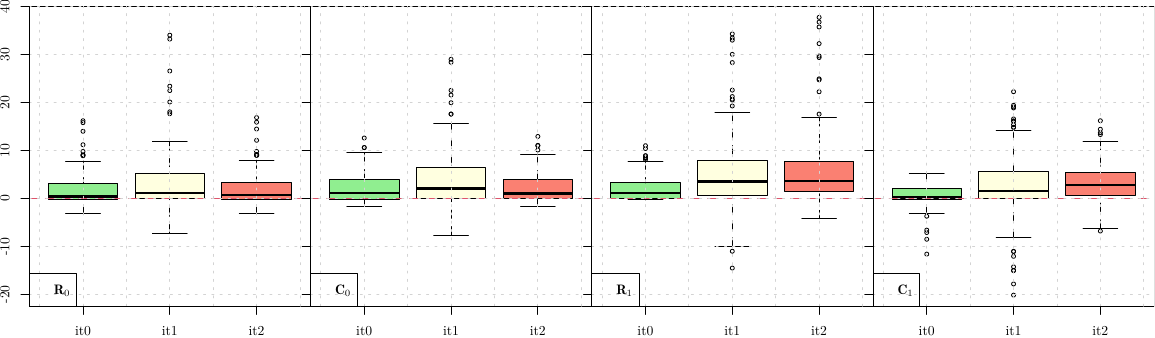}
\caption{Boxplots of the differences of percentages of correct estimation of the number of factors for the iterative criteria w.r.t. the static criterion. Positive values indicate that the iterative version is superior w.r.t. the static one. Each boxplot contains 128 percentages:  8 cases $\times$ 4 values for $p_1$ $\times$ 4 values for $T$. The extended results, stratified by $p_1$ and $T$ are available in Figure~\ref{fig:nf2b} of the Appendix~\ref{app:MC}.}\label{fig:nf2t2}
\end{figure}
Clearly, the iterative criteria can improve over the static estimation by 40\%. In particular, for $\R_0$ and $\C_0$, \textbf{it1} seems to achieve the largest gain, even if it is also prone to losing power, especially for $\R_1$ and $\C_1$. This is because it tends to overturn more often the initial estimator \textbf{it0}. In turn, \textbf{it2} is closer to \textbf{it0} for for $\R_0$ and $\C_0$ but shows a consistent gain for $\R_1$ and $\C_1$, so that it remains our recommended choice to date. As above, the extended results, stratified by $p_1$ and $T$ are available in Figure~\ref{fig:nf2b} of the Appendix~\ref{app:MC}.

\section{Conclusions\label{conclusion}}
In this paper, we study inference in the context of a factor model for a high-dimensional matrix-valued time series with the possible presence of
common stochastic trends and common stationary factors. The inferential problem is not a simple extension of existing techniques: the presence of
the common stationary factors makes it impossible to refine the rates of the estimators of the non-stationary common factor structure using e.g. the iterated projection-based estimator of \citet{hkyz2021}. Hence, we propose an entirely novel procedure, based on a preliminary step where the common stochastic trends are eliminated by projecting them away onto a large dimensional space constructed as the orthogonal complement to the loadings space of the $I(1) $ common factor structure. As mentioned in the introduction, this step goes in the opposite direction to the projection-based estimator, and we view it as an ``anti-projection'' (or an ``anti-Johnson-Lindenstrauss'') argument. After getting rid of the common $I(1)$ components, we estimate the common stationary component in a standard way, and after removing it from the data, we are able to use the full force of the iterative projection-based estimator. Our results, in terms of the rates of convergence of the estimated loadings and common factors, show that we are able to refine the rates of convergence of the estimators of the $I(1)$ factor structure. As a by-product, we also propose a technique to estimate the number of common factors in both the $I(1)$ and the stationary structures; further, building on the results on the spectrum of the second moment matrices studied in this paper, it would also be possible to propose several tests, along e.g. the lines of \citet{bt2}, for the null hypothesis that the data are $I(1)$, which would serve as a preliminary step to ascertain whether our estimation technique needs to be applied (i.e., whether $\X_{t}$ is indeed $I(1)$), or not.

Several interesting questions and possible extensions are still outstanding. In addition to deriving the limiting distributions of the estimated
loadings, common factors and common components (which, in principle, can be done as an extension of our results), and to extending our results to the presence of deterministic components (such as drifts or linear trends in the common $I(1)$ structure), here we revisit the notion of
cointegration and how this can be cast into our model (\ref{model}). In particular, we note that our current model and theory can only partly embed
a (Matrix) Error Correction Model (\citealp{johansen1991estimation}). Indeed, consider again the two-step hierarchical representation discussed in
the introduction, and, unless stated otherwise, assume for simplicity that $h_{R_{1}}=h_{C_{1}}=1$. Recalling that $\X_{\cdot j,t}$, $1\leq
j\leq p_{2}$ denotes the $j$-th column of $\X_{t}$, consider the MA$(\infty)$ representation $\Delta \X_{\cdot j,t} = \mathbf{\Gamma }_{j}\left( L\right) \beps_{\cdot j,t}$, where $\mathbf{\Gamma }_{j}\left( L\right) =\sum_{h=0}^{\infty }\mathbf{\Gamma }_{j,h}L^{h}$ is a $p_{1}\times p_{1}$-valued MA$(\infty)$ polynomial. Then, by standard arguments,\footnote{See e.g. \citet{watson1994vector}} we can represent $\X_{\cdot j,t}$ as
\begin{equation}
\X_{\cdot j,t}=\mathbf{\Gamma}_{j}(1) \sum_{s=1}^{t}
\beps_{\cdot j,s}+\mathbf{\Gamma }_{j}^{\ast }\left(L\right) \beps_{\cdot j,t},  \label{ctx}
\end{equation}%
where $\mathbf{\Gamma }_{j}(1) =\sum_{h=0}^{\infty }\mathbf{\Gamma }_{j,h}$ is a $p_{1}\times p_{1}$ matrix with rank $h_{R_{1}}$, and $\mathbf{\Gamma }_{j}^{\ast }\left( L\right) =\sum_{h=0}^{\infty }\mathbf{\Gamma }_{j,h}^{\ast }L^{h}$ with $\mathbf{\Gamma }_{j,h}^{\ast
}=-\sum_{k=h+1}^{\infty }\mathbf{\Gamma }_{j,k}$. Hence, $\mathbf{\Gamma }_{j}(1) $ can be rewritten as the product between a $p_{1}\times
h_{R_{1}}$ matrix (say $\R_{1}$) and an $h_{R_{1}}\times p_{1}$ matrix (say $\mathbf{\Pi }_{j}$), whence (\ref{ctx}) becomes 
\begin{equation*}
\X_{\cdot j,t}=\underset{p_{1}\times k_{R}}{\R_{1}}\underset{k_{R}\times p_{1}}{\mathbf{\Pi }_{j}}\sum_{s=1}^{t}\beps_{\cdot j,s}+\mathbf{\Gamma }_{j}^{\ast }\left( L\right)\beps_{\cdot j,t}.
\end{equation*}
Define, for short, the scalar common trend $\mathbf{G}_{j,t}=\boldsymbol{\Pi}_{j}\sum_{s=1}^{t}\beps_{\cdot j,s}$, and consider the $
p_{2}\times 1$ vector of common trends $\mathbf{G}_{t}=\left( \mathbf{G}_{1,t},...,\mathbf{G}_{p_{2},t}\right) ^{\prime}$. The vector $\mathbf{G}
_{t}$ itself could be cointegrated; considering the MA$(\infty)$  representation $\Delta \mathbf{G}_{t}=\mathbf{\Gamma }_{G}\left( L\right) \mathbf{u}_{t}$, it follows that
\begin{equation}
\mathbf{G}_{t}=\mathbf{\Gamma }_{G}(1) \sum_{s=1}^{t}\mathbf{u}_{s}+\mathbf{\Gamma }_{G}^{\ast }\left( L\right) \mathbf{u}_{t}.  \label{ctg}
\end{equation}
Again, $\mathbf{\Gamma }_{G}(1) $ is $p_{2}\times p_{2}$ and has rank $h_{C_{1}}=1$, so that we can write
\begin{equation*}
\mathbf{G}_{t}=\mathbf{\Gamma }_{G}(1) \sum_{s=1}^{t}\mathbf{u}_{s}+\mathbf{\Gamma }_{G}^{\ast }\left( L\right) \mathbf{u}_{t}=\underset{p_{2}\times h_{C_{1}}}{\C_{1}}\underset{h_{C_{1}}\times p_{2}}{\mathbf{\beta }_{G}^{\prime}}\sum_{s=1}^{t}\mathbf{u}_{s}+\mathbf{\Gamma}
_{G}^{\ast }\left( L\right) \mathbf{u}_{t}=\underset{p_{2}\times h_{C_{1}}}{\C_{1}}\F_{1,t}^{\prime}+\mathbf{\Gamma }_{G}^{\ast }\left( L\right) \mathbf{u}_{t},
\end{equation*}%
where $\F_{1,t}^{\prime}=\mathbf{\beta }_{G}^{\prime}\sum_{s=1}^{t} \mathbf{u}_{s}$ is the common stochastic trend. Then, by substituting, we
have%
\begin{equation*}
\X_{t}=\R_{1}\mathbf{G}_{t}^{\prime}+\widetilde{\E}_{t},
\end{equation*}
where $\widetilde{\E}_{t}$ is a stationary matrix-valued whose $j$-th column is given by $\mathbf{\Gamma }_{j}^{\ast }\left( L\right) \beps_{\cdot j,t}$ in (\ref{ctx}); using (\ref{ctg}), we receive
\begin{equation}
\X_{t}=\R_{1}\F_{1,t}\C_{1}^{\prime}+\R_{1}\mathbf{u}_{t}^{\prime}\left[ \mathbf{\Gamma }_{G}^{\ast}(1) \right] ^{\prime}+\E_{t},  \label{mecm}
\end{equation}%
where $\E_{t}$ is the overall error term. In this model, we have the same structure for the common stochastic trend(s) $\F_{1,t}$ as in
model (\ref{model}). However, the common stationary component $\R_{1}\mathbf{u}_{t}^{\prime}\left[ \mathbf{\Gamma }_{G}^{\ast }\left( L\right) 
\right] ^{\prime}$ is different to the one in (\ref{model}), in that $\mathbf{\Gamma }_{G}^{\ast }(1) $ may have full rank, thus entailing that the stationary common component $\R_{1}\mathbf{u}_{t}^{\prime}\left[ \mathbf{\Gamma }_{G}^{\ast }(1) \right]^{\prime}$ does not have a ``two-way''\ but a ``one-way''\ structure. Our assumptions hereafter are also different to the ones implicitly present in (\ref{mecm}), seeing as we assume that $\F_{1,t}$ and $\F_{0,t}$ are independent, whereas, in (\ref{mecm}), $\F_{1,t}$ and $\mathbf{u}_{t}$ clearly are not. Hence, the extension of our methodology to a cointegrated system - whilst building on the methodology developed herein - is a not entirely trivial task, which is currently under investigation by the authors.
\par
As a second extension, augmenting (\ref{model})-(\ref{model_ft}) to include linear trends could be also of interest. In such a case, (\ref{model}) would become 
\begin{equation}
\X_{t}=\mathbf{A}t\mathbf{B}^{\prime}+\R_{1}\F_{1,t}\C_{1}^{\prime}+\R_{0}\F_{0,t}\C
_{0}^{\prime}+\E_{t}.  \label{model_trend}
\end{equation}
An ``anti-projection'' based estimation strategy, in this case, could be based on estimating $\C_{1}$ and $\C_{0}$ from the first-differenced version of (\ref{model_trend})
\begin{eqnarray*}
\Delta \X_{t} &=&\mathbf{AB}^{\prime}+\R_{1}\Delta \F_{1,t}\C_{1}^{\prime}+\R_{0}\Delta \F_{0,t}\C_{0}^{\prime}+\Delta \E_{t} \\
&=&\mathbf{AB}^{\prime}+\R\Delta \F_{t}\C+\Delta \E_{t},
\end{eqnarray*}%
using e.g. the projection-based estimator of \citet{hkyz2021} (denoting this as, say, $\hat{\R}\ $and $\hat{\C}$), and subsequently
estimating $\mathbf{A}$ from the anti-projected version of $\X_{t}$, viz. $\X_{t}\hat{\C}_{\perp}$. This extension, as well
as the case where one of the common $I(1)$ factor has a drift component, is also under investigation by the authors. 

\bibliographystyle{chicago}
\bibliography{LTbiblio}

\newpage \clearpage

\begin{appendix}
    


\section{Technical lemmas\label{lemmas}}
Henceforth, whenever possible in this and in the next section, for
simplicity and without loss of generality we will carry out our proofs under
the constraints $h_{R_{1}}=1$, $h_{C_{1}}=1$, $h_{R_{0}}=1$, and $h_{C_{0}}=1$.

\begin{lemma}
\label{stout}Consider a multi-index random variable $U_{i_{1},\dots,i_{h}}$,
with $1\leq i_{1}\leq S_{1}$, $1\leq i_{2}\leq S_{2}$, etc\dots Assume that%
\begin{equation}
\sum_{S_{1}}\cdot \cdot \sum_{S_{h}}\frac{1}{S_{1}\cdot \dots\cdot S_{h}}%
P\left( \max_{1\leq i_{1}\leq S_{1},\dots,1\leq i_{h}\leq S_{h}}\left\vert
U_{i_{1},\dots,i_{h}}\right\vert >\epsilon L_{S_{1},\dots,S_{h}}\right) <\infty ,
\label{bc-1}
\end{equation}%
for some $\epsilon >0$ and a sequence $L_{S_{1},\dots,S_{h}}$ defined as%
\begin{equation*}
L_{S_{1},\dots,S_{h}}=S_{1}^{d_{1}}\cdot \dots\cdot S_{h}^{d_{h}}l_{1}\left(
S_{1}\right) \cdot \dots l_{h}\left( S_{h}\right) ,
\end{equation*}%
where $d_{1}$, $d_{2}$, etc. are non-negative numbers and $l_{1}\left( \cdot
\right) $, $l_{2}\left( \cdot \right) $, etc. are slowly varying functions
in the sense of Karamata. Then it holds that%
\begin{equation}
\limsup_{\left( S_{1},\dots,S_{h}\right) \rightarrow \infty }\frac{\left\vert
U_{S_{1},\dots,S_{h}}\right\vert }{L_{S_{1},\dots,S_{h}}}=0\text{ \textit{a.s.}}
\label{bc-2}
\end{equation}

\begin{proof}
The lemma is shown in \citet{massacci2022high} - see in particular Lemma A11
therein.
\end{proof}
\end{lemma}

\begin{lemma}
\label{berkes}We assume that Assumption \ref{as-1} is satisfied. Then it
holds that%
\begin{equation}
E\left\Vert \F_{1,t}\right\Vert _{F}^{2}\leq c_{0}t,  \label{lemma1-i}
\end{equation}%
\begin{equation}
\left\Vert \sum_{t=1}^{T}\F_{1,t}\right\Vert _{F}^{2}=O_{P}\left(
T^{2}\right)  \label{lemma1-ii}
\end{equation}

\begin{proof}
Equation (\ref{lemma1-i}) follows immediately from Proposition 4 in %
\citet{berkeshormann}. As far as (\ref{lemma1-ii}) is concerned, it follows
immediately from the FCLT for Bernoulli shifts (e.g., Theorem A.1 in %
\citealp{aue09}) and the Continuous Mapping Theorem.
\end{proof}
\end{lemma}

\begin{lemma}
\label{chung}We assume that Assumption \ref{as-1} is satisfied. Then it
holds that%
\begin{align}
\liminf_{T\rightarrow \infty }\frac{\log \log T}{T^{2}}\sum_{t=1}^{T}\F_{1,t}\F_{1,t}^{\prime} &=D_{1}\text{ \ \ a.s.},  \label{donsker-1}
\\
\liminf_{T\rightarrow \infty }\frac{\log \log T}{T^{2}}\sum_{t=1}^{T}\F_{1,t}^{\prime}\F_{1,t} &=D_{2}\text{ \ \ a.s.},  \label{donsker-2}
\end{align}%
where $D_{1}$ and $D_{2}$\ are two positive definite matrices of dimensions $%
h_{R_{1}}\times h_{R_{1}}$ and $h_{C_{1}}\times h_{C_{1}}\ $respectively.

\begin{proof}
We prove (\ref{donsker-1}) - the proof of (\ref{donsker-2}) is the same.
Assumption \ref{as-1} entails that, for each $T$, it is possible to define a
matrix valued, $h_{R_{1}}\times h_{C_{1}}$-dimensional Wiener process $\left\{ 
\W_{RC,T}\left( k\right) ,1\leq k\leq T\right\} $, with covariance
matrix $\Sigma _{F}^{\left( a\right) }$\ such that%
\begin{equation}
\max_{1\leq k\leq T}\left\Vert \F_{1,k}-\W_{RC,T}\left(
k\right) \right\Vert =O_{a.s.}\left( T^{1/2-\zeta }\right) ,  \label{sip}
\end{equation}%
for some $\zeta >0$ (a proof can be found e.g. in \citealp{aue2014}).
Considering%
\begin{align*}
&\sum_{t=1}^{T}\F_{1,t}\F_{1,t}^{\prime}\\
&=\sum_{t=1}^{T}\W_{RC,T}\left( t\right) \W_{RC,T}^{\prime}\left( t\right) +\sum_{t=1}^{T}\left( \F_{1,t}-\W%
_{RC,T}\left( t\right) \right) \W_{RC,T}^{\prime}\left( t\right) \\
&+\sum_{t=1}^{T}\W_{RC,T}\left( t\right) \left( \F_{1,t}-%
\W_{RC,T}\left( t\right) \right) +\sum_{t=1}^{T}\left( \F_{1,t}-\W_{RC,T}\left( t\right) \right) \left( \F_{1,t}-\W_{RC,T}\left( t\right) \right) ^{\prime},
\end{align*}%
equation (\ref{sip}) entails%
\begin{align*}
&\left\Vert \sum_{t=1}^{T}\left( \F_{1,t}-\W_{RC,T}\left(t\right) \right) \W_{RC,T}^{\prime}\left(t\right) \right\Vert \\
&\leq \sum_{t=1}^{T}\left\Vert \F_{1,t}-\W_{RC,T}\left(t\right) \right\Vert \left\Vert \W_{RC,T}\left( t\right) \right\Vert\leq T\max_{1\leq k\leq T}\left\Vert \F_{1,k}-\W_{RC,T}\left(t\right) \right\Vert \max_{1\leq k\leq T}\left\Vert \W_{RC,T}\left(t\right) \right\Vert \\
&=T\cdot O_{a.s.}\left( T^{1/2-\zeta }\right) O_{a.s.}\left( T^{1/2}\left( \log
\log T\right) ^{1/2}\right) =O_{a.s.}\left( T^{2-\zeta }\left( \log \log
T\right) ^{1/2}\right) ,
\end{align*}%
and by the same token%
\begin{align*}
&\left\Vert \sum_{t=1}^{T}\left( \F_{1,t}-\W_{RC,T}\left(t\right) \right) \left( \F_{1,t}-\W_{RC,T}\left( t\right)\right) ^{\prime}\right\Vert \\
&\leq T\left( \max_{1\leq k\leq T}\left\Vert \F_{1,k}-\W_{RC,T}\left( t\right) \right\Vert \right) ^{2}=O_{a.s.}\left( T^{2-2\zeta}\right) ,
\end{align*}%
whence 
\begin{equation*}
\frac{\log \log T}{T^{2}}\sum_{t=1}^{T}\F_{1,t}\F_{1,t}^{\prime}=\frac{\log \log T}{T^{2}}\sum_{t=1}^{T}\W_{RC,T}\left( t\right) 
\W_{RC,T}^{\prime}\left( t\right) +o_{a.s.}(1) .
\end{equation*}%
Note now that, letting $\W_{RC,T,j}\left( t\right) $ be the $j$-th
column of $\W_{RC,T}\left( t\right) $%
\begin{equation*}
\frac{\log \log T}{T^{2}}\sum_{t=1}^{T}\W_{RC,T}\left( t\right) 
\W_{RC,T}^{\prime}\left( t\right) =\sum_{j=1}^{p_{2}}\frac{\log
\log T}{T^{2}}\sum_{t=1}^{T}\W_{RC,T,j}\left( t\right) \W%
_{RC,T,j}^{\prime}\left( t\right) .
\end{equation*}%
We now show that the limit of the expression above is positive definite. To
this end, let $\mathbf{b}$ be a $h_{C_{1}}\times 1$ nontrivial vector and
consider%
\begin{equation*}
\sum_{j=1}^{p_{2}}\frac{\log \log T}{T^{2}}\sum_{t=1}^{T}\mathbf{b}^{\prime}%
\W_{RC,T,j}\left( t\right) \W_{RC,T,j}^{\prime}\left(
t\right) \mathbf{b};
\end{equation*}%
given that the distribution of $\W_{RC,T}\left( t\right) $,does not
depend on $T$, we have that $\left\{ \mathbf{b}^{\prime}\W_{RC,T,j}\left( t\right) \right. ,$ $\left. 1\leq t\leq T\right\} $ $\overset{\D}{=}$ $\left\{ \widetilde{W}_{j}\left( t\right) ,1\leq t\leq
T\right\} $, where $\left\{ \widetilde{W}_{j}\left( t\right) ,1\leq t\leq
T\right\} $ is a scalar Wiener process with variance such that%
\begin{equation*}
\sum_{j=1}^{p_{2}}\mathbf{b}^{\prime}E\left( \W_{RC,T,j}\left(
1\right) \W_{RC,T,j}^{\prime}(1) \right) \mathbf{b=b}%
^{\prime}\Sigma _{F}^{\left( a\right) }\mathbf{b.}
\end{equation*}%
Using equation (4.6) in \citet{donsker1977laws}, it therefore follows that%
\begin{equation*}
\liminf_{T\rightarrow \infty }\sum_{j=1}^{p_{2}}\frac{\log \log T}{T^{2}}%
\sum_{t=1}^{T}\mathbf{b}^{\prime}\W_{RC,T,j}\left( t\right) \W_{RC,T,j}^{\prime}\left( t\right) \mathbf{b}=\frac{1}{2}\mathbf{b}^{\prime}\Sigma _{F}^{\left( a\right) }\mathbf{b},
\end{equation*}%
a.s.; seeing as $\Sigma _{F}^{\left( a\right) }$ is positive definite by
Assumption \ref{as-1}\textit{(ii)}(a), $\mathbf{b}^{\prime}\Sigma
_{F}^{\left( a\right) }\mathbf{b}>0$ for all nontrivial $\mathbf{b}$. The
desired result now follows.
\end{proof}
\end{lemma}

\begin{lemma}
\label{f1-summab}We assume that Assumption \ref{as-2} is satisfied. Then it
holds that 
\begin{align}
\sum_{t,s=1}^{T}\left\Vert E\left( \F_{0,t}\F_{0,s}^{\prime}\right) \right\Vert _{F} &\leq c_{0}T,\text{ and } \sum_{t,s=1}^{T}\left \Vert E\left( \F_{0,t}^{\prime}\F_{0,s}\right) \right\Vert_{F}\leq c_{1}T,  \label{f1-sum-1} \\
\sum_{t=1}^{T}\left\Vert \F_{0,t}\right\Vert _{F}^{2} &= O_{a.s.}\left( T\right) ,  \label{f1-sum-2}
\end{align}%
for some $0<c_{0},c_{1}<\infty $.

\begin{proof}
Equation (\ref{f1-sum-1}) is a direct consequence of Assumption \ref{as-2} - see in particular the proof of Lemma A.3 in \citet{massacci2022high}. As far as (\ref{f1-sum-2}) is concerned, we begin by showing that $\left\Vert 
\F_{0,t}\right\Vert_{F}^{2}$ is an $\mathcal{L}_{2}$-decomposable Bernoulli shift, focusing on the case $h_{R_{0}}=h_{C_{0}}=1$ for simplicity and with no loss of generality. Consider the construction%
\begin{equation*}
\F_{0,t,t}^{\ast }=g_{\F_{0}}\left( \eta_{t}^{\F_{0}},\dots,\eta_{1}^{\F_{0}}, \widetilde{\eta}_{0}^{\F_{0}},\widetilde{\eta}_{-1}^{\F_{0}},\dots\right) .
\end{equation*}%
It holds that 
\begin{align*}
&\left\vert \F_{0,t}^{2}-\left( \F_{0,t,t}^{\ast }\right)^{2}\right\vert _{2}\leq \left\vert \F_{0,t}+\F%
_{0,t,t}^{\ast }\right\vert _{4}\left\vert \F_{0,t}-\F_{0,t,t}^{\ast }\right\vert _{4} \\
&\leq 2\left\vert \F_{0,t}\right\vert _{4}\left\vert \F_{0,t}-\F_{0,t,t}^{\ast }\right\vert _{4}\leq c_{0}^{-a},
\end{align*}%
which entails that $\left\Vert \F_{0,t}\right\Vert _{F}^{2}$ is a $%
\mathcal{L}_{2}$-decomposable Bernoulli shift. Hence, by Proposition 4 in %
\citet{berkeshormann}, it follows that%
\begin{equation*}
E\left\vert \sum_{t=1}^{T}\left( \left\Vert \F_{0,t}\right\Vert
_{F}^{2}-E\left\Vert \F_{0,0}\right\Vert _{F}^{2}\right) \right\vert
^{2}\leq c_{0}T,
\end{equation*}%
and therefore Theorem 3 in \citet{moricz1976moment}\ yields%
\begin{equation*}
E\max_{1\leq k\leq T}\left\vert \sum_{t=1}^{T}\left( \left\Vert \F%
_{0,t}\right\Vert _{F}^{2}-E\left\Vert \F_{0,0}\right\Vert
_{F}^{2}\right) \right\vert ^{2}\leq c_{0}T\left( \log 2T\right) ^{4}.
\end{equation*}%
Equation (\ref{f1-sum-2}) readily follows from the SLLN.
\end{proof}
\end{lemma}

\begin{lemma}
\label{crossf}We assume that Assumptions \ref{as-1}, \ref{as-2} and \ref%
{as-5}\ are satisfied. Then it holds that%
\begin{equation*}
\left\Vert \sum_{t=1}^{T}\F_{1,t}\F_{0,t}\right\Vert
_{F}=O_{P}\left( T\right) .
\end{equation*}

\begin{proof}
Clearly%
\begin{align*}
&E\left\Vert \sum_{t=1}^{T}\F_{1,t}\F_{0,t}\right\Vert_{F}^{2}=\sum_{t,s=1}^{T}E\left(F_{1,t}\F_{1,s}\F_{0,t} \F_{0,s}\right) \leq \sum_{t,s=1}^{T}\left\vert E\left(\F_{1,t}\F_{1,s}\right) \right\vert \left\vert E\left( \F_{0,t}\F_{0,s}\right) \right\vert \\
&\leq \sum_{t,s=1}^{T}\left\vert \left( E\left\Vert \F_{1,t}\right\Vert _{F}^{2}\right) ^{1/2}\left( E\left\Vert \F_{1,s}\right\Vert _{F}^{2}\right) ^{1/2}\right\vert \left\vert E\left(\F_{0,t}\F_{0,s}\right) \right\vert \\
&\leq c_{0}T\sum_{t,s=1}^{T}\left\vert E\left( \F_{0,t}\F_{0,s}\right) \right\vert \leq c_{1}T^{2},
\end{align*}%
having used Assumptions \ref{as-1} and \ref{as-2}, (\ref{lemma1-i}) in the
penultimate passage, and (\ref{f1-sum-1}) in the last one. The desired
result now follows immediately from Markov inequality.
\end{proof}
\end{lemma}

\begin{lemma}
\label{err}We assume that Assumption \ref{as-3}\ are satisfied. Then it
holds that%
\begin{align}
\lambda _{\max }\left( \sum_{t=1}^{T}\E_{t}\E_{t}^{\prime}\right) &=O\left( p_{2}T\right) +O_{P}\left(
p_{1}p_{2}^{1/2}T^{1/2}\right) ,  \label{err-1} \\
\lambda _{\max }\left( \sum_{t=1}^{T}\E_{t}^{\prime}\E%
_{t}\right) &=O\left( p_{1}T\right) +O_{P}\left(
p_{1}^{1/2}p_{2}T^{1/2}\right) .  \label{err-2}
\end{align}

\begin{proof}
We can write%
\begin{equation*}
\sum_{t=1}^{T}\E_{t}\E_{t}^{\prime}=E\left( \sum_{t=1}^{T}%
\E_{t}\E_{t}^{\prime}\right) +\sum_{t=1}^{T}\left( \E_{t}\E_{t}^{\prime}-E\left( \E_{t}\E_{t}^{\prime}\right) \right) ,
\end{equation*}%
whence by Weyl's inequality%
\begin{equation*}
\lambda _{\max }\left( \sum_{t=1}^{T}\E_{t}\E_{t}^{\prime}\right) \leq \lambda _{\max }\left( E\left( \sum_{t=1}^{T}\E_{t}%
\E_{t}^{\prime}\right) \right) +\lambda _{\max }\left( \left( 
\E_{t}\E_{t}^{\prime}-E\left( \E_{t}\E%
_{t}^{\prime}\right) \right) \right) =I+II.
\end{equation*}%
Using again Weyl's inequality, it holds that 
\begin{equation*}
\lambda _{\max }\left( E\left( \sum_{t=1}^{T}\E_{t}\E%
_{t}^{\prime}\right) \right) \leq \sum_{t=1}^{T}\lambda _{\max }\left(
E\left( \E_{t}\E_{t}^{\prime}\right) \right) ,
\end{equation*}%
and by Assumption \ref{as-3}\textit{(ii)}(b)%
\begin{equation*}
\lambda _{\max }\left( E\left( \E_{t}\E_{t}^{\prime}\right)
\right) \leq \max_{1\leq i\leq
p_{1}}\sum_{j=1}^{p_{2}}\sum_{h=1}^{p_{1}}\left\vert E\left(
e_{ij,t}e_{hj,t}\right) \right\vert \leq c_{0}p_{2},
\end{equation*}%
whence finally it follows that $\lambda _{\max }\left( E\left( \sum_{t=1}^{T}%
\E_{t}\E_{t}^{\prime}\right) \right) =O\left( p_{2}T\right) 
$. Also, by symmetry%
\begin{align*}
&\lambda _{\max }\left( \left( \E_{t}\E_{t}^{\prime}-E\left( \E_{t}\E_{t}^{\prime}\right) \right) \right) \\
&\leq \left\Vert \E_{t}\E_{t}^{\prime}-E\left( \E_{t}\E_{t}^{\prime}\right) \right\Vert _{F}=\left(\sum_{i,j=1}^{p_{1}}\left( \sum_{t=1}^{T}\sum_{k=1}^{p_{2}}\left(e_{ik,t}e_{jk,t}-E\left( e_{ik,t}e_{jk,t}\right) \right) \right) ^{2}\right)^{1/2}.
\end{align*}%
It holds that%
\begin{align*}
&E\sum_{i,j=1}^{p_{1}}\left( \sum_{t=1}^{T}\sum_{k=1}^{p_{2}}\left(
e_{ik,t}e_{jk,t}-E\left( e_{ik,t}e_{jk,t}\right) \right) \right) ^{2} \\
&=\sum_{i,j=1}^{p_{1}}\sum_{t,s=1}^{T}\sum_{h,k=1}^{p_{2}}\cov\left(
e_{ik,t}e_{jk,t},e_{ih,s}e_{jh,s}\right) \leq c_{0}p_{1}^{2}p_{2}T,
\end{align*}%
by Assumption \ref{as-3}\textit{(iii)}(a). Hence by Markov inequality it
finally follows that%
\begin{equation*}
\lambda _{\max }\left( \left( \E_{t}\E_{t}^{\prime}-E\left( 
\E_{t}\E_{t}^{\prime}\right) \right) \right) =O_{P}\left(
p_{1}p_{2}^{1/2}T^{1/2}\right) ,
\end{equation*}%
whence (\ref{err-1}) follows. Equation (\ref{err-2}) follows from the same
logic.
\end{proof}
\end{lemma}

\begin{lemma}
\label{errorfactor}We assume that Assumptions \ref{as-1}, \ref{as-3} and \ref%
{as-5} are satisfied. Then it holds that%
\begin{equation*}
\left\Vert \sum_{t=1}^{T}\F_{1,t}\E_{t}\right\Vert
_{F}=O_{P}\left( p_{1}^{1/2}p_{2}^{1/2}T\right) .
\end{equation*}

\begin{proof}
We let $h_{R_{1}}=h_{C_{1}}=1$ for simplicity. We have%
\begin{eqnarray*}
&&E\left\Vert \sum_{t=1}^{T}\F_{1,t}\E_{t}\right\Vert
_{F}^{2} \\
&=&E\sum_{i=1}^{p_{1}}\sum_{j=1}^{p_{2}}\left( \sum_{t=1}^{T}\F%
_{1,t}e_{ij,t}\right)
^{2}=\sum_{i=1}^{p_{1}}\sum_{j=1}^{p_{2}}\sum_{t,s=1}^{T}E\left( \F%
_{1,t}\F_{1,s}\right) E\left( e_{ij,s}e_{ij,t}\right)  \\
&\leq &\sum_{i=1}^{p_{1}}\sum_{j=1}^{p_{2}}\sum_{t,s=1}^{T}E\left( \F%
_{1,t}^{2}\right) \left\vert E\left( e_{ij,s}e_{ij,t}\right) \right\vert
\leq c_{0}T\sum_{i=1}^{p_{1}}\sum_{j=1}^{p_{2}}\sum_{t,s=1}^{T}\left\vert
E\left( e_{ij,s}e_{ij,t}\right) \right\vert \leq c_{1}p_{1}p_{2}T,
\end{eqnarray*}%
by Assumptions \ref{as-5} (used in the second line) and \ref{as-3}\textit{%
(ii)}(a), used in the third line.
\end{proof}
\end{lemma}

\begin{lemma}
\label{eig-mRx}We assume that Assumptions \ref{as-1}-\ref{as-5} are
satisfied. Then there exists a positive, finite constant $c_{0}$ and a
triplet of random variables $\left( p_{0,0},p_{2,0},T_{0}\right) $ such
that, for all $p_{1}\geq p_{0,0}$, $p_{2}\geq p_{2,0}$ and $T\geq T_{0}$, it
holds that%
\begin{equation*}
\left( \log \log T\right) \lambda _{j}\left( \M_{R_1}\right) \geq
c_{0}\text{, \ \ for all }j\leq h_{R_{1}},
\end{equation*}%
and, for all $\epsilon >0$%
\begin{equation*}
\lambda _{j}\left( \M_{R_1}\right) =o_{a.s.}\left( \frac{\left(
\log T\right) ^{3/2+\epsilon }}{T}\right) ,\text{ \ \ for all }j>h_{R_{1}}.
\end{equation*}

\begin{proof}
It holds that%
\begin{align}
\M_{R_1} &=\frac{1}{p_{1}p_{2}T^{2}}\sum_{t=1}^{T}\R_1\F_{1,t}\C_{1}^{\prime}\C_{1}\F_{1,t}^{\prime}\R_{1}^{\prime} + \frac{1}{p_{1}p_{2}T^{2}}\sum_{t=1}^{T}\R_{0}\F_{0,t}\C_{0}^{\prime}\C_{0}\F_{0,t}^{\prime}\R_{0}^{\prime}  \label{mxr-dec} \\
&+\frac{1}{p_{1}np_{2}T^{2}}\sum_{t=1}^{T}\E_{t}\E_{t}^{\prime}+\frac{1}{p_{1}p_{2}T^{2}}\sum_{t=1}^{T}\R_1\F_{1,t}\C_{1}^{\prime}\E_{t}^{\prime}+\frac{1}{p_{1}p_{2}T^{2}}\sum_{t=1}^{T} \E_{t}\C_{1}\F_{1,t}^{\prime}\R_{1}^{\prime}  \notag \\
&+\frac{1}{p_{1}p_{2}T^{2}}\sum_{t=1}^{T}\R_1\F_{1,t}\C_{1}^{\prime}\C_{0}\F_{0,t}^{\prime}\R_{0}^{\prime}+\frac{1}{p_{1}p_{2}T^{2}}\sum_{t=1}^{T}\R_{0}\F_{0,t}\C_{0}^{\prime}\C_{1}\F_{1,t}^{\prime}\R_{1}^{\prime}  \notag \\
&+\frac{1}{p_{1}p_{2}T^{2}}\sum_{t=1}^{T}\R_{0}\F_{0,t}\C_{0}^{\prime}\E_{t}^{\prime}+\frac{1}{p_{1}p_{2}T^{2}} \sum_{t=1}^{T}\E_{t}\C_{0}\F_{0,t}^{\prime}\R_{0}^{\prime}  \notag \\
&=I+II+III+IV+IV^{\prime}+V+V^{\prime}+VI+VI^{\prime}.  \notag
\end{align}
We begin by finding bounds for $II$, $III$, $IV$, $V$, and $VI$, using $%
h_{R_{1}}=h_{C_{1}}=h_{R_{0}}=h_{C_{1}}=1$ for simplicity whenever possible. We have%
\begin{equation*}
\left\Vert II\right\Vert _{F}\leq \frac{1}{p_{1}p_{2}T^{2}}\left\Vert 
\R_{0}\right\Vert _{F}^{2}\left\Vert \C_{0}\right\Vert
_{F}^{2}\left( \sum_{t=1}^{T}\F_{0,t}^{2}\right) =O_{a.s.}\left( 
\frac{1}{T}\right) ,
\end{equation*}%
having used Assumption \ref{as-4}\textit{(i)}(b) and (\ref{f1-sum-2}) in
Lemma \ref{f1-summab}.\ Turning to $III$, combining Lemmas \ref{stout} and %
\ref{err}, it is easy to see that%
\begin{equation*}
\lambda _{\max }\left( \frac{1}{p_{1}p_{2}T^{2}}\sum_{t=1}^{T}\E_{t}%
\E_{t}^{\prime}\right) =O\left( \frac{1}{p_{1}T}\right)
+o_{a.s.}\left( \frac{\left( \log p_{1}\log p_{2}\log T\right) ^{1+\epsilon }%
}{p_{2}^{1/2}T^{3/2}}\right) .
\end{equation*}%
We now study%
\begin{equation*}
\left\Vert IV\right\Vert _{F}=\frac{1}{p_{1}p_{2}T^{2}}\left(
\sum_{i,h=1}^{p_{1}}\left( r_{i}\sum_{t=1}^{T}\sum_{j=1}^{p_{2}}c_{j}\F_{1,t}e_{hj,t}\right) ^{2}\right) ^{1/2};
\end{equation*}%
it holds that%
\begin{align*}
&E\sum_{i,h=1}^{p_{1}}r_{i}^{2}\left( \sum_{t=1}^{T}\sum_{j=1}^{p_{2}}c_{j}%
\F_{1,t}e_{hj,t}\right) ^{2} \\
&=\sum_{i,h=1}^{p_{1}}r_{i}^{2}E\sum_{t,s=1}^{T}%
\sum_{j,k=1}^{p_{2}}c_{j}c_{k}\F_{1,t}\F_{1,s}e_{hj,t}e_{hk,s} \\
&\leq \left( \max_{1\leq i\leq p_{1}}r_{i}^{2}\right) \left( \max_{1\leq
i\leq p_{2}}c_{i}^{2}\right)
\sum_{i,h=1}^{p_{1}}\sum_{t,s=1}^{T}\sum_{j,k=1}^{p_{2}}\left\vert \left(
E\left\Vert \F_{1,t}\right\Vert _{F}^{2}\right) ^{1/2}\left(
E\left\Vert \F_{1,s}\right\Vert _{F}^{2}\right) ^{1/2}\right\vert
\left\vert E\left( e_{hj,t}e_{hk,s}\right) \right\vert \\
&\leq c_{0}p_{1}T\sum_{h=1}^{p_{1}}\sum_{t,s=1}^{T}\sum_{j,k=1}^{p_{2}}\left\vert
E\left( e_{hj,t}e_{hk,s}\right) \right\vert \leq c_{1}p_{1}^{2}p_{2}T^{2},
\end{align*}%
having used Assumption \ref{as-5} in the third line, Assumption \ref{as-4}%
\textit{(ii)}(d)\ in the last passage. Hence, using Lemma \ref{stout}%
\begin{equation*}
\left\Vert IV\right\Vert _{F}=o_{a.s.}\left( \frac{1}{p_{2}^{1/2}T}\left(
\log p_{1}\log ^{2}p_{2}\log ^{2}T\right) ^{1/2+\epsilon }\right) ,
\end{equation*}%
for all $\epsilon >0$. The same bound holds, by symmetry, for $\left\Vert
IV^{\prime}\right\Vert _{F}$. We now study%
\begin{equation*}
\left\Vert V\right\Vert _{F}\leq \left\Vert \R_1\right\Vert
_{F}\left\Vert \R_{0}^{\prime}\right\Vert _{F}\left\Vert \C_{0}\right\Vert _{F}\left\Vert \C_{1}\right\Vert _{F}\frac{1}{%
p_{1}p_{2}T^{2}}\left( \sum_{t=1}^{T}\F_{1,t}\F_{0,t}\right)
\leq c_{0}\frac{1}{T^{2}}\left( \sum_{t=1}^{T}\F_{1,t}\F%
_{0,t}\right) .
\end{equation*}%
We have%
\begin{align*}
&E\left( \sum_{t=1}^{T}\F_{1,t}\F_{0,t}\right) ^{2} \\
&\leq \sum_{t,s=1}^{T}\left( E\left\Vert \F_{1,t}\right\Vert
_{F}^{2}\right) ^{1/2}\left( E\left\Vert \F_{1,s}\right\Vert
_{F}^{2}\right) ^{1/2}\left\Vert E\left( \F_{0,t}\F%
_{0,s}\right) \right\Vert _{F}\leq c_{0}T\sum_{t,s=1}^{T}\left\Vert E\left( 
\F_{0,t}\F_{0,s}\right) \right\Vert _{F}\leq c_{1}T^{2},
\end{align*}%
having used Assumption \ref{as-5} and (\ref{lemma1-i}) in Lemma \ref{berkes} and (\ref{f1-sum-2}) in Lemma \ref{f1-summab} in the last passage. Hence, by Lemma \ref{stout}%
\begin{equation*}
\left\Vert V\right\Vert _{F}=o_{a.s.}\left( \frac{1}{T}\left( \log T\right)
^{3/2+\epsilon }\right) ,
\end{equation*}%
and the same bound holds, by symmetry, for $\left\Vert V^{\prime}\right\Vert _{F}$. Finally, we study%
\begin{align*}
\left\Vert VI\right\Vert _{F} &=\frac{1}{p_{1}p_{2}T^{2}}\left\Vert
\sum_{t=1}^{T}\R_{0}\F_{0,t}\C_{0}^{\prime}\E_{t}^{\prime}\right\Vert _{F} \leq \frac{1}{p_{1}p_{2}T^{2}}\left\Vert \R_{0}\right\Vert
_{F}\left\Vert \sum_{t=1}^{T}\F_{0,t}\C_{0}^{\prime}\E_{t}^{\prime}\right\Vert _{F}.
\end{align*}%
It holds that%
\begin{align}
&E\left\Vert \sum_{t=1}^{T}\F_{0,t}\C_{0}^{\prime}\E_{t}^{\prime}\right\Vert _{F}^{2}  \label{sum-f1e-bound} \\
&=E\left( \sum_{i=1}^{p_{1}}\left( \sum_{j=1}^{p_{2}}\sum_{t=1}^{T}c_{1,j}\F_{0,t}e_{ij,t}\right) ^{2}\right)
=E\sum_{i=1}^{p_{1}}\sum_{j,k=1}^{p_{2}}\sum_{t,s=1}^{T}c_{1,j}c_{1,k}\F_{0,t}\F_{0,s}e_{ij,t}e_{ik,s}  \notag \\
&\leq \left( \max_{1\leq i\leq p_{2}}c_{1,i}^{2}\right)\sum_{i=1}^{p_{1}}\sum_{j,k=1}^{p_{2}}\sum_{t,s=1}^{T}E\left( \left\Vert\F_{0,t}\right\Vert _{F}^{2}\right) ^{1/2}E\left( \left\Vert \F_{0,s}\right\Vert _{F}^{2}\right) ^{1/2}\left\vert E\left(e_{ij,t}e_{ik,s}\right) \right\vert .  \notag
\end{align}%
Seeing as $\F_{0,t}$\ is stationary, using Assumption \ref{as-2}%
\textit{(ii)}(d), the above is bounded by%
\begin{equation*}
c_{0}\sum_{i=1}^{p_{1}}\sum_{j,k=1}^{p_{2}}\sum_{t,s=1}^{T}\left\vert
E\left( e_{ij,t}e_{ik,s}\right) \right\vert \leq c_{1}p_{1}p_{2}T,
\end{equation*}%
whence using Lemma \ref{stout} 
\begin{equation*}
\left\Vert \sum_{t=1}^{T}\F_{0,t}\C_{0}^{\prime}\E_{t}^{\prime}\right\Vert _{F}=o_{a.s.}\left( \left( p_{1}p_{2}T\right)
^{1/2}\left( \log p_{1}\log ^{2}p_{2}\log ^{2}T\right) ^{3/2+\epsilon
}\right) .
\end{equation*}%
In turn, using Assumption \ref{as-4}\textit{(ii)}, this yields%
\begin{equation*}
\left\Vert VI\right\Vert _{F}=o_{a.s.}\left( \frac{1}{p_{2}^{1/2}T^{3/2}}%
\left( \log p_{1}\log ^{2}p_{2}\log ^{2}T\right) ^{3/2+\epsilon }\right) .
\end{equation*}%
The same bound holds, by symmetry, for $\left\Vert VI^{\prime}\right\Vert
_{F}$.

Putting together all the above, by symmetry it follows that%
\begin{equation}
\lambda _{\max }\left( II+III+IV+IV^{\prime}+V+V^{\prime}+VI+VI^{\prime}\right) =o_{a.s.}\left( \frac{\left( \log T\right) ^{3/2+\epsilon }}{T}%
\right) .  \label{l-max-R}
\end{equation}%
Finally, consider $I$ in (\ref{mxr-dec}). By construction, $\lambda
_{j}\left( I\right) =0$ a.s. for all $j>h_{R_{1}}$; when $j\leq h_{R_{1}}$, it holds
that%
\begin{equation*}
\lambda _{j}\left( \frac{\log \log T}{p_{1}p_{2}T^{2}}\sum_{t=1}^{T}\R_1\F_{1,t}\C_{1}^{\prime}\C_{1}\F_{1,t}^{\prime}\R_{1}^{\prime}\right) =\lambda _{j}\left( \frac{1}{p_{1}}\R_1\left( \frac{\log \log
T}{T^{2}}\sum_{t=1}^{T}\F_{1,t}\F_{1,t}^{\prime}\right) \R_{1}^{\prime}\right) ,
\end{equation*}%
seeing as $\C_{1}^{\prime}\C=p_{2}\I_{h_{C_{1}}}$; using
the multiplicative version of Weyl's inequality (see e.g. Theorem 7 in %
\citealp{merikoski2004inequalities})%
\begin{align}
&\lambda _{j}\left( \frac{1}{p_{1}}\R_1\left( \frac{1}{T^{2}}%
\sum_{t=1}^{T}\F_{1,t}\F_{1,t}^{\prime}\right) \R_{1}^{\prime}\right)  \label{l-chung} \\
&\geq \lambda _{j}\left( \frac{1}{p_{1}}\R_{1}^{\prime}\R_{1}\right) \lambda _{\min }\left( \left( \frac{\log \log T}{T^{2}}\sum_{t=1}^{T}%
\F_{1,t}\F_{1,t}^{\prime}\right) \right) \geq c_{0}\text{ a.s.,}
\notag
\end{align}%
having used Assumption \ref{as-4}\textit{(ii)} and Lemma \ref{chung}. The
desired results now follow by combining (\ref{l-chung}) and (\ref{l-max-R}),
using Weyl's inequality.
\end{proof}
\end{lemma}

\begin{lemma}
\label{eig-mCx}We assume that Assumptions \ref{as-1}-\ref{as-5} are
satisfied. Then there exists a positive, finite constant $c_{0}$ and a
triplet of random variables $\left( p_{0,0},p_{2,0},T_{0}\right) $ such
that, for all $p_{1}\geq p_{0,0}$, $p_{2}\geq p_{2,0}$ and $T\geq T_{0}$, it
holds that%
\begin{equation*}
\left( \log \log T\right) \lambda _{j}\left( \M_{C_1}\right) \geq
c_{0}\text{, \ \ for all }j\leq h_{C_{1}}.
\end{equation*}%
and, for all $\epsilon >0$%
\begin{equation*}
\lambda _{j}\left( \M_{C_1}\right) =o_{a.s.}\left( \frac{\left(
\log T\right) ^{3/2+\epsilon }}{T}\right) ,\text{ \ \ for all }j>h_{C_{1}}.
\end{equation*}

\begin{proof}
The proof is the same, \textit{mutatis mutandis}, as that of Lemma \ref%
{eig-mRx} and therefore we omit it.
\end{proof}
\end{lemma}

\begin{lemma}
\label{lambda}We assume that Assumptions \ref{as-1}-\ref{as-5} are
satisfied. Then it holds that%
\begin{align}
\left\Vert \Lambda _{R_{1}}^{-1}\right\Vert &=O_{P}(1) ,
\label{l-r-inv} \\
\left\Vert \Lambda _{C_{1}}^{-1}\right\Vert &=O_{P}(1) .
\label{l-c-inv}
\end{align}

\begin{proof}
We consider (\ref{l-r-inv}) only; the proof of (\ref{l-c-inv}) is similar.
Following the passages in the proof of Lemma \ref{eig-mRx} therein \textit{%
verbatim}, it is easy to see that (\ref{l-max-R}) becomes%
\begin{equation*}
\lambda _{\max }\left( II+III+IV+IV^{\prime}+V+V^{\prime}+VI+VI^{\prime}\right) =O_{P}\left( \frac{1}{T}\right) .
\end{equation*}%
Similarly, when $j\leq h_{R_{1}}$, it holds that%
\begin{equation*}
\lambda _{j}\left( \frac{1}{p_{1}p_{2}T^{2}}\sum_{t=1}^{T}\R_1\F_{1,t}%
\C_{1}^{\prime}\C_{1}\F_{1,t}^{\prime}\R_{1}^{\prime}\right)
\geq \lambda _{j}\left( \frac{1}{p_{1}}\R_{1}^{\prime}\R_1\right)
\lambda _{\min }\left( \left( \frac{1}{T^{2}}\sum_{t=1}^{T}\F_{1,t}%
\F_{1,t}^{\prime}\right) \right) ,
\end{equation*}%
recalling the identification restriction $\C_{1}^{\prime}\C%
=p_{2}\I_{h_{C_{1}}}$. Using the multiplicative version of Weyl's
inequality (see e.g. Theorem 7 in \citealp{merikoski2004inequalities})%
\begin{equation*}
\lambda _{j}\left( \frac{1}{p_{1}}\R_1\left( \frac{1}{T^{2}}%
\sum_{t=1}^{T}\F_{1,t}\F_{1,t}^{\prime}\right) \R_{1}^{\prime}\right) \geq \lambda _{j}\left( \frac{1}{p_{1}}\R_{1}^{\prime}\R_1\right) \lambda _{\min }\left( \left( \frac{1}{T^{2}}\sum_{t=1}^{T}\F_{1,t}\F_{1,t}^{\prime}\right) \right) .
\end{equation*}%
By Assumption \ref{as-4}\textit{(ii)}, $\lambda _{j}\left( \R_{1}^{\prime}\R_{1}/p_{1}\right) >0$. Let now $\mathbf{b}$ be a nonzero $%
h_{R_{1}}\times 1$ vector; it is immediate to see that $\mathbf{b}^{\prime}%
\beps_{t}$ is a decomposable Bernoulli shift with the same
rate as $\beps_{t}$. Hence, the FCLT for Bernoulli shifts
(see Theorem A.1 in \citealp{aue09}) entails that%
\begin{equation*}
\limsup_{T\rightarrow \infty }P\left( \frac{1}{T^{2}}\sum_{t=1}^{T}\mathbf{b}%
^{\prime}\F_{1,t}\F_{1,t}^{\prime}\mathbf{b}\leq \epsilon
^{\prime}\right) \leq P\left( \int_{0}^{1}W^{2}\left( r\right) dr\leq
\epsilon ^{\prime}\right) ,
\end{equation*}%
where $\left\{ W\left( r\right) ,0\leq r\leq 1\right\} $ is a Wiener process
with variance $\mathbf{b}^{\prime}\boldsymbol{\Sigma}_{F}^{\left( a\right) }%
\mathbf{b}$; Assumption \ref{as-1}\textit{(ii)} entails that $\mathbf{b}%
^{\prime}\boldsymbol{\Sigma}_{F}^{\left( a\right) }\mathbf{b>0}$. Hence, using
e.g. Theorem 1.1 in \citet{li2001small}, for any $\epsilon ^{\prime}\rightarrow 0^{+}$%
\begin{align*}
&\limsup_{T\rightarrow \infty }P\left( \frac{1}{T^{2}}\sum_{t=1}^{T}\mathbf{b}^{\prime}\F_{1,t}\F_{1,t}^{\prime}\mathbf{b}\leq \epsilon
^{\prime}\right) \\
&\leq P\left( \int_{0}^{1}W^{2}\left( r\right) dr\leq \epsilon ^{\prime}\right) =\exp \left( -\frac{1}{8\epsilon ^{\prime}}\right) =\epsilon
^{\prime \prime},
\end{align*}%
and therefore%
\begin{equation*}
P\left( \frac{1}{T^{2}}\sum_{t=1}^{T}\mathbf{b}^{\prime}\F_{1,t}%
\F_{1,t}^{\prime}\mathbf{b}>\epsilon ^{\prime}\right) \geq
1-\epsilon ^{\prime \prime}.
\end{equation*}%
In turn, this readily entails that 
\begin{equation*}
\lim_{T\rightarrow \infty }P\left( \lambda _{\min }\left( \left( \frac{1}{%
T^{2}}\sum_{t=1}^{T}\F_{1,t}\F_{1,t}^{\prime}\right) \right)
>0\right) =1,
\end{equation*}%
which implies the desired result.
\end{proof}
\end{lemma}

\begin{lemma}
\label{spec-eig-mRx}
We assume that Assumptions \ref{as-1}-\ref{as-5} are
satisfied. Then it holds that $\left\Vert \left( \Lambda _{R_{1}}^{\dagger}\right) ^{-1}\right\Vert =O_{P}(1) $.

\begin{proof}
Some arguments in the proof are the same as in the proof of Lemma \ref{eig-mRx}, to which we refer for details. It holds that%
\begin{align*}
\hat{\M}_{R_{1}}^{\dagger} &=\frac{1}{p_{1}p_{2}^{2}T^{2}}%
\sum_{t=1}^{T}\R_1\F_{1,t}\C_{1}^{\prime}\hat{\C}_{1}%
\hat{\C}_{1}^{\prime}\C_{1}\F_{1,t}^{\prime}\R_{1}^{\prime}+%
\frac{1}{p_{1}p_{2}^{2}T^{2}}\sum_{t=1}^{T}\R_{0}\F_{0,t}%
\C_{0}^{\prime}\hat{\C}_{1}\hat{\C}_{1}^{\prime}%
\C_{0}\F_{0,t}^{\prime}\R_{0}^{\prime} \\
&+\frac{1}{p_{1}p_{2}^{2}T^{2}}\sum_{t=1}^{T}\E_{t}\hat{\C}_{1}\hat{\C}_{1}^{\prime}\E_{t}^{\prime}+\frac{1}{%
p_{1}p_{2}^{2}T^{2}}\sum_{t=1}^{T}\R_1\F_{1,t}\C_{1}^{\prime}%
\hat{\C}_{1}\hat{\C}_{1}^{\prime}\E_{t}^{\prime}+%
\frac{1}{p_{1}p_{2}^{2}T^{2}}\sum_{t=1}^{T}\E_{t}\hat{\C}_{1}%
\hat{\C}_{1}^{\prime}\C_{1}\F_{1,t}^{\prime}\R_{1}^{\prime}
\\
&+\frac{1}{p_{1}p_{2}^{2}T^{2}}\sum_{t=1}^{T}\R_1\F_{1,t}\C_{1}^{\prime}\hat{\C}_{1}\hat{\C}_{1}^{\prime}\C_{0}%
\F_{0,t}^{\prime}\R_{0}^{\prime}+\frac{1}{%
p_{1}p_{2}^{2}T^{2}}\sum_{t=1}^{T}\R_{0}\F_{0,t}\C_{0}^{\prime}\hat{\C}_{1}\hat{\C}_{1}^{\prime} \C_{1}\F_{1,t}^{\prime}\R_{1}^{\prime} \\
&+\frac{1}{p_{1}p_{2}^{2}T^{2}}\sum_{t=1}^{T}\R_{0}\F_{0,t}%
\C_{0}^{\prime}\hat{\C}_{1}\hat{\C}_{1}^{\prime}%
\E_{t}^{\prime}+\frac{1}{p_{1}p_{2}^{2}T^{2}}\sum_{t=1}^{T}\E_{t}\hat{\C}_{1}\hat{\C}_{1}^{\prime}\C_{0}\F_{0,t}^{\prime}\R_{0}^{\prime} \\
&=I+II+III+IV+IV^{\prime}+V+V^{\prime}+VI+VI^{\prime}.
\end{align*}%
We begin by finding bounds for $II$, $III$, $IV$, $V$, and $VI$, using $%
h_{R_{1}}=h_{C_{1}}=h_{R_{0}}=h_{C_{1}}=1$ for simplicity whenever possible. By standard
arguments, we have%
\begin{equation*}
\left\Vert II\right\Vert _{F}\leq \frac{1}{p_{1}p_{2}^{2}T^{2}}\left\Vert 
\R_{0}\right\Vert _{F}^{2}\left\Vert \C_{0}\right\Vert
_{F}^{2}\left\Vert \hat{\C}_{1}\right\Vert _{F}^{2}\left(
\sum_{t=1}^{T}\F_{0,t}^{2}\right) =O_{a.s.}\left( \frac{1}{T}\right)
.
\end{equation*}%
Following the proof of Lemma \ref{eig-mRx}, it is immediate to see that $%
\left\Vert III\right\Vert _{F}$, $\left\Vert IV\right\Vert _{F}$, $%
\left\Vert V\right\Vert _{F}$, and $\left\Vert VI\right\Vert _{F}$ are all
dominated by $\left\Vert II\right\Vert _{F}$. This entails that%
\begin{equation*}
\lambda _{\max }\left( II+III+IV+IV^{\prime}+V+V^{\prime}+VI+VI^{\prime}\right) =O_{P}\left( \frac{1}{T}\right) ,
\end{equation*}%
and therefore%
\begin{equation*}
\lambda _{j}\left( \hat{\M}_{R_{1}}^{\dagger}\right) =O_{P}\left( 
\frac{1}{T}\right) ,
\end{equation*}%
for all $j>h_{R_{1}}$. Consider now $I$; it holds that%
\begin{align*}
I &=\frac{1}{p_{1}p_{2}^{2}T^{2}}\sum_{t=1}^{T}\R_1\F_{1,t}\C_{1}^{\prime}\C_{1}\bH_{C_{1}}\bH_{C_{1}}^{\prime}\C_{1}^{\prime}\C_{1}\F_{1,t}^{\prime}\R_{1}^{\prime}+\frac{1}{p_{1}p_{2}^{2}T^{2}} \sum_{t=1}^{T}\R_1\F_{1,t}\C_{1}^{\prime}\left( \hat{\C}_{1}- \C_{1}\bH_{C_{1}}\right) \bH_{C_{1}}^{\prime}\C_{1}^{\prime}\C_{1}\F_{1,t}^{\prime}\R_{1}^{\prime} \\
&+\frac{1}{p_{1}p_{2}^{2}T^{2}}\sum_{t=1}^{T}\R_1\F_{1,t}\C_{1}^{\prime}\C_{1}\bH_{C_{1}}\left( \hat{\C}_{1}-\C_{1}\bH_{C_{1}}\right)
^{\prime}\C_{1}\F_{1,t}^{\prime}\R_{1}^{\prime} \\
&+\frac{1}{p_{1}p_{2}^{2}T^{2}}\sum_{t=1}^{T}\R_1\F_{1,t}\C_{1}^{\prime}\left( \hat{\C}_{1}-\C_{1}\bH_{C_{1}}\right) \left( \hat{%
\C}_{1}-\C_{1}\bH_{C_{1}}\right) ^{\prime}\C_{1}\F_{1,t}^{\prime}%
\R_{1}^{\prime} \\
&=I_{a}+I_{b}+I_{b}^{\prime}+I_{c}.
\end{align*}%
It is easy to see that%
\begin{align*}
\left\Vert I_{b}\right\Vert _{F} &\leq \frac{1}{p_{1}p_{2}^{2}T^{2}}%
\left\Vert \R_1\right\Vert _{F}^{2}\left\Vert \C_{1}\right\Vert
_{F}^{3}\left\Vert \bH_{C_{1}}\right\Vert _{F}\left\Vert \hat{\C}_{1}-\C_{1}\bH_{C_{1}}\right\Vert _{F}\left( \sum_{t=1}^{T}\left\Vert \F_{1,t}\right\Vert _{F}^{2}\right) \\
&=O_{P}(1) \frac{1}{p_{1}p_{2}^{2}T^{2}}p_{1}p_{2}^{3/2}\frac{%
p_{2}^{1/2}}{T}T^{2}=O_{P}\left( \frac{1}{T}\right) ,
\end{align*}%
and the same holds for $\left\Vert I_{b}^{\prime}\right\Vert _{F}$.
Similarly%
\begin{align*}
\left\Vert I_{c}\right\Vert _{F} &\leq \frac{1}{p_{1}p_{2}^{2}T^{2}}%
\left\Vert \R_1\right\Vert _{F}^{2}\left\Vert \C_{1}\right\Vert
_{F}^{2}\left\Vert \bH_{C_{1}}\right\Vert _{F}\left\Vert \hat{\C}_{1}-\C_{1}\bH_{C_{1}}\right\Vert _{F}^{2}\left( \sum_{t=1}^{T}\left\Vert 
\F_{1,t}\right\Vert _{F}^{2}\right) \\
&=O_{P}(1) \frac{1}{p_{1}p_{2}^{2}T^{2}}p_{1}p_{2}\frac{p_{2}}{%
T^{2}}T^{2}=O_{P}\left( \frac{1}{T^{2}}\right) .
\end{align*}%
Finally, it holds that, for all $j\leq h_{R_{1}}$%
\begin{equation*}
\lambda _{j}\left( I_{a}\right) \geq \lambda _{j}\left( \frac{\R_{1}^{\prime}\R_1}{p_{1}}\right) \lambda _{\min }\left( \frac{1}{T^{2}}%
\sum_{t=1}^{T}\F_{1,t}\bH_{C_{1}}\bH_{C_{1}}^{\prime}\F_{1,t}^{\prime}\right) .
\end{equation*}%
Recall that Theorem \ref{hat-estimates} states that $\bH_{C_{1}}^{\prime} \bH_{C_{1}}=\I_{h_{R_{1}}}+o_{P}(1) $, which entails that $\bH_{C_{1}}\bH_{C_{1}}^{\prime}=\I_{h_{R_{1}}}+o_{P}\left(
1\right) $. Now the same arguments as in the proof of Lemma \ref{lambda} entail that $P\left( \lambda _{j}\left( I_{a}\right) >0\right) \geq 1-\epsilon $ for any $\epsilon >0$, which implies the desired result.
\end{proof}
\end{lemma}

\begin{lemma}
\label{spec-c-proj}We assume that Assumptions \ref{as-1}-\ref{as-5} are satisfied. Then it holds that $\left\Vert \left( \Lambda _{C_{1}}^{\dagger}\right) ^{-1}\right\Vert =O_{P}(1) $.

\begin{proof}
The proof is essentially the same as that of Lemma \ref{spec-eig-mRx}.
\end{proof}
\end{lemma}

\begin{lemma}
\label{r-orth}We assume that Assumptions \ref{as-1}-\ref{as-5} are
satisfied. Then it holds that 
\begin{equation*}
\left\Vert \hat{\R}_{1,\perp}-\R_{1,\perp}\right\Vert
_{F}=O_{P}\left( \frac{1}{T}\right) .
\end{equation*}

\begin{proof}
Recall that, by construction, $\hat{\R}_1^{\prime}\hat{\R}_1=p_{1}\I_{h_{R_{1}}}$; recall the identification restriction $\R_{1}^{\prime}\R_{1}=p_{1}\I_{h_{R_{1}}}$; and finally, recall
(\ref{h-r-orthogonal}), viz.%
\begin{equation*}
\bH_{R_{1}}^{\prime}\bH_{R_{1}}=\I_{h_{R_{1}}}+O_{P}\left( \frac{1%
}{T}\right) .
\end{equation*}%
Then, using Theorem \ref{hat-estimates}, it holds that%
\begin{align*}
&\left\Vert \hat{\R}_{1,\perp}-\R_{1,\perp}\right\Vert
_{F} \\
&=\left\Vert \I_{p_{1}}-\hat{\R}_1\left( \hat{\R}_1^{\prime}\hat{\R}_1\right) ^{-1}\hat{\R}_1^{\prime}-\left( \I_{p_{1}}-\R_1\left( \R_{1}^{\prime}\R%
\right) ^{-1}\R_{1}^{\prime}\right) \right\Vert _{F} \\
&=\frac{1}{p_{1}}\left\Vert \R_1\R_{1}^{\prime}-\hat{\R}_1%
\hat{\R}_1^{\prime}\right\Vert _{F}+O_{P}\left( \frac{1}{T}\right)
\\
&\leq \frac{1}{p_{1}}\left\Vert \left( \hat{\R}_1-\R_1\bH%
_{R_{1}}\right) \left( \hat{\R}_1-\R_1\bH_{R_{1}}\right) ^{\prime}-%
\R_1\R_{1}^{\prime}\right\Vert _{F}+O_{P}\left( \frac{1}{T}\right) \\
&\leq \frac{2}{p_{1}}\left\Vert \hat{\R}_1-\R_1\bH%
_{R_{1}}\right\Vert _{F}\left\Vert \R_1\bH_{R_{1}}^{\prime}\right\Vert _{F}+%
\frac{1}{p_{1}}\left\Vert \hat{\R}_1-\R_1\bH_{R_{1}}\right\Vert
_{F}^{2}+O_{P}\left( \frac{1}{T}\right) \\
&=O_{P}\left( \frac{1}{T}\right) +O_{P}\left( \frac{1}{T^{2}}\right) .
\end{align*}
\end{proof}
\end{lemma}

\begin{lemma}
\label{c-orth}We assume that Assumptions \ref{as-1}-\ref{as-5} are
satisfied. Then it holds that 
\begin{equation*}
\left\Vert \hat{\C}_{1,\perp}-\C_{1,\perp}\right\Vert
_{F}=O_{P}\left( \frac{1}{T}\right) .
\end{equation*}

\begin{proof}
The proof follows from the same arguments as the proof of Lemma \ref{r-orth}.
\end{proof}
\end{lemma}

\begin{lemma}
\label{eig-orth-r}We assume that Assumptions \ref{as-1}-\ref{as-5} are
satisfied. Then there exists a positive constant $c_{0}$ such that 
\begin{equation*}
\lambda _{j}\left( \M_{X}^{R_{1},\perp}\right) =c_{0}+o_{P}\left(
1\right) \text{,}
\end{equation*}%
for all $j\leq h_{R_{1}}$, and%
\begin{equation*}
\lambda _{j}\left( \M_{X}^{R_{1},\perp}\right) =O_{P}\left( \frac{1%
}{p_{2}^{1/2}T^{1/2}}\right) +O_{P}\left( \frac{1}{p_{1}p_{2}}\right) ,
\end{equation*}%
for all $j>h_{R_{1}}$.

\begin{proof}
Some arguments are repetitive and similar to the proof of Lemma \ref{eig-mRx}%
, and therefore we omit them when possible for brevity. By construction%
\begin{eqnarray*}
\M_{X}^{R_{1},\perp} &=&\frac{1}{p_{1}p_{2}^{2}T}\sum_{t=1}^{T}%
\R_{1}\F_{1,t}\C_{1}^{\prime}\hat{\C}%
_{1,\perp}\hat{\C}_{1,\perp}^{\prime}\C_{1}\F%
_{1,t}^{\prime}\R_{1}^{\prime} \\
&&+\frac{1}{p_{1}p_{2}^{2}T}\sum_{t=1}^{T}\R_{0}\F_{0,t}%
\C_{0}^{\prime}\hat{\C}_{1,\perp}\hat{\C}%
_{1,\perp}^{\prime}\C_{0}\F_{0,t}^{\prime}\R%
_{0}^{\prime} \\
&&+\frac{1}{p_{1}p_{2}^{2}T}\sum_{t=1}^{T}\E_{t}\hat{\C}%
_{1,\perp}\hat{\C}_{1,\perp}^{\prime}\E_{t}^{\prime}+%
\frac{1}{p_{1}p_{2}^{2}T}\sum_{t=1}^{T}\R_{1}\F_{1,t}\C_{1}^{\prime}\hat{\C}_{1,\perp}\hat{\C}_{1,\perp
}^{\prime}\E_{t}^{\prime} \\
&&+\left( \frac{1}{p_{1}p_{2}^{2}T}\sum_{t=1}^{T}\R_{1}\F%
_{1,t}\C_{1}^{\prime}\hat{\C}_{1,\perp}\hat{\C}_{1,\perp}^{\prime}\E_{t}^{\prime}\right) ^{\prime}+%
\frac{1}{p_{1}p_{2}^{2}T}\sum_{t=1}^{T}\R_{1}\F_{1,t}\C_{1}^{\prime}\hat{\C}_{1,\perp}\hat{\C}_{1,\perp
}^{\prime}\C_{0}\F_{0,t}^{\prime}\R_{1}^{\prime}
\\
&&+\left( \frac{1}{p_{1}p_{2}^{2}T}\sum_{t=1}^{T}\R_{1}\F%
_{1,t}\C_{1}^{\prime}\hat{\C}_{1,\perp}\hat{\C}_{1,\perp}^{\prime}\C_{0}\F_{0,t}^{\prime}%
\R_{0}^{\prime}\right) ^{\prime} \\
&&+\frac{1}{p_{1}p_{2}^{2}T}\sum_{t=1}^{T}\R_{0}\F_{0,t}%
\C_{0}^{\prime}\hat{\C}_{1,\perp}\hat{\C}%
_{1,\perp}^{\prime}\E_{t}^{\prime}+\left( \frac{1}{p_{1}p_{2}^{2}T%
}\sum_{t=1}^{T}\R_{0}\F_{0,t}\C_{0}^{\prime}%
\hat{\C}_{1,\perp}\hat{\C}_{1,\perp}^{\prime}%
\E_{t}^{\prime}\right) ^{\prime} \\
&=&I+II+III+IV+IV^{\prime}+V+V^{\prime}+VI+VI^{\prime}.
\end{eqnarray*}%
We begin by noting that we can always write%
\begin{equation}
\C_{1}^{\prime}\hat{\C}_{1,\perp}=\C%
_{1}^{\prime}\left( \hat{\C}_{1,\perp}-\C_{1,\perp
}\right) .  \label{orth-prod-c}
\end{equation}%
We will work under the restrictions $h_{R_{1}}=h_{C_{1}}=1$, for simplicity
and with no loss of generality, and, when possible, $h_{R_{1}}=h_{C_{1}}=1$.
We begin by studying%
\begin{eqnarray*}
\left\Vert I\right\Vert _{F} &=&\left\Vert \frac{1}{p_{1}p_{2}^{2}T}%
\sum_{t=1}^{T}\R_{1}\F_{1,t}\C_{1}^{\prime}\left( 
\hat{\C}_{1,\perp}-\C_{1,\perp}\right) \left( \hat{%
\C}_{1,\perp}-\C_{1,\perp}\right) ^{\prime}\C_{1}%
\F_{1,t}^{\prime}\R_{1}^{\prime}\right\Vert _{F} \\
&\leq &\frac{1}{p_{1}p_{2}^{2}T}\left\Vert \R_{1}\right\Vert
_{F}^{2}\left\Vert \C_{1}\right\Vert _{F}^{2}\left\Vert \hat{%
\C}_{1,\perp}-\C_{1,\perp}\right\Vert _{F}^{2}\left(
\sum_{t=1}^{T}\left\Vert \F_{1,t}\right\Vert _{F}^{2}\right)  \\
&=&O_{P}\left( T^{2}\right) \frac{1}{p_{1}p_{2}^{2}T}p_{1}p_{2}\frac{1}{T^{2}%
}=O_{P}\left( \frac{1}{p_{2}T}\right) ,
\end{eqnarray*}%
by virtue of (\ref{orth-prod-c}) and Lemma \ref{c-orth}. Using the fact that 
$\hat{\C}_{1,\perp}$ is symmetric and idempotent, we have%
\begin{equation*}
\hat{\C}_{1,\perp}\hat{\C}_{1,\perp}^{\prime
}=\left( \hat{\C}_{1,\perp}\right) ^{2}=\hat{\C}%
_{1,\perp};
\end{equation*}%
hence%
\begin{eqnarray*}
III &=&\frac{1}{p_{1}p_{2}^{2}T}\sum_{t=1}^{T}\E_{t}\hat{\C}_{1,\perp}\E_{t}^{\prime} \\
&=&\frac{1}{p_{1}p_{2}^{2}T}\sum_{t=1}^{T}\E_{t}\C_{1,\perp}%
\E_{t}^{\prime}+\frac{1}{p_{1}p_{2}^{2}T}\sum_{t=1}^{T}\E%
_{t}\left( \hat{\C}_{1,\perp}-\C_{1,\perp}\right) 
\E_{t}^{\prime} \\
&=&III_{a}+III_{b}.
\end{eqnarray*}%
Seeing as%
\begin{equation*}
\sum_{t=1}^{T}\E_{t}\C_{1,\perp}\E_{t}^{\prime
}=\sum_{t=1}^{T}E\left( \E_{t}\C_{1,\perp}\E%
_{t}^{\prime}\right) +\sum_{t=1}^{T}\left( \E_{t}\C%
_{1,\perp}\E_{t}^{\prime}-E\left( \E_{t}\C%
_{1,\perp}\E_{t}^{\prime}\right) \right) ,
\end{equation*}%
we have (denoting the element in position $j$, $h$ of $\C_{1,\perp}$
as $c_{\perp ,hj}$)%
\begin{eqnarray*}
&&\lambda _{\max }\left( \sum_{t=1}^{T}E\left( \E_{t}\C%
_{1,\perp}\E_{t}^{\prime}\right) \right)  \\
&\leq &\sum_{t=1}^{T}\lambda _{\max }\left( E\left( \E_{t}\C%
_{1,\perp}\E_{t}^{\prime}\right) \right) \leq \left( \max_{1\leq
h,j\leq p_{2}}\left( c_{\perp ,hj}\right) ^{2}\right)
\sum_{t=1}^{T}\max_{1\leq i\leq
p_{1}}\sum_{k=1}^{p_{1}}\sum_{j,h=1}^{p_{2}}\left\vert E\left(
e_{ij,t}e_{kh,t}\right) \right\vert  \\
&\leq &c_{0}Tp_{2},
\end{eqnarray*}%
and, after some algebra 
\begin{eqnarray*}
&&E\left\Vert \sum_{t=1}^{T}\left( \E_{t}\C_{1,\perp}%
\E_{t}^{\prime}-E\left( \E_{t}\C_{1,\perp}\E_{t}^{\prime}\right) \right) \right\Vert _{F}^{2} \\
&\leq
&c_{0}\sum_{i,j=1}^{p_{1}}\sum_{t,s=1}^{T}%
\sum_{h_{1},h_{2},h_{3},h_{4}=1}^{p_{2}}\left\vert \cov\left(
e_{ih_{1},t}e_{jh_{2},t},e_{ih_{3},s}e_{jh_{4},s}\right) \right\vert  \\
&\leq &c_{0}p_{1}^{2}Tp_{2}^{3},
\end{eqnarray*}%
by Assumption \ref{as-3}\textit{(iii)}, whence ultimately 
\begin{equation*}
\left\Vert III_{a}\right\Vert _{F}=O_{P}\left( \frac{1}{p_{1}p_{2}}\right)
+O_{P}\left( \frac{1}{p_{2}^{1/2}T^{1/2}}\right) .
\end{equation*}%
Also, by similar calculations as above, it is not hard to see that $%
\sum_{t=1}^{T}\left\Vert \E_{t}\right\Vert _{F}^{2}=O_{P}\left(
p_{1}p_{2}T\right) $, and therefore%
\begin{eqnarray*}
\left\Vert III_{b}\right\Vert _{F} &\leq &\frac{1}{p_{1}p_{2}^{2}T}%
\sum_{t=1}^{T}\left\Vert \E_{t}\right\Vert _{F}^{2}\left\Vert 
\hat{\C}_{1,\perp}-\C_{1,\perp}\right\Vert _{F} \\
&=&O_{P}(1) \frac{1}{p_{1}p_{2}^{2}T}p_{1}p_{2}T\frac{1}{T}%
=O_{P}\left( \frac{1}{p_{2}T}\right) ,
\end{eqnarray*}%
whence ultimately 
\begin{equation}
\left\Vert III\right\Vert _{F}=O_{P}\left( \frac{1}{p_{1}p_{2}}\right)
+O_{P}\left( \frac{1}{p_{2}^{1/2}T^{1/2}}\right) .  \label{3F-error}
\end{equation}%
Further, again exploiting idempotency and (\ref{orth-prod-c})%
\begin{eqnarray*}
\left\Vert IV\right\Vert _{F} &=&\left\Vert \frac{1}{p_{1}p_{2}^{2}T}%
\sum_{t=1}^{T}\R_{1}\F_{1,t}\C_{1}^{\prime}\left( 
\hat{\C}_{1,\perp}-\C_{1,\perp}\right) \E%
_{t}^{\prime}\right\Vert _{F} \\
&\leq &\frac{1}{p_{1}p_{2}^{2}T}\left\Vert \R_{1}\right\Vert
_{F}\left\Vert \C_{1}\right\Vert _{F}\left\Vert \hat{\C}%
_{1,\perp}-\C_{1,\perp}\right\Vert _{F}\left\Vert \sum_{t=1}^{T}%
\F_{1,t}\E_{t}\right\Vert _{F} \\
&=&O_{P}(1) \frac{1}{p_{1}p_{2}^{2}T}p_{1}^{1/2}p_{2}^{1/2}\frac{%
1}{T}\left( p_{1}^{1/2}p_{2}^{1/2}T\right) =O_{P}\left( \frac{1}{p_{2}T}%
\right) ,
\end{eqnarray*}%
having used Lemma \ref{errorfactor}. Similarly 
\begin{eqnarray*}
\left\Vert V\right\Vert _{F} &=&\left\Vert \frac{1}{p_{1}p_{2}^{2}T}%
\sum_{t=1}^{T}\R_{1}\F_{1,t}\C_{1}^{\prime}\left( 
\hat{\C}_{1,\perp}-\C_{1,\perp}\right) \C_{0}%
\F_{0,t}^{\prime}\R_{1}^{\prime}\right\Vert _{F} \\
&\leq &\frac{1}{p_{1}p_{2}^{2}T}\left\Vert \R_{1}\right\Vert
_{F}^{2}\left\Vert \C_{1}\right\Vert _{F}\left\Vert \C%
_{0}\right\Vert _{F}\left\Vert \hat{\C}_{1,\perp}-\C%
_{1,\perp}\right\Vert _{F}\left\Vert \sum_{t=1}^{T}\F_{1,t}\F_{0,t}^{\prime}\right\Vert _{F} \\
&=&O_{P}(1) \frac{1}{p_{1}p_{2}^{2}T}p_{1}p_{2}^{1/2}p_{2}^{1/2}%
\frac{1}{T}T=O_{P}\left( \frac{1}{p_{2}T}\right) .
\end{eqnarray*}%
Finally, we have%
\begin{eqnarray*}
VI &=&\frac{1}{p_{1}p_{2}^{2}T}\sum_{t=1}^{T}\R_{0}\F_{0,t}%
\C_{0}^{\prime}\C_{1,\perp}\E_{t}^{\prime}+\frac{1%
}{p_{1}p_{2}^{2}T}\sum_{t=1}^{T}\R_{0}\F_{0,t}\C%
_{0}^{\prime}\left( \hat{\C}_{1,\perp}-\C_{1,\perp
}\right) \E_{t}^{\prime} \\
&=&VI_{a}+VI_{b}.
\end{eqnarray*}%
It holds that%
\begin{eqnarray*}
E\left\Vert \sum_{t=1}^{T}\C_{1,\perp}\E_{t}^{\prime}%
\F_{0,t}\right\Vert _{F}^{2}
&=&E\sum_{i=1}^{p_{1}}\sum_{h=1}^{p_{2}}\left(
\sum_{t=1}^{T}\sum_{j=1}^{p_{2}}c_{\perp ,hj}e_{ij,t}\F_{0,t}\right)
^{2} \\
&=&\sum_{i=1}^{p_{1}}\sum_{h=1}^{p_{2}}\sum_{t,s=1}^{T}%
\sum_{j_{1},j_{2}=1}^{p_{2}}c_{\perp ,hj_{1}}c_{\perp ,hj_{2}}E\left(
e_{ij_{1},t}e_{ij_{2},s}\right) E\left( \F_{0,t}\F%
_{0,s}\right)  \\
&\leq &\left( \max_{1\leq h,j\leq p_{2}}c_{\perp ,hj}\right)
^{2}\sum_{i=1}^{p_{1}}\sum_{h=1}^{p_{2}}\sum_{t,s=1}^{T}%
\sum_{j_{1},j_{2}=1}^{p_{2}}\left\vert E\left(
e_{ij_{1},t}e_{ij_{2},s}\right) \right\vert \left\vert E\left( \F%
_{0,t}\F_{0,s}\right) \right\vert  \\
&\leq &c_{0}\left( \max_{1\leq h,j\leq p_{2}}c_{\perp ,hj}\right)
^{2}\sum_{i=1}^{p_{1}}\sum_{h=1}^{p_{2}}\sum_{t,s=1}^{T}%
\sum_{j_{1},j_{2}=1}^{p_{2}}\left\vert E\left(
e_{ij_{1},t}e_{ij_{2},s}\right) \right\vert \leq c_{1}p_{1}p_{2}^{2}T,
\end{eqnarray*}%
and therefore we have%
\begin{eqnarray*}
\left\Vert VI_{a}\right\Vert _{F} &\leq &\frac{1}{p_{1}p_{2}^{2}T}\left\Vert 
\R_{0}\right\Vert _{F}\left\Vert \C_{0}\right\Vert
_{F}\left\Vert \sum_{t=1}^{T}\C_{1,\perp}\E_{t}^{\prime}%
\F_{0,t}\right\Vert _{F} \\
&=&O_{P}(1) \frac{1}{p_{1}p_{2}^{2}T}%
p_{1}^{1/2}p_{2}^{1/2}p_{1}^{1/2}p_{2}T^{1/2}=O_{P}\left( \frac{1}{%
p_{2}^{1/2}T^{1/2}}\right) .
\end{eqnarray*}%
Also%
\begin{eqnarray*}
\left\Vert VI_{b}\right\Vert _{F} &\leq &\frac{1}{p_{1}p_{2}^{2}T}\left\Vert 
\R_{0}\right\Vert _{F}\left\Vert \C_{0}\right\Vert
_{F}\left\Vert \hat{\C}_{1,\perp}-\C_{1,\perp
}\right\Vert _{F}\left\Vert \sum_{t=1}^{T}\F_{0,t}\E%
_{t}^{\prime}\right\Vert _{F} \\
&=&O_{P}(1) \frac{1}{p_{1}p_{2}^{2}T}p_{1}^{1/2}p_{2}^{1/2}\frac{%
1}{T}p_{1}^{1/2}p_{2}^{1/2}T^{1/2}=O_{P}\left( \frac{1}{p_{2}T^{3/2}}\right)
,
\end{eqnarray*}%
so that ultimately 
\begin{equation*}
\left\Vert VI\right\Vert _{F}=O_{P}\left( \frac{1}{p_{2}^{1/2}T^{1/2}}%
\right) .
\end{equation*}%
Putting all together, it follows that%
\begin{eqnarray*}
&&\lambda _{\max }\left( I+III+IV+IV^{\prime}+V+V^{\prime}+VI+VI^{\prime
}\right)  \\
&=&O_{P}\left( \frac{1}{p_{2}^{1/2}T^{1/2}}\right) +O_{P}\left( \frac{1}{%
p_{1}p_{2}}\right) .
\end{eqnarray*}%
Consider now $II$; by construction, $\lambda _{j}\left( II\right) =0$
whenever $j>h_{R_{1}}$. Further%
\begin{eqnarray*}
II &=&\frac{1}{p_{1}p_{2}^{2}T}\sum_{t=1}^{T}\R_{0}\F_{0,t}%
\C_{0}^{\prime}\hat{\C}_{1,\perp}\C_{0}\F_{0,t}^{\prime}\R_{0}^{\prime} \\
&=&\frac{1}{p_{1}p_{2}^{2}T}\sum_{t=1}^{T}\R_{0}\F_{0,t}%
\C_{0}^{\prime}\C_{1,\perp}\C_{0}\F%
_{0,t}^{\prime}\R_{0}^{\prime}+\frac{1}{p_{1}p_{2}^{2}T}%
\sum_{t=1}^{T}\R_{0}\F_{0,t}\C_{0}^{\prime}\left( 
\hat{\C}_{1,\perp}-\C_{1,\perp}\right) \C_{0}%
\F_{0,t}^{\prime}\R_{0}^{\prime} \\
&=&II_{a}+II_{b}.
\end{eqnarray*}%
Note that, by standard algebra 
\begin{equation*}
\left\Vert \C_{0}^{\prime}\C_{1,\perp}\C%
_{0}\right\Vert _{F}=c_{0}p_{2}^{2},
\end{equation*}%
and therefore 
\begin{eqnarray*}
\left\Vert II_{a}\right\Vert _{F} &=&\frac{1}{p_{1}p_{2}^{2}T}\sum_{t=1}^{T}%
\R_{0}E\left( \F_{0,t}\C_{0}^{\prime}\C%
_{1,\perp}\C_{0}\F_{0,t}^{\prime}\right) \R%
_{0}^{\prime} \\
&&+\frac{1}{p_{1}p_{2}^{2}T}\sum_{t=1}^{T}\R_{0}\left( \F%
_{0,t}\C_{0}^{\prime}\C_{1,\perp}\C_{0}\F%
_{0,t}^{\prime}-E\left( \F_{0,t}\C_{0}^{\prime}\C%
_{1,\perp}\C_{0}\F_{0,t}^{\prime}\right) \right) \R%
_{0}^{\prime} \\
&=&II_{a,1}+II_{a,2}.
\end{eqnarray*}%
Using Assumptions \ref{as-2} and \ref{as-4}, it is easy to see via tedious
but elementary passages that $\lambda _{j}\left( II_{a,1}\right) \geq c_{0}$%
, and $\lambda _{\max }\left( II_{a,2}\right) =O_{P}\left( T^{-1/2}\right) $%
. Moreover%
\begin{eqnarray*}
\left\Vert II_{b}\right\Vert _{F} &\leq &\frac{1}{p_{1}p_{2}^{2}T}\left\Vert 
\R_{0}\right\Vert _{F}^{2}\left\Vert \C_{0}\right\Vert
_{F}^{2}\left\Vert \hat{\C}_{1,\perp}-\C_{1,\perp
}\right\Vert _{F}\sum_{t=1}^{T}\left\Vert \F_{0,t}\right\Vert
_{F}^{2} \\
&=&O_{P}\left( T\right) \frac{1}{p_{1}p_{2}^{2}T}p_{1}p_{2}\frac{1}{T}%
=O_{P}\left( \frac{1}{p_{2}T}\right) .
\end{eqnarray*}%
Putting all together, it follows that%
\begin{equation*}
\lambda _{j}\left( II\right) =c_{0}+o_{P}(1) ,
\end{equation*}%
for all $j\leq h_{R_{1}}$. The desired result now follows from Weyl's
inequality.
\end{proof}
\end{lemma}

\begin{lemma}
\label{eig-tilde-c}We assume that Assumptions \ref{as-1}-\ref{as-5} are
satisfied. Then there exists a positive constant $c_{0}$ such that%
\begin{equation*}
\lambda _{j}\left( \M_{C_1,\perp}\right) =c_{0}+o_{P}\left(
1\right) \text{,}
\end{equation*}%
for all $j\leq h_{C_{0}}$, and%
\begin{equation*}
\lambda _{j}\left( \M_{C_1,\perp}\right) =O_{P}\left( \frac{1}{%
p_{1}^{1/2}T^{1/2}}\right) +O_{P}\left( \frac{1}{p_{1}p_{2}}\right) ,
\end{equation*}%
for all $j>h_{C_{0}}$.

\begin{proof}
The proof is the same as the proof of Lemma \ref{eig-orth-r}, \textit{mutatis mutandis}.
\end{proof}
\end{lemma}

\begin{lemma}
\label{yong}We assume that Assumptions \ref{as-1}-\ref{as-5} hold. Then it
holds that%
\begin{eqnarray}
&&\sigma _{\max }\left[ \frac{1}{p_{1}}\hat{\R}_{1,\perp
}^{\prime}\left( \hat{\R}_{0}-\R_{0}\hat{\bH}%
_{R_{0}}\right) \right]   \label{yong1} \\
&=&O_{P}\left( \frac{p_{1}^{1/2}}{p_{1}p_{2}}\right) +O_{P}\left( \frac{%
p_{1}^{1/2}}{p_{2}T}\right) +O_{P}\left( \frac{p_{1}^{1/2}}{%
p_{1}^{1/2}p_{2}^{1/2}T^{1/2}}\right) ,  \label{yong2}
\end{eqnarray}%
and 
\begin{eqnarray}
&&\sigma _{\max }\left[ \frac{1}{p_{2}}\hat{\C}_{1,\perp
}^{\prime}\left( \hat{\C}_{0}-\C_{0}\hat{\bH}%
_{C_{0}}\right) \right]  \\
&=&O_{P}\left( \frac{p_{2}^{1/2}}{p_{1}p_{2}}\right) +O_{P}\left( \frac{%
p_{2}^{1/2}}{p_{1}T}\right) +O_{P}\left( \frac{p_{2}^{1/2}}{%
p_{1}^{1/2}p_{2}^{1/2}T^{1/2}}\right) .  \notag
\end{eqnarray}

\begin{proof}
We only show (\ref{yong1}); the proof of (\ref{yong2}) is essentially the
same. It holds that 
\begin{eqnarray*}
&&\frac{1}{p_{1}}\hat{\R}_{1,\perp}^{\prime}\left( \hat{%
\R}_{0}-\R_{0}\hat{\bH}_{R_{0}}\right)  \\
&=&\frac{1}{p_{1}^{2}p_{2}^{2}T}\hat{\R}_{1,\perp}^{\prime
}\sum_{t=1}^{T}\R_{1}\F_{1,t}\C_{1}^{\prime}%
\hat{\C}_{1,\perp}\hat{\C}_{1,\perp}^{\prime}%
\C_{1}\F_{1,t}^{\prime}\R_{1}^{\prime}\hat{%
\R}_{0}\Lambda _{R_{0}}^{-1} \\
&&+\frac{1}{p_{1}^{2}p_{2}^{2}T}\hat{\R}_{1,\perp}^{\prime
}\sum_{t=1}^{T}\R_{0}\F_{0,t}\C_{0}^{\prime}\C_{1,\perp}\left( \hat{\C}_{1,\perp}-\C_{1,\perp
}\right) ^{\prime}\C_{0}\F_{0,t}^{\prime}\R%
_{0}^{\prime}\hat{\R}_{0}\Lambda _{R_{0}}^{-1} \\
&&+\frac{1}{p_{1}^{2}p_{2}^{2}T}\hat{\R}_{1,\perp}^{\prime
}\sum_{t=1}^{T}\R_{0}\F_{0,t}\C_{0}^{\prime}\left( 
\hat{\C}_{1,\perp}-\C_{1,\perp}\right) \C%
_{1,\perp}^{\prime}\C_{0}\F_{0,t}^{\prime}\R%
_{0}^{\prime}\hat{\R}_{0}\Lambda _{R_{0}}^{-1} \\
&&+\frac{1}{p_{1}^{2}p_{2}^{2}T}\hat{\R}_{1,\perp}^{\prime
}\sum_{t=1}^{T}\R_{0}\F_{0,t}\C_{0}^{\prime}\left( 
\hat{\C}_{1,\perp}-\C_{1,\perp}\right) \left( \hat{%
\C}_{1,\perp}-\C_{1,\perp}\right) ^{\prime}\C_{0}%
\F_{0,t}^{\prime}\R_{0}^{\prime}\hat{\R}%
_{0}\Lambda _{R_{0}}^{-1} \\
&&+\frac{1}{p_{0}^{2}p_{2}^{2}T}\hat{\R}_{1,\perp}^{\prime
}\sum_{t=1}^{T}\E_{t}\hat{\C}_{1,\perp}\hat{\C}_{1,\perp}^{\prime}\E_{t}^{\prime}\hat{\R}%
_{0}\Lambda _{R_{0}}^{-1}+\frac{1}{p_{1}^{2}p_{2}^{2}T}\hat{\R}%
_{1,\perp}^{\prime}\sum_{t=1}^{T}\R_{1}\F_{1,t}\C%
_{1}^{\prime}\hat{\C}_{1,\perp}\hat{\C}_{1,\perp
}^{\prime}\E_{t}^{\prime}\hat{\R}_{0}\Lambda
_{R_{0}}^{-1} \\
&&+\frac{1}{p_{1}^{2}p_{2}^{2}T}\hat{\R}_{1,\perp}^{\prime
}\left( \sum_{t=1}^{T}\R_{1}\F_{1,t}\C_{1}^{\prime}%
\hat{\C}_{1,\perp}\hat{\C}_{1,\perp}^{\prime}%
\E_{t}^{\prime}\right) ^{\prime}\hat{\R}_{0}\Lambda
_{R_{0}}^{-1} \\
&&+\frac{1}{p_{1}^{2}p_{2}^{2}T}\hat{\R}_{1,\perp}^{\prime
}\sum_{t=1}^{T}\R_{1}\F_{1,t}\C_{1}^{\prime}%
\hat{\C}_{1,\perp}\hat{\C}_{1,\perp}^{\prime}%
\C_{0}\F_{0,t}^{\prime}\R_{0}^{\prime}\hat{%
\R}_{0}\Lambda _{R_{0}}^{-1} \\
&&+\frac{1}{p_{1}^{2}p_{2}^{2}T}\hat{\R}_{1,\perp}^{\prime
}\left( \sum_{t=1}^{T}\R_{1}\F_{1,t}\C_{1}^{\prime}%
\hat{\C}_{1,\perp}\hat{\C}_{1,\perp}^{\prime}%
\C_{0}\F_{0,t}^{\prime}\R_{0}^{\prime}\right)
^{\prime}\hat{\R}_{0}\Lambda _{R_{0}}^{-1} \\
&&+\frac{1}{p_{1}^{2}p_{2}^{2}T}\hat{\R}_{1,\perp}^{\prime
}\sum_{t=1}^{T}\R_{0}\F_{0,t}\C_{0}^{\prime}%
\hat{\C}_{1,\perp}\hat{\C}_{1,\perp}^{\prime}%
\E_{t}^{\prime}\hat{\R}_{0}\Lambda _{R_{0}}^{-1} \\
&&+\frac{1}{p_{1}^{2}p_{2}^{2}T}\hat{\R}_{1,\perp}^{\prime
}\left( \sum_{t=1}^{T}\R_{0}\F_{0,t}\C_{0}^{\prime}%
\hat{\C}_{1,\perp}\hat{\C}_{1,\perp}^{\prime}%
\E_{t}^{\prime}\right) ^{\prime}\hat{\R}_{0}\Lambda
_{R_{0}}^{-1} \\
&=&I+II+III+IV+V+VI+VII+VIII+IX+X+XI.
\end{eqnarray*}%
We have%
\begin{eqnarray*}
\left\Vert I\right\Vert _{F} &=&\left\Vert \frac{1}{p_{1}^{2}p_{2}^{2}T}%
\left( \hat{\R}_{1,\perp}-\R_{1,\perp}\right)
\sum_{t=1}^{T}\R_{1}\F_{1,t}\C_{1}^{\prime}\left( 
\hat{\C}_{1,\perp}-\C_{1,\perp}\right) \left( \hat{%
\C}_{1,\perp}-\C_{1,\perp}\right) \C_{1}\F%
_{1,t}^{\prime}\R_{1}^{\prime}\hat{\R}_{0}\Lambda
_{R_{0}}^{-1}\right\Vert _{F} \\
&\leq &\frac{1}{p_{1}^{2}p_{2}^{2}T}\left\Vert \hat{\R}_{1,\perp
}-\R_{1,\perp}\right\Vert _{F}\left\Vert \R_{1}\right\Vert
_{F}^{2}\left\Vert \hat{\R}_{0}\right\Vert _{F}\left\Vert 
\C_{1}\right\Vert _{F}^{2}\left\Vert \hat{\C}_{1,\perp}-%
\C_{1,\perp}\right\Vert _{F}^{2}\left\Vert \sum_{t=1}^{T}\F%
_{1,t}^{2}\right\Vert _{F}\left\Vert \Lambda _{R_{0}}^{-1}\right\Vert _{F} \\
&=&O_{P}(1) \frac{1}{p_{1}^{2}p_{2}^{2}T}\frac{1}{T}%
p_{1}p_{1}^{1/2}p_{2}\frac{1}{T^{2}}T^{2}=O_{P}\left( \frac{1}{%
p_{1}^{1/2}p_{2}T^{2}}\right) .
\end{eqnarray*}%
Note now that%
\begin{equation*}
\hat{\C}_{1,\perp}\left( \hat{\C}_{1,\perp}-%
\C_{1,\perp}\right) =\hat{\C}_{1,\perp}\I%
_{p_{2}}-\hat{\C}_{1,\perp}\C_{1,\perp}=\hat{%
\C}_{1,\perp}\left[ \C\left( \C^{\prime}\C%
\right) ^{-1}\C^{\prime}\right] ,
\end{equation*}%
and since $\C\left( \C^{\prime}\C\right) ^{-1}%
\C^{\prime}$ is an idempotent matrix with $h_{C_{1}}$\ nonzero
eigenvalues, we have%
\begin{equation}
\left\Vert \hat{\C}_{1,\perp}\left( \hat{\C}%
_{1,\perp}-\C_{1,\perp}\right) \right\Vert _{F}=O\left( \frac{1}{T}%
\right) ,  \label{c-orth-large-1}
\end{equation}%
and similarly%
\begin{equation}
\left\Vert \C_{1,\perp}\left( \hat{\C}_{1,\perp}-%
\C_{1,\perp}\right) \right\Vert _{F}=O\left( \frac{1}{T}\right) .
\label{c-orth-large-2}
\end{equation}%
Hence we have%
\begin{eqnarray*}
\left\Vert II\right\Vert _{F} &=&\left\Vert \frac{1}{p_{1}^{2}p_{2}^{2}T}%
\hat{\R}_{1,\perp}^{\prime}\sum_{t=1}^{T}\R_{0}\F_{0,t}\C_{0}^{\prime}\C_{1,\perp}\left( \hat{\C}_{1,\perp}-\C_{1,\perp}\right) ^{\prime}\C_{0}\mathbf{F%
}_{0,t}^{\prime}\R_{0}^{\prime}\hat{\R}_{0}\Lambda
_{R_{0}}^{-1}\right\Vert _{F} \\
&\leq &\frac{1}{p_{1}^{2}p_{2}^{2}T}\left\Vert \R_{0}\right\Vert
_{F}^{2}\left\Vert \hat{\R}_{0}\right\Vert _{F}\left\Vert 
\hat{\R}_{1,\perp}\right\Vert _{F}\left\Vert \C%
_{0}\right\Vert _{F}^{2}\left\Vert \C_{1,\perp}\left( \hat{%
\C}_{1,\perp}-\C_{1,\perp}\right) \right\Vert
_{F}\left\Vert \sum_{t=1}^{T}\F_{0,t}^{2}\right\Vert _{F}\left\Vert
\Lambda _{R_{0}}^{-1}\right\Vert _{F} \\
&=&O_{P}(1) \frac{1}{p_{1}^{2}p_{2}^{2}T}%
p_{1}p_{1}^{1/2}p_{1}p_{2}\frac{1}{T}T=O_{P}\left( \frac{p_{1}^{1/2}}{p_{2}T}%
\right) ,
\end{eqnarray*}%
and the same rate holds for $III$; further%
\begin{eqnarray*}
\left\Vert IV\right\Vert _{F} &\leq &\frac{1}{p_{1}^{2}p_{2}^{2}T}\left\Vert 
\R_{0}\right\Vert _{F}^{2}\left\Vert \hat{\R}%
_{0}\right\Vert _{F}\left\Vert \hat{\R}_{1,\perp}\right\Vert
_{F}\left\Vert \C_{0}\right\Vert _{F}^{2}\left\Vert \hat{\C}_{1,\perp}-\C_{1,\perp}\right\Vert _{F}^{2}\left\Vert
\sum_{t=1}^{T}\F_{0,t}^{2}\right\Vert _{F}\left\Vert \Lambda
_{R_{0}}^{-1}\right\Vert _{F} \\
&=&O_{P}(1) \frac{1}{p_{1}^{2}p_{2}^{2}T}p_{1}^{5/2}p_{2}\frac{1%
}{T^{2}}T=O_{P}\left( \frac{p_{1}^{1/2}}{p_{2}T^{2}}\right) .
\end{eqnarray*}%
We now study $V$, using the decomposition%
\begin{eqnarray*}
V &=&\frac{1}{p_{1}^{2}p_{2}^{2}T}\R_{1,\perp}^{\prime
}\sum_{t=1}^{T}\E_{t}\C_{1,\perp}\E_{t}^{\prime}%
\R_{0}\hat{\bH}_{R_{0}}\Lambda _{R_{0}}^{-1}+\frac{1}{%
p_{1}^{2}p_{2}^{2}T}\left( \hat{\R}_{1,\perp}-\R%
_{1,\perp}\right) ^{\prime}\sum_{t=1}^{T}\E_{t}\C_{1,\perp
}\E_{t}^{\prime}\R_{0}\hat{\bH}_{R_{0}}\Lambda
_{R_{0}}^{-1} \\
&&+\frac{1}{p_{1}^{2}p_{2}^{2}T}\R_{1,\perp}^{\prime}\sum_{t=1}^{T}%
\E_{t}\left( \hat{\C}_{1,\perp}-\C_{1,\perp
}\right) ^{\prime}\E_{t}^{\prime}\R_{0}\hat{\bH}%
_{R_{0}}\Lambda _{R_{0}}^{-1} \\
&&+\frac{1}{p_{1}^{2}p_{2}^{2}T}\left( \hat{\R}_{1,\perp}-%
\R_{1,\perp}\right) ^{\prime}\sum_{t=1}^{T}\E_{t}\left( 
\hat{\C}_{1,\perp}-\C_{1,\perp}\right) ^{\prime}%
\E_{t}^{\prime}\R_{0}\hat{\bH}_{R_{0}}\Lambda
_{R_{0}}^{-1} \\
&&+\frac{1}{p_{1}^{2}p_{2}^{2}T}\R_{1,\perp}^{\prime}\sum_{t=1}^{T}%
\E_{t}\C_{1,\perp}\E_{t}^{\prime}\left( \hat{%
\R}_{0}-\R_{0}\hat{\bH}_{R_{0}}\right) \Lambda
_{R_{0}}^{-1}+V_{f} \\
&=&V_{a}+V_{b}+V_{c}+V_{d}+V_{e}+V_{f},
\end{eqnarray*}%
where $V_{f}$ is a remainder which can be shown to be dominated by the other
terms. We begin with $V_{a}$, and note that%
\begin{eqnarray*}
&&\R_{1,\perp}^{\prime}\sum_{t=1}^{T}\E_{t}\C%
_{1,\perp}\E_{t}^{\prime}\R_{0} \\
&=&\R_{1,\perp}^{\prime}\sum_{t=1}^{T}E\left( \E_{t}%
\C_{1,\perp}\E_{t}^{\prime}\right) \R_{0}+\mathbf{R%
}_{1,\perp}^{\prime}\sum_{t=1}^{T}\left[ \E_{t}\C_{1,\perp
}\E_{t}^{\prime}-E\left( \E_{t}\C_{1,\perp}\E_{t}^{\prime}\right) \right] \R_{0};
\end{eqnarray*}%
the element in position $1\leq k\leq p_{1}$ of the vector (recall we are
assuming only one factor) $\R_{1,\perp}^{\prime
}\sum_{t=1}^{T}E\left( \E_{t}\C_{1,\perp}\E%
_{t}^{\prime}\right) \R_{0}$ is%
\begin{equation*}
\sum_{t=1}^{T}\sum_{u=1}^{p_{1}}\sum_{i=1}^{p_{1}}\sum_{h=1}^{p_{2}}\sum_{%
\ell =1}^{p_{2}}p_{uikh\ell }=r_{0,u}r_{\perp ,ik}c_{\perp ,h\ell }E\left(
e_{ih,t}e_{u\ell ,t}\right) ,
\end{equation*}%
where $p_{uikh\ell }=r_{0,u}r_{\perp ,ik}c_{\perp ,h\ell }$, and therefore%
\begin{eqnarray*}
&&\sigma _{\max }\left( \R_{1,\perp}^{\prime}\sum_{t=1}^{T}E\left( 
\E_{t}\C_{1,\perp}\E_{t}^{\prime}\right) \R%
_{0}\right)  \\
&\leq &\left(
\sum_{k=1}^{p_{1}}\sum_{t=1}^{T}\sum_{u=1}^{p_{1}}\sum_{i=1}^{p_{1}}%
\sum_{h=1}^{p_{2}}\sum_{\ell =1}^{p_{2}}\left\vert p_{uikh\ell }\right\vert
\left\vert E\left( e_{ih,t}e_{u\ell ,t}\right) \right\vert \right) ^{1/2} \\
&&\times \left( \max_{1\leq k\leq
p_{1}}\sum_{t=1}^{T}\sum_{u=1}^{p_{1}}\sum_{i=1}^{p_{1}}\sum_{h=1}^{p_{2}}%
\sum_{\ell =1}^{p_{2}}\left\vert p_{uikh\ell }\right\vert \left\vert E\left(
e_{ih,t}e_{u\ell ,t}\right) \right\vert \right) ^{1/2} \\
&\leq &c_{0}\left( p_{1}^{2}p_{2}T\right) ^{1/2}\left( p_{1}p_{2}T\right)
^{1/2}\leq c_{1}p_{1}^{3/2}p_{2}T.
\end{eqnarray*}%
Also%
\begin{eqnarray*}
&&E\left\Vert \R_{1,\perp}^{\prime}\sum_{t=1}^{T}\left[ \E%
_{t}\C_{1,\perp}\E_{t}^{\prime}-E\left( \E_{t}%
\C_{1,\perp}\E_{t}^{\prime}\right) \right] \R%
_{0}\right\Vert _{F}^{2} \\
&=&\sum_{k=1}^{p_{1}}\sum_{t,s=1}^{T}%
\sum_{i_{1},i_{2},i_{3},i_{4}=1}^{p_{1}}%
\sum_{h_{1},h_{2},h_{3},h_{4}=1}^{p_{2}}p_{i_{1}i_{2}kh_{1}h_{2}}p_{i_{3}i_{4}kh_{3}h_{4}}\cov\left( e_{i_{1}h_{1},t}e_{i_{2}h_{2},t},e_{i_{3}h_{3},s}e_{i_{4}h_{4},s}\right) 
\\
&\leq
&c_{0}\sum_{k=1}^{p_{1}}\sum_{t,s=1}^{T}%
\sum_{i_{1},i_{2},i_{3},i_{4}=1}^{p_{1}}%
\sum_{h_{1},h_{2},h_{3},h_{4}=1}^{p_{2}}\left\vert \cov\left(
e_{i_{1}h_{1},t}e_{i_{2}h_{2},t},e_{i_{3}h_{3},s}e_{i_{4}h_{4},s}\right)
\right\vert  \\
&\leq &c_{0}p_{1}^{4}p_{2}^{3}T,
\end{eqnarray*}%
whence 
\begin{equation*}
\left\Vert \R_{1,\perp}^{\prime}\sum_{t=1}^{T}\left[ \E_{t}%
\C_{1,\perp}\E_{t}^{\prime}-E\left( \E_{t}\mathbf{C%
}_{1,\perp}\E_{t}^{\prime}\right) \right] \R%
_{0}\right\Vert _{F}=O_{P}\left( p_{1}^{2}p_{2}^{3/2}T^{1/2}\right) .
\end{equation*}%
Therefore we have%
\begin{equation*}
\sigma _{\max }\left( V_{a}\right) =O\left( \frac{p_{1}^{3/2}p_{2}T}{%
p_{1}^{2}p_{2}^{2}T}\right) +O_{P}\left( \frac{p_{1}^{2}p_{2}^{3/2}T^{1/2}}{%
p_{1}^{2}p_{2}^{2}T}\right) =O\left( \frac{p_{1}^{1/2}}{p_{1}p_{2}}\right)
+O_{P}\left( \frac{p_{1}^{1/2}}{p_{1}^{1/2}p_{2}^{1/2}T^{1/2}}\right) .
\end{equation*}

We now note that%
\begin{equation*}
\sigma _{\max }\left( \sum_{t=1}^{T}\E_{t}\C_{1,\perp}%
\E_{t}^{\prime}\R_{0}\right) =O_{P}\left(
p_{1}^{1/2}p_{2}T\right) +O_{P}\left( p_{1}p_{2}^{3/2}T^{1/2}\right) .
\end{equation*}%
Indeed, the element in position $1\leq i\leq p_{1}$ of $\E_{t}%
\C_{1,\perp}\E_{t}^{\prime}\R_{0}$ is given by%
\begin{equation*}
\sum_{k=1}^{p_{1}}\sum_{j,h=1}^{p_{2}}r_{0,k}c_{\perp ,jh}e_{ij,t}e_{kh,t};
\end{equation*}%
hence%
\begin{eqnarray*}
&&\sigma _{\max }\left( \sum_{t=1}^{T}E\left( \E_{t}\C%
_{1,\perp}\E_{t}^{\prime}\R_{0}\right) \right)  \\
&\leq &\left(
\sum_{t=1}^{T}\sum_{i,k=1}^{p_{1}}\sum_{j,h=1}^{p_{2}}\left\vert
r_{0,k}c_{\perp ,jh}\right\vert \left\vert E\left( e_{ij,t}e_{kh,t}\right)
\right\vert \right) ^{1/2}\left( \sum_{t=1}^{T}\max_{1\leq i\leq
p_{1}}\sum_{k=1}^{p_{1}}\sum_{j,h=1}^{p_{2}}\left\vert r_{0,k}c_{\perp
,jh}\right\vert \left\vert E\left( e_{ij,t}e_{kh,t}\right) \right\vert
\right) ^{1/2} \\
&\leq &c_{0}\left( p_{1}p_{2}T\right) ^{1/2}\left( p_{2}T\right) ^{1/2}\leq
c_{1}p_{1}^{1/2}p_{2}T;
\end{eqnarray*}%
also%
\begin{eqnarray*}
&&E\left\Vert \sum_{t=1}^{T}\left[ \E_{t}\C_{1,\perp}%
\E_{t}^{\prime}\R_{0}-E\left( \E_{t}\C%
_{1,\perp}\E_{t}^{\prime}\R_{0}\right) \right] \right\Vert
_{F}^{2} \\
&\leq
&c_{0}\sum_{i=1}^{p_{1}}\sum_{t,s=1}^{T}\sum_{h_{1},h_{2}=1}^{p_{1}}%
\sum_{h_{1},h_{2},h_{3},h_{4}=1}^{p_{2}}\left\vert \cov\left(
e_{ih_{1},t}e_{k_{1}h_{2},t},e_{ih_{3},s}e_{k_{2}h_{4},s}\right) \right\vert 
\\
&\leq &c_{0}p_{1}^{2}p_{2}^{3}T,
\end{eqnarray*}%
whence 
\begin{equation*}
\left\Vert \sum_{t=1}^{T}\left[ \E_{t}\C_{1,\perp}\E%
_{t}^{\prime}\R_{0}-E\left( \E_{t}\C_{1,\perp}%
\E_{t}^{\prime}\R_{0}\right) \right] \right\Vert
_{F}=O_{P}\left( p_{1}p_{2}^{3/2}T^{1/2}\right) .
\end{equation*}%
We therefore have%
\begin{eqnarray*}
&&\left\Vert \frac{1}{p_{1}^{2}p_{2}^{2}T}\left( \hat{\R}%
_{1,\perp}-\R_{1,\perp}\right) ^{\prime}\sum_{t=1}^{T}\E%
_{t}\C_{1,\perp}\E_{t}^{\prime}\R_{0}\hat{%
\bH}_{R_{0}}\Lambda _{R_{0}}^{-1}\right\Vert _{F} \\
&\leq &\frac{1}{p_{1}^{2}p_{2}^{2}T}\left\Vert \hat{\R}_{1,\perp
}-\R_{1,\perp}\right\Vert _{F}\left\Vert \hat{\bH}%
_{R_{0}}\right\Vert _{F}\left\Vert \Lambda _{R_{0}}^{-1}\right\Vert
_{F}\sigma _{\max }\left( \sum_{t=1}^{T}\E_{t}\C_{1,\perp}%
\E_{t}^{\prime}\R_{0}\right)  \\
&=&O_{P}(1) \frac{1}{p_{1}^{2}p_{2}^{2}T}\frac{1}{T}\left(
p_{1}^{1/2}p_{2}T+p_{1}p_{2}^{3/2}T^{1/2}\right) =O_{P}\left( \frac{1}{%
p_{1}^{3/2}p_{2}T}\right) +O_{P}\left( \frac{1}{p_{1}p_{2}^{1/2}T^{3/2}}%
\right) .
\end{eqnarray*}%
Continuing with $V_{c}$, the same passages as in the above yield%
\begin{eqnarray*}
\left\Vert V_{c}\right\Vert  &\leq &\frac{1}{p_{1}^{2}p_{2}^{2}T}\left\Vert 
\hat{\C}_{1,\perp}-\C_{1,\perp}\right\Vert
_{F}\sum_{t=1}^{T}\left\Vert \R_{1,\perp}^{\prime}\E%
_{t}\right\Vert _{F}\left\Vert \C_{1,\perp}^{\prime}\E%
_{t}^{\prime}\R_{0}\right\Vert _{F}\left\Vert \Lambda
_{R_{0}}^{-1}\right\Vert _{F} \\
&=&O_{P}(1) \frac{1}{p_{1}^{2}p_{2}^{2}T^{2}}\left(
\sum_{t=1}^{T}\left\Vert \R_{1,\perp}^{\prime}\E%
_{t}\right\Vert _{F}^{2}\right) ^{1/2}\left( \sum_{t=1}^{T}\left\Vert 
\C_{1,\perp}^{\prime}\E_{t}^{\prime}\R%
_{0}\right\Vert _{F}^{2}\right) ^{1/2} \\
&=&O_{P}(1) \frac{1}{p_{1}^{2}p_{2}^{2}T^{2}}\left(
p_{1}^{2}p_{2}T\right) ^{1/2}\left( p_{1}p_{2}^{2}T\right)
^{1/2}=O_{P}\left( \frac{1}{p_{1}^{1/2}p_{2}^{1/2}T}\right) ,
\end{eqnarray*}%
and the same can be shown for $V_{d}$ and (with a different, but still
dominated, rate) $V_{e}$. We now turn to $VI$; omitting some passages
already considered above, we have%
\begin{eqnarray*}
\left\Vert VI\right\Vert _{F} &\leq &\frac{1}{p_{1}^{2}p_{2}^{2}T}\left\Vert 
\hat{\R}_{1,\perp}-\R_{1,\perp}\right\Vert
_{F}\left\Vert \R_{1}\right\Vert _{F}\left\Vert \C%
_{1}\right\Vert _{F}\left\Vert \hat{\C}_{1,\perp}-\C%
_{1,\perp}\right\Vert _{F} \\
&&\times \left\Vert \sum_{t=1}^{T}\F_{1,t}\C_{1,\perp
}^{\prime}\E_{t}^{\prime}\R_{0}\right\Vert _{F}\left\Vert
\Lambda _{R_{0}}^{-1}\right\Vert _{F}+r_{p_{1}p_{2}T},
\end{eqnarray*}%
where $r_{p_{1}p_{2}T}$ is a (dominated) remainder term, and%
\begin{eqnarray*}
&&\frac{1}{p_{1}^{2}p_{2}^{2}T}\left\Vert \hat{\R}_{1,\perp}-%
\R_{1,\perp}\right\Vert _{F}\left\Vert \R_{1}\right\Vert
_{F}\left\Vert \C_{1}\right\Vert _{F}\left\Vert \hat{\C}%
_{1,\perp}-\C_{1,\perp}\right\Vert _{F}\left\Vert \sum_{t=1}^{T}%
\F_{1,t}\C_{1,\perp}^{\prime}\E_{t}^{\prime}%
\R_{0}\right\Vert _{F} \\
&=&O_{P}(1) \frac{1}{p_{1}^{2}p_{2}^{2}T}\frac{1}{T}%
p_{1}^{1/2}p_{2}^{1/2}\frac{1}{T}Tp_{1}^{1/2}p_{2}=O_{P}\left( \frac{1}{%
p_{1}p_{2}^{1/2}T^{2}}\right) ;
\end{eqnarray*}%
the same can be shown for $VII$. Also%
\begin{eqnarray*}
\left\Vert VIII\right\Vert _{F} &\leq &\frac{1}{p_{1}^{2}p_{2}^{2}T}%
\left\Vert \hat{\R}_{1,\perp}-\R_{1,\perp}\right\Vert
_{F}\left\Vert \R_{1}\right\Vert _{F}\left\Vert \hat{\C}%
_{1,\perp}-\C_{1,\perp}\right\Vert _{F}\left\Vert \C%
_{1}\right\Vert _{F} \\
&&\times \left\Vert \C_{0}\right\Vert _{F}\left\Vert \sum_{t=1}^{T}%
\F_{1,t}\F_{0,t}^{\prime}\right\Vert \left\Vert \R%
_{1}\right\Vert _{F}\left\Vert \hat{\R}_{0}\right\Vert
_{F}\left\Vert \Lambda _{R_{0}}^{-1}\right\Vert _{F} \\
&=&O_{P}(1) \frac{1}{p_{1}^{2}p_{2}^{2}T}\frac{1}{T}p_{1}^{1/2}%
\frac{1}{T}p_{2}^{1/2}p_{2}^{1/2}Tp_{1}^{1/2}p_{1}^{1/2}=O_{P}\left( \frac{1%
}{p_{1}^{1/2}p_{2}T^{2}}\right) ,
\end{eqnarray*}%
and the same holds for $IX$. Also repeating the passages above, we receive%
\begin{equation*}
\frac{1}{p_{1}^{2}p_{2}^{2}T}\hat{\R}_{1,\perp}^{\prime
}\sum_{t=1}^{T}\R_{0}\F_{0,t}\C_{0}^{\prime}%
\hat{\C}_{1,\perp}\E_{t}^{\prime}\hat{\R}%
_{0}\Lambda _{R_{0}}^{-1}
\end{equation*}%
\begin{eqnarray*}
\left\Vert X\right\Vert _{F} &\leq &\frac{1}{p_{1}^{2}p_{2}^{2}T}\left\Vert 
\hat{\R}_{1,\perp}\right\Vert _{F}\left\Vert \R%
_{0}\right\Vert _{F}\left\Vert \C_{0}\right\Vert _{F}\left\Vert
\sum_{t=1}^{T}\F_{0,t}\C_{1,\perp}^{\prime}\E%
_{t}^{\prime}\R_{0}\right\Vert \left\Vert \bH%
_{R_{0}}\right\Vert _{F}\left\Vert \Lambda _{R_{0}}^{-1}\right\Vert
_{F}+r_{p_{1}p_{2}T}^{\prime} \\
&=&O_{P}(1) \frac{1}{p_{1}^{2}p_{2}^{2}T}%
p_{1}p_{1}^{1/2}p_{2}^{1/2}T^{1/2}p_{1}^{1/2}p_{2}=O_{P}\left( \frac{1}{%
p_{2}^{1/2}T^{1/2}}\right) ,
\end{eqnarray*}%
with $r_{p_{1}p_{2}T}^{\prime}$ a (dominated) remainder. Finally%
\begin{eqnarray*}
\left\Vert XI\right\Vert _{F} &\leq &\frac{1}{p_{1}^{2}p_{2}^{2}T}\left\Vert
\sum_{t=1}^{T}\R_{1,\perp}^{\prime}\E_{t}\C%
_{1,\perp}\F_{0,t}^{\prime}\right\Vert \left\Vert \C%
_{0}\right\Vert _{F}\left\Vert \R_{0}\right\Vert _{F}\left\Vert 
\hat{\R}_{0}\right\Vert _{F}\left\Vert \Lambda
_{R_{0}}^{-1}\right\Vert _{F}+r_{p_{1}p_{2}T}^{\prime \prime} \\
&=&O_{P}(1) \frac{1}{p_{1}^{2}p_{2}^{2}T}%
p_{1}p_{2}T^{1/2}p_{2}^{1/2}p_{1}^{1/2}p_{1}^{1/2}=O_{P}\left( \frac{1}{%
p_{2}^{1/2}T^{1/2}}\right) ,
\end{eqnarray*}%
with, as usual, $r_{p_{1}p_{2}T}^{\prime \prime}$ a (dominated) remainder.
The desired result now follows from putting all together.
\end{proof}
\end{lemma}

\begin{lemma}
\label{tilde-orth}We assume that Assumptions \ref{as-1}-\ref{as-5} are
satisfied. Then it holds that%
\begin{eqnarray*}
\left\Vert \widetilde{\R}_{1,\perp}-\R_{1,\perp
}\right\Vert _{F} &=&O_{P}\left( \frac{1}{p_{2}^{1/2}T}\right) +O_{P}\left( 
\frac{1}{T^{2}}\right) +O_{P}\left( \frac{1}{p_{1}T}\right) +O_{P}\left( 
\frac{1}{p_{1}^{1/2}T^{3/2}}\right) , \\
\left\Vert \widetilde{\C}_{_{1,\perp}}-\C_{1,\perp
}\right\Vert _{F} &=&O_{P}\left( \frac{1}{p_{1}^{1/2}T}\right) +O_{P}\left( 
\frac{1}{T^{2}}\right) +O_{P}\left( \frac{1}{p_{2}T}\right) +O_{P}\left( 
\frac{1}{p_{2}^{1/2}T^{3/2}}\right) ,
\end{eqnarray*}

\begin{proof}
The proof follows from a minor adaptation of the proof of Lemma \ref{r-orth}.
\end{proof}
\end{lemma}

\begin{lemma}
\label{yong-2}We assume that Assumptions \ref{as-1}-\ref{as-5} hold. Let $%
\mathbf{a}$ be a $p_{1}\times 1$ vector with $\left\Vert \mathbf{a}%
\right\Vert =O\left( p_{1}^{1/2}\right) $, and $\mathbf{b}$ be a $%
p_{2}\times 1$ vector with $\left\Vert \mathbf{b}\right\Vert =O\left(
p_{2}^{1/2}\right) $. Then it holds that%
\begin{eqnarray*}
&&\sigma _{\max }\left[ \frac{1}{p_{1}}\widetilde{\R}_{1,\perp
}^{\prime}\left( \widetilde{\R}_{0}-\R_{0}\widetilde{%
\bH}_{R_{0}}\right) \right]  \\
&=&O_{P}\left( \frac{p_{1}^{1/2}}{p_{1}p_{2}}\right) +O_{P}\left( \frac{%
p_{1}^{1/2}}{\left( p_{1}p_{2}T\right) ^{1/2}}\right) +O_{P}\left( \frac{%
p_{1}^{1/2}}{p_{2}T^{2}}\right) +O_{P}\left( \frac{p_{1}^{1/2}}{p_{2}^{2}T}%
\right) +O_{P}\left( \frac{p_{1}^{1/2}}{p_{2}^{3/2}T^{3/2}}\right) ,
\end{eqnarray*}%
and%
\begin{eqnarray*}
&&\sigma _{\max }\left[ \frac{1}{p_{2}}\widetilde{\C}_{1,\perp
}^{\prime}\left( \widetilde{\C}_{0}-\C_{0}\widetilde{%
\bH}_{C_{0}}\right) \right]  \\
&=&O_{P}\left( \frac{p_{2}^{1/2}}{p_{1}p_{2}}\right) +O_{P}\left( \frac{%
p_{2}^{1/2}}{\left( p_{1}p_{2}T\right) ^{1/2}}\right) +O_{P}\left( \frac{%
p_{2}^{1/2}}{p_{1}T^{2}}\right) +O_{P}\left( \frac{p_{2}^{1/2}}{p_{1}^{2}T}%
\right) +O_{P}\left( \frac{p_{2}^{1/2}}{p_{1}^{3/2}T^{3/2}}\right) .
\end{eqnarray*}

\begin{proof}
The proof is essentially the same as that of Lemma \ref{yong}, \textit{%
mutatis mutandis}, using (also) the fact that%
\begin{equation*}
\left\Vert \widetilde{\C}_{1,\perp}\left( \widetilde{\C}%
_{1,\perp}-\C_{1,\perp}\right) \right\Vert _{F}=O_{P}\left( \frac{1%
}{p_{2}^{1/2}T}\right) +O_{P}\left( \frac{1}{T^{2}}\right) +O_{P}\left( 
\frac{1}{p_{1}T}\right) +O_{P}\left( \frac{1}{p_{1}^{1/2}T^{3/2}}\right) ,
\end{equation*}%
and 
\begin{equation*}
\left\Vert \C_{1,\perp}\left( \widetilde{\C}_{1,\perp}-%
\C_{1,\perp}\right) \right\Vert _{F}=O_{P}\left( \frac{1}{%
p_{2}^{1/2}T}\right) +O_{P}\left( \frac{1}{T^{2}}\right) +O_{P}\left( \frac{1%
}{p_{1}T}\right) +O_{P}\left( \frac{1}{p_{1}^{1/2}T^{3/2}}\right) .
\end{equation*}%
and the same (with the rates in Lemma \ref{tilde-orth}) for $\widetilde{%
\R}_{1,\perp}$.
\end{proof}
\end{lemma}

\begin{lemma}
\label{bai}We assume that Assumptions \ref{as-1}-\ref{as-5} are satisfied.
Then it holds that%
\begin{equation*}
\frac{1}{T}\left\Vert \sum_{t=1}^{T}\left( \hat{\F}_{0,t}-\left( 
\bH_{R_{0}}\right) ^{-1}\F_{0,t}\left( \bH%
_{C_{0}}^{\prime}\right) ^{-1}\right) \F_{1,t}^{\prime}\right\Vert
_{F}=O_{P}\left( \frac{1}{\sqrt{p_{1}p_{2}}}\right) +O_{P}\left( \frac{1}{%
p_{1\wedge 2}^{1/2}T^{1/2}}\right) .
\end{equation*}

\begin{proof}
We prove the lemma for the case $h_{R_{1}}=h_{C_{1}}=h_{R_{1}}=h_{C_{1}}=1$,
with no loss of generality. We use the same arguments as in the proof of
Theorem \ref{f1-hat}, obtaining%
\begin{eqnarray*}
&&\sum_{t=1}^{T}\left( \hat{\F}_{0,t}-\left( \bH%
_{R_{0}}\right) ^{-1}\F_{0,t}\left( \bH_{C_{0}}^{\prime
}\right) ^{-1}\right) \F_{1,t} \\
&=&\left( \hat{\mathbf{D}}\right) ^{-1}\hat{\mathbf{N}}\left( 
\hat{\C}_{1,\perp}^{\prime}\otimes \hat{\R}_{\perp
}^{\prime}\right) \left( \left( \C_{0}-\hat{\C}%
_{0}\left( \bH_{C_{0}}\right) ^{-1}\right) \otimes \left( \hat{%
\R}_{0}\left( \bH_{R_{0}}\right) ^{-1}\right) \right)
\sum_{t=1}^{T}\F_{0,t}\F_{1,t} \\
&&+\left( \hat{\mathbf{D}}\right) ^{-1}\hat{\mathbf{N}}\left( 
\hat{\C}_{1,\perp}^{\prime}\otimes \hat{\R}_{\perp
}^{\prime}\right) \left( \left( \hat{\C}_{1}\left( \bH%
_{C_{0}}\right) ^{-1}\right) \otimes \left( \R_{0}-\hat{\mathbf{R%
}}_{0}\left( \bH_{R_{0}}\right) ^{-1}\right) \right) \sum_{t=1}^{T}%
\F_{0,t}\F_{1,t} \\
&&+\left( \hat{\mathbf{D}}\right) ^{-1}\hat{\mathbf{N}}\left( 
\hat{\C}_{1,\perp}^{\prime}\otimes \hat{\R}_{\perp
}^{\prime}\right) \left( \left( \C_{1}-\hat{\C}%
_{0}\left( \bH_{C_{0}}\right) ^{-1}\right) \otimes \left( \R%
_{0}-\hat{\R}_{0}\left( \bH_{R_{0}}\right) ^{-1}\right)
\right) \sum_{t=1}^{T}\F_{0,t}\F_{1,t} \\
&&+\left( \hat{\mathbf{D}}\right) ^{-1}\left( \hat{\C}%
_{0}^{\prime}\hat{\C}_{1,\perp}\otimes \hat{\R}%
_{0}^{\prime}\hat{\R}_{1,\perp}\right) \left( \C%
_{1}\otimes \R_{1}\right) \sum_{t=1}^{T}\F_{1,t}^{2} \\
&&+\left( \hat{\mathbf{D}}\right) ^{-1}\hat{\mathbf{N}}\left( 
\hat{\C}_{1,\perp}\otimes \hat{\R}_{1,\perp
}\right) \sum_{t=1}^{T}\Ve\left( \E_{t}\right) \F_{1,t} \\
&=&I+II+III+IV+V.
\end{eqnarray*}%
Recall (\ref{d-hat-inv}); it follows that 
\begin{eqnarray*}
I &=&\left( \hat{\mathbf{D}}\right) ^{-1}\left( \hat{\C}%
_{0}^{\prime}\hat{\C}_{1,\perp}\left( \C_{0}-\hat{%
\C}_{0}\left( \bH_{C_{0}}\right) ^{-1}\right) \otimes 
\hat{\R}_{0}^{\prime}\hat{\R}_{1,\perp}\hat{%
\R}_{0}\left( \bH_{R_{0}}\right) ^{-1}\right) \sum_{t=1}^{T}%
\F_{0,t}\F_{1,t} \\
&=&O_{P}(1) \frac{1}{\left( p_{1}p_{2}\right) ^{2}}%
p_{2}^{1/2}p_{2}p_{2}^{1/2}\left( \frac{1}{p_{1}p_{2}}+\frac{1}{%
p_{1}^{1/2}T^{1/2}}\right) p_{1}^{1/2}p_{1}p_{1}^{1/2}T \\
&=&O_{P}\left( \frac{T}{p_{1}p_{2}}\right) +O_{P}\left( \left( \frac{T}{p_{1}%
}\right) ^{1/2}\right) ,
\end{eqnarray*}%
and, by the same passages%
\begin{equation*}
II=O_{P}\left( \frac{T}{p_{1}p_{2}}\right) +O_{P}\left( \left( \frac{T}{p_{2}%
}\right) ^{1/2}\right) ,
\end{equation*}%
and $III$ is dominated by $I$ and $II$; also%
\begin{eqnarray*}
IV &=&\left( \hat{\mathbf{D}}\right) ^{-1}\left( \hat{\C}%
_{0}^{\prime}\left( \hat{\C}_{1,\perp}-\C_{1,\perp
}\right) \C_{1}\otimes \hat{\R}_{0}^{\prime}\left( 
\hat{\R}_{1,\perp}-\R_{1,\perp}\right) \R%
_{1}\right) \sum_{t=1}^{T}\F_{1,t}^{2} \\
&=&O_{P}(1) \frac{1}{\left( p_{1}p_{2}\right) ^{2}}p_{2}^{1/2}%
\frac{1}{T}p_{2}^{1/2}p_{1}^{1/2}\frac{1}{T}p_{1}^{1/2}T^{2}=O_{P}\left( 
\frac{1}{p_{1}p_{2}}\right) .
\end{eqnarray*}%
Finally%
\begin{eqnarray*}
V &=&\left( \hat{\mathbf{D}}\right) ^{-1}\left( \hat{\C}%
_{0}^{\prime}\hat{\C}_{1,\perp}\otimes \hat{\R}%
_{0}^{\prime}\hat{\R}_{1,\perp}\right) \left( \hat{\C}_{1,\perp}\otimes \hat{\R}_{1,\perp}\right) \sum_{t=1}^{T}%
\Ve\left( \E_{t}\right) \F_{1,t} \\
&=&\left( \hat{\mathbf{D}}\right) ^{-1}\left( \hat{\C}%
_{0}^{\prime}\otimes \hat{\R}_{0}^{\prime}\right) \left( 
\hat{\C}_{1,\perp}\otimes \hat{\R}_{1,\perp
}\right) \sum_{t=1}^{T}\Ve\left( \E_{t}\right) \F_{1,t} \\
&=&\left( \hat{\mathbf{D}}\right) ^{-1}\left( \hat{\C}%
_{0}^{\prime}\otimes \hat{\R}_{0}^{\prime}\right) \left( 
\C_{1,\perp}\otimes \R_{1,\perp}\right) \sum_{t=1}^{T}\Ve%
\left( \E_{t}\right) \F_{1,t} \\
&&+\left( \hat{\mathbf{D}}\right) ^{-1}\left( \hat{\C}%
_{0}^{\prime}\otimes \hat{\R}_{0}^{\prime}\right) \left( 
\C_{1,\perp}\otimes \left( \hat{\R}_{1,\perp}-\mathbf{R%
}_{1,\perp}\right) \right) \sum_{t=1}^{T}\Ve\left( \E_{t}\right) 
\F_{1,t} \\
&&+\left( \hat{\mathbf{D}}\right) ^{-1}\left( \hat{\C}%
_{0}^{\prime}\otimes \hat{\R}_{0}^{\prime}\right) \left(
\left( \hat{\C}_{1,\perp}-\C_{1,\perp}\right) \otimes 
\R_{1,\perp}\right) \sum_{t=1}^{T}\Ve\left( \E_{t}\right) 
\F_{1,t} \\
&&+\left( \hat{\mathbf{D}}\right) ^{-1}\left( \hat{\C}%
_{0}^{\prime}\otimes \hat{\R}_{0}^{\prime}\right) \left(
\left( \hat{\C}_{1,\perp}-\C_{1,\perp}\right) \otimes
\left( \hat{\R}_{1,\perp}-\R_{1,\perp}\right) \right)
\sum_{t=1}^{T}\Ve\left( \E_{t}\right) \F_{1,t} \\
&=&V_{a}+V_{b}+V_{c}+V_{d}.
\end{eqnarray*}%
Noting that%
\begin{eqnarray*}
&&E\left\Vert \sum_{t=1}^{T}\R_{1,\perp}^{\prime}\E_{t}%
\C_{1,\perp}\F_{1,t}\right\Vert _{F}^{2} \\
&=&E\sum_{i=1}^{p_{1}}\sum_{\ell =1}^{p_{2}}\left(
\sum_{j=1}^{p_{1}}\sum_{h=1}^{p_{2}}\sum_{t=1}^{T}r_{\perp ,ij}c_{\perp
,h\ell }e_{jh,t}\F_{1,t}\right) ^{2} \\
&=&\sum_{i=1}^{p_{1}}\sum_{\ell
=1}^{p_{2}}\sum_{j_{1},j_{2}=1}^{p_{1}}\sum_{h_{1},h_{2}=1}^{p_{2}}%
\sum_{t,s=1}^{T}r_{\perp ,ij_{1}}r_{\perp ,ij_{2}}c_{\perp ,h_{1}\ell
}c_{\perp ,h_{2}\ell }E\left( \F_{1,t}\F_{1,s}\right)
E\left( e_{j_{1}h_{1},t}e_{j_{2}h_{2},s}\right)  \\
&\leq &c_{0}\sum_{i=1}^{p_{1}}\sum_{\ell
=1}^{p_{2}}\sum_{j_{1},j_{2}=1}^{p_{1}}\sum_{h_{1},h_{2}=1}^{p_{2}}%
\sum_{t,s=1}^{T}\left( E\left( \F_{1,t}^{2}\right) E\left( \F%
_{1,s}^{2}\right) \right) ^{1/2}\left\vert E\left(
e_{j_{1}h_{1},t}e_{j_{2}h_{2},s}\right) \right\vert  \\
&\leq &c_{1}T\sum_{i=1}^{p_{1}}\sum_{\ell
=1}^{p_{2}}\sum_{j_{1},j_{2}=1}^{p_{1}}\sum_{h_{1},h_{2}=1}^{p_{2}}%
\sum_{t,s=1}^{T}\left\vert E\left( e_{ij,t}e_{hk,s}\right) \right\vert \leq
c_{2}\left( p_{1}p_{2}T\right) ^{2},
\end{eqnarray*}%
it immediately follows that 
\begin{equation*}
\left\Vert V_{a}\right\Vert _{F}=O_{P}(1) \frac{1}{p_{1}p_{2}}%
\left( p_{1}p_{2}\right) ^{1/2}\left( p_{1}p_{2}T\right) =O_{P}\left( \frac{T%
}{\left( p_{1}p_{2}\right) ^{1/2}}\right) .
\end{equation*}%
Similarly, seeing as%
\begin{eqnarray*}
&&E\left\Vert \sum_{t=1}^{T}\E_{t}\C_{1,\perp}\F%
_{1,t}\right\Vert _{F}^{2} \\
&=&E\sum_{i=1}^{p_{1}}\sum_{h=1}^{p_{2}}\left(
\sum_{j=1}^{p_{2}}\sum_{t=1}^{T}c_{\perp ,hj}^{s}e_{ij,t}\F%
_{1,t}\right) ^{2}\leq
c_{0}\sum_{i=1}^{p_{1}}\sum_{h=1}^{p_{2}}\sum_{j,k=1}^{p_{2}}%
\sum_{t,s=1}^{T}\left( E\left( \F_{1,t}^{2}\right) E\left( \F%
_{1,s}^{2}\right) \right) ^{1/2}\left\vert E\left( e_{ij,t}e_{ik,s}\right)
\right\vert  \\
&\leq &c_{1}p_{1}p_{2}^{2}T^{2},
\end{eqnarray*}%
we have%
\begin{eqnarray*}
\left\Vert V_{b}\right\Vert _{F} &\leq &\left\Vert \left( \hat{\mathbf{D}%
}\right) ^{-1}\right\Vert _{F}\left\Vert \hat{\C}_{0}\right\Vert
_{F}\left\Vert \hat{\R}_{0}\right\Vert _{F}\left\Vert \hat{%
\R}_{1,\perp}-\R_{1,\perp}\right\Vert _{F}\left\Vert
\sum_{t=1}^{T}\E_{t}\C_{1,\perp}\F_{1,t}\right\Vert
_{F}^{2} \\
&=&O_{P}(1) \frac{1}{\left( p_{1}p_{2}\right) ^{2}}%
p_{2}^{1/2}p_{1}^{1/2}\frac{1}{T}p_{1}^{1/2}p_{2}T=O_{P}\left( \frac{1}{%
p_{1}p_{2}^{1/2}}\right) .
\end{eqnarray*}%
The same holds for $V_{c}$, and $V_{d}$ is clearly dominated by $V_{b}$ and $%
V_{c}$. Hence%
\begin{equation*}
\left\Vert V\right\Vert _{F}=O_{P}\left( \frac{T}{\left( p_{1}p_{2}\right)
^{1/2}}\right) .
\end{equation*}
\end{proof}
\end{lemma}
\begin{lemma}\label{spec-m-r-tilde}
We assume that Assumptions \ref{as-1}-\ref{as-5} are
satisfied. Then there exists a positive constant $c_{0}$ such that%
\begin{equation*}
\lambda _{j}\left( \M_{X}^{R_{1},PE}\right) =c_{0}+o_{P}\left(
1\right) \text{,}
\end{equation*}%
for all $j\leq h_{R_{1}}$, and%
\begin{equation*}
\lambda _{j}\left( \M_{X}^{R_{1},PE}\right) =O_{P}\left( \frac{1}{%
p_{1\wedge 2}^{1/2}T^{3/2}}\right) +O_{P}\left( \frac{1}{p_{2}T}\right)
+O_{P}\left( \frac{1}{T^{2}}\right) +O_{P}\left( \frac{1}{%
p_{1}^{1/2}p_{2}^{1/2}T}\right) ,
\end{equation*}%
for all $j>h_{R_{1}}$.
\begin{proof}
The proof repeats several arguments already discussed above, which are
therefore omitted. Note that 
\begin{align*}
\M_{X}^{R_{1},PE} &=\frac{1}{p_{1}p_{2}^{2}T^{2}}\sum_{t=1}^{T}%
\R_{1}\F_{1,t}\C_{1}^{\prime}\hat{\C}%
_{1}\hat{\C}_{1}^{\prime}\C_{1}\F_{1,t}^{\prime
}\R_{1}^{\prime}+\frac{1}{p_{1}p_{2}^{2}T^{2}}\sum_{t=1}^{T}\E_{t}\hat{\C}_{1}\hat{\C}_{1}^{\prime}\E%
_{t}^{\prime} \\
&+\frac{1}{p_{1}p_{2}^{2}T^{2}}\sum_{t=1}^{T}\R_{1}\F_{1,t}%
\C^{\prime}\hat{\C}_{1}\hat{\C}_{1}^{\prime
}\E_{t}^{\prime}+\left( \frac{1}{p_{1}p_{2}^{2}T^{2}}\sum_{t=1}^{T}%
\R_{1}\F_{1,t}\C_{1}^{\prime}\hat{\C}%
_{1}\hat{\C}_{1}^{\prime}\E_{t}^{\prime}\right)
^{\prime} \\
&+\frac{1}{p_{1}p_{2}^{2}T^{2}}\sum_{t=1}^{T}\left( \R_{0}\F%
_{0,t}\C_{0}^{\prime}-\hat{\R}_{0}\hat{\F}%
_{0,t}\hat{\C}_{0}^{\prime}\right) \hat{\C}_{1}%
\hat{\C}_{1}^{\prime}\E_{t}^{\prime} \\
&+\left( \frac{1}{p_{1}p_{2}^{2}T^{2}}\sum_{t=1}^{T}\left( \R_{0}%
\F_{0,t}\C_{0}^{\prime}-\hat{\R}_{0}\hat{%
\F}_{0,t}\hat{\C}_{0}^{\prime}\right) \hat{\mathbf{C%
}}_{1}\hat{\C}_{1}^{\prime}\E_{t}^{\prime}\right)
^{\prime} \\
&+\frac{1}{p_{1}p_{2}^{2}T^{2}}\sum_{t=1}^{T}\R_{1}\F_{1,t}%
\C_{1}^{\prime}\hat{\C}_{1}\hat{\C}%
_{1}^{\prime}\left( \R_{0}\F_{0,t}\C_{0}^{\prime}-%
\hat{\R}_{0}\hat{\F}_{0,t}\hat{\C}%
_{0}^{\prime}\right) ^{\prime} \\
&+\left( \frac{1}{p_{1}p_{2}^{2}T^{2}}\sum_{t=1}^{T}\R_{1}\F%
_{1,t}\C_{1}^{\prime}\hat{\C}_{1}\hat{\C}%
_{1}^{\prime}\left( \R_{0}\F_{0,t}\C_{0}^{\prime}-%
\hat{\R}_{0}\hat{\F}_{0,t}\hat{\C}%
_{0}^{\prime}\right) ^{\prime}\right) ^{\prime} \\
&+\frac{1}{p_{1}p_{2}^{2}T^{2}}\sum_{t=1}^{T}\left( \R_{0}\F%
_{0,t}\C_{0}^{\prime}-\hat{\R}_{0}\hat{\F}%
_{0,t}\hat{\C}_{0}^{\prime}\right) \hat{\C}_{1}%
\hat{\C}_{1}^{\prime}\left( \R_{0}\F_{0,t}%
\C_{0}^{\prime}-\hat{\R}_{0}\hat{\F}_{0,t}%
\hat{\C}_{0}^{\prime}\right) ^{\prime} \\
&=I+II+III+III^{\prime}+IV+IV^{\prime}+V+V^{\prime}+VI.
\end{align*}
We have already studied terms $II+III+III^{\prime}$ in the proof of Lemma %
\ref{proj-negative}, with 
\begin{equation*}
II+III+III^{\prime}=O_{P}\left( \frac{1}{T^{2}}\right) +O_{P}\left( \frac{1%
}{p_{2}T}\right) +O_{P}\left( \frac{1}{p_{2}^{1/2}T^{3/2}}\right) .
\end{equation*}%
Further, write%
\begin{eqnarray}
&&\R_{0}\F_{0,t}\C_{0}^{\prime}-\hat{\R}%
_{0}\hat{\F}_{0,t}\hat{\C}_{0}^{\prime}
\label{st-cc} \\
&=&\left( \hat{\R}_{0}-\R_{0}\bH_{R_{0}}\right) 
\bH_{R_{0}}^{-1}\F_{0,t}\C_{0}^{\prime}+\R%
_{0}\F_{0,t}\left( \bH_{C_{0}}^{\prime}\right) ^{-1}\left( 
\hat{\C}_{0}-\C_{0}\bH_{C_{0}}\right) ^{\prime} 
\notag \\
&&+\R_{0}\bH_{R_{0}}\left( \hat{\F}_{0,t}-\bH_{R_{0}}^{-1}\F_{0,t}\left( \bH_{C_{0}}^{\prime}\right)
^{-1}\right) \left( \C_{0}\bH_{C_{0}}\right) ^{\prime}+ 
\notag \\
&&+\left( \hat{\R}_{0}-\R_{0}\bH_{R_{0}}\right) 
\bH_{R_{0}}^{-1}\F_{0,t}\left( \bH_{C_{0}}^{\prime
}\right) ^{-1}\left( \hat{\C}_{0}-\C_{0}\bH%
_{C_{0}}\right) ^{\prime}  \notag \\
&&+\left( \hat{\R}_{0}-\R_{0}\bH_{R_{0}}\right) 
\bH_{R_{0}}^{-1}\left( \hat{\F}_{0,t}-\bH%
_{R_{0}}^{-1}\F_{0,t}\left( \bH_{C_{0}}^{\prime}\right)
^{-1}\right) \left( \C_{0}\bH_{C_{0}}\right) ^{\prime} 
\notag \\
&&+\R_{0}\bH_{R_{0}}\left( \hat{\F}_{0,t}-\bH_{R_{0}}^{-1}\F_{0,t}\left( \bH_{C_{0}}^{\prime}\right)
^{-1}\right) \left( \hat{\C}_{0}-\C_{0}\bH%
_{C_{0}}\right) ^{\prime}  \notag \\
&&+\left( \hat{\R}_{0}-\R_{0}\bH_{R_{0}}\right) 
\bH_{R_{0}}^{-1}\left( \hat{\F}_{0,t}-\bH%
_{R_{0}}^{-1}\F_{0,t}\left( \bH_{C_{0}}^{\prime}\right)
^{-1}\right) \left( \hat{\C}_{0}-\C_{0}\bH%
_{C_{0}}\right) ^{\prime}.  \notag
\end{eqnarray}%
Then we have%
\begin{eqnarray*}
IV &=&\frac{1}{p_{1}p_{2}^{2}T^{2}}\sum_{t=1}^{T}\left( \R_{0}%
\F_{0,t}\C_{0}^{\prime}-\hat{\R}_{0}\hat{%
\F}_{0,t}\hat{\C}_{0}^{\prime}\right) \hat{\mathbf{C%
}}_{1}\bH_{C_{1}}^{\prime}\C_{1}^{\prime}\E%
_{t}^{\prime} \\
&&+\frac{1}{p_{1}p_{2}^{2}T^{2}}\sum_{t=1}^{T}\left( \R_{0}\F%
_{0,t}\C_{0}^{\prime}-\hat{\R}_{0}\hat{\F}%
_{0,t}\hat{\C}_{0}^{\prime}\right) \hat{\C}%
_{1}\left( \hat{\C}_{1}-\C_{1}\bH_{C_{1}}\right)
^{\prime}\E_{t}^{\prime}=IV_{a}+IV_{b};
\end{eqnarray*}%
using (\ref{st-cc}), we can write 
\begin{equation*}
IV_{a}=\sum_{i=1}^{7}IV_{a,i}.
\end{equation*}%
Following the passages in the above, it holds that%
\begin{eqnarray*}
\left\Vert IV_{a,1}\right\Vert _{F} &=&O_{P}(1) \frac{1}{%
p_{1}p_{2}^{2}T^{2}}\left\Vert \hat{\R}_{0}-\R_{0}%
\bH_{R_{0}}\right\Vert _{F}\left\Vert \hat{\C}%
_{1}\right\Vert _{F}\left\Vert \C_{0}\right\Vert _{F}\left\Vert
\sum_{t=1}^{T}\F_{0,t}\C_{1}^{\prime}\E%
_{t}\right\Vert _{F} \\
&=&O_{P}(1) \frac{1}{p_{1}p_{2}^{2}T^{2}}p_{1}^{1/2}\left( \frac{%
1}{p_{1}p_{2}}+\frac{1}{\left( p_{2}T\right) ^{1/2}}\right)
p_{2}^{1/2}p_{2}^{1/2}\left( p_{1}p_{2}T\right) ^{1/2} \\
&=&O_{P}\left( \frac{1}{p_{1}^{1/2}T^{3/2}}\left( \frac{1}{p_{1}p_{2}}+\frac{%
1}{\left( p_{2}T\right) ^{1/2}}\right) \right) ;
\end{eqnarray*}%
\begin{eqnarray*}
\left\Vert IV_{a,2}\right\Vert _{F} &=&O_{P}(1) \frac{1}{%
p_{1}p_{2}^{2}T^{2}}\left\Vert \hat{\C}_{0}-\C_{0}%
\bH_{C_{0}}\right\Vert _{F}\left\Vert \hat{\C}%
_{1}\right\Vert _{F}\left\Vert \R_{0}\right\Vert _{F}\left\Vert
\sum_{t=1}^{T}\F_{0,t}\C_{1}^{\prime}\E%
_{t}\right\Vert _{F} \\
&=&O_{P}\left( \frac{1}{p_{2}^{1/2}T^{3/2}}\left( \frac{1}{p_{1}p_{2}}+\frac{%
1}{\left( p_{1}T\right) ^{1/2}}\right) \right) ;
\end{eqnarray*}%
\begin{eqnarray*}
\left\Vert IV_{a,3}\right\Vert _{F} &=&O_{P}(1) \frac{1}{%
p_{1}p_{2}^{2}T^{2}}\left\Vert \hat{\C}_{1}\right\Vert
_{F}\left\Vert \R_{0}\right\Vert _{F}\left\Vert \C%
_{0}\right\Vert _{F} \\
&&\times \left\Vert \sum_{t=1}^{T}\left( \hat{\F}_{0,t}-\mathbf{H%
}_{R_{0}}^{-1}\F_{0,t}\left( \bH_{C_{0}}^{\prime}\right)
^{-1}\right) \C_{1}^{\prime}\E_{t}\right\Vert _{F} \\
&=&O_{P}(1) \frac{p_{1}^{1/2}p_{2}}{p_{1}p_{2}^{2}T^{3/2}}\left( 
\frac{1}{T}\sum_{t=1}^{T}\left\Vert \hat{\F}_{0,t}-\bH%
_{R_{0}}^{-1}\F_{0,t}\left( \bH_{C_{0}}^{\prime}\right)
^{-1}\right\Vert _{F}^{2}\right) ^{1/2}\left( \sum_{t=1}^{T}\left\Vert 
\C_{1}^{\prime}\E_{t}\right\Vert _{F}^{2}\right) ^{1/2} \\
&=&O_{P}(1) \frac{1}{p_{2}^{1/2}T}\left( \frac{1}{\sqrt{%
p_{1}p_{2}}}+\frac{1}{\left( p_{1\wedge 2}T\right) ^{1/2}}\right) ,
\end{eqnarray*}%
having used (\ref{f1-hat-l2}) in the last set of equations. By the same
token, it can be shown that $IV_{a,4}-IV_{a,7}$ are all dominated by $%
IV_{a,1}-IV_{a,3}$. Similarly, using (\ref{st-cc}), we can write 
\begin{equation*}
IV_{b}=\sum_{i=1}^{7}IV_{b,i}.
\end{equation*}%
It holds that%
\begin{eqnarray*}
\left\Vert IV_{b,1}\right\Vert _{F} &=&O_{P}(1) \frac{1}{%
p_{1}p_{2}^{2}T^{2}}\left\Vert \hat{\R}_{0}-\R_{0}%
\bH_{R_{0}}\right\Vert _{F}\left\Vert \hat{\C}%
_{1}\right\Vert _{F}\left\Vert \hat{\C}_{1}-\C_{1}%
\bH_{C_{1}}\right\Vert _{F}\left\Vert \C_{0}\right\Vert
_{F}\left\Vert \sum_{t=1}^{T}\F_{0,t}\E_{t}\right\Vert _{F}
\\
&=&O_{P}(1) \frac{1}{p_{1}p_{2}^{2}T^{2}}p_{1}^{1/2}\left( \frac{%
1}{p_{1}p_{2}}+\frac{1}{\left( p_{2}T\right) ^{1/2}}\right) p_{2}^{1/2}\frac{%
p_{2}^{1/2}}{T}\left( p_{1}p_{2}T\right) ^{1/2} \\
&=&O_{P}\left( \frac{1}{p_{1}^{1/2}T^{5/2}}\left( \frac{1}{p_{1}p_{2}}+\frac{%
1}{\left( p_{2}T\right) ^{1/2}}\right) \right) ;
\end{eqnarray*}%
\begin{eqnarray*}
\left\Vert IV_{b,2}\right\Vert _{F} &=&O_{P}(1) \frac{1}{%
p_{1}p_{2}^{2}T^{2}}\left\Vert \hat{\C}_{0}-\C_{0}%
\bH_{C_{0}}\right\Vert _{F}\left\Vert \hat{\C}%
_{1}\right\Vert _{F}\left\Vert \R_{0}\right\Vert _{F}\left\Vert 
\hat{\C}_{1}-\C_{1}\bH_{C_{1}}\right\Vert
_{F}\left\Vert \sum_{t=1}^{T}\F_{0,t}\E_{t}\right\Vert _{F}
\\
&=&O_{P}\left( \frac{1}{p_{2}^{1/2}T^{5/2}}\left( \frac{1}{p_{1}p_{2}}+\frac{%
1}{\left( p_{1}T\right) ^{1/2}}\right) \right) .
\end{eqnarray*}%
We now study, along similar lines as the proof of Lemma \ref{bai} 
\begin{eqnarray*}
&&\sum_{t=1}^{T}\left( \hat{\F}_{0,t}-\left( \bH%
_{R_{0}}\right) ^{-1}\F_{0,t}\left( \bH_{C_{0}}^{\prime
}\right) ^{-1}\right) \E_{t} \\
&=&\left( \hat{\mathbf{D}}\right) ^{-1}\hat{\mathbf{N}}\left( 
\hat{\C}_{1,\perp}^{\prime}\otimes \hat{\R}%
_{1,\perp}^{\prime}\right) \sum_{t=1}^{T}\F_{0,t}\left( \hat{%
\R}_{0}\left( \bH_{R_{0}}\right) ^{-1}\right) ^{\prime}%
\E_{t}\left( \C_{0}-\hat{\C}_{0}\left( \bH%
_{C_{0}}\right) ^{-1}\right)  \\
&&+\left( \hat{\mathbf{D}}\right) ^{-1}\hat{\mathbf{N}}\left( 
\hat{\C}_{1,\perp}^{\prime}\otimes \hat{\R}%
_{1,\perp}^{\prime}\right) \left( \R_{0}-\hat{\R}%
_{0}\left( \bH_{R_{0}}\right) ^{-1}\right) ^{\prime}\sum_{t=1}^{T}%
\F_{0,t}\E_{t}\left( \hat{\C}_{0}\left( \mathbf{H%
}_{C_{0}}\right) ^{-1}\right)  \\
&&+\left( \hat{\mathbf{D}}\right) ^{-1}\hat{\mathbf{N}}\left( 
\hat{\C}_{1,\perp}^{\prime}\otimes \hat{\R}%
_{1,\perp}^{\prime}\right) \left( \left( \C_{0}-\hat{\C%
}_{0}\left( \bH_{C_{0}}\right) ^{-1}\right) ^{\prime}\otimes \left( 
\R_{0}-\hat{\R}_{0}\left( \bH_{R_{0}}\right)
^{-1}\right) ^{\prime}\right) \sum_{t=1}^{T}\F_{0,t}\E_{t}
\\
&&+\left( \hat{\mathbf{D}}\right) ^{-1}\hat{\mathbf{N}}\left( \left(
\left( \hat{\C}_{1,\perp}-\C_{1,\perp}\right) ^{\prime
}\C_{1}\right) \otimes \left( \left( \hat{\R}_{1,\perp}-%
\R_{1,\perp}\right) ^{\prime}\R_{1}\right) \right)
\sum_{t=1}^{T}\F_{1,t}\E_{t} \\
&&+\left( \hat{\R}_{0}^{\prime}\hat{\R}_{1,\perp}%
\hat{\R}_{0}\right) ^{-1}\left( \hat{\C}_{0}^{\prime
}\hat{\C}_{1,\perp}\hat{\C}_{0}\right)
^{-1}\sum_{t=1}^{T}\hat{\R}_{0}^{\prime}\hat{\R}%
_{1,\perp}\E_{t}\hat{\C}_{1,\perp}\hat{\C}%
_{0}\E_{t} \\
&=&a+b+c+d+e.
\end{eqnarray*}%
By using the same arguments as in the above we have%
\begin{eqnarray*}
\left\Vert a\right\Vert _{F} &\leq &\left\Vert \left( \hat{\mathbf{D}}%
\right) ^{-1}\right\Vert _{F}\left\Vert \hat{\C}_{0}^{\prime}%
\hat{\C}_{1,\perp}\right\Vert _{F}\left\Vert \hat{\R%
}_{0}^{\prime}\hat{\R}_{1,\perp}\right\Vert _{F}\left\Vert
\sum_{t=1}^{T}\F_{0,t}\left( \hat{\R}_{0}\left( \mathbf{H%
}_{R_{0}}\right) ^{-1}\right) ^{\prime}\E_{t}\right\Vert _{F} \\
&&\times \left\Vert \C_{0}-\hat{\C}_{0}\left( \bH%
_{C_{0}}\right) ^{-1}\right\Vert _{F} \\
&=&O_{P}(1) \frac{1}{p_{1}^{2}p_{2}^{2}}p_{1}^{3/2}p_{2}^{3/2}%
\left( p_{1}p_{2}T\right) ^{1/2}p_{2}^{1/2}\left( \frac{1}{p_{1}p_{2}}+\frac{%
1}{p_{1}^{1/2}T^{1/2}}\right)  \\
&=&O_{P}\left( \frac{T^{1/2}}{p_{1}p_{2}^{1/2}}\right) +O_{P}\left( \frac{%
p_{2}^{1/2}}{p_{1}^{1/2}}\right) ;
\end{eqnarray*}%
\begin{eqnarray*}
\left\Vert b\right\Vert _{F} &\leq &\left\Vert \left( \hat{\mathbf{D}}%
\right) ^{-1}\right\Vert _{F}\left\Vert \hat{\C}_{0}^{\prime}%
\hat{\C}_{1,\perp}\right\Vert _{F}\left\Vert \hat{\R%
}_{0}^{\prime}\hat{\R}_{1,\perp}\right\Vert _{F}\left\Vert 
\R_{0}-\hat{\R}_{0}\left( \bH_{R_{0}}\right)
^{-1}\right\Vert _{F} \\
&&\times \left\Vert \sum_{t=1}^{T}\F_{0,t}\E_{t}\left( 
\hat{\C}_{0}\left( \bH_{C_{0}}\right) ^{-1}\right)
\right\Vert _{F} \\
&=&O_{P}(1) \frac{1}{p_{1}^{2}p_{2}^{2}}p_{1}p_{2}\left(
p_{1}p_{2}\right) ^{1/2}\left( p_{1}p_{2}T\right) ^{1/2}p_{1}^{1/2}\left( 
\frac{1}{p_{1}p_{2}}+\frac{1}{p_{2}^{1/2}T^{1/2}}\right)  \\
&=&O_{P}\left( \frac{T^{1/2}}{p_{1}^{1/2}p_{2}}\right) +O_{P}\left( \frac{%
p_{1}^{1/2}}{p_{2}^{1/2}}\right) ;
\end{eqnarray*}%
also, it follows by the same logic that $\left\Vert c\right\Vert _{F}$ is
dominated by the other two terms; 
\begin{eqnarray*}
\left\Vert d\right\Vert _{F} &\leq &\left\Vert \left( \hat{\mathbf{D}}%
\right) ^{-1}\right\Vert _{F}\left\Vert \hat{\C}_{0}\right\Vert
_{F}\left\Vert \hat{\C}_{1,\perp}\left( \hat{\C}%
_{1,\perp}-\C_{1,\perp}\right) \right\Vert _{F}\left\Vert \mathbf{C%
}_{1}\right\Vert _{F} \\
&&\times \left\Vert \hat{\R}_{0}\right\Vert _{F}\left\Vert 
\hat{\R}_{1,\perp}\left( \hat{\R}_{1,\perp}-%
\R_{1,\perp}\right) \right\Vert _{F}\left\Vert \R%
_{1}\right\Vert _{F}\left\Vert \sum_{t=1}^{T}\F_{1,t}\E%
_{t}\right\Vert _{F} \\
&=&O_{P}(1) \frac{1}{p_{1}^{2}p_{2}^{2}}p_{2}^{1/2}\frac{1}{T}%
p_{2}^{1/2}p_{1}^{1/2}\frac{1}{T}p_{1}^{1/2}\left( p_{1}p_{2}T^{2}\right)
^{1/2}=O_{P}\left( \frac{1}{p_{1}^{1/2}p_{2}^{1/2}T}\right) .
\end{eqnarray*}%
Finally%
\begin{equation*}
\left\Vert e\right\Vert _{F}\leq \left\Vert \left( \hat{\R}%
_{0}^{\prime}\hat{\R}_{1,\perp}\hat{\R}_{0}\right)
^{-1}\right\Vert _{F}\left\Vert \left( \hat{\C}_{0}^{\prime}%
\hat{\C}_{1,\perp}\hat{\C}_{0}\right)
^{-1}\right\Vert _{F}\left\Vert \sum_{t=1}^{T}\R_{0}^{\prime}%
\R_{1,\perp}\E_{t}\C_{1,\perp}\C_{0}%
\E_{t}\right\Vert _{F}+r_{p_{1}p_{2}T},
\end{equation*}%
where $r_{p_{1}p_{2}T}$\ is a dominated remainder, and%
\begin{eqnarray*}
&&\left\Vert \sum_{t=1}^{T}\R_{0}^{\prime}\R_{1,\perp}%
\E_{t}\C_{1,\perp}\C_{0}\E_{t}\right\Vert
_{F}^{2} \\
&=&\sum_{i=1}^{p_{1}}\sum_{j=1}^{p_{2}}\sum_{t,s=1}^{T}\sum_{\ell _{1},\ell
_{2}=1}^{p_{1}}\sum_{h_{1},h_{2}=1}^{p_{1}}\sum_{u_{1},u_{2}=1}^{p_{2}}%
\sum_{k_{1},k_{2}=1}^{p_{2}}r_{\perp ,h_{1}\ell _{1}}r_{\perp ,h_{2}\ell
_{2}}r_{0,\ell _{1}}r_{0,\ell _{2}}c_{\perp ,k_{1}u_{1}}c_{\perp
,k_{2}u_{2}}c_{0,u_{1}}c_{0,u_{2}}e_{h_{1}k_{1},t}e_{h_{2}k_{2},t}e_{ij,t}e_{ij,s},
\end{eqnarray*}%
which can be shown to be $O_{P}\left( p_{1}^{4}p_{2}^{4}T\right) $; putting
all together, it follows that%
\begin{equation*}
\left\Vert e\right\Vert _{F}=O_{P}\left( T^{1/2}\right) .
\end{equation*}%
Hence it follows that%
\begin{equation*}
\left\Vert \sum_{t=1}^{T}\left( \hat{\F}_{0,t}-\bH%
_{R_{0}}^{-1}\F_{0,t}\left( \bH_{C_{0}}^{\prime}\right)
^{-1}\right) \E_{t}\right\Vert _{F}=O_{P}\left( T^{1/2}\right)
+O_{P}\left( \frac{p_{2}^{1/2}}{p_{1}^{1/2}}\right) +O_{P}\left( \frac{%
p_{1}^{1/2}}{p_{2}^{1/2}}\right) ,
\end{equation*}%
and therefore%
\begin{eqnarray*}
\left\Vert IV_{b,3}\right\Vert _{F} &=&O_{P}(1) \frac{1}{%
p_{1}p_{2}^{2}T^{2}}\left\Vert \hat{\C}_{1}\right\Vert
_{F}\left\Vert \R_{0}\right\Vert _{F}\left\Vert \C%
_{0}\right\Vert _{F}\left\Vert \hat{\C}_{1}-\C_{1}%
\bH_{C_{1}}\right\Vert _{F} \\
&&\times \left\Vert \sum_{t=1}^{T}\left( \hat{\F}_{0,t}-\mathbf{H%
}_{R_{0}}^{-1}\F_{0,t}\left( \bH_{C_{0}}^{\prime}\right)
^{-1}\right) \E_{t}\right\Vert _{F} \\
&=&O_{P}(1) \frac{1}{T^{5/2}}\left( \frac{1}{\sqrt{p_{1}p_{2}}}+%
\frac{1}{p_{1\wedge 2}T^{1/2}}\right) ,
\end{eqnarray*}%
and again by the same logic, it can be shown that $IV_{b,4}-IV_{b,7}$ are
all dominated by $IV_{b,1}-IV_{b,3}$. By the same logic%
\begin{equation*}
V=\sum_{i=1}^{7}V_{i},
\end{equation*}%
with 
\begin{eqnarray*}
\left\Vert V_{1}\right\Vert _{F} &=&O_{P}(1) \frac{1}{%
p_{1}p_{2}^{2}T^{2}}\left\Vert \R_{1}\right\Vert _{F}\left\Vert 
\C_{1}\right\Vert _{F}\left\Vert \hat{\C}_{1}\right\Vert
_{F}^{2}\left\Vert \C_{0}\right\Vert _{F}\left\Vert \hat{\R}_{0}-\R_{0}\bH_{R_{0}}\right\Vert _{F}\left\Vert
\sum_{t=1}^{T}\F_{1,t}\F_{0,t}\right\Vert _{F} \\
&=&O_{P}(1) \frac{1}{p_{1}p_{2}^{2}T^{2}}%
p_{1}^{1/2}p_{2}^{2}p_{1}^{1/2}\left( \frac{1}{p_{1}p_{2}}+\frac{1}{\left(
p_{2}T\right) ^{1/2}}\right) T \\
&=&O_{P}\left( \frac{1}{T}\left( \frac{1}{p_{1}p_{2}}+\frac{1}{\left(
p_{2}T\right) ^{1/2}}\right) \right) ,
\end{eqnarray*}%
\begin{eqnarray*}
\left\Vert V_{2}\right\Vert _{F} &=&O_{P}(1) \frac{1}{%
p_{1}p_{2}^{2}T^{2}}\left\Vert \R_{1}\right\Vert _{F}\left\Vert 
\C_{1}\right\Vert _{F}\left\Vert \hat{\C}_{1}\right\Vert
_{F}^{2}\left\Vert \R_{0}\right\Vert _{F}\left\Vert \hat{\C}_{0}-\C_{0}\bH_{C_{0}}\right\Vert _{F}\left\Vert
\sum_{t=1}^{T}\F_{1,t}\F_{0,t}\right\Vert _{F} \\
&=&O_{P}\left( \frac{1}{T}\left( \frac{1}{p_{1}p_{2}}+\frac{1}{\left(
p_{1}T\right) ^{1/2}}\right) \right) ,
\end{eqnarray*}%
\begin{eqnarray*}
\left\Vert V_{3}\right\Vert _{F} &=&O_{P}(1) \frac{1}{%
p_{1}p_{2}^{2}T^{2}}\left\Vert \R_{1}\right\Vert _{F}\left\Vert 
\C_{1}\right\Vert _{F}\left\Vert \hat{\C}_{1}\right\Vert
_{F}^{2}\left\Vert \R_{0}\right\Vert _{F}\left\Vert \C%
_{0}\right\Vert _{F} \\
&&\times \sum_{t=1}^{T}\left\Vert \left( \hat{\F}_{0,t}-\mathbf{H%
}_{R_{0}}^{-1}\F_{0,t}\left( \bH_{C_{0}}^{\prime}\right)
^{-1}\right) \F_{1,t}^{\prime}\right\Vert _{F} \\
&=&O_{P}(1) \frac{1}{p_{1}p_{2}^{2}T^{2}}%
p_{1}^{1/2}p_{2}^{1/2}p_{2}p_{1}^{1/2}p_{2}^{1/2}T\left( \frac{1}{\sqrt{%
p_{1}p_{2}}}+\frac{1}{\left( p_{1\wedge 2}T\right) ^{1/2}}\right)  \\
&=&O_{P}\left( \frac{1}{T}\left( \frac{1}{\sqrt{p_{1}p_{2}}}+\frac{1}{\left(
p_{1\wedge 2}T\right) ^{1/2}}\right) \right) ,
\end{eqnarray*}%
again by Lemma \ref{bai}; similarly, it can be shown that $V_{4}-V_{7}$ are
all dominated by $V_{1}-V_{3}$. Finally we write%
\begin{equation*}
VI=\sum_{i=1}^{7}VI_{i}.
\end{equation*}%
It holds that%
\begin{eqnarray*}
\left\Vert VI_{1}\right\Vert _{F} &=&O_{P}(1) \frac{1}{%
p_{1}p_{2}^{2}T^{2}}\left\Vert \hat{\C}_{1}\right\Vert
_{F}^{2}\left\Vert \C_{0}\right\Vert _{F}^{2}\left\Vert \hat{%
\R}_{0}-\R_{0}\bH_{R_{0}}\right\Vert
_{F}^{2}\sum_{t=1}^{T}\left\Vert \F_{0,t}\right\Vert _{F}^{2} \\
&=&O_{P}(1) \frac{1}{p_{1}p_{2}^{2}T^{2}}p_{2}p_{2}p_{1}\left( 
\frac{1}{p_{1}^{2}p_{2}^{2}}+\frac{1}{p_{2}T}\right) T \\
&=&O_{P}\left( \frac{1}{T}\left( \frac{1}{p_{1}^{2}p_{2}^{2}}+\frac{1}{p_{2}T%
}\right) \right) ,
\end{eqnarray*}%
and similarly 
\begin{equation*}
\left\Vert VI_{2}\right\Vert _{F}=O_{P}\left( \frac{1}{T}\left( \frac{1}{%
p_{1}^{2}p_{2}^{2}}+\frac{1}{p_{1}T}\right) \right) ,
\end{equation*}%
and%
\begin{eqnarray*}
\left\Vert VI_{3}\right\Vert _{F} &=&O_{P}(1) \frac{1}{%
p_{1}p_{2}^{2}T^{2}}\left\Vert \hat{\C}_{1}\right\Vert
_{F}^{2}\left\Vert \C_{0}\right\Vert _{F}^{2}\left\Vert \R%
_{0}\right\Vert _{F}^{2}\sum_{t=1}^{T}\left\Vert \hat{\F}_{0,t}-%
\bH_{R_{0}}^{-1}\F_{0,t}\left( \bH_{C_{0}}^{\prime
}\right) ^{-1}\right\Vert _{F}^{2} \\
&=&O_{P}(1) \frac{1}{p_{1}p_{2}^{2}T^{2}}p_{2}p_{2}p_{1}T\left( 
\frac{1}{p_{1}p_{2}}+\frac{1}{\left( p_{1\wedge 2}T\right) ^{1/2}}\right)
^{2} \\
&=&O_{P}\left( \frac{1}{T}\left( \frac{1}{p_{1}p_{2}}+\frac{1}{\left(
p_{1\wedge 2}T\right) ^{1/2}}\right) ^{2}\right) ;
\end{eqnarray*}%
again, it can be shown by the same logic that $VI_{4}-VI_{7}$ are all
dominated by $VI_{1}-VI_{3}$. The desired result now follows from the same
logic as in the previous proofs.
\end{proof}
\end{lemma}
\begin{lemma}
\label{spec-m-c-tilde}We assume that Assumptions \ref{as-1}-\ref{as-5} are
satisfied. Then there exists a positive constant $c_{0}$ such that%
\begin{equation*}
\lambda _{j}\left( \proj{\M}_{C_{1}}\right) =c_{0}+o_{P}(1) 
\text{,}
\end{equation*}%
for all $j\leq h_{C_{1}}$, and%
\begin{equation*}
\lambda _{j}\left( \proj{\M}_{C_{1}}\right) =O_{P}\left( \frac{1}{%
p_{1\wedge 2}^{1/2}T^{3/2}}\right) +O_{P}\left( \frac{1}{p_{1}T}\right)
+O_{P}\left( \frac{1}{T^{2}}\right) +O_{P}\left( \frac{1}{%
p_{1}^{1/2}p_{2}^{1/2}T}\right) ,
\end{equation*}%
for all $j>h_{C_{1}}$.

\begin{proof}
The proof is the same as the proof of Lemma \ref{spec-m-r-tilde}, \textit{%
mutatis mutandis}.
\end{proof}
\end{lemma}

\begin{lemma}
\label{eig-tilde}We assume that Assumptions \ref{as-1}-\ref{as-5} are
satisfied. Then it holds that%
\begin{equation*}
\left\Vert \widetilde{\Lambda }_{R_{1}}^{-1}\right\Vert =O_{P}(1) ,%
\text{ \ \ and \ \ }\left\Vert \widetilde{\Lambda }_{C_{1}}^{-1}\right\Vert
=O_{P}(1) .
\end{equation*}

\begin{proof}
The proof follows from Lemmas \ref{spec-m-r-tilde} and \ref{spec-m-c-tilde},
in the same way as the proof of Lemma \ref{spec-eig-mRx}.
\end{proof}
\end{lemma}

\clearpage
\newpage



\section{Proofs\label{proofs}}
Henceforth, we will use the following notation: $\log \left( x\right) $ is
the natural log of $x$;

\begin{proof}[Proof of Theorem \protect\ref{hat-estimates}]
We begin by studying the estimator of $\R_1$. By construction, it holds that%
\begin{equation*}
\hat{\R}_1=\M_{R_1}\hat{\R}_1\Lambda_{R_{1}}^{-1},
\end{equation*}%
where recall that, by Lemma \ref{lambda}, $\left\Vert \Lambda
_{R_{1}}^{-1}\right\Vert =O_{P}(1) $. Hence%
\begin{align}
\hat{\R}_1 &=\frac{1}{p_{1}p_{2}T^{2}}\sum_{t=1}^{T}\R_1\F_{1,t}\C_{1}^{\prime}\C_{1}\F_{1,t}^{\prime}\R_{1}^{\prime}\hat{\R}_1\Lambda _{R_{1}}^{-1}+\frac{1}{p_{1}p_{2}T^{2}}\sum_{t=1}^{T}\R_{0}\F_{0,t}\C_{0}^{\prime}\C_{0}\F_{0,t}^{\prime}\R_{0}^{\prime}\hat{\R}_1\Lambda
_{R_{1}}^{-1}  \label{r-hat-dec} \\
&+\frac{1}{p_{1}p_{2}T^{2}}\sum_{t=1}^{T}\E_{t}\E_{t}^{\prime}\hat{\R}_1\Lambda _{R_{1}}^{-1}+\frac{1}{p_{1}p_{2}T^{2}}\sum_{t=1}^{T}\R_1\F_{1,t}\C_{1}^{\prime}\C_{0}\F%
_{0,t}^{\prime}\R_{0}^{\prime}\hat{\R}_1\Lambda _{R_{1}}^{-1}
\notag \\
&+\left( \frac{1}{p_{1}p_{2}T^{2}}\sum_{t=1}^{T}\R_1\F_{1,t}\C_{1}^{\prime}\C_{0}\F_{0,t}^{\prime}\R_{0}^{\prime}%
\hat{\R}_1\Lambda _{R_{1}}^{-1}\right) ^{\prime}  \notag \\
&+\frac{1}{p_{1}p_{2}T^{2}}\sum_{t=1}^{T}\R_1\F_{1,t}\C_{1}^{\prime}%
\E_{t}^{\prime}\hat{\R}_1\Lambda _{R_{1}}^{-1}+\left( \frac{1%
}{p_{1}p_{2}T^{2}}\sum_{t=1}^{T}\R_1\F_{1,t}\C_{1}^{\prime}\E%
_{t}^{\prime}\hat{\R}_1\Lambda _{R_{1}}^{-1}\right) ^{\prime}  \notag
\\
&+\frac{1}{p_{1}p_{2}T^{2}}\sum_{t=1}^{T}\R_{0}\F_{0,t}%
\C_{0}^{\prime}\E_{t}^{\prime}\hat{\R}_1\Lambda
_{R_{1}}^{-1}+\left( \frac{1}{p_{1}p_{2}T^{2}}\sum_{t=1}^{T}\R_{0}%
\F_{0,t}\C_{0}^{\prime}\E_{t}^{\prime}\hat{\R}_1\Lambda _{R_{1}}^{-1}\right) ^{\prime}  \notag \\
&=I+II+III+IV+IV^{\prime}+V+V^{\prime}+VI+VI^{\prime}.  \notag
\end{align}%
Define%
\begin{equation}
\bH_{R_{1}}=\frac{1}{p_{1}T^{2}}\sum_{t=1}^{T}\F_{1,t}\frac{\C_{1}^{\prime}\C_{1}}{p_{2}}\F_{1,t}^{\prime}\R_{1}^{\prime}\hat{\R}_1\Lambda _{R_{1}}^{-1}=\left( \frac{1}{T^{2}}\sum_{t=1}^{T}%
\F_{1,t}\F_{1,t}^{\prime}\right) \left( \frac{\R_{1}^{\prime}\hat{\R}_1}{p_{1}}\right) \Lambda _{R_{1}}^{-1};  \label{h-r}
\end{equation}%
then it is immediate to see that%
\begin{align*}
\left\Vert \bH_{R_{1}}\right\Vert _{F} &\leq\frac{1}{p_{1}T^{2}}%
\left\Vert \R_1\right\Vert _{F}\left\Vert \hat{\R}_1\right\Vert _{F}\left\Vert \sum_{t=1}^{T}\F_{1,t}\F_{1,t}^{\prime}\right\Vert _{F}\left\Vert \Lambda _{R_{1}}^{-1}\right\Vert_{F} \\
&=O_{P}(1)
\end{align*}%
having used the identification restriction $\C_{1}^{\prime}\C_{1}=p_{2}\I_{h_{C_{1}}}$, and the facts that $\left\Vert \R%
\right\Vert _{F}=O\left( p_{1}^{1/2}\right) $, $\left\Vert \hat{\R}_1\right\Vert _{F}=p_{1}^{1/2}\ $by construction, and $\left\Vert \Lambda
_{R_{1}}^{-1}\right\Vert _{F}=O_{P}(1) $, and (\ref{lemma1-ii}) in
Lemma \ref{berkes}. Further (using $h_{R_{1}}=h_{C_{1}}=h_{R_{0}}=h_{C_{0}}=1$)%
\begin{equation*}
\left\Vert II\right\Vert _{F}\leq \frac{1}{p_{1}p_{2}T^{2}}\left\Vert 
\R_{0}\right\Vert _{F}^{2}\left\Vert \C_{0}\right\Vert
_{F}^{2}\left\Vert \hat{\R}_1\right\Vert _{F}\left\Vert \Lambda
_{R_{1}}^{-1}\right\Vert _{F}\left( \sum_{t=1}^{T}\F_{0,t}^{2}\right)
=O_{P}(1) \frac{p_{1}^{1/2}}{T},
\end{equation*}%
having used Assumption \ref{as-4}, and Lemmas \ref{f1-summab} and \ref%
{lambda}. Consider now%
\begin{equation*}
\left\Vert III\right\Vert _{F}\leq \frac{1}{p_{1}p_{2}T^{2}}\lambda _{\max
}\left( \sum_{t=1}^{T}\E_{t}\E_{t}^{\prime}\right)
\left\Vert \hat{\R}_1\right\Vert _{F}\left\Vert \Lambda
_{R_{1}}^{-1}\right\Vert _{F}=O_{P}\left( \frac{1}{p_{1}^{1/2}T}\right)
+O_{P}\left( \frac{p_{1}^{1/2}}{p_{2}^{1/2}T^{3/2}}\right) ,
\end{equation*}%
by Lemma \ref{err}. Similarly%
\begin{align*}
\left\Vert IV\right\Vert _{F} &\leq\frac{1}{p_{1}p_{2}T^{2}}\left\Vert 
\R_1\right\Vert _{F}\left\Vert \R_{0}\right\Vert
_{F}\left\Vert \hat{\R}_1\right\Vert _{F}\left\Vert \Lambda
_{R_{1}}^{-1}\right\Vert _{F}\left\Vert \C_{1}\right\Vert _{F}\left\Vert 
\C_{0}\right\Vert _{F}\left\vert \sum_{t=1}^{T}\F_{1,t}\F_{0,t}^{\prime}\right\vert \\
&=O_{P}\left( T\right) \frac{1}{p_{1}p_{2}T^{2}}%
p_{1}^{1/2}p_{1}^{1/2}p_{1}^{1/2}p_{2}^{1/2}p_{2}^{1/2}=O_{P}\left( \frac{%
p_{1}^{1/2}}{T}\right) ,
\end{align*}%
having used Assumption \ref{as-4} and Lemma \ref{crossf}. The same holds for 
$\left\Vert IV^{\prime}\right\Vert _{F}$. By the same token%
\begin{align*}
\left\Vert V\right\Vert _{F} &\leq\frac{1}{p_{1}p_{2}T^{2}}\left\Vert 
\R_1\right\Vert _{F}\left\Vert \hat{\R}_1\right\Vert
_{F}\left\Vert \Lambda _{R_{1}}^{-1}\right\Vert _{F}\left\Vert \sum_{t=1}^{T}%
\F_{1,t}\C_{1}^{\prime}\E_{t}^{\prime}\right\Vert _{F} \\
&=O_{P}\left( p_{1}^{1/2}p_{2}^{1/2}T\right) \frac{1}{p_{1}p_{2}T^{2}}%
p_{1}^{1/2}p_{1}^{1/2}=O_{P}\left( \frac{p_{1}^{1/2}}{p_{2}^{1/2}T}\right) ,
\end{align*}%
by Assumption \ref{as-4}, the fact that $\left\Vert \sum_{t=1}^{T}\F_{1,t}\C_{1}^{\prime}\E_{t}^{\prime}\right\Vert _{F}=O_{P}\left(
p_{1}^{1/2}p_{2}^{1/2}T\right) $, and Lemma \ref{lambda}; and the same also
holds for $\left\Vert V^{\prime}\right\Vert _{F}$. Finally, using the same
arguments as above, it holds that%
\begin{align*}
\left\Vert VI\right\Vert _{F} &\leq\frac{1}{p_{1}p_{2}T^{2}}\left\Vert 
\R_{0}\right\Vert _{F}\left\Vert \hat{\R}_1\right\Vert
_{F}\left\Vert \Lambda _{R_{1}}^{-1}\right\Vert _{F}\left\Vert \sum_{t=1}^{T}%
\F_{0,t}\C_{0}^{\prime}\E_{t}^{\prime}\right\Vert
_{F} \\
&=O_{P}\left( p_{1}^{1/2}p_{2}^{1/2}T^{1/2}\right) \frac{1}{p_{1}p_{2}T^{2}}%
p_{1}^{1/2}p_{1}^{1/2}=O_{P}\left( \frac{p_{1}^{1/2}}{p_{2}^{1/2}T^{3/2}}%
\right) ,
\end{align*}%
and the same holds for $\left\Vert VI^{\prime}\right\Vert _{F}$. Then,
putting all together, it follows that%
\begin{equation*}
\left\Vert \hat{\R}_1-\R_1\bH_{R_{1}}\right\Vert _{F}=O_{P}\left( 
\frac{p_{1}^{1/2}}{T}\right) .
\end{equation*}%
We conclude the proof by showing that $\left\Vert \left( \bH%
_{R_{1}}\right) ^{-1}\right\Vert _{F}=O_{P}(1) $. Recall that, by
construction, $\hat{\R}_1^{\prime}\hat{\R}_1=p_{1}%
\I_{h_{R_{1}}}$; recall also the identification restriction $\R_{1}^{\prime}\R_{1}=p_{1}\I_{h_{R_{1}}}$; hence%
\begin{align*}
\I_{h_{R_{1}}} &=\frac{1}{p_{1}}\hat{\R}_1^{\prime}\hat{%
\R_1} \\
&=\bH_{R_{1}}^{\prime}\left( \frac{1}{p_{1}}\R_{1}^{\prime}\R_1\right) \bH_{R_{1}}+\frac{1}{p_{1}}\left( \hat{\R}_1-\R_1\bH_{R_{1}}\right) ^{\prime}\hat{\R}_1 \\
&+\frac{1}{p_{1}}\hat{\R}_1^{\prime}\left( \hat{\R}_1-%
\R_1\bH_{R_{1}}\right) +\frac{1}{p_{1}}\left( \hat{\R}_1-\R_1\bH_{R_{1}}\right) ^{\prime}\left( \hat{\R}_1-\R_1\bH_{R_{1}}\right) \\
&=\bH_{R_{1}}^{\prime}\bH_{R_{1}}+I+I^{\prime}+II.
\end{align*}%
Clearly%
\begin{equation*}
\left\Vert I\right\Vert _{F}\leq \frac{1}{p_{1}}\left\Vert \hat{\R}_1\right\Vert _{F}\left\Vert \hat{\R}_1-\R_1\bH_{R_{1}}\right\Vert
_{F}=O_{P}\left( \frac{1}{T}\right) ;
\end{equation*}%
the same holds for $\left\Vert I^{\prime}\right\Vert _{F}$, and, by the
same token, $\left\Vert II\right\Vert _{F}$ is dominated. Hence%
\begin{equation}
\bH_{R_{1}}^{\prime}\bH_{R_{1}}=\I_{h_{R_{1}}}+O_{P}\left( \frac{1%
}{T}\right) .  \label{h-r-orthogonal}
\end{equation}%
Thus, as $\min \left\{ p_{1},p_{2},T\right\} \rightarrow \infty $, $\bH_{R_{1}}$ is an orthogonal matrix, and therefore $\left( \bH_{R_{1}}\right)
^{-1}=\bH_{R_{1}}^{\prime}+o_{P}(1) $. Now $\left\Vert 
\bH_{R_{1}}\right\Vert _{F}=O_{P}(1) $ follows immediately.

We now turn to studyingt the estimator of $\C_1$. Observing that%
\begin{equation*}
\hat{\C}_{1}=\M_{C_1}\hat{\C}_{1}\Lambda _{C_{1}}^{-1},
\end{equation*}%
and that, by Lemma \ref{lambda}, $\left\Vert \Lambda _{C_{1}}^{-1}\right\Vert
=O_{P}(1) $, it holds that%
\begin{align*}
\hat{\C}_{1} &=\frac{1}{p_{1}p_{2}T^{2}}\sum_{t=1}^{T}\C_{1}\F_{1,t}^{\prime}\R_{1}^{\prime}\R_1\F_{1,t}\C_{1}^{\prime}%
\hat{\C}_{1}\Lambda _{C_{1}}^{-1}+\frac{1}{p_{1}p_{2}T^{2}}\sum_{t=1}^{T}%
\C_{0}\F_{0,t}^{\prime}\R_{0}^{\prime}\R_{0}\F_{0,t}\C_{0}^{\prime}\hat{\C}_{1}\Lambda
_{C_{1}}^{-1} \\
&+\frac{1}{p_{1}p_{2}T^{2}}\sum_{t=1}^{T}\E_{t}^{\prime}\E%
_{t}\hat{\C}_{1}\Lambda _{C_{1}}^{-1}+\frac{1}{p_{1}p_{2}T^{2}}%
\sum_{t=1}^{T}\C_{1}\F_{1,t}^{\prime}\R_{1}^{\prime}\R_{0}%
\F_{0,t}\C_{0}^{\prime}\hat{\C}_{1}\Lambda _{C_{1}}^{-1}
\\
&+\left( \frac{1}{p_{1}p_{2}T^{2}}\sum_{t=1}^{T}\C_{1}\F_{1,t}^{\prime}%
\R_{1}^{\prime}\R_{0}\F_{0,t}\C_{0}^{\prime}%
\hat{\C}_{1}\Lambda _{C_{1}}^{-1}\right) ^{\prime} \\
&+\frac{1}{p_{1}p_{2}T^{2}}\sum_{t=1}^{T}\C_{1}\F_{1,t}^{\prime}\R_{1}^{\prime}\E_{t}\hat{\C}_{1}\Lambda _{C_{1}}^{-1}+\left( \frac{1%
}{p_{1}p_{2}T^{2}}\sum_{t=1}^{T}\R_1\F_{1,t}\C_{1}^{\prime}\E%
_{t}^{\prime}\hat{\C}_{1}\Lambda _{C_{1}}^{-1}\right) ^{\prime} \\
&+\frac{1}{p_{1}p_{2}T^{2}}\sum_{t=1}^{T}\C_{0}\F%
_{0,t}^{\prime}\R_{0}^{\prime}\E_{t}\hat{\C}_{1}%
\Lambda _{C_{1}}^{-1}+\left( \frac{1}{p_{1}p_{2}T^{2}}\sum_{t=1}^{T}\R_{0}\F_{0,t}\C_{0}^{\prime}\E_{t}^{\prime}\hat{%
\C}_{1}\Lambda _{C_{1}}^{-1}\right) ^{\prime} \\
&=I+II+III+IV+IV^{\prime}+V+V^{\prime}+VI+VI^{\prime}.
\end{align*}%
Letting 
\begin{equation}
\bH_{C_{1}}=\frac{1}{p_{2}T^{2}}\sum_{t=1}^{T}\F_{1,t}^{\prime}%
\frac{\R_{1}^{\prime}\R_1}{p_{1}}\F_{1,t}\C_{1}^{\prime}\hat{\C}_{1}\Lambda _{C_{1}}^{-1}=\left( \frac{1}{T^{2}}\sum_{t=1}^{T}%
\F_{1,t}^{\prime}\F_{1,t}\right) \left( \frac{\C_{1}^{\prime}\hat{\C}_{1}}{p_{2}}\right) \Lambda _{C_{1}}^{-1},  \label{h-c}
\end{equation}%
the proof proceeds as above, \textit{mutatis
mutandis}.
\end{proof}

\begin{proof}[Proof of Lemma \protect\ref{f-hat-negative}]
Recall that%
\begin{align*}
\hat{\F}_{1,t} &=\frac{1}{p_{1}p_{2}}\hat{\R}_1^{\prime}\X_{t}\hat{\C}_{1} \\
&=\frac{1}{p_{1}p_{2}}\hat{\R}_1^{\prime}\R_1\F_{1,t}\C_{1}^{\prime}\hat{\C}_{1}+\frac{1}{p_{1}p_{2}}\hat{\R}_1%
^{\prime}\R_{0}\F_{0,t}\C_{0}^{\prime}\hat{%
\C}_{1}+\frac{1}{p_{1}p_{2}}\hat{\R}_1^{\prime}\E_{t}%
\hat{\C}_{1} \\
&=I+II+III.
\end{align*}%
We will use the decompositions%
\begin{equation*}
\R_{1}=\R_{1} \pm \hat{\R}_1\left( \bH_{R_{1}}\right) ^{-1},\text{ \ \ and \ \ }\C_{1} = \C_{1} \pm \hat{\C}_{1}\left( \bH_{C_{1}}\right) ^{-1}.
\end{equation*}%
Consider $I$; it holds that
\begin{align*}
I &=\frac{\hat{\R}_1^{\prime}\hat{\R}_1}{p_{1}}\left( \bH_{R_{1}}\right) ^{-1}\F_{1,t}\left(\bH_{C_{1}}^{\prime}\right) ^{-1}\frac{\hat{\C}_{1}^{\prime}\hat{\C}_{1}}{p_{2}}-\frac{\hat{\R}_1^{\prime} \left(\hat{\R}_1\left(\bH_{R_{1}}\right)^{-1}-\R_1\right) }{p_{1}}\F_{1,t}\left( \bH_{C_{1}}^{\prime}\right) ^{-1} \frac{\hat{\C}_{1}^{\prime}\hat{\C}_{1}}{p_{2}} \\
&-\frac{\hat{\R}_1^{\prime}\hat{\R}_1}{p_{1}}\left( \bH_{R_{1}}\right) ^{-1}\F_{1,t}^{\prime}\frac{\left( 
\hat{\C}_{1}\left( \bH_{C_{1}}\right) ^{-1}-\C_{1}\right) ^{\prime}\hat{\C}_{1}}{p_{2}}+\frac{\hat{\R}_1^{\prime}\left( \hat{\R}_1\left( \bH_{R_{1}}\right)^{-1}-\R_1\right) }{p_{1}}\F_{1,t}\frac{\left(\hat{\C}_{1}\left( \bH_{C_{1}}\right)^{-1}-\C_{1}\right) ^{\prime}\hat{\C}_{1}}{p_{2}} \\
&=\left( \bH_{R_{1}}\right) ^{-1}\F_{1,t}\left( \bH_{C_{1}}^{\prime}\right) ^{-1}-I_{a}-I_{b}+I_{c}.
\end{align*}%
By (\ref{lemma1-i}) in Lemma \ref{berkes}, it immediately follows that $%
\left\Vert \F_{1,t}\right\Vert _{F}=O_{P}\left( T^{1/2}\right) $. Hence%
\begin{equation*}
\left\Vert I_{a}\right\Vert _{F}\leq \frac{\left\Vert \hat{\R}_1%
\right\Vert _{F}\left\Vert \hat{\R}_1\left( \bH%
_{R_{1}}\right) ^{-1}-\R_1\right\Vert _{F}}{p_{1}}\left\Vert \F_{1,t}\right\Vert _{F}\left\Vert \left( \bH_{C_{1}}^{\prime}\right)
^{-1}\right\Vert _{F}=O_{P}\left( T^{-1/2}\right) ,
\end{equation*}%
having used Theorem \ref{hat-estimates}; similarly%
\begin{equation*}
\left\Vert I_{b}\right\Vert _{F}\leq \frac{\left\Vert \hat{\C}_{1}%
\right\Vert _{F}\left\Vert \hat{\C}_{1}\left( \bH%
_{C_{1}}\right) ^{-1}-\C_{1}\right\Vert _{F}}{p_{1}}\left\Vert \F_{1,t}\right\Vert _{F}\left\Vert \left( \bH_{R_{1}}\right) ^{-1}\right\Vert
_{F}=O_{P}\left( T^{-1/2}\right) ,
\end{equation*}%
by Theorem \ref{hat-estimates}. By the same token, it is easy to see that $\left\Vert
I_{c}\right\Vert _{F}=O_{P}\left( T^{-3/2}\right) $. Further%
\begin{equation*}
\left\Vert II\right\Vert _{F}\leq \frac{1}{p_{1}p_{2}}\left\Vert \hat{%
\R_1}\right\Vert _{F}\left\Vert \R_{0}\right\Vert
_{F}\left\Vert \C_{0}\right\Vert _{F}\left\Vert \hat{\C}_{1}%
\right\Vert _{F}\left\Vert \F_{0,t}\right\Vert _{F}=O_{P}\left(
1\right) ,
\end{equation*}%
which is a consequence of the fact that, by Assumption \ref{as-2}, $%
\left\Vert \F_{0,t}\right\Vert _{F}=O_{P}(1) $. Finally 
\begin{align*}
&\left\Vert III\right\Vert _{F} \\
&\leq\left\Vert \frac{1}{p_{1}p_{2}}\left( \R_1\bH_{R_{1}}\right) ^{\prime}\E_{t}\C_{1}\bH_{C_{1}}\right\Vert _{F}+\left\Vert \frac{1}{p_{1}p_{2}%
}\left( \hat{\R}_1-\R_1\bH_{R_{1}}\right) ^{\prime}\E_{t}%
\C_{1}\bH_{C_{1}}\right\Vert _{F} \\
&+\left\Vert \frac{1}{p_{1}p_{2}}\left( \R_1\bH_{R_{1}}\right) ^{\prime}%
\E_{t}\left( \hat{\C}_{1}-\C_{1}\bH_{C_{1}}\right) \right\Vert
_{F}+\left\Vert \frac{1}{p_{1}p_{2}}\left( \hat{\R}_1-\R_1\bH%
_{R_{1}}\right) ^{\prime}\E_{t}\left( \hat{\C}_{1}-\C_{1}\bH%
_{C_{1}}\right) \right\Vert _{F} \\
&\leq\frac{1}{p_{1}p_{2}}\left\Vert \R_{1}^{\prime}\E_{t}%
\C_{1}\right\Vert _{F}\left\Vert \bH_{R_{1}}\right\Vert
_{F}\left\Vert \bH_{C_{1}}\right\Vert _{F}+\frac{1}{p_{1}p_{2}}\left\Vert 
\hat{\R}_1-\R_1\bH_{R_{1}}\right\Vert _{F}\left\Vert \E_{t}%
\C_{1}\right\Vert _{F}\left\Vert \bH_{C_{1}}\right\Vert _{F} \\
&+\frac{1}{p_{1}p_{2}}\left\Vert \R_{1}^{\prime}\E%
_{t}\right\Vert _{F}\left\Vert \bH_{R_{1}}\right\Vert _{F}\left\Vert 
\hat{\C}_{1}-\C_{1}\bH_{C_{1}}\right\Vert _{F}+\frac{1}{p_{1}p_{2}}%
\left\Vert \E_{t}\right\Vert _{F}\left\Vert \hat{\R}_1-%
\R_1\bH_{R_{1}}\right\Vert _{F}\left\Vert \hat{\C}_{1}-\C_{1}\bH%
_{C_{1}}\right\Vert _{F} \\
&=O_{P}\left( \frac{1}{p_{1}^{1/2}p_{2}^{1/2}}\right) +O_{P}\left( \frac{1}{%
p_{2}^{1/2}T}\right) +O_{P}\left( \frac{1}{p_{1}^{1/2}T}\right) +O_{P}\left( 
\frac{1}{T^{2}}\right) .
\end{align*}%
This follows because%
\begin{align*}
&E\left\Vert \R_{1}^{\prime}\E_{t}\C_{1}\right\Vert
_{F}^{2} \\
&=E\left( \sum_{i=1}^{p_{1}}\sum_{j=1}^{p_{2}}r_{i}c_{j}e_{ij,t}\right)
^{2}\leq \left( \max_{1\leq i\leq p_{1}}r_{i}^{2}\right) \left( \max_{1\leq
i\leq p_{2}}c_{i}^{2}\right) \sum_{i,i^{\prime}=1}^{p_{1}}\sum_{j,j^{\prime}=1}^{p_{2}}\left\vert E\left( e_{ij,t}e_{i^{\prime}j^{\prime},t}\right)
\right\vert \leq c_{0}p_{1}p_{2},
\end{align*}%
by Assumption \ref{as-3}\textit{(ii)}(e); also%
\begin{align*}
&E\left\Vert \E_{t}\C_{1}\right\Vert _{F}^{2} \\
&=E\left( \sum_{i=1}^{p_{1}}\left( \sum_{j=1}^{p_{2}}c_{j}e_{ij,t}\right)
^{2}\right) \leq \left( \max_{1\leq i\leq p_{2}}c_{i}^{2}\right)
\sum_{i=1}^{p_{1}}\sum_{j,j^{\prime}=1}^{p_{2}}\left\vert E\left(
e_{ij,t}e_{ij^{\prime},t}\right) \right\vert \leq c_{0}p_{1}p_{2},
\end{align*}%
by Assumption \ref{as-3}\textit{(ii)}(c); further, we have $E\left\Vert 
\R_{1}^{\prime}\E_{t}\right\Vert _{F}^{2}\leq c_{0}p_{1}p_{2}$,
by Assumption \ref{as-3}\textit{(ii)}(b) and the same arguments as above;
and, finally, we also have $E\left\Vert \E_{t}\right\Vert _{F}^{2}$ $%
=$ $\sum_{i=1}^{p_{1}}\sum_{j=1}^{p_{2}}E\left( e_{ij,t}^{2}\right) $ $\leq $
$c_{0}p_{1}p_{2}$.
\end{proof}

\begin{proof}[Proof of Lemma \protect\ref{proj-negative}]
We study the estimator of $\R_1$ first. Some of the arguments in the proof are
based on repeating some of the passages above, and we therefore omit them
for brevity. It holds that 
\begin{align}
\hat{\R}_1^{\dagger} &=\frac{1}{p_{1}p_{2}^{2}T^{2}}%
\sum_{t=1}^{T}\R_1\F_{1,t}\C_{1}^{\prime}\hat{\C}_{1}%
\hat{\C}_{1}^{\prime}\C_{1}\F_{1,t}^{\prime}\R_{1}^{\prime}%
\hat{\R}_1^{\dagger}\left( \Lambda _{R_{1}}^{\dagger}\right) ^{-1}
\label{r-hat-dagger} \\
&+\frac{1}{p_{1}p_{2}^{2}T^{2}}\sum_{t=1}^{T}\R_{0}\F_{0,t}%
\C_{0}^{\prime}\hat{\C}_{1}\hat{\C}_{1}^{\prime}%
\C_{0}\F_{0,t}^{\prime}\R_{0}^{\prime}\hat{%
\R_1}^{\dagger}\left( \Lambda _{R_{1}}^{\dagger}\right) ^{-1}  \notag \\
&+\frac{1}{p_{1}p_{2}^{2}T^{2}}\sum_{t=1}^{T}\E_{t}\hat{\C}_{1}\hat{\C}_{1}^{\prime}\E_{t}^{\prime}\hat{\R}_1%
^{\dagger}\left( \Lambda _{R_{1}}^{\dagger}\right) ^{-1}  \notag \\
&+\frac{1}{p_{1}p_{2}^{2}T^{2}}\sum_{t=1}^{T}\R_1\F_{1,t}\C_{1}^{\prime}\hat{\C}_{1}\hat{\C}_{1}^{\prime}\E%
_{t}^{\prime}\hat{\R}_1^{\dagger}\left( \Lambda _{R_{1}}^{\dagger}\right) ^{-1}  \notag \\
&+\left( \frac{1}{p_{1}p_{2}^{2}T^{2}}\sum_{t=1}^{T}\R_1\F_{1,t}\C_{1}^{\prime}\hat{\C}_{1}\hat{\C}_{1}^{\prime}\E%
_{t}^{\prime}\hat{\R}_1^{\dagger}\left( \Lambda _{R_{1}}^{\dagger}\right) ^{-1}\right) ^{\prime}  \notag \\
&+\frac{1}{p_{1}p_{2}^{2}T^{2}}\sum_{t=1}^{T}\R_1\F_{1,t}\C_{1}^{\prime}\hat{\C}_{1}\hat{\C}_{1}^{\prime}\C_{0}%
\F_{0,t}^{\prime}\R_{0}^{\prime}\hat{\R}_1%
^{\dagger}\left( \Lambda _{R_{1}}^{\dagger}\right) ^{-1}  \notag \\
&+\left( \frac{1}{p_{1}p_{2}^{2}T^{2}}\sum_{t=1}^{T}\R_1\F_{1,t}\C_{1}^{\prime}\hat{\C}_{1}\hat{\C}_{1}^{\prime}\C_{0}%
\F_{0,t}^{\prime}\R_{0}^{\prime}\hat{\R}_1%
^{\dagger}\left( \Lambda _{R_{1}}^{\dagger}\right) ^{-1}\right) ^{\prime} 
\notag \\
&+\frac{1}{p_{1}p_{2}^{2}T^{2}}\sum_{t=1}^{T}\R_{0}\F_{0,t}%
\C_{0}^{\prime}\hat{\C}_{1}\hat{\C}_{1}^{\prime}%
\E_{t}^{\prime}\hat{\R}_1^{\dagger}\left( \Lambda
_{R_{1}}^{\dagger}\right) ^{-1}  \notag \\
&+\left( \frac{1}{p_{1}p_{2}^{2}T^{2}}\sum_{t=1}^{T}\R_{0}\F%
_{0,t}\C_{0}^{\prime}\hat{\C}_{1}\hat{\C}_{1}%
^{\prime}\E_{t}^{\prime}\hat{\R}_1^{\dagger}\left(
\Lambda _{R_{1}}^{\dagger}\right) ^{-1}\right) ^{\prime}  \notag \\
&=I+II+III+IV+IV^{\prime}+V+V^{\prime}+VI+VI^{\prime}.  \notag
\end{align}%
We have%
\begin{equation*}
I=\R_1\bH_{R_{1}}^{\dagger},
\end{equation*}%
having defined%
\begin{equation*}
\bH_{R_{1}}^{\dagger}=\frac{1}{T^{2}}\sum_{t=1}^{T}\F_{1,t}\left( 
\frac{\C_{1}^{\prime}\hat{\C}_{1}}{p_{2}}\right) \left( \frac{%
\C_{1}^{\prime}\hat{\C}_{1}}{p_{2}}\right) ^{\prime}\F_{1,t}^{\prime}\frac{\R_{1}^{\prime}\hat{\R}_1^{\dagger}}{%
p_{1}}\left( \Lambda _{R_{1}}^{\dagger}\right) ^{-1}.
\end{equation*}%
Lemma \ref{spec-eig-mRx} entails that $\left\Vert \left( \Lambda
_{R_{1}}^{\dagger}\right) ^{-1}\right\Vert =O_{P}(1) $; hence%
\begin{align*}
\left\Vert \bH_{R_{1}}^{\dagger}\right\Vert _{F} &\leq\left\Vert \frac{%
1}{T^{2}}\sum_{t=1}^{T}\F_{1,t}\left( \frac{\C_{1}^{\prime}%
\hat{\C}_{1}}{p_{2}}\right) \left( \frac{\C_{1}^{\prime}%
\hat{\C}_{1}}{p_{2}}\right) ^{\prime}\F_{1,t}^{\prime}\right\Vert _{F}\left\Vert \frac{\R_{1}^{\prime}\hat{\R}_1%
^{\dagger}}{p_{1}}\right\Vert _{F}\left\Vert \left( \Lambda _{R_{1}}^{\dagger}\right) ^{-1}\right\Vert _{F} \\
&\leq\frac{1}{T^{2}}\sum_{t=1}^{T}\left\Vert \F_{1,t}\right\Vert
_{F}^{2}\left\Vert \frac{\C_{1}^{\prime}\hat{\C}_{1}}{p_{2}}%
\right\Vert _{F}^{2}\left\Vert \frac{\R_{1}^{\prime}\hat{\R}_1%
^{\dagger}}{p_{1}}\right\Vert _{F}\left\Vert \left( \Lambda _{R_{1}}^{\dagger}\right) ^{-1}\right\Vert _{F}=O_{P}(1) .
\end{align*}%
Further, recalling that $h_{R_{1}}=h_{C_{1}}=h_{R_{0}}=h_{C_{0}}=1$%
\begin{align*}
\left\Vert II\right\Vert _{F} &\leq\frac{1}{p_{1}p_{2}^{2}T^{2}}\left\Vert 
\R_{0}\right\Vert _{F}\left\Vert \R_{0}\right\Vert
_{F}\left\Vert \hat{\R}_1^{\dagger}\right\Vert _{F}\left\Vert 
\hat{\C}_{1}\right\Vert _{F}^{2}\left\Vert \C_{0}\right\Vert
_{F}^{2}\left( \sum_{t=1}^{T}\F_{0,t}^{2}\right) \left\Vert \left(
\Lambda _{R_{1}}^{\dagger}\right) ^{-1}\right\Vert _{F} \\
&=\frac{1}{p_{1}p_{2}^{2}T^{2}}p_{1}^{3/2}p_{2}^{2}O_{P}\left( T\right)
=O_{P}\left( \frac{p_{1}^{1/2}}{T}\right) .
\end{align*}%
Moreover, 
\begin{align*}
III &=\frac{1}{p_{1}p_{2}^{2}T^{2}}\sum_{t=1}^{T}\E_{t}\C_{1}\bH%
_{C_{1}}\bH_{C_{1}}^{\prime}\C_{1}^{\prime}\E_{t}^{\prime}%
\hat{\R}_1^{\dagger}\left( \Lambda _{R_{1}}^{\dagger}\right) ^{-1} \\
&+\frac{1}{p_{1}p_{2}^{2}T^{2}}\sum_{t=1}^{T}\E_{t}\left( \hat{%
\C}_{1}-\C_{1}\bH_{C_{1}}\right) \bH_{C_{1}}^{\prime}\C_{1}^{\prime}\E_{t}^{\prime}\hat{\R}_1^{\dagger}\left(
\Lambda _{R_{1}}^{\dagger}\right) ^{-1} \\
&+\frac{1}{p_{1}p_{2}^{2}T^{2}}\sum_{t=1}^{T}\E_{t}\C_{1}\bH%
_{C_{1}}\left( \hat{\C}_{1}-\C_{1}\bH_{C_{1}}\right) ^{\prime}\E%
_{t}^{\prime}\hat{\R}_1^{\dagger}\left( \Lambda _{R_{1}}^{\dagger}\right) ^{-1} \\
&+\frac{1}{p_{1}p_{2}^{2}T^{2}}\sum_{t=1}^{T}\E_{t}\left( \hat{%
\C}_{1}-\C_{1}\bH_{C_{1}}\right) \left( \hat{\C}_{1}-\C_{1}\bH%
_{C_{1}}\right) ^{\prime}\E_{t}^{\prime}\hat{\R}_1^{\dagger}\left( \Lambda _{R_{1}}^{\dagger}\right) ^{-1} \\
&=III_{a}+III_{b}+III_{b}^{\prime}+III_{c}.
\end{align*}%
Under $h_{R_{1}}=h_{C_{1}}=h_{R_{0}}=h_{C_{0}}=1$, $\bH_{C_{1}}$ is a random sign under
our identification restrictions, so we will omit it; it holds that%
\begin{align*}
\left\Vert III_{a}\right\Vert _{F} &\leq\lambda _{\max }\left( \frac{1}{%
p_{1}p_{2}^{2}T^{2}}\sum_{t=1}^{T}\E_{t}\C_{1}\C_{1}^{\prime}\E_{t}^{\prime}\right) \left\Vert \hat{\R}_1^{\dagger}\right\Vert _{F}\left\Vert \left( \Lambda _{R_{1}}^{\dagger}\right)
^{-1}\right\Vert _{F} \\
&\leq p_{1}^{1/2}\lambda _{\max }\left( \frac{1}{p_{1}p_{2}^{2}T^{2}}%
\sum_{t=1}^{T}\E_{t}\C_{1}\C_{1}^{\prime}\E_{t}^{\prime}\right) .
\end{align*}%
Also%
\begin{align*}
&\lambda _{\max }\left( \frac{1}{p_{1}p_{2}^{2}T^{2}}\sum_{t=1}^{T}\E_{t}\C_{1}\C_{1}^{\prime}\E_{t}^{\prime}\right) \leq \lambda
_{\max }\left( \frac{1}{p_{1}p_{2}^{2}T^{2}}\sum_{t=1}^{T}E\left( \E%
_{t}\C_{1}\C_{1}^{\prime}\E_{t}^{\prime}\right) \right) \\
&+\lambda _{\max }\left( \frac{1}{p_{1}p_{2}^{2}T^{2}}\sum_{t=1}^{T}\left( 
\E_{t}\C_{1}\C_{1}^{\prime}\E_{t}^{\prime}-E\left( \E_{t}\C_{1}\C_{1}^{\prime}\E_{t}^{\prime}\right) \right) \right) .
\end{align*}%
It holds that%
\begin{align*}
&\lambda _{\max }\left( \frac{1}{p_{1}p_{2}^{2}T^{2}}\sum_{t=1}^{T}E\left( 
\E_{t}\C_{1}\C_{1}^{\prime}\E_{t}^{\prime}\right) \right) \\
&\leq\frac{1}{p_{1}p_{2}^{2}T}\max_{1\leq h\leq
p_{1}}\sum_{k=1}^{p_{1}}\left\vert E\left(
\sum_{j=1}^{p_{2}}c_{j}e_{hj,t}\right) \left(
\sum_{j=1}^{p_{2}}c_{j}e_{kj,t}\right) \right\vert \\
&\leq c_{0}\frac{1}{p_{1}p_{2}^{2}T}\max_{1\leq h\leq
p_{1}}\sum_{k=1}^{p_{1}}\sum_{h,j=1}^{p_{2}}\left\vert E\left(
e_{hj,t}e_{kj,t}\right) \right\vert \leq c_{1}\frac{1}{p_{2}T};
\end{align*}%
also%
\begin{align*}
&\left\Vert \frac{1}{p_{1}p_{2}^{2}T^{2}}\sum_{t=1}^{T}\left( \E_{t}%
\C_{1}\C_{1}^{\prime}\E_{t}^{\prime}-E\left( \E_{t}\C_{1}\C_{1}^{\prime}\E_{t}^{\prime}\right) \right) \right\Vert _{F} \\
&=\frac{1}{p_{1}p_{2}^{2}T^{2}}\left( \sum_{i,j=1}^{p_{1}}\left(
\sum_{t=1}^{T}\left( \sum_{h=1}^{p_{2}}c_{h}\left( e_{ih,t}-E\left(
e_{ih,t}\right) \right) \right) \left( \sum_{k=1}^{p_{2}}c_{k}\left(
e_{ik,t}-E\left( e_{ik,t}\right) \right) \right) \right) ^{2}\right) ^{1/2},
\end{align*}%
and%
\begin{align*}
&E\sum_{i,j=1}^{p_{1}}\left( \sum_{t=1}^{T}\left(
\sum_{h=1}^{p_{2}}c_{h}\left( e_{ih,t}-E\left( e_{ih,t}\right) \right)
\right) \left( \sum_{k=1}^{p_{2}}c_{k}\left( e_{ik,t}-E\left(
e_{ik,t}\right) \right) \right) \right) ^{2} \\
&\leq\left( \max_{1\leq h\leq p_{2}}\left\vert c_{h}\right\vert \right)
^{4}\sum_{i,j=1}^{p_{1}}\sum_{h_{1},h_{2},h_{3},h_{4}=1}^{p_{2}}%
\sum_{t,s=1}^{T}\left\vert \cov\left(
e_{ih_{1},t}e_{jh_{2},t},e_{ih_{3},s}e_{jh_{4},s}\right) \right\vert \leq
c_{0}p_{1}^{2}p_{2}^{3}T,
\end{align*}%
whence%
\begin{equation*}
\left\Vert \frac{1}{p_{1}p_{2}^{2}T^{2}}\sum_{t=1}^{T}\left( \E_{t}%
\C_{1}\C_{1}^{\prime}\E_{t}^{\prime}-E\left( \E_{t}\C_{1}\C_{1}^{\prime}\E_{t}^{\prime}\right) \right) \right\Vert
_{F}=O_{P}\left( \frac{1}{p_{2}^{1/2}T^{3/2}}\right) .
\end{equation*}%
Hence%
\begin{equation*}
\lambda _{\max }\left( \frac{1}{p_{1}p_{2}^{2}T^{2}}\sum_{t=1}^{T}\E%
_{t}\C_{1}\C_{1}^{\prime}\E_{t}^{\prime}\right) =O\left( \frac{1}{%
p_{2}T}\right) +O_{P}\left( \frac{1}{p_{2}^{1/2}T^{3/2}}\right) ,
\end{equation*}%
which in turn entails that%
\begin{equation*}
\left\Vert III_{a}\right\Vert _{F}=O\left( \frac{p_{1}^{1/2}}{p_{2}T}\right)
+O_{P}\left( \frac{p_{1}^{1/2}}{p_{2}^{1/2}T^{3/2}}\right) .
\end{equation*}%
Following the proof of Lemma C.5 in \citet{he2023one}, it can be shown that $%
III_{b}$ and $III_{c}$ are both dominated by $III_{a}$. We also have (recall
that we are assuming $h_{R_{1}}=h_{C_{1}}=1$, and that therefore $\bH_{C_{1}}$ is
a random sign)%
\begin{align*}
IV &=\frac{1}{p_{1}p_{2}^{2}T^{2}}\R_1\C_{1}^{\prime}\hat{\C}_{1}%
\sum_{t=1}^{T}\F_{1,t}\left( \C_{1}\bH_{C_{1}}\right) ^{\prime}\E_{t}^{\prime}\hat{\R}_1^{\dagger}\left( \Lambda _{R_{1}}^{\dagger}\right) ^{-1} \\
&+\frac{1}{p_{1}p_{2}^{2}T^{2}}\R_1\C_{1}^{\prime}\hat{\C}_{1}%
\sum_{t=1}^{T}\F_{1,t}\left( \hat{\C}_{1}-\C_{1}\bH%
_{C_{1}}\right) ^{\prime}\E_{t}^{\prime}\hat{\R}_1^{\dagger}\left( \Lambda _{R_{1}}^{\dagger}\right) ^{-1} \\
&=IV_{a}+IV_{b}.
\end{align*}%
Consider now%
\begin{equation*}
\left\Vert \sum_{t=1}^{T}\F_{1,t}\C_{1}^{\prime}\E%
_{t}^{\prime}\right\Vert _{F}=\left( \sum_{i=1}^{p_{1}}\left\vert
\sum_{h=1}^{p_{2}}\sum_{t=1}^{T}\F_{1,t}c_{h}e_{ih,t}\right\vert
^{2}\right) ^{1/2},
\end{equation*}%
with%
\begin{align*}
&E\sum_{i=1}^{p_{1}}\left\vert \sum_{h=1}^{p_{2}}\sum_{t=1}^{T}\F_{1,t}c_{h}e_{ih,t}\right\vert ^{2} \\
&=E\sum_{i=1}^{p_{1}}\sum_{h,k=1}^{p_{2}}\sum_{t,s=1}^{T}c_{h}c_{k}\F_{1,t}\F_{1,s}e_{ih,t}e_{ik,s}=\sum_{i=1}^{p_{1}}\sum_{h,k=1}^{p_{2}}%
\sum_{t,s=1}^{T}c_{h}c_{k}E\left( \F_{1,t}\F_{1,s}\right)
E\left( e_{ih,t}e_{ik,s}\right) \\
&\leq\max_{1\leq h\leq
p_{2}}c_{h}^{2}\sum_{i=1}^{p_{1}}\sum_{h,k=1}^{p_{2}}\sum_{t,s=1}^{T}\left%
\vert E\left( \F_{1,t}\F_{1,s}\right) \right\vert \left\vert
E\left( e_{ih,t}e_{ik,s}\right) \right\vert \\
&\leq c_{0}\sum_{i=1}^{p_{1}}\sum_{h,k=1}^{p_{2}}\sum_{t,s=1}^{T}\left\vert
E\left( \F_{1,t}^{2}\right) E\left( \F_{1,s}^{2}\right)
\right\vert ^{1/2}\left\vert E\left( e_{ih,t}e_{ik,s}\right) \right\vert \\
&\leq c_{1}T\sum_{i=1}^{p_{1}}\sum_{h,k=1}^{p_{2}}\sum_{t,s=1}^{T}\left\vert
E\left( e_{ih,t}e_{ik,s}\right) \right\vert \leq c_{2}p_{1}p_{2}T^{2},
\end{align*}%
whence%
\begin{equation}
\left\Vert \sum_{t=1}^{T}\F_{1,t}\C_{1}^{\prime}\E%
_{t}^{\prime}\right\Vert _{F}=O_{P}\left( p_{1}^{1/2}p_{2}^{1/2}T\right) .
\label{frob-new}
\end{equation}%
Hence it follows that 
\begin{equation*}
\left\Vert IV_{a}\right\Vert _{F}\leq \frac{1}{p_{1}p_{2}^{2}T^{2}}%
\left\Vert \R_1\right\Vert _{F}\left\Vert \hat{\R}_1%
^{\dagger}\right\Vert _{F}\left\Vert \C_{1}\right\Vert _{F}\left\Vert 
\hat{\C}_{1}\right\Vert _{F}\left\Vert \left( \Lambda _{R_{1}}^{\dagger}\right) ^{-1}\right\Vert _{F}\left\Vert \sum_{t=1}^{T}\F_{1,t}\C_{1}^{\prime}\E_{t}^{\prime}\right\Vert _{F}=O_{P}\left( \frac{%
p_{1}^{1/2}}{p_{2}^{1/2}T}\right) ;
\end{equation*}%
also%
\begin{align*}
&\left\Vert IV_{b}\right\Vert _{F} \\
&\leq\frac{1}{p_{1}p_{2}^{2}T^{2}}\left\Vert \R_1\right\Vert
_{F}\left\Vert \hat{\R}_1^{\dagger}\right\Vert _{F}\left\Vert 
\C_{1}\right\Vert _{F}\left\Vert \hat{\C}_{1}\right\Vert
_{F}\left\Vert \left( \Lambda _{R_{1}}^{\dagger}\right) ^{-1}\right\Vert
_{F}\left\Vert \hat{\C}_{1}-\C_{1}\bH_{C_{1}}\right\Vert
_{F}\left\Vert \sum_{t=1}^{T}\F_{1,t}\E_{t}^{\prime}\right\Vert _{F}=O_{P}\left( \frac{p_{1}^{1/2}}{T^{2}}\right) ;
\end{align*}%
thus%
\begin{equation*}
\left\Vert IV\right\Vert _{F}=O_{P}\left( \frac{p_{1}^{1/2}}{T^{2}}\right) .
\end{equation*}%
Similarly, it is not hard to see that%
\begin{align*}
\left\Vert V\right\Vert _{F} &\leq\frac{1}{p_{1}p_{2}^{2}T^{2}}\left\Vert 
\R_1\right\Vert _{F}\left\Vert \R_{0}\right\Vert
_{F}\left\Vert \hat{\R}_1^{\dagger}\right\Vert _{F}\left\Vert 
\C_{1}\right\Vert _{F}\left\Vert \C_{0}\right\Vert
_{F}\left\Vert \hat{\C}_{1}\right\Vert _{F}^{2}\left\vert
\sum_{t=1}^{T}\F_{1,t}\F_{0,t}\right\vert \left\Vert \left(
\Lambda _{R_{1}}^{\dagger}\right) ^{-1}\right\Vert _{F} \\
&=O_{P}(1) \frac{1}{p_{1}p_{2}^{2}T^{2}}%
p_{1}^{3/2}p_{2}^{2}T=O_{P}\left( \frac{p_{1}^{1/2}}{T^{2}}\right) ,
\end{align*}%
and 
\begin{align*}
\left\Vert VI\right\Vert _{F} &\leq\frac{1}{p_{1}p_{2}^{2}T^{2}}\left\Vert 
\R_{0}\right\Vert _{F}\left\Vert \hat{\R}_1^{\dagger}\right\Vert _{F}\left\Vert \C_{0}\right\Vert _{F}\left\Vert 
\hat{\C}_{1}\right\Vert _{F}\left\Vert \sum_{t=1}^{T}\F_{0,t}%
\hat{\C}_{1}^{\prime}\E_{t}^{\prime}\right\Vert
_{F}\left\Vert \left( \Lambda _{R_{1}}^{\dagger}\right) ^{-1}\right\Vert _{F} \\
&=O_{P}\left( \left( p_{1}p_{2}T\right) ^{1/2}+\frac{\left( p_{1}T\right)
^{1/2}p_{2}}{T}\right) \frac{1}{p_{1}p_{2}^{2}T^{2}}p_{1}p_{2} \\
&=O_{P}\left( \frac{p_{1}^{1/2}}{p_{2}^{1/2}T^{3/2}}\right) +O_{P}\left( 
\frac{p_{1}^{1/2}}{T^{5/2}}\right) .
\end{align*}%
The desired result now follows.

The proof of the other result follows from the same arguments as that of
Lemma \ref{eig-mRx}, upon\ defining%
\begin{equation*}
\bH_{C_{1}}^{\dagger}=\frac{1}{T^{2}}\sum_{t=1}^{T}\F_{1,t}^{\prime}\left( \frac{\R_{1}^{\prime}\hat{\R}_1}{p_{1}}%
\right) \left( \frac{\R_{1}^{\prime}\hat{\R}_1}{p_{1}}\right)
^{\prime}\F_{1,t}\frac{\C_{1}^{\prime}\hat{\C}_{1}%
^{\dagger}}{p_{2}}\left( \Lambda _{C_{1}}^{\dagger}\right) ^{-1}.
\end{equation*}
\end{proof}

\begin{proof}[Proof of Theorem \protect\ref{rc1-hat}]
Some arguments are similar to the proof of Lemma \ref{eig-orth-r}, and we
therefore omit them when possible. We begin by showing (\ref{r1-hat}); by
definition, it holds that%
\begin{equation*}
\hat{\R}_{0}=\M_{X}^{R_{1},\perp}\hat{\R}%
_{0}\Lambda _{R_{0}}^{-1},
\end{equation*}%
whence%
\begin{eqnarray}
&&\hat{\R}_{0}  \label{r1-hat-dec} \\
&=&\frac{1}{p_{1}p_{2}^{2}T}\sum_{t=1}^{T}\R_{1}\F_{1,t}%
\C_{1}^{\prime}\hat{\C}_{1,\perp}\hat{\C}%
_{1,\perp}^{\prime}\C_{1}\F_{1,t}^{\prime}\R%
_{1}^{\prime}\hat{\R}_{0}\Lambda _{R_{0}}^{-1}  \notag \\
&&+\frac{1}{p_{1}p_{2}^{2}T}\sum_{t=1}^{T}\R_{0}\F_{0,t}%
\C_{0}^{\prime}\hat{\C}_{1,\perp}\hat{\C}%
_{1,\perp}^{\prime}\C_{0}\F_{0,t}^{\prime}\R%
_{0}^{\prime}\hat{\R}_{0}\Lambda _{R_{0}}^{-1}  \notag \\
&&+\frac{1}{p_{1}p_{2}^{2}T}\sum_{t=1}^{T}\E_{t}\hat{\C}%
_{1,\perp}\hat{\C}_{1,\perp}^{\prime}\E_{t}^{\prime}%
\hat{\R}_{0}\Lambda _{R_{0}}^{-1}+\frac{1}{p_{1}p_{2}^{2}T}%
\sum_{t=1}^{T}\R_{1}\F_{1,t}\C_{1}^{\prime}\hat{%
\C}_{1,\perp}\hat{\C}_{1,\perp}^{\prime}\E%
_{t}^{\prime}\hat{\R}_{0}\Lambda _{R_{0}}^{-1}  \notag \\
&&+\left( \frac{1}{p_{1}p_{2}^{2}T}\sum_{t=1}^{T}\R_{1}\F%
_{1,t}\C_{1}^{\prime}\hat{\C}_{1,\perp}\hat{%
\C}_{1,\perp}^{\prime}\E_{t}^{\prime}\right) ^{\prime}%
\hat{\R}_{0}\Lambda _{R_{0}}^{-1}+\frac{1}{p_{1}p_{2}^{2}T}%
\sum_{t=1}^{T}\R_{1}\F_{1,t}\C_{1}^{\prime}\hat{%
\C}_{1,\perp}\hat{\C}_{1,\perp}^{\prime}\C_{0}%
\F_{0,t}^{\prime}\R_{0}^{\prime}\hat{\R}%
_{0}\Lambda _{R_{0}}^{-1}  \notag \\
&&+\left( \frac{1}{p_{1}p_{2}^{2}T}\sum_{t=1}^{T}\R_{1}\F%
_{1,t}\C_{1}^{\prime}\hat{\C}_{1,\perp}\hat{%
\C}_{1,\perp}^{\prime}\C_{0}\F_{0,t}^{\prime}%
\R_{0}^{\prime}\right) ^{\prime}\hat{\R}_{0}\Lambda
_{R_{0}}^{-1}  \notag \\
&&+\frac{1}{p_{1}p_{2}^{2}T}\sum_{t=1}^{T}\R_{0}\F_{0,t}%
\C_{0}^{\prime}\hat{\C}_{1,\perp}\hat{\C}%
_{\perp}^{\prime}\E_{t}^{\prime}\hat{\R}_{0}\Lambda
_{R_{0}}^{-1}+\left( \frac{1}{p_{1}p_{2}^{2}T}\sum_{t=1}^{T}\R_{0}%
\F_{0,t}\C_{0}^{\prime}\hat{\C}_{\perp}%
\hat{\C}_{1,\perp}^{\prime}\E_{t}^{\prime}\right)
^{\prime}\hat{\R}_{0}\Lambda _{R_{0}}^{-1}  \notag \\
&=&I+II+III+IV+IV^{\prime}+V+V^{\prime}+VI+VI^{\prime}.  \notag
\end{eqnarray}%
We begin by noting that, by Lemma \ref{eig-orth-r}, $\left\Vert \Lambda
_{R_{0}}^{-1}\right\Vert _{R_{1}}=O_{P}(1) $. Consider $II$ first%
\begin{eqnarray*}
II &=&\frac{1}{p_{1}p_{2}^{2}T}\sum_{t=1}^{T}\R_{0}\F_{0,t}%
\C_{0}^{\prime}\hat{\C}_{1,\perp}\C_{0}\F_{0,t}^{\prime}\R_{0}^{\prime}\hat{\R}_{0}\Lambda
_{R_{0}}^{-1} \\
&=&\R_{0}\bH_{R_{0}}+\frac{1}{p_{1}p_{2}^{2}T}\sum_{t=1}^{T}%
\R_{0}\F_{0,t}\C_{0}^{\prime}\left( \hat{%
\C}_{1,\perp}-\C_{1,\perp}\right) ^{\prime}\C_{0}%
\F_{0,t}^{\prime}\R_{0}^{\prime}\hat{\R}%
_{0}\Lambda _{R_{0}}^{-1} \\
&=&II_{a}+II_{b},
\end{eqnarray*}%
where we have defined%
\begin{equation*}
\bH_{R_{0}}=\left( \frac{1}{T}\sum_{t=1}^{T}\F_{0,t}\frac{%
\C_{0}^{\prime}\C_{1,\perp}\C_{0}}{p_{2}^{2}}%
\F_{0,t}^{\prime}\right) \frac{\R_{0}^{\prime}\hat{%
\R}_{0}}{p_{1}}\Lambda _{R_{0}}^{-1}
\end{equation*}%
By similar arguments as in the above, it is easy to see that%
\begin{equation*}
\left\Vert \bH_{R_{0}}\right\Vert _{F}\leq \left( \frac{1}{T}%
\sum_{t=1}^{T}\left\Vert \F_{0,t}\right\Vert _{F}^{2}\right) \left( 
\frac{\left\Vert \C_{0}^{\prime}\C_{1,\perp}\C%
_{0}\right\Vert _{F}}{p_{2}}\right) ^{2}\frac{\left\Vert \R%
_{0}\right\Vert _{F}\left\Vert \hat{\R}_{0}\right\Vert _{F}}{%
p_{1}}\left\Vert \Lambda _{R_{0}}^{-1}\right\Vert _{F}=O_{P}(1) .
\end{equation*}%
Further%
\begin{eqnarray*}
\left\Vert II_{b}\right\Vert _{F} &\leq &\frac{1}{p_{1}p_{2}^{2}T}\left\Vert 
\R_{0}\right\Vert _{F}^{2}\left\Vert \hat{\R}%
_{0}\right\Vert _{F}\sum_{t=1}^{T}\left\Vert \F_{0,t}\right\Vert
_{F}^{2}\left\Vert \C_{0}\right\Vert _{F}^{2}\left\Vert \hat{%
\C}_{1,\perp}-\C_{1,\perp}\right\Vert _{F}\left\Vert
\Lambda _{R_{0}}^{-1}\right\Vert _{F} \\
&=&O_{P}(1) \frac{1}{p_{1}p_{2}^{2}T}p_{1}p_{1}^{1/2}Tp_{2}\frac{%
1}{T}=O_{P}\left( \frac{p_{1}^{1/2}}{p_{2}T}\right) ,
\end{eqnarray*}%
so that ultimately%
\begin{equation*}
II=\R_{0}\bH_{R_{0}}+O_{P}\left( \frac{p_{1}^{1/2}}{p_{2}T}%
\right) .
\end{equation*}%
We also note that Assumption \ref{as-4}\textit{(ii)}(b) $\C%
_{0}^{\prime}\C_{1,\perp}\neq 0$, together with Assumption \ref%
{as-2}\textit{(ii)} entails that $\bH_{R_{0}}$ has full rank $%
h_{R_{1}}$. We now consider the other terms, starting from%
\begin{equation*}
I=\frac{1}{p_{1}p_{2}^{2}T}\sum_{t=1}^{T}\R_{1}\F_{1,t}%
\C_{1}^{\prime}\left( \hat{\C}_{1,\perp}-\C%
_{1,\perp}\right) \left( \hat{\C}_{1,\perp}-\C%
_{1,\perp}\right) ^{\prime}\C_{1}\F_{1,t}^{\prime}\R_{1}^{\prime}\hat{\R}_{0}\Lambda _{R_{0}}^{-1},
\end{equation*}%
having used (\ref{orth-prod-c}) in the first passage, whence%
\begin{eqnarray*}
\left\Vert I\right\Vert _{F} &\leq &\left\Vert \R_{1}\right\Vert
_{F}\left( \frac{1}{T}\sum_{t=1}^{T}\left\Vert \F_{1,t}\right\Vert
_{F}^{2}\right) \frac{\left\Vert \C_{1}\right\Vert
_{F}^{2}\left\Vert \hat{\C}_{1,\perp}-\C_{1,\perp
}\right\Vert _{F}^{2}}{p_{2}^{2}}\frac{\left\Vert \R_{0}\right\Vert
_{F}\left\Vert \hat{\R}_{0}\right\Vert _{F}}{p_{1}}\left\Vert
\Lambda _{R_{0}}^{-1}\right\Vert _{F} \\
&=&p_{1}^{1/2}O_{P}\left( \frac{1}{p_{2}T}\right) =O_{P}\left( \frac{%
p_{1}^{1/2}}{p_{2}T}\right) .
\end{eqnarray*}%
Similarly, using (\ref{3F-error})%
\begin{eqnarray}
\left\Vert III\right\Vert _{F} &=&\left\Vert \frac{1}{p_{1}p_{2}^{2}T}%
\sum_{t=1}^{T}\E_{t}\hat{\C}_{1,\perp}\E%
_{t}^{\prime}\hat{\R}_{0}\Lambda _{R_{0}}^{-1}\right\Vert _{F}
\label{3f1} \\
&=&\left( O_{P}\left( \frac{1}{p_{1}p_{2}}\right) +O_{P}\left( \frac{1}{%
p_{2}^{1/2}T^{1/2}}\right) \right) p_{1}^{1/2}  \notag \\
&=&O_{P}\left( \frac{p_{1}^{1/2}}{p_{1}p_{2}}\right) +O_{P}\left( \frac{%
p_{1}^{1/2}}{p_{2}^{1/2}T^{1/2}}\right) .  \notag
\end{eqnarray}%
We now study (explicitly considering the case $%
h_{R_{1}}=h_{C_{1}}=h_{R_{1}}=h_{C_{1}}=1$ for simplicity)%
\begin{eqnarray*}
IV &=&\frac{1}{p_{1}p_{2}^{2}T}\sum_{t=1}^{T}\R_{1}\F_{1,t}%
\C_{1}^{\prime}\left( \hat{\C}_{1,\perp}-\C%
_{1,\perp}\right) \hat{\C}_{1,\perp}^{\prime}\E%
_{t}^{\prime}\hat{\R}_{0}\Lambda _{R_{0}}^{-1} \\
&=&\frac{1}{p_{1}p_{2}^{2}T}\sum_{t=1}^{T}\R_{1}\F_{1,t}%
\C_{1}^{\prime}\left( \hat{\C}_{1,\perp}-\C%
_{1,\perp}\right) \hat{\C}_{1,\perp}^{\prime}\E%
_{t}^{\prime}\hat{\R}_{0}\Lambda _{R_{0}}^{-1} \\
&&+\frac{1}{p_{1}p_{2}^{2}T}\sum_{t=1}^{T}\R_{1}\F_{1,t}%
\C_{1}^{\prime}\left( \hat{\C}_{1,\perp}-\C%
_{1,\perp}\right) \left( \hat{\C}_{1,\perp}-\C%
_{1,\perp}\right) ^{\prime}\E_{t}^{\prime}\hat{\R}%
_{0}\Lambda _{R_{0}}^{-1} \\
&=&IV_{a}+IV_{b},
\end{eqnarray*}%
again having used (\ref{orth-prod-c}). Hence, using the fact that (as can be
verified with a similar logic as above)%
\begin{equation*}
\left\Vert \sum_{t=1}^{T}\F_{1,t}\C_{1,\perp}^{\prime}%
\E_{t}^{\prime}\right\Vert _{F}=O_{P}\left(
p_{1}^{1/2}p_{2}T\right) ,
\end{equation*}%
it holds that%
\begin{eqnarray*}
\left\Vert IV_{a}\right\Vert _{F} &\leq &\frac{1}{T}\frac{\left\Vert \R_{1}\right\Vert _{F}\left\Vert \hat{\R}_{0}\right\Vert _{F}}{%
p_{1}}\left\Vert \sum_{t=1}^{T}\F_{1,t}\C_{1,\perp}^{\prime
}\E_{t}^{\prime}\right\Vert _{F}\frac{\left\Vert \C%
_{1}\right\Vert _{F}\left\Vert \hat{\C}_{1,\perp}-\C%
_{1,\perp}\right\Vert _{F}}{p_{2}^{2}}\left\Vert \Lambda
_{R_{0}}^{-1}\right\Vert _{F} \\
&=&O_{P}(1) \frac{1}{T}p_{1}^{1/2}p_{2}T\frac{1}{p_{2}^{3/2}T}%
=O_{P}\left( \frac{p_{1}^{1/2}}{p_{2}^{1/2}T}\right) ,
\end{eqnarray*}%
and%
\begin{eqnarray*}
\left\Vert IV_{b}\right\Vert _{F} &\leq &\frac{1}{T}\frac{\left\Vert \R_{1}\right\Vert _{F}\left\Vert \hat{\R}_{0}\right\Vert _{F}}{%
p_{1}}\left\Vert \sum_{t=1}^{T}\F_{1,t}\E_{t}^{\prime
}\right\Vert _{F}\frac{\left\Vert \C_{1}\right\Vert _{F}\left\Vert 
\hat{\C}_{1,\perp}-\C_{1,\perp}\right\Vert _{F}^{2}}{%
p_{2}^{2}}\left\Vert \Lambda _{R_{0}}^{-1}\right\Vert _{F} \\
&=&O_{P}(1) \frac{1}{T}p_{1}^{1/2}p_{2}^{1/2}T\frac{1}{%
p_{2}^{3/2}T^{2}}=O_{P}\left( \frac{p_{1}^{1/2}}{p_{2}T^{2}}\right) ,
\end{eqnarray*}%
by Lemma \ref{errorfactor}, whence%
\begin{equation*}
\left\Vert IV\right\Vert _{F}=O_{P}\left( \frac{p_{1}^{1/2}}{p_{2}^{1/2}T}%
\right) .
\end{equation*}%
By the same token 
\begin{equation*}
V=\frac{1}{p_{1}p_{2}^{2}T}\sum_{t=1}^{T}\R_{1}\F_{1,t}%
\C_{1}^{\prime}\left( \hat{\C}_{1,\perp}-\C%
_{1,\perp}\right) \C_{0}\F_{0,t}^{\prime}\R%
_{0}^{\prime}\hat{\R}_{0}\Lambda _{R_{0}}^{-1},
\end{equation*}%
whence%
\begin{eqnarray*}
\left\Vert V\right\Vert _{F} &\leq &\left\Vert \R_{1}\right\Vert
_{F}\left\Vert \frac{1}{T}\sum_{t=1}^{T}\F_{1,t}\F%
_{0,t}\right\Vert _{F}\frac{\left\Vert \C_{1}\right\Vert
_{F}\left\Vert \C_{0}\right\Vert _{F}\left\Vert \hat{\C}%
_{1,\perp}-\C_{1,\perp}\right\Vert _{F}}{p_{2}^{2}}\frac{%
\left\Vert \R_{0}\right\Vert _{F}\left\Vert \hat{\R}%
_{0}\right\Vert _{F}}{p_{1}}\left\Vert \Lambda _{R_{0}}^{-1}\right\Vert _{F}
\\
&=&O_{P}(1) p_{1}^{1/2}\frac{1}{p_{2}T}=O_{P}\left( \frac{%
p_{1}^{1/2}}{p_{2}T}\right) .
\end{eqnarray*}%
Also%
\begin{eqnarray*}
VI &=&\frac{1}{p_{1}p_{2}^{2}T}\sum_{t=1}^{T}\R_{0}\F_{0,t}%
\C_{0}^{\prime}\hat{\C}_{1,\perp}\E%
_{t}^{\prime}\hat{\R}_{0}\Lambda _{R_{0}}^{-1} \\
&=&\frac{1}{p_{1}p_{2}^{2}T}\sum_{t=1}^{T}\R_{0}\F_{0,t}%
\C_{0}^{\prime}\C_{1,\perp}\E_{t}^{\prime}%
\hat{\R}_{0}\Lambda _{R_{0}}^{-1} \\
&&+\frac{1}{p_{1}p_{2}^{2}T}\sum_{t=1}^{T}\R_{0}\F_{0,t}%
\C_{0}^{\prime}\left( \hat{\C}_{1,\perp}-\C%
_{1,\perp}\right) \E_{t}^{\prime}\hat{\R}_{0}\Lambda
_{R_{0}}^{-1} \\
&=&VI_{a}+VI_{b}.
\end{eqnarray*}%
Using the fact that%
\begin{eqnarray*}
&&E\left\Vert \sum_{t=1}^{T}\F_{0,t}\C_{1,\perp}\E%
_{t}^{\prime}\right\Vert _{F}^{2} \\
&=&E\sum_{i=1}^{p_{1}}\sum_{j=1}^{p_{2}}\left(
\sum_{t=1}^{T}\sum_{h=1}^{p_{2}}c_{\perp ,jh}e_{ih,t}\F_{0,t}\right)
^{2}\leq
c_{0}\sum_{i=1}^{p_{1}}\sum_{j=1}^{p_{2}}\sum_{t=1}^{T}%
\sum_{h_{1},h_{2}=1}^{p_{2}}\left\vert E\left(
e_{ih_{1},t}e_{ih_{2},s}\right) \right\vert \left\vert E\left( \F%
_{0,t}\F_{0,s}\right) \right\vert  \\
&\leq
&c_{0}\sum_{i=1}^{p_{1}}\sum_{j=1}^{p_{2}}\sum_{t=1}^{T}%
\sum_{h_{1},h_{2}=1}^{p_{2}}\left\vert E\left(
e_{ih_{1},t}e_{ih_{2},s}\right) \right\vert \leq O\left(
p_{1}p_{2}^{2}T\right) ,
\end{eqnarray*}%
we have%
\begin{eqnarray*}
\left\Vert VI_{a}\right\Vert _{F} &\leq &\frac{1}{p_{1}p_{2}^{2}T}\left\Vert 
\hat{\R}_{0}\right\Vert _{F}\left\Vert \R_{0}\right\Vert
_{F}\left\Vert \C_{0}\right\Vert _{F}\left\Vert \sum_{t=1}^{T}%
\F_{0,t}\C_{1,\perp}\E_{t}^{\prime}\right\Vert
_{F}\left\Vert \Lambda _{R_{0}}^{-1}\right\Vert _{F} \\
&=&O_{P}(1) \frac{1}{p_{1}p_{2}^{2}T}%
p_{1}^{1/2}p_{1}^{1/2}p_{2}^{1/2}p_{1}^{1/2}p_{2}T^{1/2}=O_{P}\left( \frac{1%
}{p_{2}^{1/2}T^{1/2}}\right) ;
\end{eqnarray*}%
also%
\begin{eqnarray*}
\left\Vert VI_{b}\right\Vert _{F} &\leq &\frac{1}{p_{1}p_{2}^{2}T}\left\Vert 
\R_{0}\right\Vert _{F}\left\Vert \C_{0}\right\Vert
_{F}\left\Vert \hat{\C}_{1,\perp}-\C_{1,\perp
}\right\Vert _{F}\left\Vert \sum_{t=1}^{T}\F_{0,t}\E%
_{t}^{\prime}\right\Vert _{F}\left\Vert \hat{\R}_{0}\right\Vert
_{F}\left\Vert \Lambda _{R_{0}}^{-1}\right\Vert _{F} \\
&=&O_{P}(1) \frac{1}{p_{1}p_{2}^{2}T}p_{1}^{1/2}p_{2}^{1/2}\frac{%
1}{T}p_{1}^{1/2}p_{2}^{1/2}T^{1/2}p_{1}^{1/2}=O_{P}\left( \frac{p_{1}^{1/2}}{%
p_{2}T^{3/2}}\right) ,
\end{eqnarray*}%
having used the fact that%
\begin{eqnarray*}
&&E\left\Vert \sum_{t=1}^{T}\F_{0,t}\E_{t}^{\prime
}\right\Vert _{F}^{2} \\
&=&E\sum_{i=1}^{p_{1}}\sum_{j=1}^{p_{2}}\left( \sum_{t=1}^{T}\F%
_{0,t}e_{ij,t}\right) ^{2}\leq
\sum_{i=1}^{p_{1}}\sum_{j=1}^{p_{2}}\sum_{t,s=1}^{T}\left\vert E\left( 
\F_{0,t}\F_{0,s}\right) \right\vert \left\vert E\left(
e_{ij,t}e_{ij,s}\right) \right\vert  \\
&\leq &c_{0}\sum_{i=1}^{p_{1}}\sum_{j=1}^{p_{2}}\sum_{t,s=1}^{T}\left\vert
E\left( e_{ij,t}e_{ij,s}\right) \right\vert \leq c_{1}p_{1}p_{2}T.
\end{eqnarray*}%
Hence%
\begin{equation*}
\left\Vert VI\right\Vert _{F}=O_{P}\left( \frac{p_{1}^{1/2}}{%
p_{2}^{1/2}T^{1/2}}\right) .
\end{equation*}%
The desired result now follows from putting everything together. As far as
the invertibility of $\bH_{R_{0}}$ is concerned, it follows from
similar arguments as in the proof of Theorem \ref{hat-estimates}.

The proof of (\ref{c1-hat}) is similar, upon noting%
\begin{eqnarray*}
\hat{\C}_{0} &=&\frac{1}{p_{1}^{2}p_{2}T}\sum_{t=1}^{T}\C%
_{1}\F_{1,t}^{\prime}\R_{1}^{\prime}\hat{\R}%
_{1,\perp}\hat{\R}_{1,\perp}^{\prime}\R_{1}\F%
_{1,t}\C_{1}^{\prime}\hat{\C}_{1,\perp}^{s}\hat{%
\C}_{0}\Lambda _{C_{0}}^{-1} \\
&&+\frac{1}{p_{1}^{2}p_{2}T}\sum_{t=1}^{T}\C_{0}\F%
_{0,t}^{\prime}\R_{0}^{\prime}\hat{\R}_{1,\perp}%
\hat{\R}_{1,\perp}^{\prime}\R_{0}\F_{0,t}%
\C_{0}^{\prime}\hat{\C}_{0}\Lambda _{C_{0}}^{-1} \\
&&+\frac{1}{p_{1}^{2}p_{2}T}\sum_{t=1}^{T}\E_{t}^{\prime}\hat{%
\R}_{1,\perp}\hat{\R}_{1,\perp}^{\prime}\E_{t}%
\hat{\C}_{0}\Lambda _{C_{0}}^{-1}+\frac{1}{p_{1}^{2}p_{2}T}%
\sum_{t=1}^{T}\E_{t}^{\prime}\hat{\R}_{1,\perp}%
\hat{\R}_{1,\perp}^{\prime}\R_{1}\F_{1,t}%
\C_{1}^{\prime}\hat{\C}_{0}\Lambda _{C_{0}}^{-1} \\
&&+\left( \frac{1}{p_{1}^{2}p_{2}T}\sum_{t=1}^{T}\E_{t}^{\prime}%
\hat{\R}_{1,\perp}\hat{\R}_{1,\perp}^{\prime}%
\R_{1}\F_{1,t}\C_{1}^{\prime}\hat{\C}%
_{0}\Lambda _{C_{0}}^{-1}\right) ^{\prime} \\
&&+\frac{1}{p_{1}^{2}p_{2}T}\sum_{t=1}^{T}\C_{0}\F%
_{0,t}^{\prime}\R_{0}^{\prime}\hat{\R}_{1,\perp}%
\hat{\R}_{1,\perp}^{\prime}\R_{1}\F_{1,t}%
\C_{1}^{\prime}\hat{\C}_{0}\Lambda _{C_{0}}^{-1} \\
&&+\left( \frac{1}{p_{1}^{2}p_{2}T}\sum_{t=1}^{T}\C_{0}\F%
_{0,t}^{\prime}\R_{0}^{\prime}\hat{\R}_{1,\perp}%
\hat{\R}_{1,\perp}^{\prime}\R_{1}\F_{1,t}%
\C_{1}^{\prime}\hat{\C}_{0}\Lambda _{C_{0}}^{-1}\right)
^{\prime} \\
&&+\frac{1}{p_{1}^{2}p_{2}T}\sum_{t=1}^{T}\E_{t}^{\prime}\hat{%
\R}_{1,\perp}\hat{\R}_{1,\perp}^{\prime}\R_{0}%
\F_{0,t}\C_{0}^{\prime}\hat{\C}_{0}\Lambda
_{C_{0}}^{-1} \\
&&+\left( \frac{1}{p_{1}^{2}p_{2}T}\sum_{t=1}^{T}\E_{t}^{\prime}%
\hat{\R}_{1,\perp}\hat{\R}_{1,\perp}^{\prime}%
\R_{0}\F_{0,t}\C_{0}^{\prime}\hat{\C}%
_{0}\Lambda _{C_{0}}^{-1}\right) ^{\prime} \\
&=&I+II+III+IV+IV^{\prime}+V+V^{\prime}+VI+VI^{\prime}.
\end{eqnarray*}%
and defining%
\begin{equation*}
\bH_{C_{0}}=\left( \frac{1}{T}\sum_{t=1}^{T}\F_{0,t}^{\prime}%
\frac{\R_{0}^{\prime}\R_{1,\perp}}{p_{1}}\frac{\R%
_{1,\perp}^{\prime}\R_{0}}{p_{1}}\F_{0,t}\right) \frac{%
\C_{0}^{\prime}\hat{\C}_{0}}{p_{2}}\Lambda
_{C_{0}}^{-1},
\end{equation*}%
\begin{equation*}
\bH_{R_{0}}=\left( \frac{1}{T}\sum_{t=1}^{T}\F_{0,t}\frac{%
\C_{0}^{\prime}\C_{1,\perp}}{p_{2}}\frac{\C%
_{1,\perp}^{\prime}\C_{0}}{p_{2}}\F_{0,t}^{\prime}\right) 
\frac{\R_{0}^{\prime}\hat{\R}_{0}}{p_{1}}\Lambda
_{R_{0}}^{-1}.
\end{equation*}
\end{proof}

\begin{proof}[Proof of Theorem \protect\ref{f1-hat}]
Let%
\begin{eqnarray*}
\hat{\mathbf{D}} &=&\left( \hat{\C}_{0}^{\prime}\hat{%
\C}_{1,\perp}\hat{\C}_{1,\perp}^{\prime}\hat{%
\C}_{0}\right) ^{-1}\otimes \left( \hat{\R}_{0}^{\prime}%
\hat{\R}_{1,\perp}\hat{\R}_{1,\perp}^{\prime}%
\hat{\R}_{0}\right) ^{-1}, \\
\hat{\mathbf{N}} &=&\left( \hat{\C}_{0}^{\prime}\hat{%
\C}_{1,\perp}\right) \otimes \left( \hat{\R}%
_{1}^{\prime}\hat{\R}_{1,\perp}\right) ;
\end{eqnarray*}%
note that, by standard algebra 
\begin{equation}
\left\Vert \left( \hat{\mathbf{D}}\right) ^{-1}\right\Vert
_{F}=O_{P}\left( \frac{1}{\left( p_{1}p_{2}\right) ^{2}}\right) ,
\label{d-hat-inv}
\end{equation}%
and consider the decompositions%
\begin{equation*}
\R_{0}=\R_{0}\pm \hat{\R}_{0}\left( \bH%
_{R_{0}}\right) ^{-1},\text{ \ \ and \ \ }\C_{0}=\C_{0}\pm 
\hat{\C}_{0}\left( \bH_{C_{0}}\right) ^{-1}.
\end{equation*}%
We are now ready to start the proof. It holds that%
\begin{eqnarray}
&&\Ve\hat{\F}_{0,t}  \label{vecf1t-1} \\
&=&\left( \hat{\mathbf{D}}\right) ^{-1}\hat{\mathbf{N}}\left( \left( 
\hat{\C}_{1,\perp}\right) ^{\prime}\otimes \left( \hat{%
\R}_{1,\perp}\right) ^{\prime}\right) \left( \C_{0}\otimes 
\R_{0}\right) \Ve\F_{0,t}  \notag \\
&&+\left( \hat{\mathbf{D}}\right) ^{-1}\hat{\mathbf{N}}\left( \left( 
\hat{\C}_{1,\perp}\right) ^{\prime}\otimes \left( \hat{%
\R}_{1,\perp}\right) ^{\prime}\right) \left( \C_{1}\otimes 
\R_{1}\right) \Ve\F_{1,t}  \notag \\
&&+\left( \hat{\mathbf{D}}\right) ^{-1}\hat{\mathbf{N}}\left( \left( 
\hat{\C}_{1,\perp}\right) ^{\prime}\otimes \left( \hat{%
\R}_{1,\perp}\right) ^{\prime}\right) \Ve\E_{t}  \notag \\
&=&I+II+III.  \notag
\end{eqnarray}%
Note%
\begin{eqnarray}
&&I  \label{vecf1t-2} \\
&=&\left[ \left( \hat{\mathbf{D}}\right) ^{-1}\hat{\mathbf{N}}\left(
\left( \hat{\C}_{1,\perp}\right) ^{\prime}\otimes \left( 
\hat{\R}_{1,\perp}\right) ^{\prime}\right) \left( \left( 
\hat{\C}_{0}\left( \bH_{C_{0}}\right) ^{-1}\right)
\otimes \left( \hat{\R}_{0}\left( \bH_{R_{0}}\right)
^{-1}\right) \right) \right.   \notag \\
&&+\left( \hat{\mathbf{D}}\right) ^{-1}\hat{\mathbf{N}}\left( \left( 
\hat{\C}_{1,\perp}\right) ^{\prime}\otimes \left( \hat{%
\R}_{1,\perp}\right) ^{\prime}\right) \left( \left( \C_{0}-%
\hat{\C}_{0}\left( \bH_{C_{0}}\right) ^{-1}\right)
\otimes \left( \hat{\R}_{0}\left( \bH_{R_{0}}\right)
^{-1}\right) \right)   \notag \\
&&+\left( \hat{\mathbf{D}}\right) ^{-1}\hat{\mathbf{N}}\left( \left( 
\hat{\C}_{1,\perp}\right) ^{\prime}\otimes \left( \hat{%
\R}_{1,\perp}\right) ^{\prime}\right) \left( \left( \hat{%
\C}_{0}\left( \bH_{C_{0}}\right) ^{-1}\right) \otimes \left( 
\R_{0}-\hat{\R}_{0}\left( \bH_{R_{0}}\right)
^{-1}\right) \right)   \notag \\
&&+\left. \left( \hat{\mathbf{D}}\right) ^{-1}\hat{\mathbf{N}}\left(
\left( \hat{\C}_{1,\perp}\right) ^{\prime}\otimes \left( 
\hat{\R}_{1,\perp}\right) ^{\prime}\right) \left( \left( 
\C_{0}-\hat{\C}_{0}\left( \bH_{C_{0}}\right)
^{-1}\right) \otimes \left( \R_{0}-\hat{\R}_{0}\left( 
\bH_{R_{0}}\right) ^{-1}\right) \right) \right]   \notag \\
&&\times \Ve\F_{0,t}  \notag \\
&=&I_{a}+I_{b}+I_{c}+I_{d}.  \notag
\end{eqnarray}%
By standard algebraic manipulations%
\begin{equation*}
I_{a}=\Ve\left( \left( \bH_{R_{0}}\right) ^{-1}\F_{0,t}\left( 
\bH_{C_{0}}^{\prime}\right) ^{-1}\right) .
\end{equation*}%
Further%
\begin{eqnarray*}
&&\left\Vert I_{b}\right\Vert _{F} \\
&=&\left\Vert \left( \hat{\mathbf{D}}\right) ^{-1}\left( \hat{%
\C}_{0}^{\prime}\hat{\C}_{1,\perp}\hat{\C}%
_{1,\perp}\left( \C_{0}-\hat{\C}_{0}\left( \bH%
_{C_{0}}\right) ^{-1}\right) \right) \otimes \left( \hat{\R}%
_{1}^{\prime}\hat{\R}_{1,\perp}\hat{\R}_{1,\perp}%
\hat{\R}_{0}\left( \bH_{R_{0}}\right) ^{-1}\right)
\right\Vert _{F}\left\Vert \Ve\F_{0,t}\right\Vert _{F} \\
&=&\left\Vert \left( \hat{\mathbf{D}}\right) ^{-1}\left( \hat{%
\C}_{0}^{\prime}\hat{\C}_{1,\perp}\left( \C%
_{0}-\hat{\C}_{0}\left( \bH_{C_{0}}\right) ^{-1}\right)
\right) \otimes \left( \hat{\R}_{1}^{\prime}\hat{\R}%
_{1,\perp}\hat{\R}_{0}\left( \bH_{R_{0}}\right)
^{-1}\right) \right\Vert _{F}\left\Vert \Ve\F_{0,t}\right\Vert _{F}
\\
&\leq &\left\Vert \left( \hat{\mathbf{D}}\right) ^{-1}\right\Vert
_{F}\left\Vert \hat{\C}_{0}\right\Vert _{F}\sigma _{\max }\left[ 
\hat{\C}_{1,\perp}\left( \C_{0}-\hat{\C}%
_{0}\left( \bH_{C_{0}}\right) ^{-1}\right) \right]  \\
&&\times \left\Vert \hat{\R}_{1}\right\Vert _{F}\left\Vert 
\hat{\R}_{1,\perp}\right\Vert _{F}\left\Vert \hat{\R%
}_{0}\right\Vert _{F}\left\Vert \left( \bH_{R_{0}}\right)
^{-1}\right\Vert _{F}\left\Vert \Ve\F_{0,t}\right\Vert _{F} \\
&=&O_{P}(1) \frac{1}{\left( p_{1}p_{2}\right) ^{2}}p_{2}^{1/2}%
p_{2}\left[ \left( O_{P}\left( \frac{p_{2}^{1/2}}{p_{1}T}\right)
+O_{P}\left( \frac{p_{2}^{1/2}}{p_{1}p_{2}}\right) +O_{P}\left( \frac{%
p_{2}^{1/2}}{p_{1}^{1/2}p_{2}^{1/2}T^{1/2}}\right) \right) \right]
p_{1}^{1/2}p_{1}p_{1}^{1/2} \\
&=&O_{P}\left( \frac{1}{p_{1}^{1/2}p_{2}^{1/2}T^{1/2}}\right) +O_{P}\left( 
\frac{1}{p_{1}p_{2}}\right) +O_{P}\left( \frac{1}{p_{1}T}\right) ,
\end{eqnarray*}%
having used (\ref{d-hat-inv}), Lemma \ref{yong} and Theorem \ref{rc1-hat}.
Similarly we can show that%
\begin{equation*}
\left\Vert I_{c}\right\Vert _{F}=O_{P}\left( \frac{1}{%
p_{1}^{1/2}p_{2}^{1/2}T^{1/2}}\right) +O_{P}\left( \frac{1}{p_{1}p_{2}}%
\right) +O_{P}\left( \frac{1}{p_{2}T}\right) ,
\end{equation*}%
and by the same logic, it can be shown that $\left\Vert I_{d}\right\Vert _{F}
$ is dominated. Further, using (\ref{orth-prod-c}) and the similar result $%
\R_{1}^{\prime}\hat{\R}_{1,\perp}=\R%
_{1}^{\prime}\left( \hat{\R}_{1,\perp}-\R_{1,\perp
}\right) $, it holds that%
\begin{eqnarray*}
II &=&\left( \hat{\mathbf{D}}\right) ^{-1}\left( \hat{\C}%
_{0}^{\prime}\hat{\C}_{1,\perp}\hat{\C}_{1,\perp
}^{\prime}\C_{1}\right) \otimes \left( \hat{\R}%
_{0}^{\prime}\hat{\R}_{1,\perp}\hat{\R}_{1,\perp
}^{\prime}\R_{1}\right) \Ve\F_{1,t} \\
&=&\left( \hat{\mathbf{D}}\right) ^{-1}\left( \left( \hat{\C}%
_{0}^{\prime}\hat{\C}_{1,\perp}\C_{1}\right) \otimes
\left( \hat{\R}_{0}^{\prime}\hat{\R}_{1,\perp}%
\R_{1}\right) \right) \Ve\F_{1,t} \\
&=&\left( \hat{\mathbf{D}}\right) ^{-1}\left( \left( \hat{\C}%
_{0}^{\prime}\left( \hat{\C}_{1,\perp}-\C_{1,\perp
}\right) \C_{1}\right) \otimes \left( \hat{\R}%
_{0}^{\prime}\left( \hat{\R}_{1,\perp}-\R_{1,\perp
}\right) \R_{1}\right) \right) \Ve\F_{1,t}
\end{eqnarray*}%
whence%
\begin{eqnarray*}
\left\Vert II\right\Vert _{F} &\leq &\left\Vert \left( \hat{\mathbf{D}}%
\right) ^{-1}\right\Vert _{F}\left\Vert \hat{\C}_{0}\right\Vert
_{F}\left\Vert \C_{1}\right\Vert _{F}\left\Vert \hat{\C}%
_{1,\perp}-\C_{1,\perp}\right\Vert _{F} \\
&&\times \left\Vert \hat{\R}_{0}\right\Vert _{F}\left\Vert 
\R_{1}\right\Vert _{F}\left\Vert \hat{\R}_{1,\perp}-%
\R_{1,\perp}\right\Vert _{F}\left\Vert \Ve\F%
_{1,t}\right\Vert _{F} \\
&=&O_{P}(1) \frac{1}{\left( p_{1}p_{2}\right) ^{2}}%
p_{2}^{1/2}p_{2}^{1/2}\frac{1}{T}p_{1}^{1/2}p_{1}^{1/2}\frac{1}{T}%
T^{1/2}=O_{P}\left( \frac{1}{p_{1}p_{2}T^{3/2}}\right) ,
\end{eqnarray*}%
having used (\ref{lemma1-i}) in Lemma \ref{berkes} and Lemmas \ref{r-orth}
and \ref{c-orth}. Finally, using the same logic as in the above, it can be
shown that (modulo some higher order terms)%
\begin{eqnarray*}
\left\Vert III\right\Vert _{F} &\leq &\left\Vert \left( \hat{\mathbf{D}}%
\right) ^{-1}\right\Vert _{F}\left\Vert \hat{\C}_{0}\right\Vert
_{F}\left\Vert \hat{\R}_{0}\right\Vert _{F}\left\Vert \left( 
\hat{\R}_{1,\perp}^{s}\right) ^{\prime}\E_{0,t}%
\hat{\C}_{1,\perp}^{s}\right\Vert _{F} \\
&=&O_{P}(1) \frac{1}{\left( p_{1}p_{2}\right) ^{2}}%
p_{1}^{1/2}p_{2}^{1/2}p_{1}p_{2}=O_{P}\left( \frac{1}{p_{1}^{1/2}p_{2}^{1/2}}%
\right) .
\end{eqnarray*}%
Then (\ref{f1-hat-point}) follows from putting everything together. Equation
(\ref{f1-hat-l2}) can be shown by noting that, using Minkowski's inequality 
\begin{eqnarray*}
&&\frac{1}{T}\sum_{t=1}^{T}\left\Vert \hat{\F}_{0,t}-\left( 
\bH_{R_{0}}\right) ^{-1}\F_{0,t}\left( \bH%
_{C_{0}}^{\prime}\right) ^{-1}\right\Vert _{F}^{2} \\
&\leq &\frac{1}{T}\sum_{t=1}^{T}\left( \left\Vert \widetilde{I}\right\Vert
_{F}^{2}+\left\Vert II\right\Vert _{F}^{2}+\left\Vert III\right\Vert
_{F}^{2}\right) ,
\end{eqnarray*}%
where $II$ and $III$ are defined in (\ref{vecf1t-1}), and $\widetilde{I}%
=I_{b}+I_{c}+I_{d}$, with $I_{b}$, $I_{c}$ and$\ I_{d}$\ defined in (\ref%
{vecf1t-2}). The desired result can now be shown by applying the same logic
as above.
\end{proof}

\begin{proof}[Proof of Theorem \protect\ref{rc-tilde}]
Some arguments are similar to the proof of Lemma \ref{rc1-hat}, and we
therefore omit them when possible. We begin by showing (\ref{r-hat-tilde});
by definition, it holds that%
\begin{equation*}
\widetilde{\R}_1=\proj{\M}_{R_{1}}\widetilde{\R}_1\widetilde{%
\Lambda }_{R_{1}}^{-1},
\end{equation*}%
whence 
\begin{align}
&\widetilde{\R}_1  \label{r-tilde-dec} \\
&=\frac{1}{p_{1}p_{2}^{2}T^{2}}\sum_{t=1}^{T}\R_1\F_{1,t}\C_{1}^{\prime}\hat{\C}_{1}\hat{\C}_{1}^{\prime}\C_{1}\F_{1,t}^{\prime}\R_{1}^{\prime}\widetilde{\R}_1\widetilde{\Lambda }%
_{R_{1}}^{-1}+\frac{1}{p_{1}p_{2}^{2}T^{2}}\sum_{t=1}^{T}\E_{t}\hat{\C}_{1}\hat{\C}_{1}^{\prime}\E_{t}^{\prime}\widetilde{%
\R_1}\widetilde{\Lambda }_{R_{1}}^{-1}  \notag \\
&+\frac{1}{p_{1}p_{2}^{2}T^{2}}\sum_{t=1}^{T}\R_1\F_{1,t}\C_{1}^{\prime}\hat{\C}_{1}\hat{\C}_{1}^{\prime}\E_{t}^{\prime}\widetilde{\R}_1\widetilde{\Lambda }_{R_{1}}^{-1}+\left( 
\frac{1}{p_{1}p_{2}^{2}T^{2}}\sum_{t=1}^{T}\R_1\F_{1,t}\C_{1}^{\prime}\hat{\C}_{1}\hat{\C}_{1}^{\prime}\E_{t}^{\prime}%
\widetilde{\R}_1\widetilde{\Lambda }_{R_{1}}^{-1}\right) ^{\prime}  \notag
\\
&+\frac{1}{p_{1}p_{2}^{2}T^{2}}\sum_{t=1}^{T}\left( \R_{0}\F_{0,t}\C_{0}^{\prime}-\hat{\R}_{0}\hat{\F}%
_{0,t}\hat{\C}_{0}^{\prime}\right) \hat{\C}_{1}\hat{\C}_{1}^{\prime}\E_{t}^{\prime}\widetilde{\R}_1%
\widetilde{\Lambda }_{R_{1}}^{-1}  \notag \\
&+\left( \frac{1}{p_{1}p_{2}^{2}T^{2}}\sum_{t=1}^{T}\left( \R_{0}%
\F_{0,t}\C_{0}^{\prime}-\hat{\R}_{0}\hat{%
\F}_{0,t}\hat{\C}_{0}^{\prime}\right) \hat{\C}_{1}\hat{\C}_{1}^{\prime}\E_{t}^{\prime}\widetilde{\R%
}\widetilde{\Lambda }_{R_{1}}^{-1}\right) ^{\prime}  \notag \\
&+\frac{1}{p_{1}p_{2}^{2}T^{2}}\sum_{t=1}^{T}\R_1\F_{1,t}\C_{1}^{\prime}\hat{\C}_{1}\hat{\C}_{1}^{\prime}\left( \R_{0}\F_{0,t}\C_{0}^{\prime}-\hat{\R}_{0}%
\hat{\F}_{0,t}\hat{\C}_{0}^{\prime}\right) ^{\prime}\widetilde{\R}_1\widetilde{\Lambda }_{R_{1}}^{-1}  \notag \\
&+\left( \frac{1}{p_{1}p_{2}^{2}T^{2}}\sum_{t=1}^{T}\R_1\F_{1,t}\C_{1}^{\prime}\hat{\C}_{1}\hat{\C}_{1}^{\prime}\left( \R_{0}\F_{0,t}\C_{0}^{\prime}-\hat{\R}_{0}%
\hat{\F}_{0,t}\hat{\C}_{0}^{\prime}\right) ^{\prime}\widetilde{\R}_1\widetilde{\Lambda }_{R_{1}}^{-1}\right) ^{\prime} 
\notag \\
&+\frac{1}{p_{1}p_{2}^{2}T^{2}}\sum_{t=1}^{T}\left( \R_{0}\F%
_{0,t}\C_{0}^{\prime}-\hat{\R}_{0}\hat{\F}%
_{0,t}\hat{\C}_{0}^{\prime}\right) \hat{\C}_{1}\hat{\C}_{1}^{\prime}\left( \R_{0}\F_{0,t}\C_{0}^{\prime}-\hat{\R}_{0}\hat{\F}_{0,t}\hat{%
\C}_{0}^{\prime}\right) ^{\prime}\widetilde{\R}_1\widetilde{\Lambda }_{R_{1}}^{-1}  \notag \\
&=I+II+III+III^{\prime}+IV+IV^{\prime}+V+V^{\prime}+VI.  \notag
\end{align}%
Upon letting%
\begin{equation*}
\widetilde{\bH}_{R_{1}}=\left( \frac{1}{T^{2}}\sum_{t=1}^{T}\F_{1,t} \frac{\C_{1}^{\prime}\hat{\C}_{1}}{p_{2}}\frac{\hat{\C}_{1}^{\prime}\C_{1}}{p_{2}}\F_{1,t}^{\prime}\right) \left( \frac{%
\R_{1}^{\prime}\widetilde{\R}_1}{p_{1}}\right) \widetilde{\Lambda 
}_{R_{1}}^{-1},
\end{equation*}%
the same logic as in the above yields that $\left\Vert \widetilde{\bH}%
_{R_{1}}\right\Vert _{F}=O_{P}(1) $. We now carry out the proof
under $h_{R_{1}}=h_{C_{1}}=1$ when possible, so that $\bH_{C_{1}}$ is a random
sign. It holds that 
\begin{align*}
II &=\frac{1}{p_{1}p_{2}^{2}T^{2}}\sum_{t=1}^{T}\E_{t}\C_{1}\bH%
_{C_{1}}\bH_{C_{1}}^{\prime}\C_{1}^{\prime}\E_{t}^{\prime}%
\widetilde{\R}_1\widetilde{\Lambda }_{R_{1}}^{-1}+\frac{1}{%
p_{1}p_{2}^{2}T^{2}}\sum_{t=1}^{T}\E_{t}\left( \hat{\C}_{1}-%
\C_{1}\bH_{C_{1}}\right) \bH_{C_{1}}^{\prime}\C_{1}^{\prime}\E_{t}^{\prime}\widetilde{\R}_1\widetilde{\Lambda }_{R_{1}}^{-1} \\
&+\frac{1}{p_{1}p_{2}^{2}T^{2}}\sum_{t=1}^{T}\E_{t}\C_{1}\bH%
_{C_{1}}\left( \hat{\C}_{1}-\C_{1}\bH_{C_{1}}\right) ^{\prime}\E%
_{t}^{\prime}\widetilde{\R}_1\widetilde{\Lambda }_{R_{1}}^{-1} \\
&+\frac{1}{p_{1}p_{2}^{2}T^{2}}\sum_{t=1}^{T}\E_{t}\left( \hat{%
\C}_{1}-\C_{1}\bH_{C_{1}}\right) \left( \hat{\C}_{1}-\C_{1}\bH%
_{C_{1}}\right) ^{\prime}\E_{t}^{\prime}\widetilde{\R}_1%
\widetilde{\Lambda }_{R_{1}}^{-1} \\
&=II_{a}+II_{b}+II_{b}^{\prime}+II_{c},
\end{align*}%
and%
\begin{equation*}
\left\Vert II_{a}\right\Vert _{F}\leq \frac{1}{p_{1}p_{2}^{2}T^{2}}%
\left\Vert \sum_{t=1}^{T}\E_{t}\C_{1}\C_{1}^{\prime}\E%
_{t}^{\prime}\right\Vert _{F}\left\Vert \widetilde{\R}_1\right\Vert
_{F}\left\Vert \widetilde{\Lambda }_{R_{1}}^{-1}\right\Vert _{F}=O_{P}\left( 
\frac{p_{1}^{1/2}}{p_{2}T}\right) +O_{P}\left( \frac{p_{1}^{1/2}}{%
p_{2}^{1/2}T^{3/2}}\right) ,
\end{equation*}%
using/adapting (\ref{3F-error}). We now consider%
\begin{equation*}
\left\Vert II_{b}\right\Vert _{F}\leq \frac{1}{p_{1}p_{2}^{2}T^{2}}%
\left\Vert \sum_{t=1}^{T}\E_{t}\left( \hat{\C}_{1}-\C_{1}\bH_{C_{1}}\right) \C_{1}^{\prime}\E_{t}^{\prime}\right\Vert
_{F}\left\Vert \widetilde{\R}_1\right\Vert _{F}\left\Vert \widetilde{%
\Lambda }_{R_{1}}^{-1}\right\Vert _{F}.
\end{equation*}%
Noting that%
\begin{align*}
&\left\Vert \sum_{t=1}^{T}\E_{t}\left( \hat{\C}_{1}-\C_{1}\bH_{C_{1}}\right) \C_{1}^{\prime}\E_{t}^{\prime}\right\Vert
_{F}^{2} \\
&=\sum_{i,h=1}^{p_{1}}\left( \sum_{t=1}^{T}\left( \sum_{j=1}^{p_{2}}\left( 
\hat{c}_{j}-c_{j}\right) e_{ij,t}\right) \left(
\sum_{j=1}^{p_{2}}c_{j}e_{hj,t}\right) \right) ^{2} \\
&\leq\sum_{i,h=1}^{p_{1}}\left( \sum_{t=1}^{T}\left(
\sum_{j=1}^{p_{2}}\left( \hat{c}_{j}-c_{j}\right) e_{ij,t}\right)
^{2}\right) \left( \sum_{t=1}^{T}\left(
\sum_{j=1}^{p_{2}}c_{j}e_{hj,t}\right) ^{2}\right) \\
&\leq\sum_{i,h=1}^{p_{1}}\sum_{t,s=1}^{T}\left( \sum_{j=1}^{p_{2}}\left( 
\hat{c}_{j}-c_{j}\right) ^{2}\right) \left(
\sum_{j=1}^{p_{2}}e_{ij,t}^{2}\right) \left(
\sum_{j=1}^{p_{2}}c_{j}e_{hj,s}\right) ^{2} \\
&=\left( \left\Vert \hat{\C}_{1}-\C_{1}\bH_{C_{1}}\right\Vert
_{F}^{2}\right) \sum_{i,h=1}^{p_{1}}\sum_{t,s=1}^{T}\left(
\sum_{j=1}^{p_{2}}e_{ij,t}^{2}\right) \left(
\sum_{j=1}^{p_{2}}c_{j}e_{hj,s}\right) ^{2} \\
&=\left( \left\Vert \hat{\C}_{1}-\C_{1}\bH_{C_{1}}\right\Vert
_{F}^{2}\right) \left(
\sum_{i=1}^{p_{1}}\sum_{t=1}^{T}\sum_{j=1}^{p_{2}}e_{ij,t}^{2}\right) \left(
\sum_{i=1}^{p_{1}}\sum_{t=1}^{T}\left(
\sum_{j=1}^{p_{2}}c_{j}e_{hj,s}\right) ^{2}\right) ,
\end{align*}%
it is easy to see that%
\begin{equation*}
\left\Vert \sum_{t=1}^{T}\E_{t}\left( \hat{\C}_{1}-\C_{1}\bH_{C_{1}}\right) \C_{1}^{\prime}\E_{t}^{\prime}\right\Vert
_{F}^{2}=O_{P}\left( p_{1}^{2}p_{2}^{3}\right) ,
\end{equation*}%
whence it immediately follows that%
\begin{equation*}
\left\Vert II_{b}\right\Vert _{F}=O_{P}\left( \frac{p_{1}^{1/2}}{%
p_{2}^{1/2}T^{2}}\right) ,
\end{equation*}%
and the same holds for $II_{b}^{\prime}$; similarly, it is not hard to show
that $II_{c}$ is dominated by $II_{a}$ and $II_{b}$. Turning to $III$, 
\begin{align*}
III &=\frac{1}{p_{1}p_{2}^{2}T^{2}}\sum_{t=1}^{T}\R_1\F_{1,t}\C_{1}^{\prime}\hat{\C}_{1}\bH_{C_{1}}^{\prime}\C_{1}^{\prime}%
\E_{t}^{\prime}\widetilde{\R}_1\widetilde{\Lambda }_{R_{1}}^{-1}
\\
&+\frac{1}{p_{1}p_{2}^{2}T^{2}}\sum_{t=1}^{T}\R_1\F_{1,t}\C_{1}^{\prime}\hat{\C}_{1}\left( \hat{\C}_{1}-\C_{1}\bH%
_{C_{1}}\right) ^{\prime}\E_{t}^{\prime}\widetilde{\R}_1%
\widetilde{\Lambda }_{R_{1}}^{-1}=III_{a}+III_{b},
\end{align*}%
and we have

\begin{align*}
\left\Vert III_{a}\right\Vert _{F} &\leq\frac{1}{p_{1}p_{2}^{2}T^{2}}%
\left\Vert \R_1\right\Vert _{F}\left\Vert \bH_{C_{1}}\right\Vert
_{F}\left\Vert \C_{1}\right\Vert _{F}\left\Vert \hat{\C}_{1}%
\right\Vert _{F}\left\Vert \sum_{t=1}^{T}\F_{1,t}\C_{1}^{\prime}%
\E_{t}^{\prime}\right\Vert _{F}\left\Vert \widetilde{\R}_1%
\right\Vert _{F}\left\Vert \widetilde{\Lambda }_{R_{1}}^{-1}\right\Vert _{F} \\
&=O_{P}(1) \frac{1}{p_{1}p_{2}^{2}T^{2}}%
p_{1}^{1/2}p_{2}^{1/2}p_{2}^{1/2}p_{1}^{1/2}\left( p_{1}p_{2}T^{2}\right)
^{1/2}=O_{P}\left( \frac{p_{1}^{1/2}}{p_{2}^{1/2}T}\right) ,
\end{align*}%
and 
\begin{align*}
\left\Vert III_{b}\right\Vert _{F} &\leq\frac{1}{p_{1}p_{2}^{2}T^{2}}%
\left\Vert \R_1\right\Vert _{F}\left\Vert \C_{1}\right\Vert
_{F}\left\Vert \hat{\C}_{1}\right\Vert _{F}\left\Vert \hat{%
\C}_{1}-\C_{1}\bH_{C_{1}}\right\Vert _{F}\left\Vert \sum_{t=1}^{T}\F_{1,t}\E_{t}^{\prime}\right\Vert _{F}\left\Vert \widetilde{\R}_{1}\right\Vert _{F}\left\Vert \widetilde{\Lambda }_{R_{1}}^{-1}\right\Vert _{F} \\
&=O_{P}(1) \frac{1}{p_{1}p_{2}^{2}T^{2}}%
p_{1}^{1/2}p_{2}^{1/2}p_{2}^{1/2}\frac{p_{2}^{1/2}}{T}p_{1}^{1/2}\left(
p_{1}p_{2}T^{2}\right) ^{1/2}=O_{P}\left( \frac{p_{1}^{1/2}}{T^{2}}\right) .
\end{align*}%
We now study%
\begin{align*}
IV &=\frac{1}{p_{1}p_{2}^{2}T^{2}}\sum_{t=1}^{T}\R_{0}\F%
_{0,t}\left( \bH_{C_0}^{\prime}\right) ^{-1}\left( \hat{\C}_{0}-\C_{0}\bH_{C_0}\right) ^{\prime}\hat{\C}_{1}%
\hat{\C}_{1}^{\prime}\E_{t}^{\prime}\widetilde{\R}_1%
\widetilde{\Lambda }_{R_{1}}^{-1} \\
&+\frac{1}{p_{1}p_{2}^{2}T^{2}}\sum_{t=1}^{T}\left( \hat{\R}_{0}-\R_{0}\bH_{R_0}\right) \bH_{R_0}^{-1}\F%
_{0,t}\C_{0}^{\prime}\hat{\C}_{1}\hat{\C}_{1}%
^{\prime}\E_{t}^{\prime}\widetilde{\R}_1\widetilde{\Lambda }%
_{R_{1}}^{-1} \\
&+\frac{1}{p_{1}p_{2}^{2}T^{2}}\sum_{t=1}^{T}\R_{0}\left( \hat{%
\F}_{0,t}-\left( \bH_{R_0}\right) ^{-1}\F_{0,t}\left( 
\bH_{C_0}^{\prime}\right) ^{-1}\right) \C_{0}^{\prime}%
\hat{\C}_{1}\hat{\C}_{1}^{\prime}\E_{t}^{\prime}%
\widetilde{\R}_1\widetilde{\Lambda }_{R_{1}}^{-1}+IV_{d} \\
&=IV_{a}+IV_{b}+IV_{c}+IV_{d},
\end{align*}%
where $IV_{d}$ is a remainder which, by the same logic as above, can be
shown to be dominated by $IV_{a}-IV_{c}$. It holds that%
\begin{align*}
\left\Vert IV_{a}\right\Vert _{F} &=O_{P}(1) \frac{1}{%
p_{1}p_{2}^{2}T^{2}}\left\Vert \R_{0}\right\Vert _{F}\left\Vert 
\hat{\C}_{1}\right\Vert _{F}\left\Vert \hat{\C}_{0}-%
\C_{0}\bH_{C_0}\right\Vert _{F}\left\Vert \sum_{t=1}^{T}%
\F_{0,t}\hat{\C}_{1}^{\prime}\E_{t}^{\prime}\right\Vert _{F}\left\Vert \widetilde{\R}_1\right\Vert _{F}\left\Vert 
\widetilde{\Lambda }_{R_{1}}^{-1}\right\Vert _{F} \\
&=O_{P}(1) \frac{1}{p_{2}T^{2}}\left( \frac{1}{p_{1}p_{2}}+%
\frac{1}{p_{1}^{1/2}T^{1/2}}\right) \left\Vert \sum_{t=1}^{T}\F_{0,t}%
\hat{\C}_{1}^{\prime}\E_{t}^{\prime}\right\Vert _{F};
\end{align*}%
noting that%
\begin{align}
\left\Vert \sum_{t=1}^{T}\F_{0,t}\hat{\C}_{1}^{\prime}%
\E_{t}^{\prime}\right\Vert _{F} &\leq\left\Vert \bH%
_{C_{1}}^{\prime}\right\Vert _{F}\left\Vert \sum_{t=1}^{T}\F_{0,t}%
\bH_{C_{1}}^{\prime}\C_{1}\E_{t}^{\prime}\right\Vert _{F}+\left\Vert
\sum_{t=1}^{T}\F_{0,t}\left( \hat{\C}_{1}-\C_{1}\bH%
_{C_{1}}\right) ^{\prime}\E_{t}^{\prime}\right\Vert _{F}
\label{f1e_chat} \\
&=O_{P}\left( p_{1}^{1/2}p_{2}^{1/2}T^{1/2}\right) +O_{P}\left( \frac{1}{%
T^{1/2}}p_{1}^{1/2}p_{2}\right) ,  \notag
\end{align}%
it now follows that%
\begin{equation*}
\left\Vert IV_{a}\right\Vert _{F}=O_{P}(1) \left( \frac{1}{%
p_{1}^{1/2}p_{2}^{3/2}T^{3/2}}+\frac{1}{p_{2}^{1/2}T^{2}}+\frac{1}{%
p_{1}^{1/2}p_{2}T^{5/2}}+\frac{1}{T^{3}}\right) .
\end{equation*}%
Similarly, using again (\ref{f1e_chat})%
\begin{align*}
\left\Vert IV_{b}\right\Vert _{F} &=O_{P}(1) \frac{1}{%
p_{1}p_{2}^{2}T^{2}}\left\Vert \hat{\R}_{0}-\R_{0}%
\bH_{R_0}\right\Vert _{F}\left\Vert \bH_{R_0}^{-1}\right\Vert
_{F}\left\Vert \C_{0}\right\Vert _{F}\left\Vert \hat{\C}_{1}%
\right\Vert _{F} \\
&\times \left\Vert \sum_{t=1}^{T}\F_{0,t}\hat{\C}_{1}%
^{\prime}\E_{t}^{\prime}\right\Vert _{F}\left\Vert \widetilde{%
\R_1}\right\Vert _{F}\left\Vert \widetilde{\Lambda }%
_{R_{1}}^{-1}\right\Vert _{F} \\
&=O_{P}(1) \left( \frac{p_{1}^{1/2}}{p_{2}T^{2}}+\frac{%
p_{1}^{1/2}}{p_{2}^{1/2}T^{3}}+\frac{1}{p_{1}^{1/2}p_{2}^{3/2}T^{3/2}}+\frac{%
1}{p_{1}^{1/2}p_{2}T^{5/2}}\right) ,
\end{align*}%
and%
\begin{align*}
\left\Vert IV_{c}\right\Vert _{F} &=O_{P}(1) \frac{1}{%
p_{1}p_{2}^{2}T^{2}}\left\Vert \R_{0}\right\Vert _{F}\left\Vert 
\C_{0}\right\Vert _{F}\left\Vert \hat{\C}_{1}\right\Vert
_{F}\left\Vert \widetilde{\R}_1\right\Vert _{F}\left\Vert \widetilde{%
\Lambda }_{R_{1}}^{-1}\right\Vert _{F} \\
&\times \left( \sum_{t=1}^{T}\left\Vert \hat{\F}_{0,t}-\left( 
\bH_{R_0}\right) ^{-1}\F_{0,t}\left( \bH_{C_0}^{\prime}\right) ^{-1}\right\Vert _{F}^{2}\right) ^{1/2}\left(
\sum_{t=1}^{T}\left\Vert \hat{\C}_{1}^{\prime}\E_{t}^{\prime}\right\Vert _{F}^{2}\right) ^{1/2} \\
&=O_{P}(1) \frac{1}{p_{2}T^{2}}T^{1/2}\left( \frac{1}{%
p_{1}^{1/2}p_{2}^{1/2}}+\frac{1}{\left( p_{1\wedge 2}T\right) ^{1/2}}\right)
\left( p_{1}^{1/2}p_{2}^{1/2}+\frac{p_{1}^{1/2}p_{2}}{T}\right)
\end{align*}%
by the same logic as above. We now study 
\begin{align*}
V^{\prime} &=\frac{1}{p_{1}p_{2}^{2}T^{2}}\sum_{t=1}^{T}\R_{0}%
\F_{0,t}\left( \bH_{C_0}^{\prime}\right) ^{-1}\left( 
\hat{\C}_{0}-\C_{0}\bH_{C_0}\right) ^{\prime}%
\hat{\C}_{1}\hat{\C}_{1}^{\prime}\C_{1}\F_{1,t}^{\prime}\R_{1}^{\prime}\widetilde{\R}_1\widetilde{\Lambda }%
_{R_{1}}^{-1} \\
&+\frac{1}{p_{1}p_{2}^{2}T^{2}}\sum_{t=1}^{T}\left( \hat{\R}_{0}-\R_{0}\bH_{R_0}\right) \F_{0,t}\C_{0}^{\prime}\hat{\C}_{1}\hat{\C}_{1}^{\prime}\C_{1}\F_{1,t}^{\prime}\R_{1}^{\prime}\widetilde{\R}_1%
\widetilde{\Lambda }_{R_{1}}^{-1} \\
&+\frac{1}{p_{1}p_{2}^{2}T^{2}}\sum_{t=1}^{T}\R_{0}\left( \hat{%
\F}_{0,t}-\left( \bH_{R_0}\right) ^{-1}\F_{0,t}\left( 
\bH_{C_0}^{\prime}\right) ^{-1}\right) \C_{0}^{\prime}%
\hat{\C}_{1}\hat{\C}_{1}^{\prime}\C_{1}\F_{1,t}^{\prime}\R_{1}^{\prime}\widetilde{\R}_1\widetilde{\Lambda }%
_{R_{1}}^{-1}+V_{d} \\
&=V_{a}+V_{b}+V_{c}+V_{d},
\end{align*}%
where $V_{d}$ is a remainder which, by the same logic as above, can be shown
to be dominated by $V_{a}-V_{c}$. We have%
\begin{align*}
\left\Vert V_{a}\right\Vert _{F} &=O_{P}(1) \frac{1}{%
p_{1}p_{2}^{2}T^{2}}\left\Vert \R_{0}\right\Vert _{F}\left\Vert 
\R_1\right\Vert _{F}\left\Vert \hat{\C}_{1}\right\Vert
_{F}^{2}\left\Vert \C_{1}\right\Vert _{F}\left\Vert \hat{\C}_{0}-\C_{0}\bH_{C_0}\right\Vert _{F} \\
&\times \left\Vert \sum_{t=1}^{T}\F_{0,t}\F_{1,t}\right\Vert _{F}\left\Vert \widetilde{\R}_1\right\Vert
_{F}\left\Vert \widetilde{\Lambda }_{R_{1}}^{-1}\right\Vert _{F} \\
&=p_{1}^{1/2}O_{P}(1) \frac{1}{p_{1}p_{2}^{2}T^{2}}%
p_{1}^{1/2}p_{1}^{1/2}p_{2}p_{2}^{1/2}p_{2}^{1/2}\left( \frac{1}{p_{1}p_{2}}+%
\frac{1}{p_{1}^{1/2}T^{1/2}}\right) T \\
&=O_{P}\left( \frac{p_{1}^{1/2}}{T}\left( \frac{1}{p_{1}p_{2}}+\frac{1}{%
p_{1}^{1/2}T^{1/2}}\right) \right) ,
\end{align*}%
\begin{align*}
\left\Vert V_{b}\right\Vert _{F} &=O_{P}(1) \frac{1}{%
p_{1}p_{2}^{2}T^{2}}\left\Vert \R_1\right\Vert _{F}\left\Vert \hat{%
\R}_{0}-\R_{0}\bH_{R_0}\right\Vert _{F}\left\Vert 
\hat{\C}_{1}\right\Vert _{F}^{2} \\
&\times \left\Vert \C_{1}\right\Vert _{F}\left\Vert \C_{0}\right\Vert _{F}\left\Vert \sum_{t=1}^{T}\F_{0,t}\F_{1,t}\right\Vert _{F}\left\Vert \widetilde{\R}_1\right\Vert
_{F}\left\Vert \widetilde{\Lambda }_{R_{1}}^{-1}\right\Vert _{F} \\
&=O_{P}(1) \frac{1}{p_{1}p_{2}^{2}T^{2}}%
p_{1}^{1/2}p_{1}^{1/2}p_{2}p_{2}^{1/2}p_{2}^{1/2}p_{1}^{1/2}\left( \frac{1}{%
p_{1}p_{2}}+\frac{1}{p_{2}^{1/2}T^{1/2}}\right) T \\
&=O_{P}\left( \frac{p_{1}^{1/2}}{T}\left( \frac{1}{p_{1}p_{2}}+\frac{1}{%
p_{2}^{1/2}T^{1/2}}\right) \right) ,
\end{align*}%
and, using Lemma \ref{bai}%
\begin{align*}
\left\Vert V_{c}\right\Vert _{F} &=O_{P}(1) \frac{1}{%
p_{1}p_{2}^{2}T^{2}}\left\Vert \R_1\right\Vert _{F}\left\Vert \R_{0}\right\Vert _{F}\left\Vert \hat{\C}_{1}\right\Vert
_{F}^{2}\left\Vert \C_{1}\right\Vert _{F}\left\Vert \C_{0}\right\Vert _{F} \\
&\times \left\Vert \sum_{t=1}^{T}\left( \hat{\F}_{0,t}-\left( 
\bH_{R_0}\right) ^{-1}\F_{0,t}\left( \bH_{C_0}^{\prime}\right) ^{-1}\right) \F_{1,t}^{\prime}\right\Vert
_{F}\left\Vert \widetilde{\R}_1\right\Vert _{F}\left\Vert \widetilde{%
\Lambda }_{R_{1}}^{-1}\right\Vert _{F} \\
&=O_{P}(1) \frac{1}{p_{1}p_{2}^{2}T^{2}}%
p_{1}^{1/2}p_{1}^{1/2}p_{1}^{1/2}p_{2}p_{2}^{1/2}p_{2}^{1/2}T\left( \frac{1}{%
p_{1}^{1/2}p_{2}^{1/2}}+\frac{1}{p_{1\wedge 2}^{1/2}T^{1/2}}\right) \\
&=O_{P}\left( \frac{p_{1}^{1/2}}{T}\left( \frac{1}{p_{1}^{1/2}p_{2}^{1/2}}+%
\frac{1}{p_{1\wedge 2}^{1/2}T^{1/2}}\right) \right) .
\end{align*}%
Finally, using (\ref{f1-hat-l2})%
\begin{align*}
\left\Vert VI\right\Vert _{F} &\leq\frac{1}{p_{1}p_{2}^{2}T}\left( \frac{1%
}{T}\sum_{t=1}^{T}\left\Vert \R_{0}\F_{0,t}\C_{0}^{\prime}-\hat{\R}_{0}\hat{\F}_{0,t}\hat{\C}_{0}^{\prime}\right\Vert _{F}^{2}\right) \left\Vert \hat{%
\C}_{1}\right\Vert _{F}^{2}\left\Vert \widetilde{\R}_1\right\Vert
_{F}\left\Vert \widetilde{\Lambda }_{R_{1}}^{-1}\right\Vert _{F} \\
&=O_{P}(1) \frac{1}{p_{1}p_{2}^{2}T}\left( \frac{1}{%
p_{1}^{1/2}p_{2}^{1/2}}+\frac{1}{p_{1\wedge 2}^{1/2}T^{1/2}}\right)
^{2}p_{2}p_{1}^{1/2} \\
&=O_{P}\left( \frac{1}{p_{1}^{1/2}p_{2}T}\left( \frac{1}{%
p_{1}^{1/2}p_{2}^{1/2}}+\frac{1}{p_{1\wedge 2}^{1/2}T^{1/2}}\right)
^{2}\right) ,
\end{align*}%
which is dominated. The desired result now follows from putting all
together; finally, the invertibility of $\widetilde{\bH}_{R_{1}}$ can be
shown in a similar way as in the above.

As far as (\ref{c-hat-tilde}) is concerned, recall%
\begin{equation*}
\widetilde{\C}_{1}=\proj{\M}_{C_{1}}\widetilde{\C}_{1}\widetilde{%
\Lambda }_{C_{1}}^{-1},
\end{equation*}%
with%
\begin{align*}
\widetilde{\C}_{1} &=\frac{1}{p_{1}^{2}p_{2}T^{2}}\sum_{t=1}^{T}\C_{1}\F_{1,t}^{\prime}\R_{1}^{\prime}\hat{\R}_1\hat{\R}_1%
^{\prime}\R_1\F_{1,t}\C_{1}^{\prime}\widetilde{\C}_{1}%
\widetilde{\Lambda }_{C_{1}}^{-1}+\frac{1}{p_{1}^{2}p_{2}T^{2}}\sum_{t=1}^{T}%
\E_{t}^{\prime}\hat{\R}_1\hat{\R}_1^{\prime}%
\E_{t}\widetilde{\C}_{1}\widetilde{\Lambda }_{C_{1}}^{-1} \\
&+\frac{1}{p_{1}^{2}p_{2}T^{2}}\sum_{t=1}^{T}\C_{1}\F_{1,t}^{\prime}%
\R_{1}^{\prime}\hat{\R}_1\hat{\R}_1^{\prime}%
\E_{t}\widetilde{\C}_{1}\widetilde{\Lambda }_{C_{1}}^{-1}+\left( 
\frac{1}{p_{1}^{2}p_{2}T^{2}}\sum_{t=1}^{T}\C_{1}\F_{1,t}^{\prime}\R_{1}^{\prime}\hat{\R}_1\hat{\R}_1^{\prime}\E_{t}%
\widetilde{\C}_{1}\widetilde{\Lambda }_{C_{1}}^{-1}\right) ^{\prime} \\
&+\frac{1}{p_{1}^{2}p_{2}T^{2}}\sum_{t=1}^{T}\left( \R_{0}\F%
_{0,t}\C_{0}^{\prime}-\hat{\R}_{0}\hat{\F}%
_{0,t}\hat{\C}_{0}^{\prime}\right) ^{\prime}\hat{\R%
}\hat{\R}_1^{\prime}\E_{t}\widetilde{\C}_{1}%
\widetilde{\Lambda }_{C_{1}}^{-1} \\
&+\left( \frac{1}{p_{1}^{2}p_{2}T^{2}}\sum_{t=1}^{T}\left( \R_{0}%
\F_{0,t}\C_{0}^{\prime}-\hat{\R}_{0}\hat{%
\F}_{0,t}\hat{\C}_{0}^{\prime}\right) ^{\prime}%
\hat{\R}_1\hat{\R}_1^{\prime}\E_{t}\widetilde{%
\C}_{1}\widetilde{\Lambda }_{C_{1}}^{-1}\right) ^{\prime} \\
&+\frac{1}{p_{1}^{2}p_{2}T^{2}}\sum_{t=1}^{T}\C_{1}\F_{1,t}^{\prime}%
\R_{1}^{\prime}\hat{\R}_1\hat{\R}_1^{\prime}\left( 
\R_{0}\F_{0,t}\C_{0}^{\prime}-\hat{\R}_{0}\hat{\F}_{0,t}\hat{\C}_{0}^{\prime}\right) 
\widetilde{\C}_{1}\widetilde{\Lambda }_{C_{1}}^{-1} \\
&+\left( \frac{1}{p_{1}^{2}p_{2}T^{2}}\sum_{t=1}^{T}\C_{1}\F_{1,t}^{\prime}\R_{1}^{\prime}\hat{\R}_1\hat{\R}_1^{\prime}\left( \R_{0}\F_{0,t}\C_{0}^{\prime}-\hat{%
\R}_{0}\hat{\F}_{0,t}\hat{\C}_{0}^{\prime}\right) \widetilde{\C}_{1}\widetilde{\Lambda }_{C_{1}}^{-1}\right) ^{\prime} \\
&+\frac{1}{p_{1}^{2}p_{2}T^{2}}\sum_{t=1}^{T}\left( \R_{0}\F%
_{0,t}\C_{0}^{\prime}-\hat{\R}_{0}\hat{\F}%
_{0,t}\hat{\C}_{0}^{\prime}\right) ^{\prime}\hat{\R}_{1}\hat{\R}_1^{\prime}\left( \R_{0}\F_{0,t}\C_{0}^{\prime}-\hat{\R}_{0}\hat{\F}_{0,t}\hat{%
\C}_{0}^{\prime}\right) \widetilde{\C}_{1}\widetilde{\Lambda }%
_{C_{1}}^{-1} \\
&=I+II+III+III^{\prime}+IV+IV^{\prime}+V+V^{\prime}+VI.
\end{align*}%
Upon defining%
\begin{equation*}
\widetilde{\bH}_{C_{1}}=\left( \frac{1}{T^{2}}\sum_{t=1}^{T}\F_{1,t}^{\prime}\left( \frac{\R_{1}^{\prime}\hat{\R}_1}{p_{1}}%
\right) \left( \frac{\hat{\R}_1^{\prime}\R_1}{p_{1}}\right) 
\F_{1,t}\left( \frac{\C_{1}^{\prime}\widetilde{\C}_{1}}{p_{2}}%
\right) \right) \widetilde{\Lambda }_{C_{1}}^{-1},
\end{equation*}%
the proof of the theorem is the same as above, and we therefore omit it.
\end{proof}

\begin{proof}[Proof of Theorem \protect\ref{f-tilde}]
Recall that%
\begin{equation*}
\widetilde{\F}_{1,t}=\frac{1}{p_{1}p_{2}}\widetilde{\R}_1%
^{\prime}\proj{\X}_t\widetilde{\C}_{1}.
\end{equation*}%
We will use the decompositions%
\begin{align}
\R_{1} &=\widetilde{\R}_1\left( \widetilde{\bH}_{R_{1}}\right)
^{-1}-\left( \widetilde{\R}_1-\R_1\widetilde{\bH}%
_{R_{1}}\right) \left( \widetilde{\bH}_{R_{1}}\right) ^{-1},
\label{r-r-tilde} \\
\C_{1} &=\widetilde{\C}_{1}\left( \widetilde{\bH}_{C_{1}}\right)
^{-1}-\left( \widetilde{\C}_{1}-\C_{1}\widetilde{\bH}%
_{C_{1}}\right) \left( \widetilde{\bH}_{C_{1}}\right) ^{-1}.
\label{c-c-tilde}
\end{align}%
It holds that%
\begin{align*}
\widetilde{\F}_{1,t} &=\left( \widetilde{\bH}_{R_{1}}\right) ^{-1}%
\F_{1,t}\left( \widetilde{\bH}_{C_{1}}^{\prime}\right) ^{-1}-\frac{%
1}{p_{1}p_{2}}\widetilde{\R}_1^{\prime}\left( \widetilde{\R}_1-%
\R_1\widetilde{\bH}_{R_{1}}\right) \left( \widetilde{\bH}%
_{R_{1}}\right) ^{-1}\F_{1,t}\C_{1}^{\prime}\widetilde{\C}_{1} \\
&-\frac{1}{p_{1}p_{2}}\widetilde{\R}_1^{\prime}\R_1\F_{1,t}%
\C_{1}^{\prime}\left( \widetilde{\C}_{1}-\C_{1}\widetilde{%
\bH}_{C_{1}}\right) \left( \widetilde{\bH}_{C_{1}}\right) ^{-1} \\
&+\frac{1}{p_{1}p_{2}}\widetilde{\R}_1^{\prime}\left( \widetilde{\R}_{1}-\R_1\widetilde{\bH}_{R_{1}}\right) \left( \widetilde{\bH}_{R_{1}}\right) ^{-1}\F_{1,t}\C_{1}^{\prime}\left( 
\widetilde{\C}_{1}-\C_{1}\widetilde{\bH}_{C_{1}}\right) \left( 
\widetilde{\bH}_{C_{1}}\right) ^{-1} \\
&+\frac{1}{p_{1}p_{2}}\widetilde{\R}_1^{\prime}\E_{t}%
\widetilde{\C}_{1}+\frac{1}{p_{1}p_{2}}\widetilde{\R}_1^{\prime}\left( \R_{0}\F_{0,t}\C_{0}^{\prime}-\hat{%
\R}_{0}\hat{\F}_{0,t}\hat{\C}_{0}^{\prime}\right) \widetilde{\C}_{1} \\
&=\left( \widetilde{\bH}_{R_{1}}\right) ^{-1}\F_{1,t}\left( 
\widetilde{\bH}_{C_{1}}^{\prime}\right) ^{-1}+I+II+III+IV+V.
\end{align*}%
It holds that%
\begin{align*}
\left\Vert I\right\Vert _{F} &\leq\frac{1}{p_{1}p_{2}}\left\Vert 
\widetilde{\R}_1\right\Vert _{F}\left\Vert \widetilde{\R}_1-%
\R_1\widetilde{\bH}_{R_{1}}\right\Vert _{F}\left\Vert \left( 
\widetilde{\bH}_{R_{1}}\right) ^{-1}\right\Vert _{F}\left\Vert \C%
\right\Vert _{F}\left\Vert \widetilde{\C}_{1}\right\Vert _{F}\left\Vert 
\F_{1,t}\right\Vert _{F} \\
&=O_{P}(1) \frac{1}{p_{1}p_{2}}p_{1}^{1/2}\left( \frac{%
p_{1}^{1/2}}{p_{2}^{1/2}T}+\frac{p_{1}^{1/2}}{T^{2}}+\frac{1}{p_{1}^{1/2}T}+%
\frac{1}{T^{3/2}}\right) p_{2}^{1/2}p_{2}^{1/2}T^{1/2} \\
&=O_{P}(1) \left( \frac{1}{p_{2}^{1/2}T^{1/2}}+\frac{1}{T^{3/2}}%
+\frac{1}{p_{1}T^{1/2}}+\frac{1}{p_{1}^{1/2}T}\right) ;
\end{align*}%
by the same token, it can be shown that%
\begin{equation*}
\left\Vert II\right\Vert _{F}=O_{P}(1) \left( \frac{1}{%
p_{1}^{1/2}T^{1/2}}+\frac{1}{T^{3/2}}+\frac{1}{p_{2}T^{1/2}}+\frac{1}{%
p_{2}^{1/2}T}\right) ,
\end{equation*}%
and $\left\Vert III\right\Vert _{F}$ is clearly dominated by $\left\Vert
I\right\Vert _{F}$ and $\left\Vert II\right\Vert _{F}$. Using the convention 
$h_{R_{1}}=h_{C_{1}}=h_{R_{0}}=h_{C_{0}}=1$%
\begin{align*}
IV &=\frac{1}{p_{1}p_{2}}\widetilde{\bH}_{R_{1}}^{\prime}\R_{1}^{\prime}\E_{t}\C_{1}\widetilde{\bH}_{C_{1}}+\frac{1}{%
p_{1}p_{2}}\left( \widetilde{\R}_1-\R_1\widetilde{\bH}_{R_{1}}\right) ^{\prime}\E_{t}\C_{1}\widetilde{\bH}_{C_{1}} \\
&+\frac{1}{p_{1}p_{2}}\widetilde{\bH}_{R_{1}}^{\prime}\R_{1}^{\prime}\E_{t}\left( \widetilde{\C}_{1}-\C_{1}\widetilde{%
\bH}_{C_{1}}\right) +\frac{1}{p_{1}p_{2}}\left( \widetilde{\R}_{1}-\R_{1}\widetilde{\bH}_{R_{1}}\right) ^{\prime}\E_{t}\left( 
\widetilde{\C}_{1}-\C_{1}\widetilde{\bH}_{C_{1}}\right) \\
&=IV_{a}+IV_{b}+IV_{c}+IV_{d}.
\end{align*}%
In the above we showed that $\left\Vert \R_{1}^{\prime}\E_{t}%
\C_{1}\right\Vert _{F}=O_{p}\left( p_{1}^{1/2}p_{2}^{1/2}\right) $ and $%
\left\Vert \E_{t}\C_{1}\right\Vert _{F}=O_{p}\left(
p_{1}^{1/2}p_{2}^{1/2}\right) $; hence it follows that%
\begin{equation*}
\left\Vert IV_{a}\right\Vert _{F}=O_{P}\left( \frac{1}{\left(
p_{1}p_{2}\right) ^{1/2}}\right) ,
\end{equation*}%
and 
\begin{align*}
\left\Vert IV_{b}\right\Vert _{F} &\leq\frac{1}{p_{1}p_{2}}\left\Vert 
\widetilde{\R}_1-\R_1\widetilde{\bH}_{R_{1}}\right\Vert
_{F}\left\Vert \E_{t}\C_{1}\right\Vert _{F}\left\Vert \widetilde{%
\bH}_{C_{1}}\right\Vert _{F} \\
&=O_{P}(1) \frac{1}{p_{1}p_{2}}p_{1}^{1/2}\left( \frac{%
p_{1}^{1/2}}{p_{2}^{1/2}T}+\frac{p_{1}^{1/2}}{T^{2}}+\frac{1}{p_{1}^{1/2}T}+%
\frac{1}{T^{3/2}}\right) \left( p_{1}p_{2}\right) ^{1/2} \\
&=O_{P}(1) \left( \frac{1}{p_{2}^{1/2}T}+\frac{1}{T^{2}}+\frac{1%
}{p_{1}T}+\frac{1}{p_{1}^{1/2}T^{3/2}}\right) ,
\end{align*}%
and likewise%
\begin{equation*}
\left\Vert IV_{c}\right\Vert _{F}=O_{P}(1) \left( \frac{1}{%
p_{1}^{1/2}T}+\frac{1}{T^{2}}+\frac{1}{p_{2}T}+\frac{1}{p_{1}^{1/2}T^{3/2}}%
\right) .
\end{equation*}%
Finally, it is not hard to see that $IV_{d}$ is dominated by $IV_{a}-IV_{c}$%
. Finally, after some algebra%
\begin{align*}
\left\Vert V\right\Vert _{F} &\leq\frac{1}{p_{1}p_{2}}\left\Vert 
\widetilde{\R}_1\right\Vert _{F}\left\Vert \R_{0}\F%
_{0,t}\C_{0}^{\prime}-\hat{\R}_{0}\hat{\F}%
_{0,t}\hat{\C}_{0}^{\prime}\right\Vert _{F}\left\Vert 
\widetilde{\C}_{1}\right\Vert _{F} \\
&=O_{P}(1) \left( \frac{1}{\left( p_{1}p_{2}\right) ^{1/2}}+%
\frac{1}{p_{1\wedge 2}^{1/2}T^{1/2}}\right) .
\end{align*}%
The final result follows from putting all together. The proof of (\ref%
{f-tilde-l2}) is similar to that of (\ref{f1-hat-l2}), and we omit it to
save space.
\end{proof}

\begin{proof}[Proof of Lemma \protect\ref{neg-rtilde}]
Consider (\ref{r-tilde-neg}). The result follows immediately upon
considering the term%
\begin{equation*}
\frac{1}{p_{1}p_{2}^{2}T}\sum_{t=1}^{T}\E_{t}\widetilde{\C}_{1,\perp}^{s}\left( \widetilde{\C}_{\perp}^{s}\right) ^{\prime}%
\E_{t}^{\prime}\widetilde{\R}_{0}\widetilde{\Lambda }%
_{R_{0}}^{-1},
\end{equation*}%
in the expansion of $\widetilde{\R}_{0}-\R_{0}\widetilde{%
\bH}_{R_0}$, which can be derived along the same lines as (\ref%
{r1-hat-dec}). In particular, the term%
\begin{equation*}
\frac{1}{p_{1}p_{2}^{2}T}\sum_{t=1}^{T}\E_{t}\C_{1,\perp}^{s}\left( \C_{1,\perp}^{s}\right) ^{\prime}\E_{t}^{\prime}%
\widetilde{\R}_{0}\widetilde{\Lambda }_{R_{0}}^{-1},
\end{equation*}%
is of order $O_{P}\left( \frac{p_{1}^{1/2}}{p_{1}p_{2}}\right) +O_{P}\left( 
\frac{p_{1}^{1/2}}{p_{2}^{1/2}T^{1/2}}\right) $; again, this can be shown
following exactly the proof of (\ref{3f1}). Since this is the dominant rate
in $\left\Vert \hat{\R}_{0}-\R_{0}\hat{\bH}%
_{R_0}\right\Vert _{F}$, the desired result follows. The same arguments
yield also (\ref{c-tilde-neg}).
\end{proof}

\begin{proof}[Proof of Lemma \protect\ref{bai03}]
The method of proof is the same for all theorems and lemmas, and it is based
on an argument in \citet{bai03} - see in particular Footnote 5 on p. 143. Consider
the random variable $\mathcal{A}$, with $\mathcal{A}=1$ if the relevant
result holds, and $\mathcal{A}=0$ otherwise; and the random variable $%
\mathcal{B}$, with $\mathcal{B}=1$ if $\widetilde{h}_{R_{1}}=h_{R_{1}}$, and $%
\widetilde{h}_{C_{1}}=h_{C_{1}}$, and $\widetilde{h}_{R_{1}}=h_{R_{0}}$, and $\widetilde{h}%
_{C_{1}}=h_{C_{0}}$, and $\mathcal{B}=0$ otherwise. Then we have%
\begin{equation*}
P\left( \left\{ \mathcal{A}=1\right\} \right) =P\left( \left\{ \mathcal{A}%
=1\right\} \cap \left\{ \mathcal{B}=1\right\} \right) +P\left( \left\{ 
\mathcal{A}=1\right\} \cap \left\{ \mathcal{B}=0\right\} \right) .
\end{equation*}%
Note also that%
\begin{equation*}
P\left( \left\{ \mathcal{A}=1\right\} \cap \left\{ \mathcal{B}=0\right\}
\right) \leq P\left( \left\{ \mathcal{B}=0\right\} \right) =o(1)
,
\end{equation*}%
under the assumption that $\widetilde{h}_{R_{1}}=h_{R_{1}}+o_{P}(1) $, $%
\widetilde{h}_{C_{1}}=h_{C_{1}}+o_{P}(1) $, $\widetilde{h}%
_{R_{1}}=h_{R_{0}}+o_{P}(1) $, $\widetilde{h}_{C_{1}}=h_{C_{0}}+o_{P}\left(
1\right) $, which also entails that $P\left( \left\{ \mathcal{B}=1\right\}
\right) \rightarrow 1$. Hence we have%
\begin{align*}
P\left( \left\{ \mathcal{A}=1\right\} \right) &=P\left( \left\{ \mathcal{A}%
=1\right\} \cap \left\{ \mathcal{B}=1\right\} \right) +o(1) \\
&=P\left( \left\{ \mathcal{A}=1\right\} |\left\{ \mathcal{B}=1\right\}
\right) P\left( \left\{ \mathcal{B}=1\right\} \right) +o(1) \\
&=P\left( \left\{ \mathcal{A}=1\right\} |\left\{ \mathcal{B}=1\right\}
\right) +o(1) ,
\end{align*}%
which proves the desired result.
\end{proof}

\begin{proof}[Proof of Theorem \protect\ref{er}]
We only show that $\widetilde{h}_{R_{1}}=h_{R_{1}}+o_{P}(1) $; the other
results follow from the same arguments. Recall that, by Lemma \ref%
{spec-m-r-tilde}, 
\begin{equation}
\lambda _{j}\left( \proj{\M}_{R_{1}}\right) =c_{0}+o_{P}(1) ,
\label{khat1}
\end{equation}%
for all $j\leq h_{R_{1}}$, where $c_{0}>0$; and, also%
\begin{equation}
\lambda _{j}\left( \proj{\M}_{R_{1}}\right) =O_{P}\left( \frac{1}{%
p_{1\wedge 2}^{1/2}T^{3/2}}\right) +O_{P}\left( \frac{1}{p_{2}T}\right)
+O_{P}\left( \frac{1}{T^{2}}\right) +O_{P}\left( \frac{1}{%
p_{1}^{1/2}p_{2}^{1/2}T}\right) ,  \label{khat2}
\end{equation}%
for all $j>h_{R_{1}}$. Hence, by elementary arguments, (\ref{khat1}) entails that%
\begin{equation*}
\max_{1\leq j\leq h_{R_{1}}-1}\frac{\lambda _{j}\left( \M%
_{X}^{R_{1}\diamond}\right) }{\lambda _{j+1}\left( \proj{\M}_{R_{1}}\right) +%
\widetilde{c}_{R_{1}}\delta _{R,p_{1},p_{2},T}^{k}}\leq \max_{1\leq j\leq
h_{R_{1}}-1}\frac{\lambda _{j}\left( \proj{\M}_{R_{1}}\right) }{\lambda
_{j+1}\left( \proj{\M}_{R_{1}}\right) }=O_{P}(1) .
\end{equation*}%
Similarly, using (\ref{khat2}) and the definition of $\widetilde{\delta }%
_{R,p_{1},p_{2},T}^{k}$%
\begin{equation*}
\max_{h_{R_{1}}+1\leq j\leq h_{\max }}\frac{\lambda _{j}\left( \M%
_{X}^{R_{1}\diamond}\right) }{\lambda _{j+1}\left( \proj{\M}_{R_{1}}\right) +%
\widetilde{c}_{R_{1}}\delta _{R,p_{1},p_{2},T}^{k}}\leq \max_{1\leq j\leq
h_{R_{1}}-1}\frac{\lambda _{j}\left( \proj{\M}_{R_{1}}\right) }{\widetilde{c}%
_{R_{1}}\delta _{R,p_{1},p_{2},T}^{k}}=O_{P}(1) .
\end{equation*}%
Finally, combining (\ref{khat1}) and (\ref{khat2}), as $\min \left\{
p_{1},p_{2},T\right\} \rightarrow \infty $ we have that, for some $%
0<c_{0}<\infty $%
\begin{equation*}
P\left( \frac{\lambda _{h_{R_{1}}}\left( \proj{\M}_{R_{1}}\right) }{\lambda
_{h_{R_{1}}+1}\left( \proj{\M}_{R_{1}}\right) +\widetilde{c}_{R_{1}}\delta
_{R,p_{1},p_{2},T}^{k}}\geq c_{0}\left( \delta _{R,p_{1},p_{2},T}^{k}\right)
^{-1}\lambda _{h_{R_{1}}}\left( \proj{\M}_{R_{1}}\right) \right) =1.
\end{equation*}%
The desired result follows from noting that, by (\ref{khat1})%
\begin{equation*}
\lim_{\min \left\{ p_{1},p_{2},T\right\} \rightarrow \infty }\left( \delta
_{R,p_{1},p_{2},T}^{k}\right) ^{-1}\lambda _{h_{R_{1}}}\left( \M%
_{X}^{R_{1}\diamond}\right) =\infty \text{. }
\end{equation*}%
When using the mock eigenvalue, note that if $h_{R_{1}}>0$%
\begin{equation*}
\frac{\lambda _{0}\left( \proj{\M}_{R_{1}}\right) }{\lambda _{1}\left( 
\proj{\M}_{R_{1}}\right) +\widetilde{c}_{R_{1}}\delta _{R,p_{1},p_{2},T}^{k}}%
\leq \frac{\lambda _{0}\left( \proj{\M}_{R_{1}}\right) }{\lambda_{1}\left( \proj{\M}_{R_{1}}\right) }=o_{P}(1) ,
\end{equation*}%
by the definition of $\lambda _{0}\left( \proj{\M}_{R_{1}}\right) $;
conversely, if $h_{R_{1}}=0$, then by the same token as above%
\begin{equation*}
P\left( \frac{\lambda _{0}\left( \proj{\M}_{R_{1}}\right) }{\lambda
_{1}\left( \proj{\M}_{R_{1}}\right) +\widetilde{c}_{R_{1}}\delta
_{R,p_{1},p_{2},T}^{k}}\geq c_{0}\left( \delta _{R,p_{1},p_{2},T}^{k}\right)
^{-1}\lambda _{0}\left( \proj{\M}_{R_{1}}\right) \right) =1.
\end{equation*}%
for some $0<c_{0}<\infty $, and, by the construction of $\lambda
_{h_{R_{1}}}\left( \proj{\M}_{R_{1}}\right) $ 
\begin{equation*}
\lim_{\min \left\{ p_{1},p_{2},T\right\} \rightarrow \infty }\left( \delta
_{R,p_{1},p_{2},T}^{k}\right) ^{-1}\lambda _{0}\left( \M%
_{X}^{R_{1}\diamond}\right) =\infty \text{,}
\end{equation*}%
whence the desired result again follows.
\end{proof}

\newpage
\section{Additional Monte Carlo results}\label{app:MC}
In this section we report extended simulation studies that could not fit in the main article due to space constraints. 

\begin{figure}
    \centering
    \includegraphics[width=0.45\linewidth]{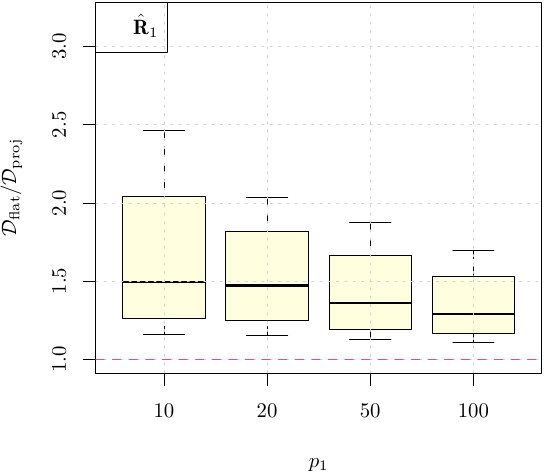}
    \includegraphics[width=0.45\linewidth]{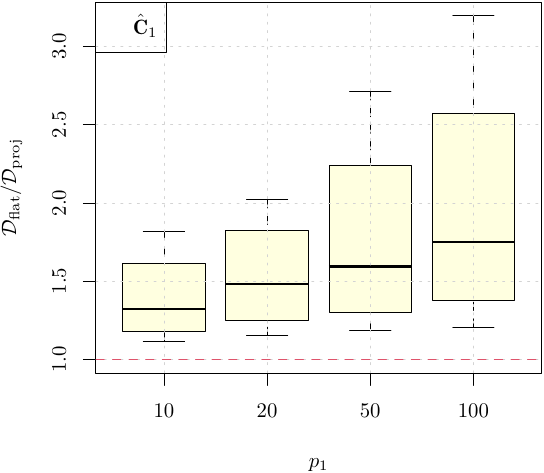}
    \caption{Boxplots of the ratio $\D_{\text{flat}}/\D_{\text{proj}}$ between the initial flattened  and the refined projected estimators for $\R_1$ (left) and $\C_1$ (right) against $p_1$.}
    \label{fig:mc2box1p}
\end{figure}

\begin{figure}
    \centering
    \includegraphics[width=0.45\linewidth]{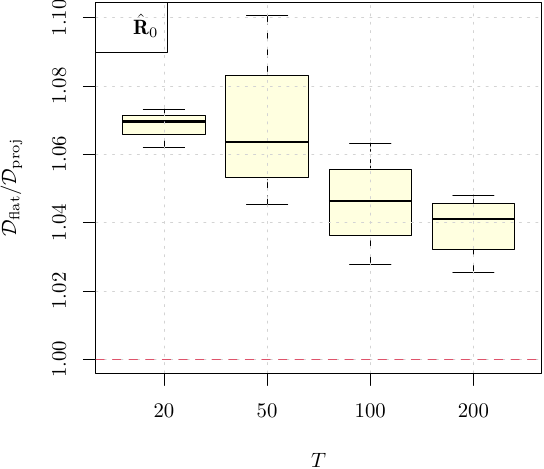}
    \includegraphics[width=0.45\linewidth]{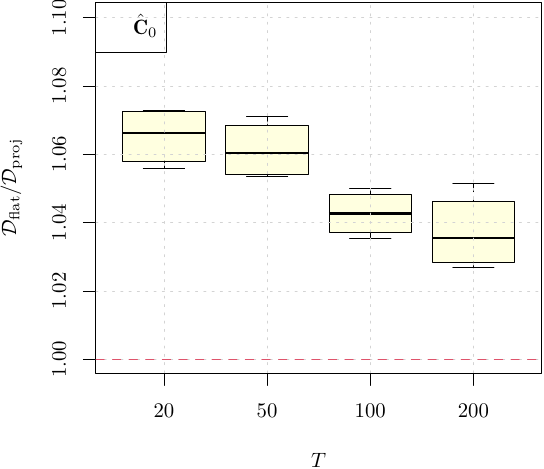}
    \caption{Boxplots of the ratio $\D_{\text{flat}}/\D_{\text{proj}}$ between the initial flattened  and the refined projected estimators for $\R_0$ (left) and $\C_0$ (right) against $T$.}
    \label{fig:mc2box0T}
\end{figure}

\begin{figure}
    \centering
    \includegraphics[width=0.45\linewidth]{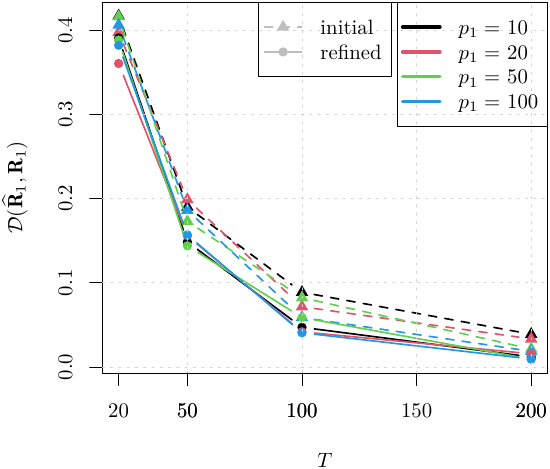}
    \includegraphics[width=0.45\linewidth]{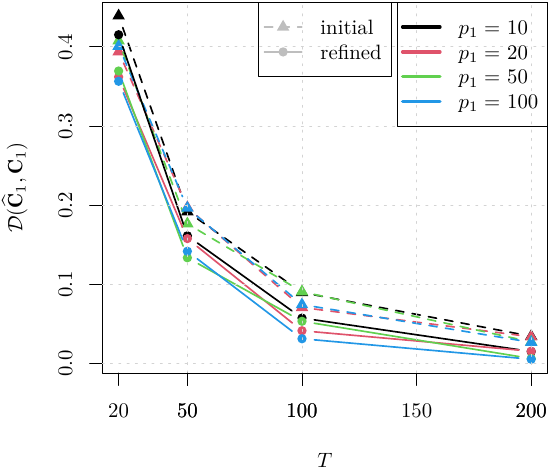}
    \includegraphics[width=0.45\linewidth]{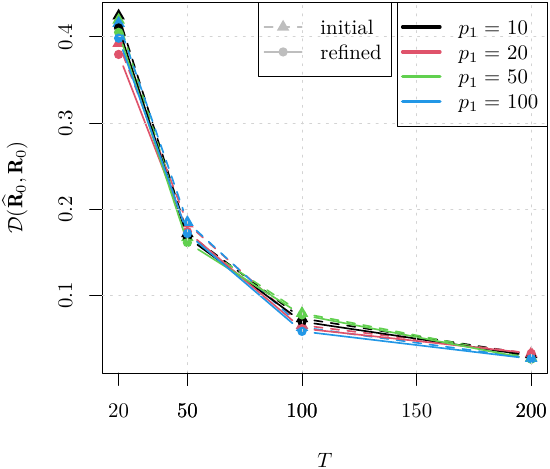}
    \includegraphics[width=0.45\linewidth]{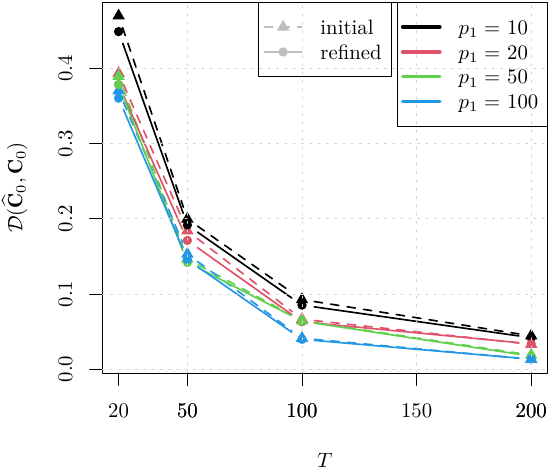}
    \caption{Case 1.2: estimation of $\R_1$, $\C_1$, $\R_0$, $\C_0$  for varying series length $T$ and row dimension $p_1$. Also, $p_2=20$. Triangles with dashed lines indicate the initial ``flattened'' estimator, circles with full lines indicate the refined projected estimator.}\label{fig:1.2}
\end{figure}

\begin{figure}
    \centering
    \includegraphics[width=0.45\linewidth]{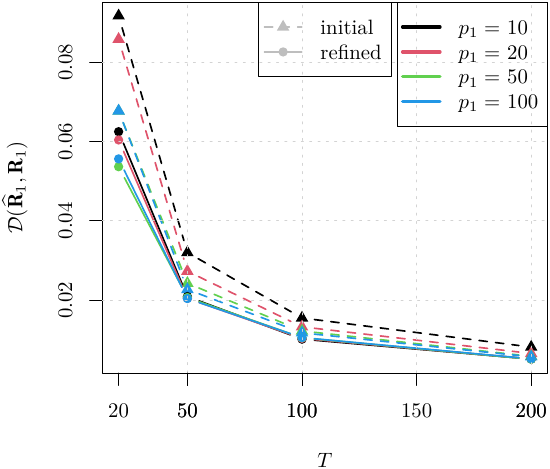}
    \includegraphics[width=0.45\linewidth]{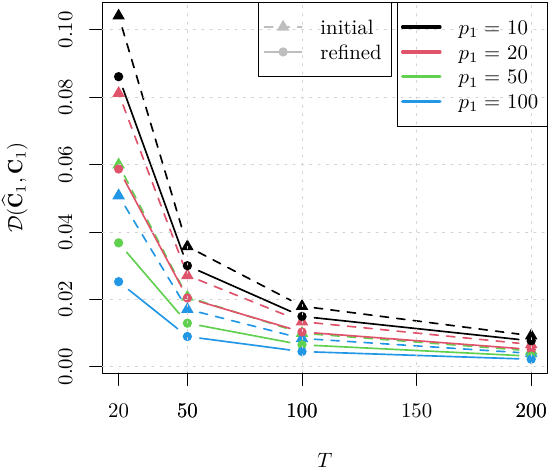}
    \includegraphics[width=0.45\linewidth]{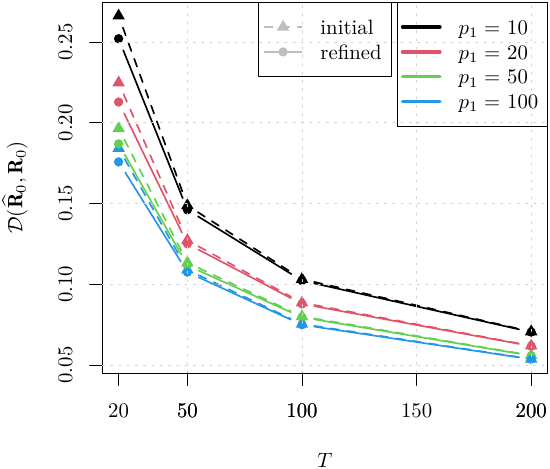}
    \includegraphics[width=0.45\linewidth]{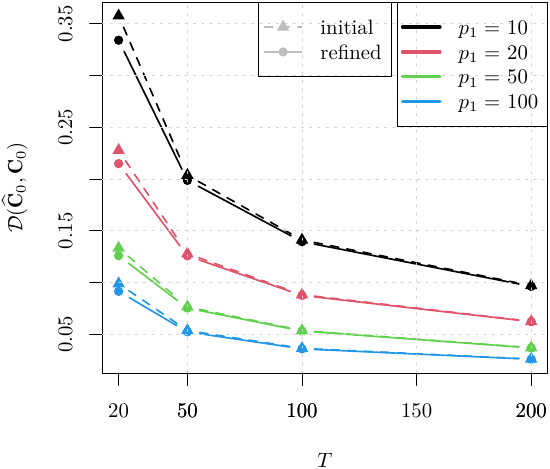}
    \caption{Case 2.1: estimation of $\R_1$, $\C_1$, $\R_0$, $\C_0$  for varying series length $T$ and row dimension $p_1$. Also, $p_2=20$. Triangles with dashed lines indicate the initial ``flattened'' estimator, circles with full lines indicate the refined projected estimator.}\label{fig:2.1}
\end{figure}

\begin{figure}
    \centering
    \includegraphics[width=0.45\linewidth]{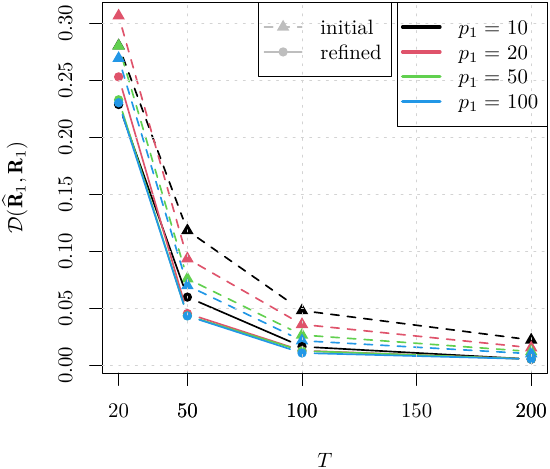}
    \includegraphics[width=0.45\linewidth]{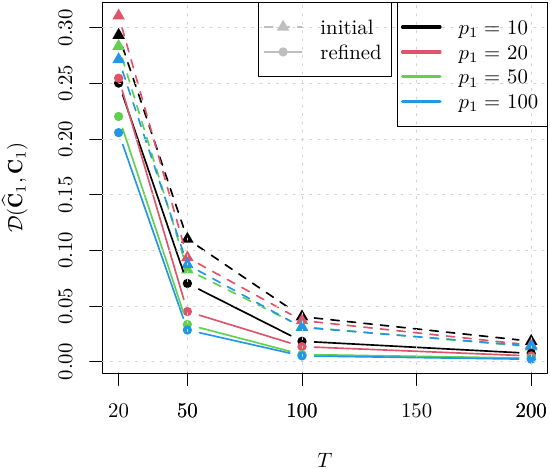}
    \includegraphics[width=0.45\linewidth]{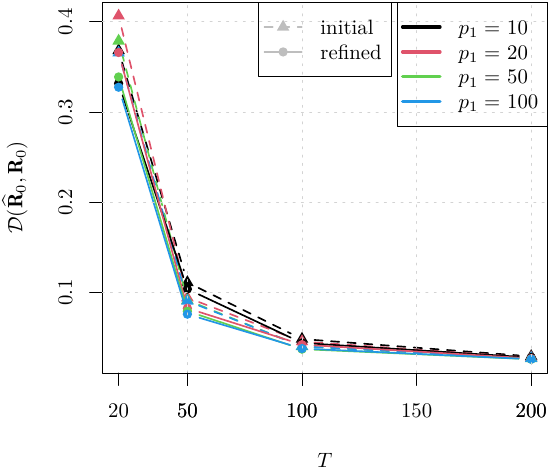}
    \includegraphics[width=0.45\linewidth]{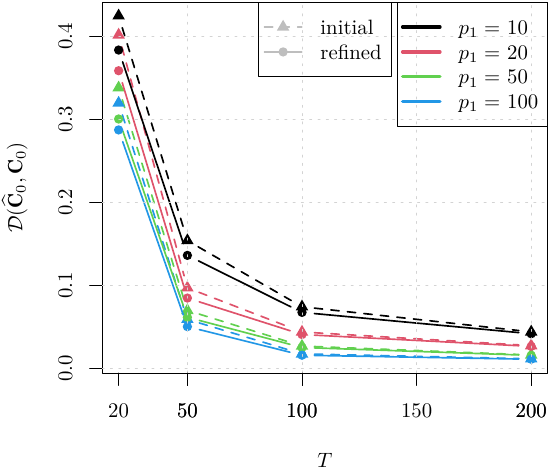}
    \caption{Case 2.2: estimation of $\R_1$, $\C_1$, $\R_0$, $\C_0$  for varying series length $T$ and row dimension $p_1$. Also, $p_2=20$. Triangles with dashed lines indicate the initial ``flattened'' estimator, circles with full lines indicate the refined projected estimator.}\label{fig:2.2}
\end{figure}

\begin{figure}
    \centering
    \includegraphics[width=0.45\linewidth]{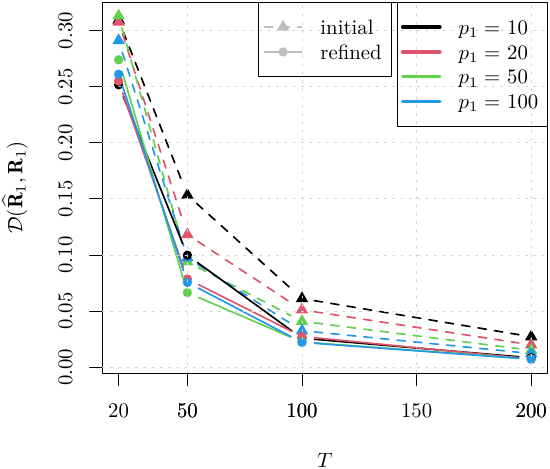}
    \includegraphics[width=0.45\linewidth]{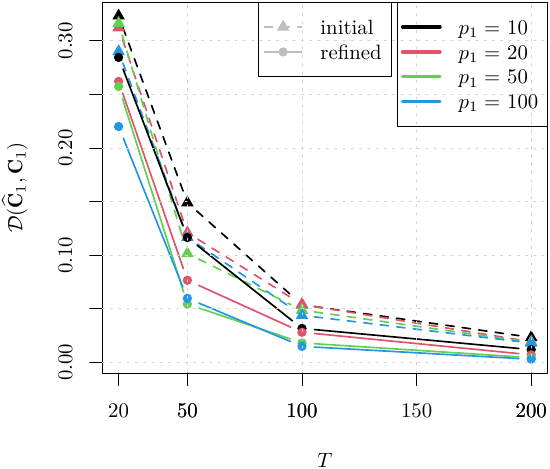}
    \includegraphics[width=0.45\linewidth]{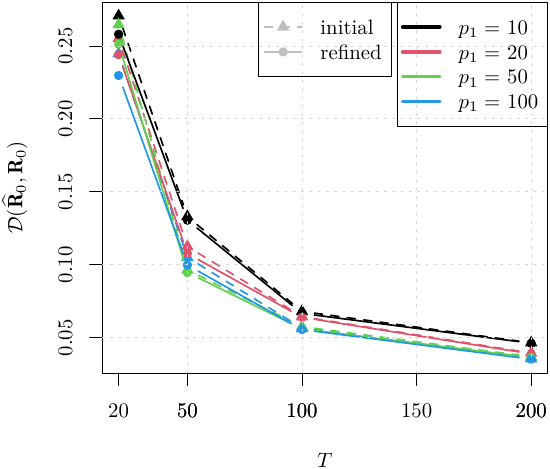}
    \includegraphics[width=0.45\linewidth]{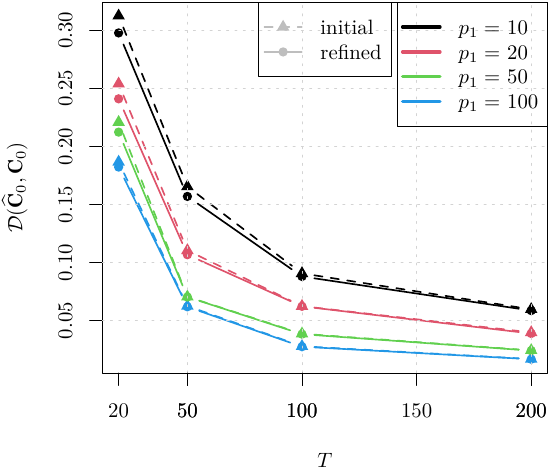}
    \caption{Case 3.1: estimation of $\R_1$, $\C_1$, $\R_0$, $\C_0$  for varying series length $T$ and row dimension $p_1$. Also, $p_2=20$. Triangles with dashed lines indicate the initial ``flattened'' estimator, circles with full lines indicate the refined projected estimator.}\label{fig:3.1}
\end{figure}

\begin{figure}
    \centering
    \includegraphics[width=0.45\linewidth]{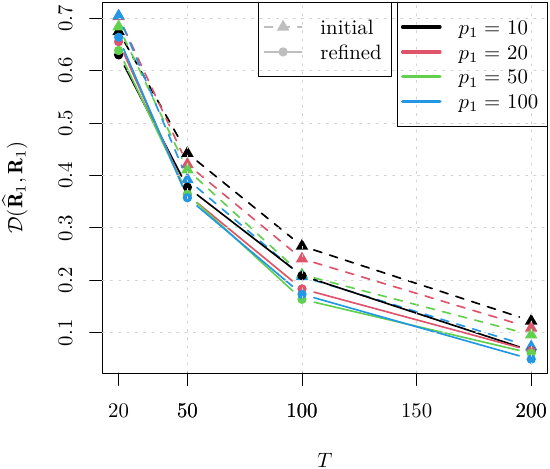}
    \includegraphics[width=0.45\linewidth]{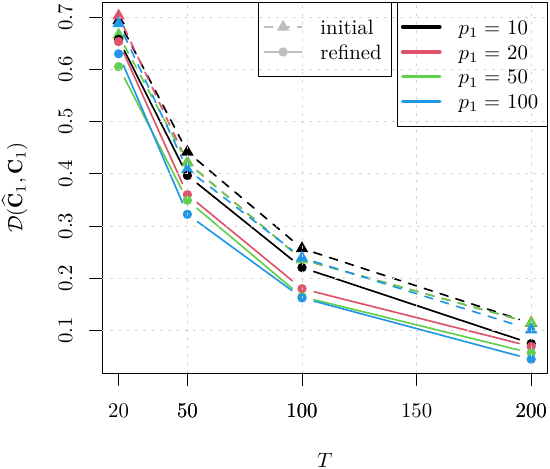}
    \includegraphics[width=0.45\linewidth]{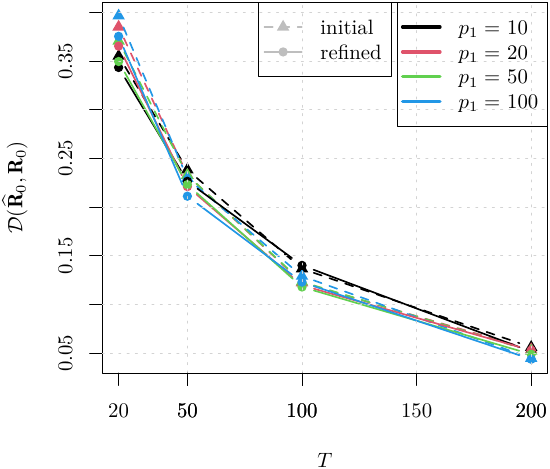}
    \includegraphics[width=0.45\linewidth]{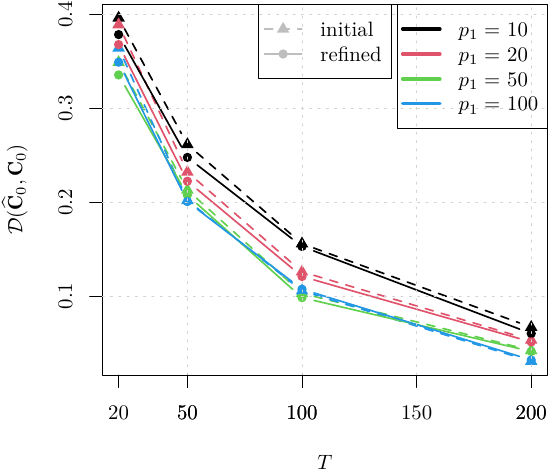}
    \caption{Case 3.2: estimation of $\R_1$, $\C_1$, $\R_0$, $\C_0$  for varying series length $T$ and row dimension $p_1$. Also, $p_2=20$. Triangles with dashed lines indicate the initial ``flattened'' estimator, circles with full lines indicate the refined projected estimator.}\label{fig:3.2}
\end{figure}

\begin{figure}
    \centering
    \includegraphics[width=0.45\linewidth]{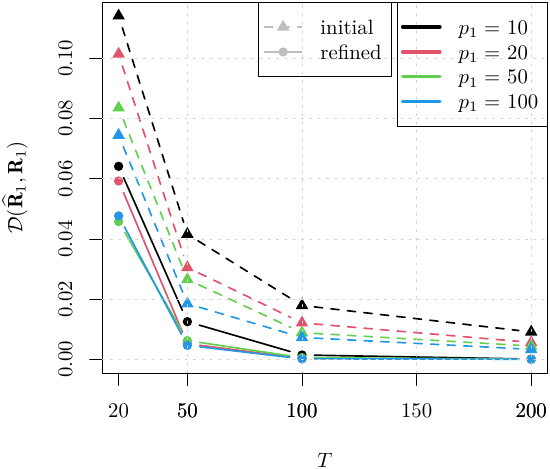}
    \includegraphics[width=0.45\linewidth]{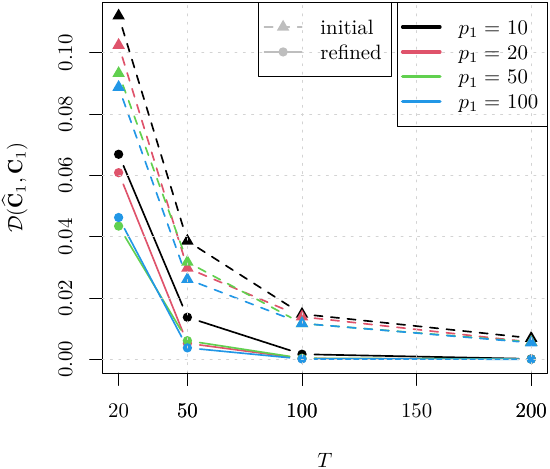}
    \includegraphics[width=0.45\linewidth]{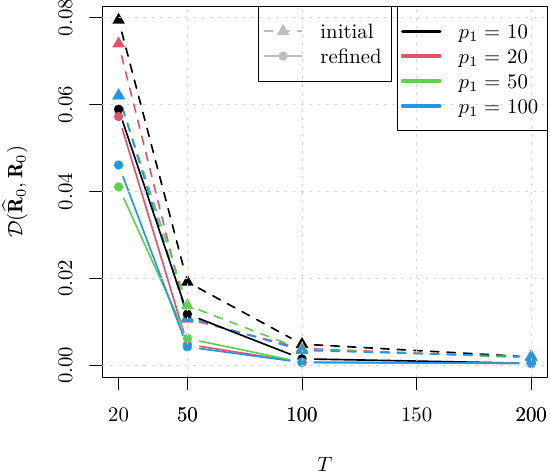}
    \includegraphics[width=0.45\linewidth]{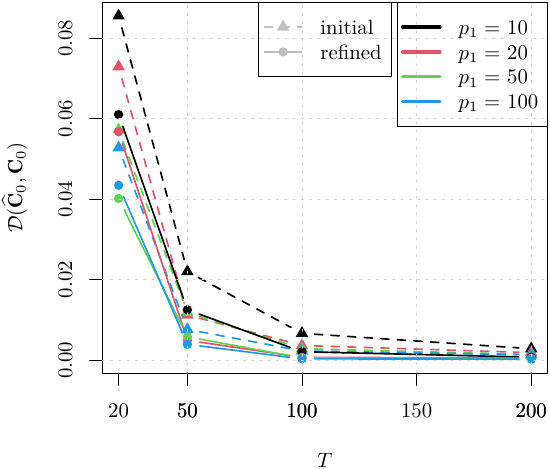}
    \caption{Case 4.1: estimation of $\R_1$, $\C_1$, $\R_0$, $\C_0$  for varying series length $T$ and row dimension $p_1$. Also, $p_2=20$. Triangles with dashed lines indicate the initial ``flattened'' estimator, circles with full lines indicate the refined projected estimator.}\label{fig:4.1}
\end{figure}

\begin{figure}
    \centering
    \includegraphics[width=0.45\linewidth]{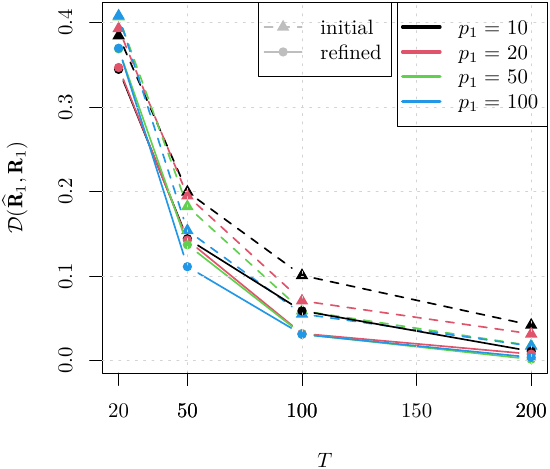}
    \includegraphics[width=0.45\linewidth]{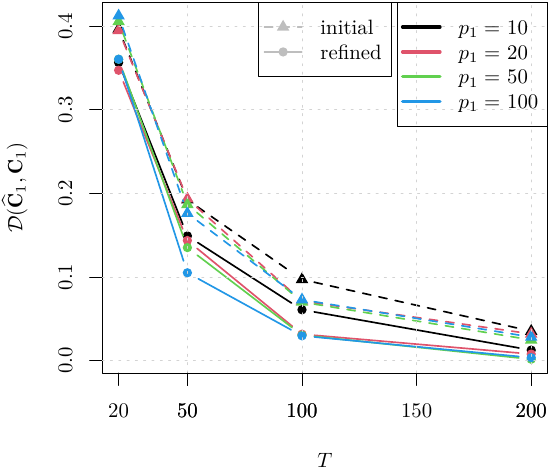}
    \includegraphics[width=0.45\linewidth]{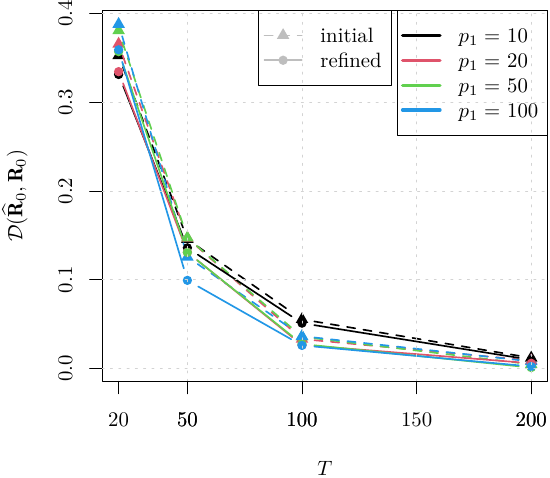}
    \includegraphics[width=0.45\linewidth]{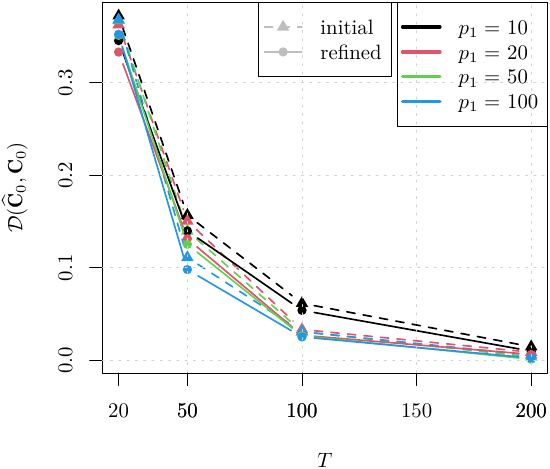}
    \caption{Case 4.2: estimation of $\R_1$, $\C_1$, $\R_0$, $\C_0$  for varying series length $T$ and row dimension $p_1$. Also, $p_2=20$. Triangles with dashed lines indicate the initial ``flattened'' estimator, circles with full lines indicate the refined projected estimator.}\label{fig:4.2}
\end{figure}

\begin{figure}
    \centering
    \includegraphics[width=0.95\linewidth]{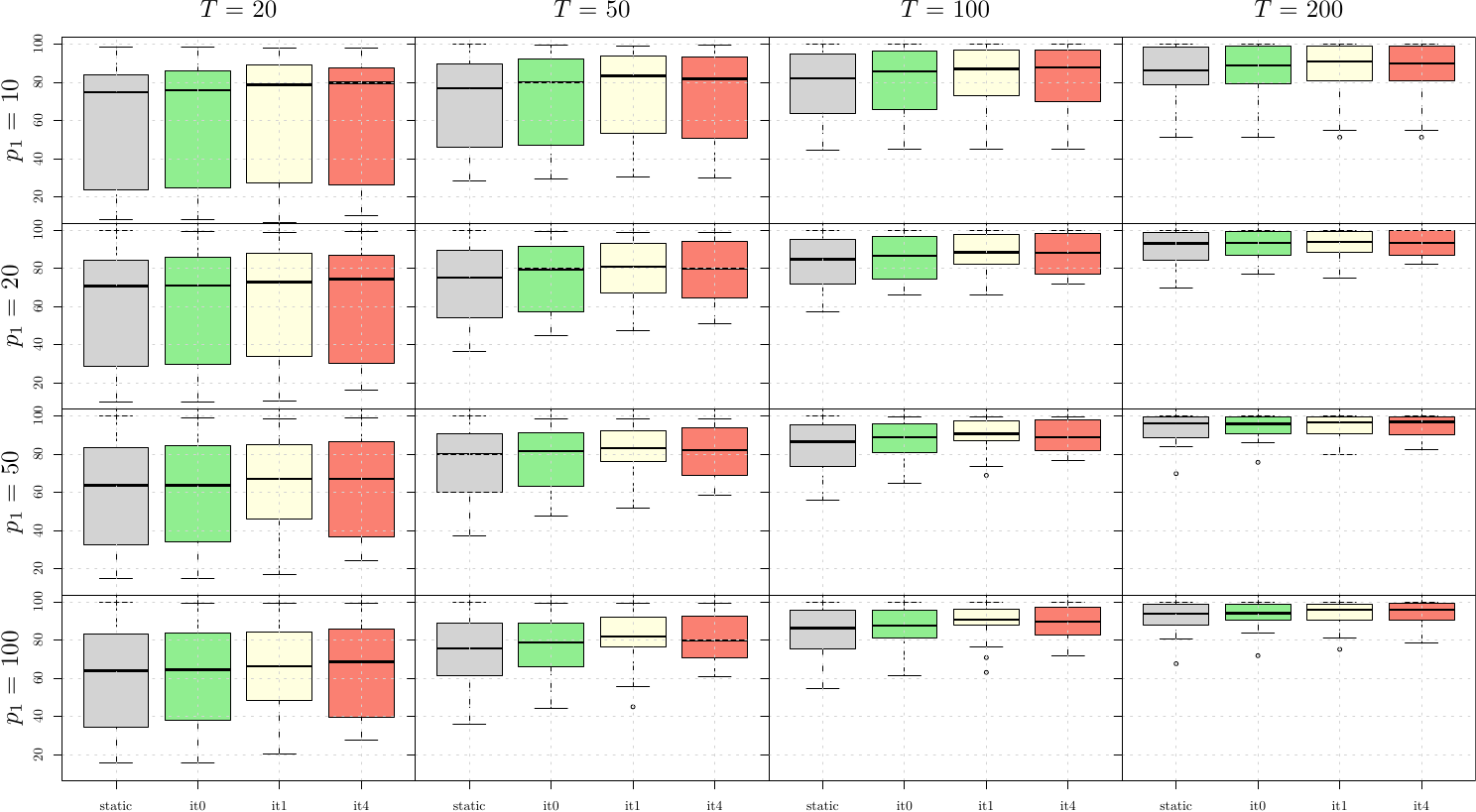}
    \caption{Boxplots of the percentages of correct estimation of the number of factors for the 4 criteria and varying  $p_1$ and sample size $T$. The percentages for the 8 cases and the 4 different parameters are aggregated in a single boxplot.}\label{fig:nf1b}
\end{figure}

\begin{figure}
    \centering
    \includegraphics[width=0.95\linewidth]{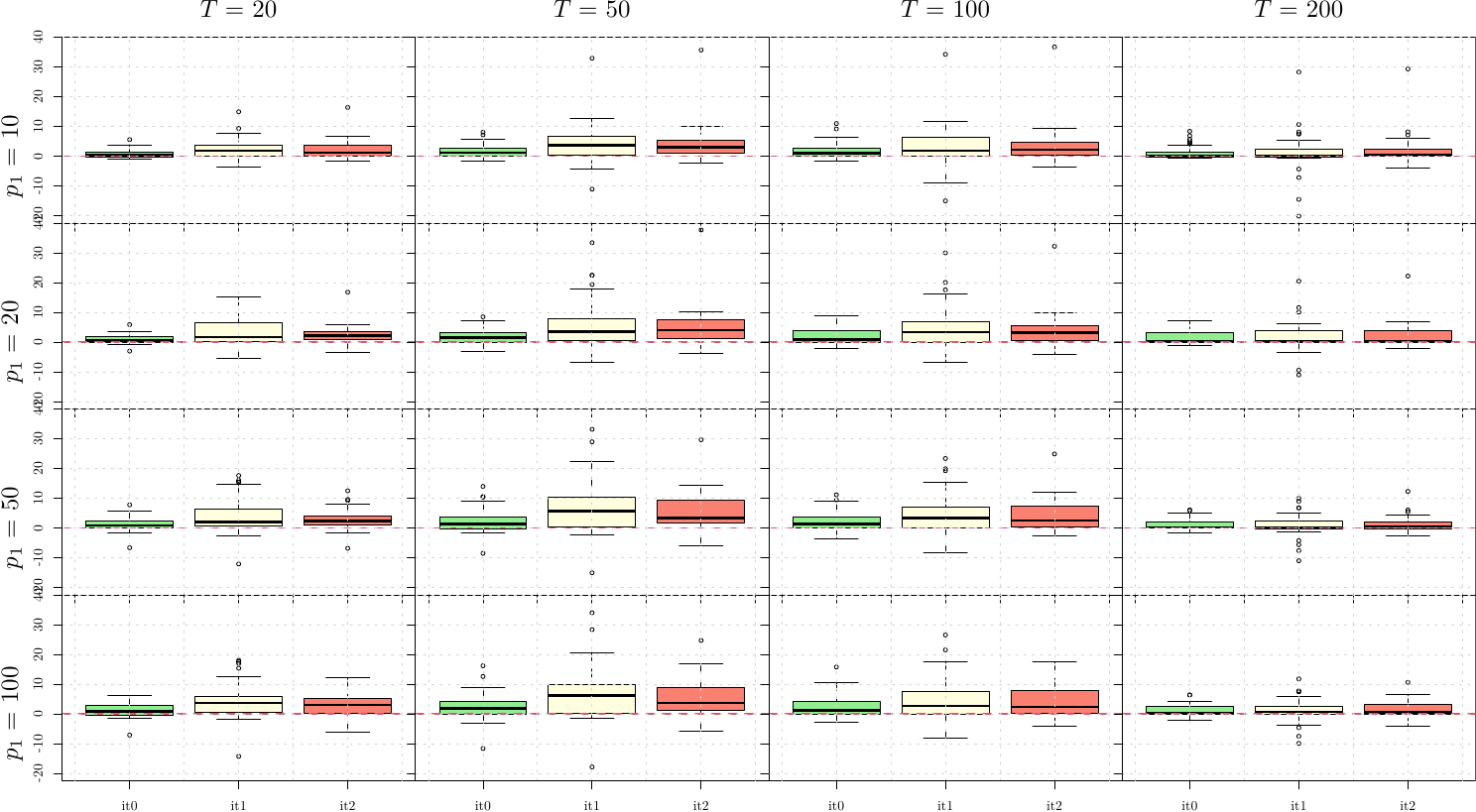}
    \caption{Boxplots of the differences of percentages of correct estimation of the number of factors for the iterative criteria w.r.t. the static criterion. Positive values indicate that the iterative version is superior w.r.t. the static one. The percentages for the 8 cases and the 4 different parameters are aggregated in a single boxplot.}\label{fig:nf2b}
\end{figure}

\end{appendix}

\end{document}